\documentclass {statsoc}

\usepackage{amsmath}

\usepackage{amssymb}
\usepackage{stmaryrd}
\usepackage{graphicx, natbib, color}
 \usepackage{multirow}
 \usepackage{textcomp}
 \usepackage{stmaryrd, array}
\usepackage{geometry}

\newtheorem{thm}{\textbf{Theorem}}
\newtheorem{rem}[thm]{\protect\remarkname}
  \newtheorem{prop}[thm]{\protect\propositionname}

\providecommand{\remarkname}{Remark}
\providecommand{\theoremname}{Theorem}
\providecommand{\propositionname}{Proposition}

\title[Discrete optimal control for estimation of ordinary differential equations]{A regularization method for the parameter estimation problem in ordinary differential equations via discrete optimal control theory.}

\author{Quentin Clairon}
\address{ Newcastle University, School of Mathematics and Statistics,
	 Newcastle Upon Tyne,
	United Kingdom.}
\email{Quentin.Clairon@newcastle.ac.uk}

\begin{document}

\begin{abstract}

We present a parameter estimation method in Ordinary Differential
Equation (ODE) models.
Due to complex relationships between parameters and states the
use of standard techniques such as nonlinear least squares can lead
to the presence of poorly identifiable parameters. Moreover, ODEs
are generally approximations of the true process and the influence
of misspecification on inference is often neglected. Methods based
on control theory have emerged to regularize the ill posed problem
of parameter estimation in this context. However, they are computationally intensive and
rely on a nonparametric state estimator known to be biased in the
sparse sample case. In this paper, we construct criteria based on
discrete control theory which are computationally efficient
and  bypass the presmoothing step of signal estimation while
retaining the benefits of control theory for estimation.  We describe how the estimation problem can
be turned into a control one and present the numerical methods used to solve it. 
We show convergence of our estimator in the parametric
and well-specified case. For small sample sizes, numerical experiments
with models containing poorly identifiable parameters and with various
sources of model misspecification demonstrate the acurracy of our method.
We finally test our approach on a real data example.
\end{abstract}


\section{Introduction}
We are interested by parameter estimation in Ordinary Differential
Equation (ODE) models of the form
\begin{equation}
\left\{ \begin{array}{l}
\dot{x}(t)=f\left(t,x(t),\theta,\vartheta(t)\right)\\
x(0)=x_{0}
\end{array}\right.\label{eq:ODEmodel-gen}
\end{equation}
where the state $x$ is in $\mathbb{R}^{d}$, $f$ is a time-dependent
vector field from $\left[0,\,T\right]\times\mathbb{R}^{d}\times\Theta\times\Theta_{f}$
to $\mathbb{R}^{d}$, $\theta$ is a parameter that belongs to a subset
$\Theta$ of $\mathbb{R}^{p}$ , $\vartheta$ is a functional parameter
from $\left[0,T\right]$ to $\Theta_{f}\subseteq\mathbb{R}^{d_{f}}$
and $x_{0}$ is the initial condition that belongs to a subset $\chi$
of $\mathbb{R}^{d}$. ODEs are much used in practice as they provide
an efficient framework for analyzing and predicting complex systems
(see eg \cite{CellComputl2002,albert1997biochemical,Mirsky2009,Wu2014}).
In particular, there has recently been focus on joint use of ODE models
and control theory methods for the purpose of optimal treatment design
at the individual \cite{Guo2012,Orellana2010} and population level
\cite{Agusto2014,Zhang2016}. 

Our aim is to estimate the true parameters, denoted $\theta^{*}$
and $\vartheta^{*}$, starting from data $y_{1},\dots,y_{n}$, that
are realizations of an observation process for $i=1,\dots,n$
\begin{equation}
Y_{i}=CX^{*}(t_{i})+\epsilon_{i}\label{eq:ObservationProcess}
\end{equation}
on the observation interval $\left[0,\,T\right]$ where $X^{*}:=X_{\theta^{*},\vartheta^{*},x_{0}^{*}}$
is the solution of (\ref{eq:ODEmodel-gen}) for $\theta=\theta^{*}$,
$\vartheta=\vartheta^{*}$ and $x_{0}=x_{0}^{*}$, $C$ is a $d^{'}\times d$
observation matrix and $\epsilon_{i}$ is centered observation noise.
That is, we want to estimate the true parameters $\left(\theta^{*},\vartheta^{*}\right)$
starting from discrete, partial and noisy observations of the true
solution $X^{*}$ at observation times $0=t_{1}<t_{2}\dots<t_{n}=T$.
If there is no functional parameter $\vartheta^{*}$, estimation of
$\theta^{*}$ is a standart parametric nonlinear regression problem
and can be solved by classical methods such as Nonlinear Least Squares
(NLS), Maximum Likelihood Estimation (MLE), or Bayesian Inference
\cite{Esposito2000,LiOsborne2005,Rodriguez2006,Wu2011}. However,
in the case of ODE models, there is a risk of an ill-posed inverse
problem \cite{Engl2009,Stuart2010}. 

To explain why, let us denote as $X_{\theta,x_{0}}$ the solution
to (\ref{eq:ODEmodel-gen}). The Fisher information matrix which controls
the Cramer-Rao bound is proportional to $\mathcal{I}_{n}\left(\theta,x_{0}\right)=\sum_{i=1}^{n}\left(CJ_{\theta,x_{0}}^{i}\right)^{T}CJ_{\theta,x_{0}}^{i}$
where $J_{\theta,x}^{i}$ is the sensitivity matrix of $X_{\theta,x_{0}}$
with respect to $\left(\theta,x_{0}\right)$ at time $t_{i}$. Instabilities
in estimation arise when the matrices $CJ_{\theta,x_{0}}^{i}$ are
badly-conditioned because in this case the inverse problem is very
sensitive to any source of perturbations and the objective function
(NLS or MLE criteria) is nearly flat around its minimum. This practical
identifiability problem can be measured by computing the spectrum
$\mu_{1}\geq\dots\geq\mu_{p}$ of $\mathcal{I}_{n}\left(\theta,x_{0}\right)$
and the associated condition number $\kappa(\mathcal{I}_{n})=\frac{\mu_{1}}{\mu_{p}}$.
The problem arises in part from the observation process, i.e. the
observation matrix $C$, the sparsity and location of the observation
times and also from the need to estimate the nuisance parameter $x_{0}^{*}$.
Complication in ODEs also arises due to the complex geometry of the
manifold $\left\{ CX_{\theta,x_{0}},\theta\in\Theta,x_{0}\in\chi\right\} $
induced by the mapping $\left(\theta,x_{0}\right)\longmapsto CX_{\theta,x_{0}}$
where there can be a small number (in comparison with $p$) of important
directions of variation very skewed from the original parameter axes
\cite{Gutenkunst2007,Transtrum2011,Transtrum2015}. This situation
is termed sloppiness and leads to a regular and widespread distribution
of the eigenvalues $\mu_{1},\dots,\mu_{p}$ with no clear one to one
correspondence between the eigenvectors of $\mathcal{I}_{n}\left(\theta,x_{0}\right)$
and the original ODE parametrization. Numerous ODEs used for example
in systems biology \cite{Gutenkunst2007} and neuroscience \cite{OLeary2015}
have been identified as sloppy. Sloppiness in ODE models has been
investigated in \cite{Tonsing2014} and shown to be mainly due to
the sparse and block structure of $CJ_{\theta,x_{0}}^{i}$, with highly
correlated entries. Sloppiness is a phenomenon due to interactions
between intrinsic system properties and the experimental design. Since
we cannot clearly distinguish important parameters from the others,
there is no clear mechanism to suppress irrelevant parameters in the
model. Moreover, methods based on optimal experimental design to circumvent
sloppiness can lead to experiments which render important ignored
parameters in the model and thus reduce accuracy and limit predictive
ability \cite{White2016}. Despite that sloppiness and practical identifiability
are not rigorously the same problem \cite{White2016}, the former
often induces the latter by making some subset of parameters unidentifiable.
Thus, there is a need to improve estimation methods which use the
sole training data set.

Another issue in ODE parameter estimation comes from the fact that
the selected model is often derived after successive simplifying assumptions
and approximations. One can think of regulation networks in systems
biology, where interactions are modeled by pairwise products while
higher order terms and the influence of external factors (forcing
functions) are unknown or neglected. Moreover, many biological processes
are known to be stochastic, and the justification of deterministic
modeling comes from the approximation of stochastic processes by ODE
solutions see \cite{Kurtz1970,Kurtz1978,Gillespie2000,Kampen1992}.
Hence, inference of the parameters has to be done while recognising
that the model is false \cite{kirk2016reverse,Brynjarsdottir2014}.

In this work, we propose a new estimation procedure to address these
challenges, based on an approximate solution of the original ODE.
The use of approximate solutions for statistical inference, such as
the classical two-step approaches \cite{Brunel2008,Brunel2014,GugushviliKlaassen2010,LiangMiaoWuaoas2010,Varah1982,Dattner2015},
Generalized Profiling (GP) \cite{Hooker2011,Ramsay2007} or even in
a Bayesian framework \cite{Chkrebtii2016,Jaeger2011}, has already proven to
be useful for regularizing the inverse problem of parameter estimation. 

Our proposed method here is seen as an attempt at improving the methods
proposed in \cite{BrunelClairon_Riccati2014,BrunelClairon_Kalman2014,BrunelClairon_Pontryagin2017},
where an approximation $X_{\theta,x_{0},u}$ is a solution of the
perturbed ODE $\dot{x}(t)=f\left(t,x(t),\theta\right)+Bu(t)$ and
where the perturbation $t\mapsto Bu(t)$ captures different sources
of model misspecification. After a pre-smoothing step to obtain a
nonparametric curve estimator $\widehat{Y}$, the parametric estimator
$\left(\widehat{\theta},\,\widehat{x_{0}}\right)$ is then defined
as the minimizer of the cost $C_{\lambda}(\theta,x_{0},u)=\left\Vert CX_{\theta,x_{0},u}-\widehat{Y}\right\Vert _{L^{2}}^{2}+\lambda\left\Vert u\right\Vert _{L^{2}}^{2}$
profiled on the possible perturbations $u$: $\left(\widehat{\theta},\,\widehat{x_{0}}\right)=\arg\min_{\left(\theta,x_{0}\right)}S(\theta,x_{0})$,
where $S(\theta,x_{0})=\min_{u}C_{\lambda}(\theta,x_{0},u)$. This
estimator, called the Tracking Estimator (TE), is thus defined as
the parameter which needs the smallest perturbation $u$ in order
to track $\widehat{Y}$, the balance between the two contrary objectives
of data fidelity (i.e. $\left\Vert CX_{\theta,x_{0},u}-\widehat{Y}\right\Vert _{L^{2}}^{2}$)
and original model fidelity (i.e. $\left\Vert u\right\Vert _{L^{2}}^{2}$)
is done through the choice of an hyperparameter value $\lambda>0$.
For each value $\left(\theta,x_{0}\right)$, the optimal control problem
$\min_{u}C_{\lambda}(\theta,x_{0},u)$ is solved by using the Pontryagin
maximum principle \cite{pontryagin1962minimum} in the non-linear
case \cite{BrunelClairon_Pontryagin2017} and the linear-quadratic
theory \cite{Sontag1998} for linear models \cite{BrunelClairon_Riccati2014,BrunelClairon_Kalman2014}.
In comparison with GP and NLS, the TE generally has a lower variance
and mean square error with the difference in performance even more
marked in the presence of model misspecification. In the parametric
case and for well-specified models, the TE has been proven to be consistent
with a $\sqrt{n}$- convergence rate under very mild model regularity
conditions and provided $\lambda>\overline{\lambda_{1}}$, with $\overline{\lambda_{1}}$
a positive model dependent bound. Moreover, an attractive feature
of the tracking framework is the seamless estimation of finite-dimensional
and time-varying parameters. The estimation of $t\mapsto\vartheta(t)$
in $\dot{x}=f\left(t,x,\theta,\vartheta\right)$ can be turned into
an optimal control problem and our estimator $\hat{\vartheta}$ is
a by-product of $\theta^{*}$ estimation which does not require the
use of standard approximations such as sieves or basis expansions
\cite{LiangMiaoWuaoas2010,XueMiaoWuaos2010,Hooker2011,Wang2014}.
However, two main limitations for the method given in \cite{BrunelClairon_Pontryagin2017}.
First, the computational time: solving the optimal control problem
by using the Pontryagin maximum principle leads to a boundary value
problem (BVP) for each new $\left(\theta,x_{0}\right)$ value and
$x_{0}^{*}$ has to be estimated as nuisance parameter. Second, the
method requires a nonparametric estimator $\widehat{Y}$. In the sparse
data case, the reconstruction $\widehat{Y}$ can be biased and this
nonparametric bias can then be spread to the parametric estimation.
Here, while we keep the same formal approach as in \cite{BrunelClairon_Pontryagin2017},
we solve the related optimal control problem by relying on discrete
control theory and a numerical method inspired by \cite{Cimen2004}.
This allows us to construct an estimation method which: 
\begin{enumerate}
	\item replaces the BVP by a sequence of finite difference equations, our
	procedure can be then applied to ODE systems of higher dimension than
	in \cite{BrunelClairon_Pontryagin2017};
	\item removes the pre-smoothing step, we can deal with sparse data cases
	which are consistent with most real observation framework;
	\item gives a consistent estimator with parametric convergence rate with
	only a strictly positive condition on the hyperparameter, i.e $\lambda>0$;
	\item can be easily adapted to avoid estimation of $x_{0}^{*}$ if it is
	not required. 
\end{enumerate}
In order to define our estimator, we present in the next section the
optimal control problem required to introduce our functional criteria
which is again a profiled cost $S_{n}$. We also describe how a semi-parametric
estimation problem can be turned into an optimal control problem similar
to that used for parametric estimation. In Subsection \ref{sec:S_tractable_form_Pontryagin_as_analysis},
we derive a tractable form for $S_{n}$ and describe the related numerical
procedure, which is based for linear models on discrete linear-quadratic
theory and for nonlinear models on the work of \cite{Cimen2004,CimenBanks2004}.
We present in Section \ref{sec:Estimator_theoretical_analysis} sufficient
conditions to ensure $S_{n}$ is well-defined on the parameter space
as well as consistency with $\sqrt{n}$- convergence in the parametric
case and for well-specified models. In Section \ref{sec:Experiment},
we use Monte Carlo experiments in order to compare the Tracking, Nonlinear
Least Squares and Generalized Profiling estimators on ODE examples
from chemistry and biology with both well-specified and misspecified
models. This section is concluded by simulations where we perform
the joint estimation of the finite dimensional and time-varying parameters
$\theta$ and $\vartheta$. In Section \ref{sec:Real-data-case},
we consider parameter estimation with real data in a model used to
study microbiotal population evolution. 

\section{Model and methodology\label{sec:Model-and-methodology}}

We recall the aim of this work is to estimate $\left(\theta^{*},\vartheta^{*}\right)$
from the data $Y=\left\{ y_{1},\dots y_{n}\right\} $ defined as the
minimizer of functional criteria $S_{n}$. First, we derive $S_{n}$
in the parametric case where there is no functional parameter $\vartheta^{*}$.

\subsection{Formal parametric estimator definition}

We denote by $X_{\theta,x_{0}}$ the solution of the Initial Value
Problem (IVP):
\begin{equation}
\left\{ \begin{array}{l}
\dot{x}(t)=f(t,x(t),\theta)\\
x(0)=x_{0}.
\end{array}\right.\label{eq:ODEmodel_gen_form}
\end{equation}
First, we need to reformulate the model (\ref{eq:ODEmodel_gen_form})
into a pseudo-linear form:
\begin{equation}
\left\{ \begin{array}{l}
\dot{x}(t)=A_{\theta}(x(t),t)x(t)\\
x(0)=x_{0}.
\end{array}\right.\label{eq:ODEmodel}
\end{equation}
This formulation is crucial for solving in a computationally efficient
way the optimal control problem defining our estimator. Of course,
linear models already fit in this formalism with $A_{\theta}(t)\,=A_{\theta}(x(t),t)$.
For nonlinear models, the pseudo-linear representation is not unique
but always exists \cite{Cimen2004}. Now, we introduce the solution
$X_{\theta,x_{0},u}$ of the perturbed ODE:
\begin{equation}
\left\{ \begin{array}{l}
\dot{x}(t)=A_{\theta}(x(t),t)x(t)+Bu(t)\\
x(0)=x_{0}
\end{array}\right.\label{eq:ControlledODEmodel}
\end{equation}
where the function $t\mapsto Bu(t)$ is a linear perturbation, $B$
is a $d\times d_{u}$ matrix and $u$ is in $L^{2}\left(\left[0,T\right],\mathbb{R}^{d_{u}}\right)$.
To proceed to parametric estimation, we consider a discretized version
of the perturbed ODE (\ref{eq:ControlledODEmodel}). The discretization
will be made at $m+1$ time points $\left\{ t_{j}^{d}\right\} _{0\leq j\leq m}$
with $t_{0}^{d}=0$ and $t_{m}^{d}=T$. Letting $\Delta_{j}=t_{j+1}^{d}-t_{j}^{d}$
being the mesh size between two discretization time-points and $u=\left(u_{0},\ldots,u_{m-1}\right)$
the set of discrete values taken by the control at each time step,
the discretized version is: 

\begin{equation}
\left\{ \begin{array}{l}
x(t_{j+1}^{d})=\left(I_{d}+\Delta_{j}A_{\theta}(x(t_{j}^{d}),t_{j}^{d})\right)x(t_{j}^{d})+B\Delta_{j}u_{j}\\
x(0)=x_{0}.
\end{array}\right.\label{eq:discretized_equation}
\end{equation}
The set of discretization time-point has to contain the observation
time points i.e. $\left\{ t_{i}\right\} _{0\leq i\leq n}\subset\left\{ t_{j}^{d}\right\} _{0\leq j\leq m}$
but can be bigger, this is an important feature of the discretization
scheme which allows us to accurately estimate $X_{\theta,x_{0},u}$
even when the observations are sparse on $\left[0,\,T\right]$. We
denote: 
\begin{itemize}
	\item $X_{\theta,x_{0},u}^{d}(t_{j}^{d})$, the solution of (\ref{eq:discretized_equation})
	for the parameter $\theta$, initial condition $x_{0}$ and the perturbation
	$u$ at time $t_{j}^{d}$. 
	\item $w_{j}=1_{\left\{ \exists t_{i}\,\textrm{s.t}\,t_{i}=t_{j}^{d}\right\} }$
	i.e. $w_{i}$ is equal to $1$ if $t_{j}^{d}$ corresponds to an observation
	time $t_{i}$, otherwise $w_{j}=0$.
	\item $\mathbf{y}_{j}$ is equal to $y_{i}$ if $t_{j}^{d}$ corresponds
	to the observation time $t_{i}$ and $0$ otherwise. 
\end{itemize}
The weights $w_{j}$ and the set of extended data $\left\{ \mathbf{y}_{i}\right\} $
are introduced to have a vector of observation which has the same
length as the discretization grid $\left\{ t_{j}^{d}\right\} _{0\leq j\leq m}$.
Now, we can introduce the discretized cost we want to minimize:
\begin{equation}
\begin{array}{llll}
C_{T}^{d}(Y;\theta,x_{0},u,U) & = & \sum_{i=0}^{n}\triangle_{i}\left\Vert CX_{\theta,x_{0},u}^{d}(t_{i})-y_{i}\right\Vert _{2}^{2}+\sum_{j=0}^{m-1}\triangle_{j}u_{j}{}^{T}Uu_{j}\\
& = & \left\Vert CX_{\theta,x_{0},u}^{d}(t_{n})-y_{n}\right\Vert _{2}^{2}\\
& + & \sum_{j=0}^{m-1}\triangle_{j}\left(\left\Vert CX_{\theta,x_{0},u}^{d}(t_{j}^{d})-\mathbf{y}_{j}\right\Vert _{2}^{2}w_{j}+u_{j}{}^{T}Uu_{j}\right),
\end{array}\label{eq:discrete_general_cost_function}
\end{equation}
and for each $\left(\theta,x_{0}\right)$ in $\Theta\times\chi$,
the profiled cost:
\begin{equation}
S_{n}(Y;\theta,x_{0},U)\,:=\inf_{u\in L_{a}}C_{T}^{d}(Y;\theta,x_{0},u,U)\label{eq:profiled_cost_general_expression}
\end{equation}
where $L_{a}$ is the set of admissible perturbations defined as the
set of controls generating trajectories bounded on $\left[0,\,T\right]$.
Here $U$ is a symmetric definite positive matrix used as a weighting
parameter balancing the amount of model and data fidelity. Similarly
to \cite{BrunelClairon_Pontryagin2017}, the tracking estimator (TE)
is defined as:
\begin{equation}
\left(\widehat{\theta}^{T},\widehat{x_{0}}^{T}\right):=\arg\min_{\left(\theta,x_{0}\right)\in\Theta\times\chi}S_{n}(Y;\theta,x_{0},U),\label{eq:TrackingEstimator}
\end{equation}
i.e. as the parameter that gives the closest trajectory $CX_{\theta,x_{0},u}^{d}$
to the observed data on $\left[0,T\right]$, while allowing a small
divergence from (\ref{eq:discretized_equation}). To compute $S_{n}$
in practice we need to solve the optimization problem:
\begin{equation}
\begin{array}{l}
\min_{u}C_{T}^{d}(Y;\theta,x_{0},u,U)\\
\begin{array}{l}
\textrm{such that }x(t_{j+1}^{d})=\left(I_{d}+\Delta_{j}A_{\theta}(x(t_{j}^{d}),t_{j}^{d})\right)x(t_{i})+B\Delta_{j}u_{j}\\
\textrm{and }x(0)=x_{0}.
\end{array}
\end{array}\label{eq:Original_problem}
\end{equation}
The problem (\ref{eq:Original_problem}) is an optimal control one
belonging to the subclass of tracking problems where the aim is to
find the smallest control possible to apply to a given dynamical system
in order to track a signal. For linear models, these problems have
been efficiently solved as they fit into the framework of discrete
linear-quadratic problems, which ensures the existence and uniqueness
of the solution and gives a computationally efficient way to find
it. For non-linear models, \cite{Cimen2004} proposes an iterative
method to solve continuous time tracking problems, the main idea being
to replace the original problem by a sequence of linear-quadratic
ones. We will use the same method adapted to discrete models, but
first, in the next subsection, we explain how the estimation of a
time-varying parameter $\vartheta\,:\,\left[0,\,T\right]\,\longrightarrow\mathbb{R}^{d_{f}}$
in ODEs $\dot{x}(t)=A_{\theta}(x(t),\vartheta(t),t)x$ is straightforward
within our framework. 

\subsection{\label{subsec:Semi-parametric-estimation}Semi-parametric estimation}

For this, let us introduce the extended state $x_{e}=\left(x,z_{1},z_{2}\right)$
in $\mathbb{R}^{d+2d_{f}}$, the extended pseudo-linear representation:
\[
\begin{array}{lll}
A_{\theta}^{e}(x_{e}(t),t) & = & \left(\begin{array}{lll}
A_{\theta}(x(t),z_{1}(t),t) & 0 & 0\\
0 & 0 & 1\\
0 & 0 & 0
\end{array}\right)\end{array}
\]
and the perturbed solution $X_{\theta,x_{0}^{e},u}^{e}$ of the parametric
extended ODE:
\begin{equation}
\left\{ \begin{array}{l}
\dot{x_{e}}(t)=A_{\theta}^{e}(x_{e}(t),t)x_{e}(t)+B_{ext}u(t)\\
x_{e}(0)=x_{0}^{e}
\end{array}\right.\label{eq:functional_param_perturbated_ode}
\end{equation}
with 
\[
B_{ext}=\left(\begin{array}{cc}
I_{d} & 0_{d,d_{f}}\\
0_{d_{f},d} & 0_{d_{f}}\\
0_{d_{f},d} & I_{d_{f}}
\end{array}\right).
\]
Here, $u$ is split into two parts, $u(t)=\left(u_{1}(t),\,u_{2}(t)\right)$,
and $X_{\theta,x_{0}^{e},u}^{e}$ is solution of
\begin{equation}
\begin{array}{l}
\dot{x}=A(t,x,z_{1},\theta)x+u_{1}\\
\dot{z}_{1}=z_{2}\\
\dot{z}_{2}=u_{2}.
\end{array}\label{eq:extended_vector_field}
\end{equation}
One can see that $z_{1}$ plays the role of $\vartheta$ , and $z_{2}$
of $\dot{\vartheta}$. If we get a state variable estimator $\widehat{X^{e}}$,
which is the case in our method, $\widehat{X^{e}}$ being obtained
as a byproduct of $\theta^{*}$ estimation, then we define $\widehat{\vartheta}=\widehat{z_{1}}$.
Let us introduce $U=\left(\begin{array}{cc}
\lambda_{1}I_{d} & 0_{d,d_{f}}\\
0_{d_{f},d} & \lambda_{2}I_{d_{f}}
\end{array}\right)$and the cost
\begin{equation}
\begin{array}{lll}
C_{T}^{d}(Y;\theta,x_{0}^{e},u,U) & = & \sum_{i=0}^{n}\triangle_{i}\left\Vert CX_{\theta,x_{0}^{e},u}^{e,d}(t_{i})-y_{i}\right\Vert _{2}^{2}\\
 &+ & \sum_{j=0}^{m-1}\triangle_{j}\left(\lambda_{1}u_{1,j}{}^{T}u_{1,j}+\lambda_{2}u_{2,j}{}^{T}u_{2,j}\right)\\
& = & \sum_{i=0}^{n}\triangle_{i}\left\Vert CX_{\theta,x_{0}^{e},u}^{e,d}(t_{i})-y_{i}\right\Vert _{2}^{2}\\
& + & \sum_{j=0}^{m-1}\triangle_{j}\left(\lambda_{1}u_{1,j}{}^{T}u_{1,j}+\lambda_{2}\ddot{z_{1,j}}{}^{T}\ddot{z_{1,j}}\right).
\end{array}\label{eq:semi_param_cost_with_discrepancy_model}
\end{equation}
Here $u_{1,j}{}^{T}u_{1,j}$ is used to quantify model discrepancy
as in the parametric case, and the last term $\lambda_{2}\ddot{z_{1,j}}{}^{T}\ddot{z_{1,j}}$
is the standard penalty used for functional estimation. Thus, a good
choice of hyperparameter for cost (\ref{eq:semi_param_cost_with_discrepancy_model})
would be a large value for $\lambda_{1}$ (in order to select a small
$u$), and $\lambda_{2}$ tending to $0$ when the sample size $n$
grows, as for standard nonparametric estimation. 

\subsection{\label{sec:S_tractable_form_Pontryagin_as_analysis}Tractable form
	for $S_{n}$}

In this subsection, we derive a tractable expression for $S_{n}$.
First, we deal with linear models then we extend the derived method
to nonlinear models by following \cite{Cimen2004}. 

\subsubsection{Linear models}

Here, we focus on the linear case i.e. when $A_{\theta}(t)\,=A_{\theta}(x,t)$
in (\ref{eq:ODEmodel_gen_form}). For a given initial condition $x_{0}$,
linear-quadratic theory ensures the existence and uniqueness of the
optimal control $\overline{u_{\theta,x_{0}}^{d}}=\arg\min_{u\in L_{a}}C_{T}^{d}(Y;\theta,x_{0},u,U)$
and that $\inf_{u\in L_{a}}C_{T}^{d}(Y;\theta,x_{0},u,U)$ can be
computed by solving a discrete final value problem, denoted the Riccati
equation. The formal computational details are left in \textit{supplementary
	materials} (Section 2). 

\begin{prop}
	\label{prop:discrete_Riccati_accurate_representation}For a given
	$\left(\theta,x_{0}\right)$ in $\Theta\times\chi$, the profiled
	cost value $S_{n}(Y;\theta,x_{0},U)$ is equal to:
	\begin{equation}
	\begin{array}{l}
	S_{n}(Y;\theta,x_{0},U)=x_{0}^{T}R_{\theta,0}^{d}x_{0}+2h_{\theta,0}^{d}(Y)^{T}x_{0}+\triangle_{m}\mathbf{y}_{m}^{T}\mathbf{y}_{m}\\
	+\sum_{j=0}^{m-1}\triangle_{j}\left(\mathbf{y}_{j}^{T}\mathbf{y}_{j}-h_{\theta,j+1}^{d}(Y)^{T}BG(R_{\theta,j+1}^{d})B^{T}h_{\theta,j+1}^{d}(Y)\right)
	\end{array}\label{eq:S_n_linear_case}
	\end{equation}
	with $G(R_{\theta,j+1}^{d})\,:=\left[U+\triangle_{j}B^{T}R_{\theta,j+1}^{d}B\right]^{-1}$
	and $(R_{\theta,j}^{d},\,h_{\theta,j}^{d}(Y))$ for $1\leq j\leq m$,
	the solution of the discrete Riccati equation: 
	\begin{equation}
	\begin{array}{lll}
	R_{\theta,j}^{d} & = & R_{\theta,j+1}^{d}+\triangle_{j}w_{j}C^{T}C+\Delta_{j}\left(R_{\theta,j+1}^{d}A_{\theta}(t_{j}^{d})+A_{\theta}(t_{j}^{d})^{T}R_{\theta,j+1}^{d}\right)\\
	& + & \triangle_{j}^{2}A_{\theta}(t_{j}^{d})^{T}R_{\theta,j+1}^{d}A_{\theta}(t_{j}^{d})\\
	& - & \triangle_{j}(I_{d}+\triangle_{j}A_{\theta}(t_{j}^{d})^{T})R_{\theta,j+1}^{d}BG(R_{\theta,j+1}^{d})B^{T}R_{\theta,j+1}^{d}(I_{d}+\triangle_{j}A_{\theta}(t_{j}^{d}))\\
	h_{\theta,j}^{d}(Y) & = & h_{\theta,j+1}^{d}(Y)-\triangle_{j}w_{j}C^{T}\mathbf{y}_{j}+\triangle_{j}A_{\theta}(t_{j}^{d})^{T}h_{\theta,j+1}^{d}(Y)\\
	& - & \triangle_{j}(I_{d}+\triangle_{j}A_{\theta}(t_{j}^{d})^{T})R_{\theta,j+1}^{d}BG(R_{\theta,j+1}^{d})B^{T}h_{\theta,j+1}^{d}(Y)
	\end{array}\label{eq:discrete_accurate_Riccati_equation_lincase}
	\end{equation}
	with final condition 	$(R_{\theta,m}^{d},\,h_{\theta,m}^{d}(Y)) = (\triangle_{m}C^{T}C,\,-\triangle_{m}C^{T}\mathbf{y}_{m}).$
	Moreover, the control $\overline{u_{\theta,x_{0}}^{d}}$ which minimizes
	the cost (\ref{eq:discrete_general_cost_function}) is unique and
	equal to:
	\begin{equation}
	\overline{u_{\theta,x_{0},j}^{d}}=-G(R_{\theta,j+1}^{d})B^{T}\left(R_{\theta,j+1}^{d}\left(I_{d}+\triangle_{j}A_{\theta}(t_{j}^{d})\right)\overline{X_{\theta,x_{0}}^{d}}(t_{j})+h_{\theta,j+1}^{d}(Y)\right)\label{eq:opt_control_expression}
	\end{equation}
	where $\overline{X_{\theta,x_{0}}^{d}}$ is the optimal trajectory,
	i.e. the solution of the initial value problem 
	\begin{equation}
	\left\{ \begin{array}{lll}
	\overline{X_{\theta,x_{0}}^{d}}(t_{j+1}^{d}) & = & \left(I_{d}+\triangle_{j}A_{\theta}(t_{j}^{d})\right)\overline{X_{\theta,x_{0}}^{d}}(t_{j}^{d})\\
	& - & \triangle_{j}BG(R_{\theta,j+1}^{d})B^{T}\left(R_{\theta,j+1}^{d}\left(I_{d}+\triangle_{j}A_{\theta}(t_{j}^{d})\right)\overline{X_{\theta,x_{0}}^{d}}(t_{j}^{d})+h_{\theta,j+1}^{d,l}(Y)\right)\\
	\overline{X_{\theta,x_{0}}^{d}}(0) & = & x_{0}.
	\end{array}\right.\label{eq:opt_trajectory_expression}
	\end{equation}
\end{prop}

\subsubsection{Non-linear models}

Here, we adapt the method proposed by \cite{Cimen2004} to solve the
tracking problem for discrete time models. The outline of the method
is as follows. We replace the original problem (\ref{eq:Original_problem})
by a recursive sequence of control problems, with iteration $l$ defined
by
\begin{equation}
\begin{array}{l}
\begin{array}{lll}
\min_{u}C_{T}^{d,l}(Y;\theta,x_{0},u,U) & \,:= & \triangle_{m}\left\Vert CX_{\theta,x_{0},u}^{l}(t_{m}^{d})-\mathbf{y}_{m}\right\Vert _{2}^{2}\\
& + & \sum_{j=0}^{m-1}\triangle_{j}\left(\left\Vert CX_{\theta,x_{0},u}^{l}(t_{j}^{d})-\mathbf{y}_{j}\right\Vert _{2}^{2}w_{j}+u_{j}^{T}u_{j}\right)
\end{array}\\
\textrm{such that }x(t_{j+1}^{d})=(I_{d}+\triangle_{j}A_{\theta}^{l}(t_{j}^{d}))x_{j}(t_{j}^{d})+\triangle_{j}B_{\sigma}u_{j}\\
\textrm{and }x(0)=x_{0}
\end{array}\label{eq:Sequence_LQ_problem}
\end{equation}
where $A_{\theta}^{l}(t_{j}^{d}):=A_{\theta}(\overline{X_{\theta,x_{0}}^{l-1}}(t_{j}^{d}),t_{j}^{d})$
and $A_{\theta}^{0}(t_{j}^{d}):=A_{\theta}(x_{0},t_{j}^{d})$. Here
$\overline{X_{\theta,x_{0}}^{l-1}}$ is the optimal trajectory corresponding
the optimal control problem (\ref{eq:Sequence_LQ_problem}) at iteration
$l-1$. Thus, for each $l$ the problem (\ref{eq:Sequence_LQ_problem})
fits into the framework of our previous subsection, and for each $l$,
we have access to a solution of the Riccati equation $(R_{\theta}^{d,l},\,h_{\theta}^{d,l}(Y))$,
an optimal control $\overline{u_{\theta,x_{0}}^{d,l}}$, an optimal
trajectory $\overline{X_{\theta,x_{0}}^{d,l}}$ and a profiled cost
value $S_{n}^{l}(Y;\theta,x_{0},U)$. Moreover, the sequences $\left\{ R_{\theta}^{d,l},\,h_{\theta}^{d,l}(Y)\right\} _{l\in\mathbb{N}}$,
$\left\{ \overline{u_{\theta,x_{0}}^{l}}\right\} _{l\in\mathbb{N}}$
, $\left\{ \overline{X_{\theta,x_{0}}^{d,l}}\right\} _{l\in\mathbb{N}}$
and $\left\{ S_{n}^{l}(Y;\theta,U)\right\} _{l\in\mathbb{N}}$ are
uniformly convergent in $l$ \cite{Cimen2004,CimenBanks2004}. Thus,
we can propose the following algorithm to compute $(R_{\theta}^{d},\,h_{\theta}^{d}(Y))$
, $\overline{u_{\theta,x_{0}}^{d}}$, $\overline{X_{\theta,x_{0}}^{d}}$
and $S_{n}(Y;\theta,x_{0},U)$.
\begin{enumerate}
	\item Initialization phase: $\overline{X_{\theta,x_{0}}^{d,0}}(t_{j}^{d})=x_{0}$
	and $A_{\theta}^{0}(t_{j}^{d})=A_{\theta}(x_{0},t_{j}^{d})$ for all
	$j\in\left\llbracket 0,\,m\right\rrbracket $. 
	\item At iteration $l$: use Proposition \ref{prop:discrete_Riccati_accurate_representation}
	to obtain $(R_{\theta}^{d,l},\,h_{\theta}^{d,l}(Y))$ , $\overline{u_{\theta x_{0}}^{d,l}}$,
	$\overline{X_{\theta,x_{0}}^{d,l}},$ $S_{n}^{l}(Y;\theta,x_{0},U)$. 
	\item If $\sum_{j=1}^{m}\left\Vert \overline{X_{\theta,x_{0}}^{d,l}}(t_{j}^{d})-\overline{X_{\theta,x_{0}}^{d,l-1}}(t_{j}^{d})\right\Vert _{2}^{2}<\varepsilon_{1}$
	and 
		\begin{equation}
		\left|S_{n}^{l}(Y;\theta,x_{0},U)-S_{n}^{l-1}(Y;\theta,x_{0},U)\right|<\varepsilon_{2},
		\end{equation}
	with $\left(\epsilon_{1},\,\epsilon_{2}\right)$ two strictly positive
	constants, then step 4; otherwise return to step 2.
	\item Set $(R_{\theta}^{d},\,h_{\theta}^{d}(Y))\,=(R_{\theta}^{d,l},\,h_{\theta}^{d,l}(Y))$
	, $\overline{u_{\theta,x_{0}}^{d}}\,=\overline{u_{\theta,x_{0}}^{d,l}}$,
	$\overline{X_{\theta,x_{0}}^{d}}\,=\overline{X_{\theta,x_{0}}^{d,l}}$
	and $S_{n}(Y;\theta,x_{0},U)\,=S_{n}^{l}(Y;\theta,x_{0},U)$. 
\end{enumerate}
In this section, we have seen how to compute $S_{n}$ in practice.
However, before looking for its minimum $\left(\widehat{\theta}^{T},\widehat{x_{0}}^{T}\right)$
by a numerical optimization method, we need to prove that the function
$\left(\theta,x_{0}\right)\longmapsto S_{n}(Y;\theta,x_{0},U)$ is
well defined on a parameter space containing $\left(\theta^{*},x_{0}^{*}\right)$
and study its regularity. This achieved in the next section, where
we also derive the conditions under which $\left(\widehat{\theta}^{T},\widehat{x_{0}}^{T}\right)$
is consistent and even asymptotically normal with a $\sqrt{n}$- rate
convergence rate. Before that, we describe an extension of our method
which allows us to avoid estimation of $x_{0}^{*}$ and reduces the
dimension of the optimization problem (\ref{eq:TrackingEstimator}),
and which is expected to be more computationally efficient.

\subsection{Profiling on $x_{0}$}

We introduce a new estimator $\widehat{\theta}^{T,CI}$ of the true
parameter $\theta^{*}$:
	\begin{equation}
	\begin{array}{lll}
	\widehat{\theta}^{T,CI} & = & \arg\min_{\theta\in\Theta}S_{n}^{CI}(Y;\theta,U)\\
	& := & \arg\min_{\theta\in\Theta}\left\{ \min_{x_{0}}S_{n}(Y;\theta,x_{0},U)\right\} .
	\end{array}\label{eq:TrackingEstimator_profiled_CI}
	\end{equation}
	We now profile the cost function on $x_{0}$ in addition to $u$.
	In equation (\ref{eq:S_n_linear_case}), one can see that $S_{n}(Y;\theta,x_{0},U)$
	is a quadratic form with respect to $x_{0}$, hence the profiling
	is straightforward. Interestingly, the formal computation used to
	derive $S_{n}(Y;\theta,U)$ follows the same step as the deterministic
	Kalman Filter state estimator derivation \cite{Sontag1998}. We derive
	the following expression for $S_{n}^{CI}(Y;\theta,U)$ in the linear
	case:
	\begin{prop}
		\label{prop:discrete_Riccati_accurate_representation_prof_CI}For
		a given $\theta$ in $\Theta$, $S_{n}^{CI}(Y;\theta,U)$ is equal
		to:
		\begin{equation}
		\begin{array}{l}
		S_{n}^{CI}(Y;\theta,U)=-h_{\theta,0}^{d}(Y)^{T}\left(R_{\theta,0}^{d}\right)^{-1}h_{\theta,0}^{d}(Y)+\triangle_{m}\mathbf{y}_{m}^{T}\mathbf{y}_{m}\\
		+\sum_{j=0}^{m-1}\triangle_{j}\left(\mathbf{y}_{j}^{T}\mathbf{y}_{j}-h_{\theta,j+1}^{d}(Y)^{T}BG(R_{\theta,j+1}^{d})B^{T}h_{\theta,j+1}^{d}(Y)\right).
		\end{array}\label{eq:S_n_CI_linear_case}
		\end{equation}
		where $(R_{\theta}^{d},\,h_{\theta}^{d}(Y))$ , $\overline{u_{\theta}^{d}}$
		are given by equations (\ref{eq:discrete_accurate_Riccati_equation_lincase}),
		(\ref{eq:opt_control_expression}) and $\overline{X_{\theta}^{d}}$
		is given by equation (\ref{eq:opt_trajectory_expression}) but with
		initial condition $\overline{X_{\theta}^{d}}(0)=-\left(R_{\theta,0}^{d}\right)^{-1}h_{\theta,0}.$
	\end{prop}
	In the non-linear case, the algorithm has to be adapted as follows
	\begin{enumerate}
		\item Initialization phase: $\overline{X_{\theta}^{d,0}}(t_{j}^{d})=x_{0}^{r}$
		and $A_{\theta}^{0}(t_{j}^{d})=A_{\theta}(x_{0}^{r},t_{j}^{d})$ for
		all $j\in\left\llbracket 0,\,m\right\rrbracket $ where $x_{0}^{r}$
		is an arbitrary starting point.
		\item At iteration $l$: use Proposition \ref{prop:discrete_Riccati_accurate_representation}
		to obtain $(R_{\theta}^{d,l},\,h_{\theta}^{d,l}(Y))$ and $\overline{u_{\theta}^{d,l}}$,
		then Proposition \ref{prop:discrete_Riccati_accurate_representation_prof_CI}
		to obtain $\overline{X_{\theta}^{d,l}}$ and $S_{n}^{CI,l}(Y;\theta,U)$. 
		\item If $\sum_{j=1}^{m}\left\Vert \overline{X_{\theta}^{d,l}}(t_{j}^{d})-\overline{X_{\theta}^{d,l-1}}(t_{j}^{d})\right\Vert _{2}^{2}<\varepsilon_{1}$
		and 	
		\begin{equation}
		\left|S_{n}^{CI,l}(Y;\theta,U)-S_{n}^{CI,l-1}(Y;\theta,U)\right|<\varepsilon_{2},
		\end{equation}
		then step 4; otherwise return to step 2.
		\item Set $(R_{\theta}^{d},\,h_{\theta}^{d}(Y))\,=(R_{\theta}^{d,l},\,h_{\theta}^{d,l}(Y))$
		, $\overline{u_{\theta}^{d}}\,=\overline{u_{\theta}^{d,l}}$, $\overline{X_{\theta}^{d}}\,=\overline{X_{\theta}^{d,l}}$
		and $S_{n}^{CI}(Y;\theta,U)\,=S_{n}^{CI,l}(Y;\theta,U)$. 
	\end{enumerate}
	
	\begin{rem}
		The state extension required for semi-parametric estimation involves
		the addition of new initial conditions $\left(\vartheta(0),\dot{\vartheta}(0)\right)$
		which need to be estimated in \cite{BrunelClairon_Pontryagin2017}.
		Interestingly here, since we profile on $x_{0}$, our approach does
		not add nuisance parameters to estimate. However, we still need to
		consider a model of larger dimension than the original.
	\end{rem}

\section{\label{sec:Estimator_theoretical_analysis}Theoretical analysis}

\subsection{Existence and regularity of $S_{n}$ and $S_{n}^{CI}$ }

First, we introduce the required conditions to ensure the existence
and regularity of $\theta\longmapsto(R_{\theta}^{d,l},\,h_{\theta}^{d,l}(Y))$
, $\left(\theta,x_{0}\right)\longmapsto S_{n}^{l}(Y;\theta,x_{0},U)$
and $\theta\longmapsto S_{n}^{CI,l}(Y;\theta,U)$ for each $l\in\mathbb{N}$: 
\begin{description}
	\item [{Condition}] 1: For all $t\in\left[0,\,T\right]$ and for all $\theta\in\Theta$,
	$x\longmapsto A_{\theta}(x,t)$ has a compact support $\varLambda$.
\end{description}
\begin{description}
	\item [{Condition}] 2: For all $x\in\varLambda$, $\theta\longmapsto A_{\theta}(x,.)$
	is continuous on $\varTheta$ and $\forall\theta\in\varTheta$, $(x,t)\longmapsto A_{\theta}(x,t)$
	is continuous on $\varLambda\times\left[0,\,T\right]$.
\end{description}
Condition 1 ensures the existence of a unique bounded function $X_{\theta,x_{0}}$
defined on $\left[0,T\right]$ for all $\left(\theta\times x_{0}\right)$
in $\Theta\times\varLambda$. In practice, the tracked signal is always
bounded, thus it is legitimate to focus on bounded solutions. Moreover,
let us recall that, for any function $f$, we can construct a function
$\widetilde{f}$ with the same level of smoothness as $f$ such that
for all $x\in\varLambda,\,\widetilde{f}(x)=f(x)$ and for all $x$
such that $d(x,\varLambda)>\epsilon,\,\widetilde{f}(x)=0$ for any
$\epsilon>0$. Hence, any matrix $A_{\theta}(x,.)$ can be replaced
by a counterpart satisfying Condition 1 with a compact support restricted
to relevant values for the state variables. From Condition 2 it is
straightforward to derive by induction $\theta\longmapsto(R_{\theta}^{d,l},\,h_{\theta}^{d,l}(Y))$
continuity with respect to $\theta$. In order to derive the same
property for $\left(\theta,x_{0}\right)\longmapsto S_{n}^{l}(Y;\theta,x_{0},U)$
and $\theta\longmapsto S_{n}^{CI,l}(Y;\theta,U)$, we need to ensure
boundedness of $\left(R_{\theta}^{d,l},\,h_{\theta}^{d,l}(Y)\right)$
as well as $R_{\theta,0}^{d,l}$ invertibility for $S_{n}^{l,CI}$
which is why we need to introduce Condition 3a. The proofs are left
in Section 3 of \textit{supplementary materials.}
\begin{prop}
	\label{prop:h_R_continuity_limits}Under conditions 1-2 for each $l\in\mathbb{N}$,
	$\left(\theta,x_{0}\right)\longmapsto S_{n}^{l}(Y;\theta,x_{0},U)$
	is continuous on $\Theta\times\varLambda$. Moreover if $R_{\theta,0}^{d,l}$
	is nonsingular then $\theta\longmapsto S_{n}^{CI,l}(Y;\theta,U)$
	is continuous on $\Theta$.
\end{prop}
The $R_{\theta,0}^{d,l}$ nonsingularity condition is somewhat ad-hoc,
so we present a necessary and sufficient testable criterion:
\begin{prop}
	\label{prop:discrete_R0_invertibility_criteria} Given $\theta\in\Theta$
	and $l\in\mathbb{N}$, $R_{\theta,0}^{d,l}$ is invertible if and
	only the matrix 
	\begin{equation}
	O_{\theta}^{d,l}(T)=C^{T}C+\sum_{i=1}^{n-1}\prod_{j=0}^{i-1}\left(I_{d}+\Delta A_{\theta}(t_{j}^{d},\overline{X_{\theta,x_{0}}^{d,l-1}}(t_{j}^{d}))\right)^{T}C^{T}C\prod_{j=0}^{i-1}\left(I_{d}+\Delta A_{\theta}(t_{j}^{d}\overline{X_{\theta,x_{0}}^{d,l-1}}(t_{j}^{d}))\right)\label{eq:discrete_invertibility_criteria}
	\end{equation}
	is invertible. 
\end{prop}
\begin{rem}
	When the system is totally observed, $C$ is of full rank and the
	matrix $O_{\theta}^{d,l}(T)$ is always nonsingular for all $\theta$
	in $\Theta$ and all $l\in\mathbb{N}$. Intuitively in the general
	case, $O_{\theta}^{d,l}(T)$ is invertible when the matrix $CA_{\theta}$
	has a sufficient number of non-zeros entries among those
	corresponding to the unobserved section of the system. That is, $O_{\theta}^{d,l}(T)$
	is invertible when the observed state variables give enough information
	on the whole system. This corresponds to the notion of observability
	in control theory (see \cite{Sontag1998}).
\end{rem}
 
\subsection{Asymptotic analysis of $\left(\widehat{\theta}^{T},\,\widehat{x_{0}}^{T}\right)$
	in parametric case}

In this section, we assume the discretization grid is the set of observation
time points i.e. $\left\{ t_{j}^{d}\right\} =\left\{ t_{i}\right\} $
which are regularly spaced so $\triangle_{i}=\triangle=\frac{T}{n}$
. Further, we assume i.i.d $\epsilon_{i}\sim N(0,\sigma^{2}I_{d'})$.
We also introduce the following notation.
\begin{enumerate}
	\item $o_{l}(1)$ is an arbitrary function $g$ (possibly vector or matrix
	valued) such that $\lim_{l\longmapsto\infty}g(l)=0$.
	\item $o_{n}(f(\triangle))$ (resp $O_{n}(f(\triangle))$ ) is an arbitrary
	function $g$ (possibly vector or matrix valued) such that $\lim_{n\longmapsto\infty}\frac{g(\triangle)}{f(\triangle)}=\lim_{\triangle\longmapsto0}\frac{g(\triangle)}{f(\triangle)}=0$
	(resp $\lim_{n\longmapsto\infty}\frac{g(\triangle)}{f(\triangle)}=L$
	with $L$ constant and finite).
	\item $o_{p,n}(f(\triangle))$ (resp $O_{p,n}(f(\triangle))$ ) is a random
	variable $g$ such that $\frac{g(\triangle)}{f(\triangle)}$ tends
	to $0$ in probability when $n\longrightarrow+\infty$ ( resp $\frac{g(\triangle)}{f(\triangle)}$
	is bounded in probability when $n\longrightarrow+\infty$).
\end{enumerate}
The required conditions on $l$ for consistency and $\sqrt{n}$-convergence
rate are necessary only for non-linear systems. Indeed, for linear
models we directly use (\ref{eq:discrete_accurate_Riccati_equation_lincase})
to compute $S_{n}^{l}(Y;\theta,U)$ instead of the algorithm inspired
by \cite{Cimen2004} and we can take $o_{l}(1)=0$. 

The proofs are given in \textit{supplementary materials}, Sections
4 and 5.

\subsubsection{Required conditions}

In order to proceed to asymptotic analysis, we introduce the asymptotic
counterpart of $(R_{\theta}^{d,l},\,h_{\theta}^{d,l}(Y))$, $\overline{X_{\theta,x_{0}}^{d,l}}$
and $S_{n}^{l}(Y;\theta,x_{0},U)$ when $n\longrightarrow\infty$,
or equivalently $\triangle\longrightarrow0$. In this asymptotic framework,
we have access to the true continuous signal $t\longrightarrow Y^{*}(t)=CX_{\theta^{*},x_{0}^{*}}(t)$
and so we can define the continuous cost:
\begin{equation}
\begin{array}{lll}
C_{T}^{l}(\theta,x_{0},u,U) & = & d^{'}\sigma^{2}+\int_{0}^{T}\left(\left\Vert CX_{\theta,x_{0},u}^{l}(t)-Y^{*}(t)\right\Vert _{2}^{2}+u(t)^{T}Uu(t)\right)dt\end{array}\label{eq:Continous_cost}
\end{equation}
with $X_{\theta,x_{0},u}^{l}$ the solution of the ODE
\begin{equation}
\left\{ \begin{array}{l}
\dot{X^{l}}_{\theta,x_{0},u}(t)=A_{\theta}(\overline{X_{\theta,x_{0}}^{l-1}}(t),t)X_{\theta,x_{0},u}^{l}(t)+Bu(t)\\
X_{\theta,x_{0},u}^{l}(0)=x_{0}.
\end{array}\right.\label{eq:extended_ODEmodel}
\end{equation}
As in the discrete case, we introduce the cost sequentially profiled
on $u$ , $S^{l}(\theta,x_{0},U)\,:=\inf_{u}C_{T}^{l}(\theta,x_{0},u,U)$,
and we derive a closed-form expression
\[
\begin{array}{lll}
S^{l}(\theta,x_{0},U) & = & x_{0}^{T}R_{\theta}^{l}(0)x_{0}+2x_{0}^{T}h_{\theta}^{l}(0)\\
  & + & \int_{0}^{T}\left(Y^{*}(t)^{T}Y^{*}(t)+d^{'}\sigma^{2}-h_{\theta}^{l}(t)^{T}BU^{-1}B^{T}h_{\theta}^{l}(t)\right)dt
 \end{array}
\]
by introducing $\left(R_{\theta}^{l},\,h_{\theta}^{l}\right)$, the
continuous time Riccati equation, and $\overline{X_{\theta,x_{0}}^{l}}$
the optimal trajectory given by the ODE
\begin{equation}
\left\{ \begin{array}{l}
\dot{R_{\theta}^{l}}(t)=-C^{T}C-A_{\theta}(\overline{X_{\theta}^{l-1}}(t),t)^{T}R_{\theta}^{l}(t)-R_{\theta}^{l}(t)A_{\theta}(\overline{X_{\theta}^{l-1}}(t),t)\\
+R_{\theta}^{l}(t)BU^{-1}B^{T}R_{\theta}^{l}(t)\\
\dot{h_{\theta}^{l}}(t)=C^{T}Y^{*}(t)-A_{\theta}(\overline{X_{\theta}^{l-1}}(t),t)^{T}h_{\theta}^{l}(t)+R_{\theta}^{l}(t)BU^{-1}B^{T}h_{\theta}^{l}(t)\\
\dot{\overline{X_{\theta}^{l}}(}t)=A_{\theta}(\overline{X_{\theta}^{l-1}}(t),t)\overline{X_{\theta}^{l}}(t)-BU^{-1}B^{T}(R_{\theta}^{l}(t)\overline{X_{\theta}^{l}}(t)+h_{\theta}^{l}(t))\\
\left(R_{\theta}^{l}(T),\,h_{\theta}^{l}(T),\,\overline{X_{\theta,x_{0}}^{l}}(0)\right)=\left(0_{d,d},\,0_{d,1},\,x_{0}\right).
\end{array}\right.\label{eq:continous_accurate_Riccati_equation}
\end{equation}
The previous functions correspond to $n\longrightarrow\infty$, but
for non-linear models we need to consider asymptotics in $l$ also,
so we introduce the asymptotic cost
\begin{equation}
\begin{array}{lll}
C_{T}^{\infty}(\theta,x_{0},u,U) & = & d^{'}\sigma^{2}+\int_{0}^{T}\left(\left\Vert CX_{\theta,x_{0},u}^{\infty}(t)-Y^{*}(t)\right\Vert _{2}^{2}+u(t)^{T}Uu(t)\right)dt\end{array}\label{eq:as_Continous_cost}
\end{equation}
associated with ODE
\begin{equation}
\left\{ \begin{array}{l}
\dot{X_{\theta,x_{0},u}^{\infty}}=A_{\theta}(\overline{X_{\theta,x_{0}}^{\infty}}(t),t)X_{\theta,x_{0},u}^{\infty}+Bu(t)\\
X_{\theta,x_{0},u}^{\infty}(0)=x_{0},
\end{array}\right.\label{eq:asymptotic_ODEmodel}
\end{equation}
together with the profiled cost value $S^{\infty}(\theta,x_{0},U)\,:=\inf_{u}C_{T}^{\infty}(\theta,x_{0},u,U).$
Again, thanks to the continuous LQ-theory, we characterize $S^{\infty}$
with $R_{\theta}^{l},\,h_{\theta}^{l},\,\overline{X_{\theta}^{l}}$,
the solution of (\ref{eq:continous_accurate_Riccati_equation}) where
$l$ and $l-1$ are replaced by $\infty$. Now that $S^{l},\,R_{\theta}^{l},\,h_{\theta}^{l}$
and $S^{\infty},\,R_{\theta}^{\infty},\,h_{\theta}^{\infty}$ have
been introduced, we can present the conditions required to derive
$\left(\widehat{\theta}^{T},\widehat{x_{0}}^{T}\right)$ consistency
with $\sqrt{n}$-convergence rate.

\begin{description}
	\item [{Condition}] 3: Matrix $B$ has independent columns.
\end{description}
\begin{description}
	\item [{Condition}] 4: The true parameters $\left(\theta^{*},x_{0}^{*}\right)$
	belong to the interior of $\varTheta\times\chi$.
\end{description}
\begin{description}
	\item [{Condition}] 5: The solution $X_{\theta,x_{0}}$ of (\ref{eq:ODEmodel})
	is such that if $CX_{\theta,x_{0}}(t)=CX_{\theta^{*},x_{0^{*}}}(t)$
	for all $t\in\left[0,\,T\right]$ then $\left(\theta,x_{0}\right)=\left(\theta^{*},x_{0^{*}}\right)$.
\end{description}
\begin{description}
	\item [{Condition}] 6: For all $x\in\varLambda$, $\theta\longmapsto A_{\theta}(x,.)$
	is differentiable on $\varTheta$, for all $\theta\in\varTheta$,
	$(x,t)\longmapsto\frac{\partial A_{\theta}(x,t)}{\partial\theta}$
	is continuous on $\varLambda\times\left[0,\,T\right]$.
\end{description}
\begin{description}
	\item [{Condition}] 7: For all $x\in\varLambda$, $\theta\longmapsto A_{\theta}(x,.)$
	is twice differentiable on $\varTheta$, for all $\theta\in\varTheta$,
	$(x,t)\longmapsto\frac{\partial^{2}A_{\theta}(x,t)}{\partial^{2}\theta}$
	is continuous on $\varLambda\times\left[0,\,T\right]$.
\end{description}
\begin{description}
	\item [{Condition}] 8: The asymptotic hessian matrix $\frac{\partial^{2}S^{\infty}(\theta^{*},x_{0}^{*})}{\partial^{2}\left(\theta,x_{0}\right)}$
	is nonsingular.
\end{description}
Condition 3 is required for the uniform convergence of $R_{\theta}^{l},\,h_{\theta}^{l}$
to $R_{\theta}^{\infty},\,h_{\theta}^{\infty}$ and $S^{l}$ to $S^{\infty}$.
Conditions 4 and 5 ensure $\left(\theta^{*},x_{0}^{*}\right)$ constitutes
a well-separated minimum of $S^{\infty}$ and conditions 6 to 8 guarantee
that $\frac{\partial^{2}S^{\infty}(\theta^{*},x_{0}^{*},U)}{\partial^{2}\left(\theta,x_{0}\right)}$
exists and the asymptotic variance-covariance of $\theta^{*}$ is
non singular.

\begin{rem}
	Condition 3 on $B$ ranking is not as restrictive as it seems. Let
	us assume $\textrm{rank}(B)=r'$ with $r'<d_{u}$ and consider the
	singular value decomposition
	\[
	B=V_{1}\left(\begin{array}{cccccc}
	\sigma_{1} & 0 & \cdots & \cdots & \cdots & 0\\
	0 & \ddots & \ddots &  &  & \vdots\\
	\vdots & \ddots & \sigma_{r'} & \ddots &  & \vdots\\
	\vdots &  & \ddots & 0 & \ddots & \vdots\\
	\vdots &  &  & \ddots & \ddots & 0\\
	0 & \cdots & \cdots & \cdots & 0 & 0
	\end{array}\right)V_{2}^{*}
	\]
	and the control reparametrization $v=V_{2}^{*}u$ with $V_{1}$ and
	$V_{2}$ two unitary matrices. The controlled ODE (\ref{eq:ControlledODEmodel})
	becomes:
	\begin{equation}
	\dot{x}(t)=A_{\theta}(x(t),t)x(t)+V_{1}\left(\begin{array}{c}
	\sigma_{1}v_{1}\\
	\vdots\\
	\sigma_{r'}v_{r'}\\
	0_{1,d_{u-r'}}
	\end{array}\right).\label{eq:ControlledODEmodel_rep}
	\end{equation}
	From this one see the components $\left(v_{r'+1},\ldots v_{d_{u}}\right)$
	do not affect the system. Hence, a degenerate matrix $B$ underlines
	the presence of useless controls and can be replaced by a new full
	rank matrix $B_{1}$ with fewer columns for the same result.
\end{rem}

\subsubsection{\label{subsec:Consistency}Consistency}

The estimator $\left(\widehat{\theta}^{T},\,\widehat{x_{0}}^{T}\right)$
is defined as an M-estimator, so in order to derive consistency we
need to show $S^{\infty}(\theta,x_{0},U)$ has a global well-separated
minimum at $\left(\theta,\,x_{0}\right)=\left(\theta^{*},\,x_{0}^{*}\right)$and
that $S_{n}^{l}(\theta,x_{0},U)$ converges uniformly to $S_{n}^{\infty}(Y;\theta,x_{0},U)$
on $\Theta\times\chi$.

This is the point of the next two propositions.
\begin{prop}
	\label{prop:identifiability_condition}Under conditions C1 to C5 then
	$\left(\theta^{*},\,x_{0}^{*}\right)$ is the unique global minimizer
	of $S^{\infty}(\theta,x_{0},U)$ on $\Theta\times\chi.$ 
\end{prop}
\begin{prop}
	\label{prop:uniform_boundedness_as}Under conditions C1 to C5 we have
	\[
	\sup_{\left(\theta,x_{0}\right)\in\Theta\times\chi}\left|S^{\infty}(\theta,x_{0},U)-S_{n}^{l}(Y;\theta,x_{0},U)\right|=o_{l}(1)+o_{p,n}(1).
	\]
\end{prop}
From this, we use Theorem 5.7 in \cite{Vaart1998} to conclude about
the consistency. 
\begin{thm}
	\label{thm:consistency} Under conditions C1 to C5, we have
	\[
	\left(\widehat{\theta}^{T},\,\widehat{x_{0}}^{T}\right)\longrightarrow\left(\theta^{*},\,x_{0}^{*}\right)\:\text{in probability}
	\]
	when $\left(l,n\right)\longrightarrow\infty$.
\end{thm}
\begin{rem}
	Interestingly, in \cite{BrunelClairon_Pontryagin2017}, for the weighting
	matrix under the form $U=\lambda I_{d_{u}}$, consistency proof for
	$\widehat{\theta}^{T}$ requires the lower bound condition $\lambda>\overline{\lambda_{1}}$
	with $\overline{\lambda_{1}}$ a positive model-dependent bound. Here,
	we just need to have $U$ positive definite. 
\end{rem}

\subsubsection{\label{subsec:Asymptotics-of}Asymptotic normality}

We show the asymptotic normality with $\sqrt{n}$-convergence rate
in two steps. First, we derive a linear asymptotic representation
of $\left(\widehat{\theta}^{T},\,\widehat{x_{0}}^{T}\right)-\left(\theta^{*},\,x_{0}^{*}\right)$
through a second order Taylor expansion of $\left(\theta,\,x_{0}\right)\longmapsto S_{n}^{l}(Y;\theta,x_{0},U)$,
Second, we approximate this linear asymptotic representation in order
to make explicit its dependence with respect to measurement noise.
\begin{prop}
	\label{prop:param_as_representation}Under conditions C1 to C7, we
	have: 
	\[
	-\nabla_{\left(\theta,x_{0}\right)}S_{n}^{l}(Y;\theta^{*},x_{0}^{*},U)=(\frac{\partial^{2}S^{\infty}(\theta^{*},x_{0}^{*},U)}{\partial^{2}\left(\theta,x_{0}\right)}+o_{p,n}(1)+o_{l}(1))\left(\widehat{\theta}^{T}-\theta^{*},\,\widehat{x_{0}}^{T}-x_{0}^{*}\right).
	\]
\end{prop}
\begin{prop}
	\label{prop:gradient_S_true_param_as_representation}Under conditions
	C1 to C7, we have
	\[
	\begin{array}{l}
	-\nabla_{\upsilon}S_{n}^{l}(Y;\theta^{*},x_{0}^{*})\\
	=\left(\triangle\sum_{j=0}^{n}\epsilon_{j}^{T}\right)(K_{\left(\theta^{*},x_{0}^{*}\right)}^{l}+o_{n}(1))+L\left(\triangle\sum_{j=0}^{n}\epsilon_{j}\right)+o_{p,n}(\sqrt{\Delta})+o_{l}(1)
	\end{array}
	\]
	with $K_{\left(\theta^{*},x_{0}^{*}\right)}^{l}=2CBU^{-1}B^{T}\int_{0}^{T}\frac{\partial h_{\theta^{*}}^{l}(t)}{\partial\left(\theta^{*},x_{0}^{*}\right)}dt$
	and $L=\left(\begin{array}{c}
	0_{p,d^{'}}\\
	-2C^{T}
	\end{array}\right).$
\end{prop}
From this, we recall the nonsingularity of $\frac{\partial^{2}S^{\infty}(\theta^{*},x_{0}^{*},U)}{\partial^{2}\left(\theta,x_{0}\right)}$
and the central limit theorem to obtain the following.
\begin{thm}
	\label{thm:Asymptotic_Normality}Under conditions C1 to C8 and if
	$l$ is such that $l=O_{n}(\sqrt{\Delta})$, then $(\widehat{\theta},\widehat{x_{0}})$
	is asymptotically normal and 
	\[
	(\widehat{\theta},\widehat{x_{0}})-\left(\theta^{*},x_{0}^{*}\right)=o_{p,n}(n^{-\frac{1}{2}}).
	\]
\end{thm}

\subsection{Asymptotic analysis of $\widehat{\theta}^{T,CI}$ for linear models
	in parametric case}

For the asymptotic analysis of $\widehat{\theta}^{T,CI}$, we restrict
to the linear models. Since $A_{\theta}$ does not depend then on
$x$, we have $C_{T}^{d,l_{1}}=C_{T}^{d,l_{2}}$ and $C_{T}^{l_{1}}=C_{T}^{l_{2}}=C_{T}^{\infty}$
for all $l_{1}$ and $l_{2}$ belonging to $\mathbb{N}$. Thus, there
is no need to consider asymptotics in $l$ and we drop the dependence
on $l$ in all quantities. The conditions we have derived for ensuring
$\widehat{\theta}^{T,CI}$ consistency with $\sqrt{n}$-convergence
rate are shown below.
\begin{description}
	\item [{Condition}] L1: For all $\theta\in\varTheta$, $t\longmapsto A_{\theta}(t)$
	is differentiable on $\left[0,\,T\right]$.
\end{description}
\begin{description}
	\item [{Condition}] L2: $\theta\longmapsto A_{\theta}$ is continuous on
	$\varTheta$.
\end{description}
\begin{description}
	\item [{Condition}] L3: For all $\theta\in\varTheta$, $R_{\theta}(0)$
	is nonsingular, where $R_{\theta}$ is defined by ODE (\ref{eq:continous_accurate_Riccati_equation}).
\end{description}
\begin{description}
	\item [{Condition}] L4: The true parameter $\theta^{*}$ belongs to the
	interior of $\varTheta$.
\end{description}
\begin{description}
	\item [{Condition}] L5: The solution $X_{\theta,x_{0}}$ of (\ref{eq:ODEmodel})
	is such that if $CX_{\theta,x_{0}}(t)=CX_{\theta^{*},x_{0^{*}}}(t)$
	for all $t\in\left[0,\,T\right]$ then $\left(\theta,x_{0}\right)=\left(\theta^{*},x_{0^{*}}\right)$.
\end{description}
\begin{description}
	\item [{Condition}] L6: $\theta\longmapsto A_{\theta}$ is $C^{2}$ on
	$\varTheta$.
\end{description}
\begin{description}
	\item [{Condition}] L7: The asymptotic hessian matrix $\frac{\partial^{2}S^{CI}(\theta^{*})}{\partial^{2}\theta}$
	is a nonsingular matrix.
\end{description}
The proofs follow the same steps as in the previous sections, hence
we just present the theorems. The proofs are also detailed in \textit{supplementary
	materials}, Sections 4 and 5.
\begin{thm}
	\label{thm:Consistency_theorem_prof_ci}Under conditions LC1 to LC5,
	we have $\widehat{\theta}^{T,CI}\longrightarrow\theta^{*}$ in probability
	when $n\longrightarrow\infty$.
\end{thm}
\begin{thm}
	\label{thm:Asymptotic_Normality_prof_ci}Under conditions LC1 to LC7,
	$\widehat{\theta}^{T,CI}$ is asymptotically normal and $\widehat{\theta}^{T,CI}-\theta^{*}=o_{p,n}(n^{-\frac{1}{2}})$.
\end{thm}
As in the discrete case, Condition L3 is an ad-hoc hypothesis, but
here again we can derive a necessary and sufficient testable condition
ensuring $R_{\theta}(0)$ nonsingularity. 
\begin{prop}
	\label{prop:continuous_R0_invertibility_criteria} Given $\theta\in\Theta$
	, $R_{\theta}(0)$ is invertible if and only if:
	
	1) the matrix
	\begin{equation}
	O_{\theta}(T)=\int_{0}^{T}\left(C\Phi_{\theta}(t,0)\right)^{T}C\Phi_{\theta}(t,0)dt\label{eq:continuous_invertibility_criteria}
	\end{equation}
	is invertible, where $\Phi_{\theta}$ is the resolvant of (\ref{eq:extended_ODEmodel}),
	and
	
	2) the following holds:
	\begin{equation}
	\left\Vert CX_{\theta,x_{0}^{1}}-CX_{\theta,x_{0}^{2}}\right\Vert _{L^{2}}^{2}=0\Longrightarrow x_{0}^{1}=x_{0}^{2}.\label{eq:second_continuous_invertibility_criteria}
	\end{equation}
\end{prop}
Interestingly, it is again equivalent to the notion of observability
in control theory, but now for a continuous model \cite{Sontag1998}.
\begin{rem}
	The difficulty in deriving the asymptotic behavior of $\widehat{\theta}^{T,CI}$
	in all generality comes from the initialisation point $x_{0}^{r}$
	required by the algorithm. So far, we have been unable to analyze
	the mapping $Q_{\theta}\,:\,x_{0}^{r}\longmapsto\overline{X_{\theta}}(.,\,x_{0}^{r})$
	where $\overline{X_{\theta}}(.,\,x_{0}^{r})$ is the trajectory given
	by the algorithm in the limit case $n=\infty$ and $l=\infty$. If
	for $\theta=\theta^{*}$, the true trajectory $X^{*}$ is a global
	attractor of $Q_{\theta^{*}}$, the demonstrations will be completed,
	but our attempts to prove it remain unfruitful.
\end{rem}

\section{\label{sec:Experiment}Experiments}

We use Monte-Carlo simulations on different models, for several sample
sizes $n$ and measurement noise variances $\sigma^{2}$. We compare
four estimators: the ones presented here $\hat{\theta}^{T}$ and $\hat{\theta}^{T,CI}$
, the classic nonlinear least square (NLS) estimator $\widehat{\theta}^{NLS}$
and the generalized profiling (GP) estimator $\widehat{\theta}^{GP}$
introduced in \cite{Ramsay2007}. The latter is the regularization
method of reference for the estimation problem in ODEs. As we have
said in the Introduction, there are two main problems for parameter
estimation, practical identifiability issues possibly due to the sloppiness
phenomenon, and estimation inaccuracy due to model misspecifications.
Thus, we compare $\hat{\theta}^{T}$, $\hat{\theta}^{T,CI}$, $\widehat{\theta}^{NLS}$
, $\widehat{\theta}^{GP}$ on models facing practical identifiability
problems in correctly and misspecified frameworks. For a given choice
of $\left(n,\sigma\right)$, we compute the following by Monte Carlo
based on $N_{MC}=100$.
\begin{enumerate}
	\item The variance $V(\widehat{\theta_{i}})$ for each element $\theta_{i}$
	of $\theta$ to analyze how each estimator behaves specifically for
	the components suffering from identifiability issues.
	\item The estimator variance-covariance norm $\left\Vert V\left(\widehat{\theta}\right)\right\Vert _{2}$
	to analyze how each estimator behaves for the whole parameter set.
	\item The componentwise mean square error $\textrm{M}(\widehat{\theta_{i}})=\left|\theta_{i}^{*}-\widehat{\mathbb{E}}\left[\widehat{\theta_{i}}\right]\right|^{2}+V(\widehat{\theta_{i}})$
	and the global $\textrm{M}(\widehat{\theta})=\sum_{i=1}^{p}\left|\theta_{i}^{*}-\widehat{\mathbb{E}}\left[\widehat{\theta_{i}}\right]\right|^{2}+\left\Vert V\left(\widehat{\theta}\right)\right\Vert _{2}$
	to measure estimator accuracy, in particular its degradation when
	facing misspecification. 
\end{enumerate}
Since model parameters can have different orders of magnitude, the
results will be given for normalized estimated values $\widehat{\theta}./\theta^{*}$.
Here the division has to be understood componentwise.

For each run, the observations are obtained by integrating the ODE
with a Runge-Kutta algorithm (ode45 in Matlab), with added centered
Gaussian noise of variance $\sigma^{2}$.

The GP method uses an approximate solution $\tilde{X}_{\theta}^{\lambda}$
of the ODE defined as the spline basis decomposition minimizing $\sum_{i=1}^{n}\left\Vert y_{i}-C\tilde{X}_{\theta}^{\lambda}(t_{i})\right\Vert ^{2}+\lambda\left\Vert \frac{d}{dt}\tilde{X}_{\theta}^{\lambda}-f\left(\cdot,\tilde{X}_{\theta}^{\lambda},\theta\right)\right\Vert _{L^{2}}^{2}$
. GP requires a selection method for both the knots location and the
hyperparameter $\lambda$ which has a similar role as in our method.
The knots are placed on the observations and $\lambda$ is selected
by using the method presented in \cite{Campbell2011,QiZhao2010}:
the value of $\lambda$ is increased until $\left\Vert X_{\hat{\theta}_{\lambda}^{GP},\widetilde{x_{0}}}-\tilde{X}_{\hat{\theta}_{\lambda}^{GP}}\right\Vert _{L^{2}}^{2}$
starts increasing, that is when $\tilde{X}_{\hat{\theta}_{\lambda}^{GP}}$
starts to differ significantly from the exact solution $X_{\hat{\theta}_{\lambda}^{GP},\widetilde{x_{0}}}$
where $\widetilde{x_{0}}=\tilde{X}_{\hat{\theta}_{\lambda}^{GP}}^{\lambda}(0)$.

For $\widehat{\theta}^{T}$ and $\widehat{\theta}^{T,CI}$, we need
to select both the discretization grid $\left\{ t_{j}^{d}\right\} _{0\leq j\leq m}$
and the matrix $U.$ For the grid, we take $m=k_{n}n$ points and
we place uniformly $k_{n}$ discretization points between two consecutive
observation times. As in \cite{BrunelClairon_Pontryagin2017}, for
$U$ we consider scalar matrices $U=\lambda I_{d}$ . The $\left(k_{n},\lambda\right)$
selection is done by minimizing the forward cross-validation method
presented in \cite{Hooker2011} among a model-specific trial of values.
Let us denote $\hat{\theta}_{k_{n},\lambda}^{T}$ and $\hat{\theta}_{k_{n},\lambda}^{T,CI}$
, the tracking estimator obtained for a given $\left(k_{n},\lambda\right)$
value. We split $\left[0,T\right]$ into $H$ subintervals $\left[t_{h},\,t_{h+1}\right]$,
such that $t_{1}=0$ and $t_{H}=T$ and we denote $X_{\theta}(.,t_{h},x_{h})$
the solution of:
\begin{equation}
\left\{ \begin{array}{l}
\dot{x}(t)=f\left(t,x(t),\theta\right)\\
x(t_{h})=x_{h}
\end{array}\right.\label{eq:ODEmodel-gen-1}
\end{equation}
defined on the interval $\left[t_{h},\,t_{h+1}\right]$. The forward
cross-validation uses the causal relation imposed on the data by the
ODE to quantify the prediction error caused by $\hat{\theta}_{k_{n},\lambda}$
(equal to $\hat{\theta}_{k_{n},\lambda}^{T}$ or $\hat{\theta}_{k_{n},\lambda}^{T,CI}$):
\[
\textrm{EP}(k_{n},\lambda)=\sum_{h=1}^{H}\sum_{\left\{ t_{i}\in\left[t_{h},\,t_{h+1}\right]\right\} }\left\Vert y_{i}-CX_{\hat{\theta}_{k_{n},\lambda}}(t_{i},t_{h},X_{\hat{\theta}_{k_{n},\lambda},\bar{u}_{\hat{\theta}_{k_{n},\lambda},}}(t_{h}))\right\Vert _{2}^{2}.
\]
Here, we choose $H=2$ subintervals and we now denote by $\hat{\theta}^{T}$
and $\hat{\theta}^{T,CI}$ the estimators corresponding to the value
$\left(k_{n},\lambda\right)$ minimizing EP.

\subsection{$\alpha-$Pinene model}

We begin with a linear ODE considered in \cite{Rodriguez2006} and
used for modeling the isomerization of $\alpha-$Pinene. This is
\begin{equation}
\left\{ \begin{array}{l}
\dot{x}_{1}=-(\theta_{1}+\theta_{2})x_{1}\\
\dot{x}_{2}=\theta_{1}x_{1}\\
\dot{x}_{3}=\theta_{2}x_{1}-(\theta_{3}+\theta_{4})x_{3}+\theta_{5}x_{5}\\
\dot{x}_{4}=\theta_{3}x_{3}\\
\dot{x}_{5}=\theta_{4}x_{3}-\theta_{5}x_{5}
\end{array}\right.\label{eq:apinene_wellpse}
\end{equation}
on the observation interval $\left[0,\,T\right]=\left[0,\,100\right]$.
Here
\[
A_{\theta}(t)=\left(\begin{array}{ccccc}
-(\theta_{1}+\theta_{2}) & 0 & 0 & 0 & 0\\
\theta_{1} & 0 & 0 & 0 & 0\\
\theta_{2} & 0 & -(\theta_{3}+\theta_{4}) & 0 & \theta_{5}\\
0 & 0 & \theta_{3} & 0 & 0\\
0 & 0 & \theta_{4} & 0 & -\theta_{5}
\end{array}\right).
\]
The initial condition is $x_{0}^{*}=(100,\,0,\,0,\,0,\,0)$ and the
true parameter value is $\theta^{*}=\left(5.93\,,2.96,\,2.05,\,27.5,\,4\right)\times10^{-2}$.
We plot in Figure \ref{fig:apinene_example} the solution of (\ref{eq:apinene_wellpse}) corresponding to $\theta^{*}$
and an example of simulated observations.

\begin{figure}
	\centering
\includegraphics[scale=0.5]{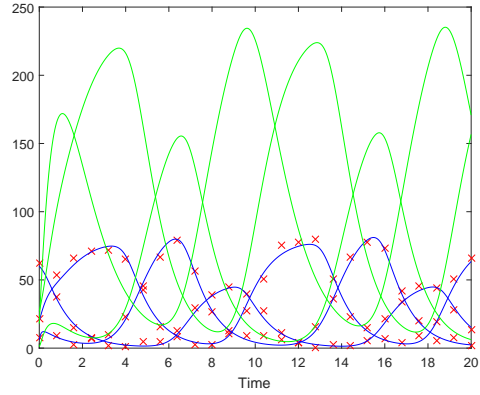}\caption{\label{fig:apinene_example}Solution of (\ref{eq:apinene_wellpse}) (blue) and corresponding noisy observations for $\sigma=2.5$ (red).}
\end{figure}

In \cite{Rodriguez2006}, model (\ref{eq:apinene_wellpse}) is used
as a benchmark estimation comparison as many approaches fail to converge
due to the difficulty of estimating $\theta_{4}^{*}$ and $\theta_{5}^{*}$
because of the high correlation between them. This is confirmed by
the eigendecomposition of $V^{*}D^{*}\left(V^{*}\right)^{-1}=\mathcal{I}_{n}\left(\theta^{*},x_{0}^{*}\right)$
for $n=10$, $V^{*}$ being the matrix composed of the eigenvectors
and $D^{*}$ the diagonal matrix containing the eigenvalues $\left\{ \mu_{l}\right\} _{1\leq l\leq10}$.
The lowest eigenvalues correspond almost exclusively to $x_{0}^{*}$
and $\left(\theta_{4}^{*},\,\theta_{5}^{*}\right)$ when one analyses
the corresponding eigenvectors. We select $\left(k_{n},\lambda\right)$
among the set $\left\{ 30,\,40,\,50\right\} \times\left\{ 10^{i},\,5\times10^{i}\right\} _{-3\leq i\leq3}.$

\paragraph{Influence of measurement noise}

We consider one sample size $n=10$ and three levels of measurement
noise ($\sigma=2.5$, $\sigma=5$ and $\sigma=10$). Results are presented
in Tables \ref{tab:var_a_pinene} and \ref{tab:mse_a_pinene}.

\begin{table}
	\caption{\label{tab:var_a_pinene} Scaled variance for $\alpha-$Pinene model.}
	\begin{tabular}{|c|c|c|c|c|c|c|c|}
		\hline 
		$\sigma$ & $\times10^{-2}$ & $V(\widehat{\theta_{1}})$ & $V(\widehat{\theta_{2}})$ & $V(\widehat{\theta_{3}})$ & $V(\widehat{\theta_{4}})$ & $V(\widehat{\theta_{5}})$ & \multicolumn{1}{c|}{$\left\Vert V\left(\widehat{\theta}\right)\right\Vert _{2}$}\tabularnewline
		\hline 
		\hline 
		\multirow{4}{*}{$2.5$} & $\widehat{\theta}^{T,CI}$& 0.06 & 0.12 & 0.30 & 0.20 & 0.29 & 0.65\tabularnewline
		\cline{2-8} 
		& $\widehat{\theta}^{T}$ & 0.06 & 0.12 & 0.30 & 0.25 & 0.37 & 0.79\tabularnewline
		\cline{2-8} 
		& $\widehat{\theta}^{NLS}$ & 0.05 & 0.10 & 0.12 & 0.34 & 0.48 & 0.89\tabularnewline
		\cline{2-8} 
		& $\widehat{\theta}^{GP}$ & 0.03 & 0.03 & 0.04 & 0.35 & 0.48 & 0.74\tabularnewline
		\hline 
		\hline 
		\multirow{4}{*}{$5$} & $\widehat{\theta}^{T,CI}$ & 0.27 & 0.44 & 0.98 & 0.82 & 1.21 & 2.27\tabularnewline
		\cline{2-8} 
		& $\widehat{\theta}^{T}$ & 0.28 & 0.44 & 0.99 & 1.02 & 1.54 & 2.84\tabularnewline
		\cline{2-8} 
		& $\widehat{\theta}^{NLS}$& 0.25 & 0.43 & 0.38 & 1.26 & 1.83 & 3.19\tabularnewline
		\cline{2-8} 
		& $\widehat{\theta}^{GP}$& 0.17 & 0.15 & 0.13 & 1.27 & 1.63 & 2.92\tabularnewline
		\hline 
		\hline 
		\multirow{4}{*}{$10$} & $\widehat{\theta}^{T,CI}$& 0.61 & 1.22 & 2.63 & 3.40 & 5.36 & 9.04\tabularnewline
		\cline{2-8} 
		& $\widehat{\theta}^{T}$ & 0.62 & 1.22 & 2.64 & 4.08 & 6.21 & 10.7\tabularnewline
		\cline{2-8} 
		& $\widehat{\theta}^{NLS}$& 0.60 & 1.49 & 1.37 & 5.48 & 8.37 & 14.3\tabularnewline
		\cline{2-8} 
		& $\widehat{\theta}^{GP}$ & 0.42 & 0.69 & 0.51 & 4.92 & 8.37 & 11.8\tabularnewline
		\hline 
\end{tabular}
\end{table}

\begin{table}
	\caption{Scaled mean square error for  $\alpha-$Pinene model. \label{tab:mse_a_pinene}}
		\begin{tabular}{|c|c|c|c|c|c|c|c|}
			\hline 
			$\sigma$& $\times10^{-2}$ & $\textrm{M}(\widehat{\theta_{1}})$ & $\textrm{M}(\widehat{\theta_{2}})$ & $\textrm{M}(\widehat{\theta_{3}})$ & $\textrm{M}(\widehat{\theta_{4}})$ & $\textrm{M}(\widehat{\theta_{5}})$ & \multicolumn{1}{c|}{$\textrm{M}(\widehat{\theta})$}\tabularnewline
			\hline 
			\hline 
			\multirow{4}{*}{$2.5$} & $\widehat{\theta}^{T,CI}$ & 0.09 & 0.13 & 0.32 & 0.21 & 0.30 & 0.74\tabularnewline
			\cline{2-8} 
			& $\widehat{\theta}^{T}$& 0.08 & 0.14 & 0.32 & 0.26 & 0.39 & 0.88\tabularnewline
			\cline{2-8} 
			& $\widehat{\theta}^{NLS}$ & 0.05 & 0.10 & 0.12 & 0.34 & 0.48 & 0.90\tabularnewline
			\cline{2-8} 
			& $\widehat{\theta}^{GP}$ & 0.64 & 1.57 & 0.04 & 12.62 & 15.3 & 30.1\tabularnewline
			\hline 
			\hline 
			\multirow{4}{*}{$5$} & $\widehat{\theta}^{T,CI}$ & 0.37 & 0.50 & 1.01 & 0.82 & 1.22 & 2.46\tabularnewline
			\cline{2-8} 
			&  $\widehat{\theta}^{T}$& 0.37 & 0.44 & 0.99 & 1.02 & 1.54 & 3.02\tabularnewline
			\cline{2-8} 
			& $\widehat{\theta}^{NLS}$ & 0.25 & 0.43 & 0.38 & 1.27 & 1.85 & 3.22\tabularnewline
			\cline{2-8} 
			&$\widehat{\theta}^{GP}$ & 0.80 & 1.75 & 0.14 & 14.0 & 17.7 & 34.0\tabularnewline
			\hline 
			\hline 
			\multirow{4}{*}{$10$}  & $\widehat{\theta}^{T,CI}$ & 0.80 & 1.30 & 2.67 & 3.43 & 5.36 & 9.04\tabularnewline
			\cline{2-8} 
			& $\widehat{\theta}^{T}$ & 0.81 & 1.31 & 2.66 & 4.08 & 6.22 & 10.9\tabularnewline
			\cline{2-8} 
			& $\widehat{\theta}^{NLS}$ & 0.61 & 1.53 & 1.40 & 5.49 & 8.39 & 14.3\tabularnewline
			\cline{2-8} 
			& $\widehat{\theta}^{GP}$ & 1.10 & 1.93 & 0.52 & 17.7 & 23.2 & 42.8\tabularnewline
			\hline 
	\end{tabular}
\end{table}

 For $\theta_{4}$ and $\theta_{5}$, we observe that $\widehat{\theta}^{T}$
and $\widehat{\theta}^{T,CI}$ give the smallest variance followed
by $\widehat{\theta}^{GP}$ for $\sigma=5$ and $\sigma=10$. These
approximate methods manage to regularize the estimation of parameters
facing a practical identifiability problem in comparison with NLS.
Moreover, we notice the same pattern for $\left\Vert V\left(\theta\right)\right\Vert _{2}$
which takes into account covariance among parameters. However, TE
and GP are methods based on approximated solutions and so are likely
to produced biased estimates. That is why we estimated the mean square error to verify
that the price to pay to decrease the variance is not too high in
terms of bias. Our methods have lower global mean square error than NLS which indicate
a reasonable bias. GP on the other hand can have a very large mean square error.
The reason, already been discussed in \cite{BrunelClairon_Pontryagin2017,BrunelClairon_Riccati2014},
is linked to the limited ability of $\tilde{X}_{\theta}^{\lambda}$
to approach the true solution. This is contrary to our method where
the mesh size can be arbitrarily small and thus $X_{\theta,u}^{d}$
can be arbitrarily close to the original ODE model. 

\paragraph{Influence of model misspecification }

We still consider the sample size $n=10$ with one level of measurement
noise $\sigma=2.5$. However, the observations are now generated by
using the stochastically perturbed model:
\begin{equation}
\left\{ \begin{array}{l}
dx_{1}=-(\theta_{1}+\theta_{2})x_{1}dt+c_{t}x_{1}dt\\
dx_{2}=\theta_{1}x_{1}dt+c_{t}x_{2}dt\\
dx_{3}=\left(\theta_{2}x_{1}-(\theta_{3}+\theta_{4})x_{3}+\theta_{5}x_{5}\right)dt+c_{t}x_{3}dt\\
dx_{4}=\theta_{3}x_{3}dt+c_{t}x_{4}dt\\
dx_{5}=\left(\theta_{4}x_{3}-\theta_{5}x_{5}\right)dt+c_{t}x_{5}dt
\end{array}\right.\label{eq:apinene_misspse}
\end{equation}
with $c_{t}\sim N(0,\,\sigma_{c}^{2})$, we still estimate $\theta^{*}$
by using model (\ref{eq:apinene_misspse}), which is now a deterministic
approximation of the true process. We plot in Figure \ref{fig:apinene_example_misspe} the solution of (\ref{eq:apinene_wellpse})
	and one realization of (\ref{eq:apinene_misspse}) for the sake of comparison. 
\begin{figure}
\centering
\includegraphics[scale=0.5]{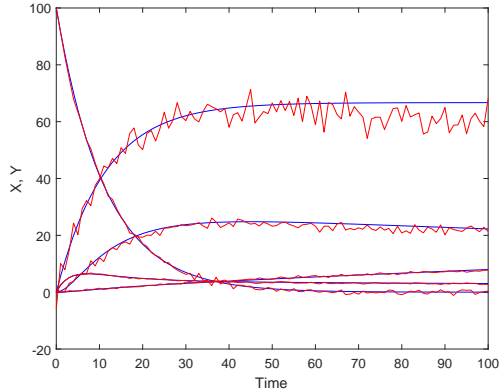}\caption{\label{fig:apinene_example_misspe}Solution of (\ref{eq:apinene_wellpse})
		(blue) and a realization of (\ref{eq:apinene_misspse}) (red) for
		$\sigma_{c}^{2}=0.002$.}

\end{figure}
This experimental design has been chosen to mimic a real case of
data analysis for chemical processes where the deterministic reaction
rate equations are used as an approximation of stochastic differential
equations \cite{Gillespie2000}. We study the effect of misspecification
by varying the value of $\sigma_{c}^{2}$ and results are presented
in Tables \ref{tab:var_a_pinene_misspe} and \ref{tab:mse_a_pinene_misspe}.

\begin{table}
	\caption{Scaled variance for  misspecified $\alpha-$Pinene model. \label{tab:var_a_pinene_misspe}}
	\begin{tabular}{|c|c|c|c|c|c|c|c|}
		\hline 
		$\sigma_{c}^{2}$ & $\times10^{-2}$ & $V(\widehat{\theta_{1}})$ & $V(\widehat{\theta_{2}})$ & $V(\widehat{\theta_{3}})$ & $V(\widehat{\theta_{4}})$ & $V(\widehat{\theta_{5}})$ & \multicolumn{1}{c|}{$\left\Vert V\left(\widehat{\theta}\right)\right\Vert _{2}$}\tabularnewline
		\hline 
		\hline 
		\multirow{4}{*}{$0.002$} &$\widehat{\theta}^{T,CI}$ & 0.09 & 0.25 & 1.19 & 0.11 & 0.18 & 1.41\tabularnewline
		\cline{2-8} 
		& $\widehat{\theta}^{T}$& 0.09 & 0.29 & 1.32 & 0.13 & 0.22 & 1.61\tabularnewline
		\cline{2-8} 
		& $\widehat{\theta}^{NLS}$& 0.06 & 0.61 & 0.87 & 0.87 & 1.62 & 3.19\tabularnewline
		\cline{2-8} 
		&$\widehat{\theta}^{GP}$ & 0.07 & 0.16 & 0.35 & 0.48 & 0.61 & 1.16\tabularnewline
		\hline 
		\hline 
	\multirow{4}{*}{$0.004$} & $\widehat{\theta}^{T,CI}$ & 0.12 & 0.43 & 2.17 & 0.15 & 0.27 & 2.51\tabularnewline
		\cline{2-8} 
		& $\widehat{\theta}^{T}$ & 0.12 & 0.42 & 2.17 & 0.18 & 0.29 & 2.48\tabularnewline
		\cline{2-8} 
		& $\widehat{\theta}^{NLS}$ & 0.12 & 1.78 & 1.99 & 1.52 & 2.93 & 6.95\tabularnewline
		\cline{2-8} 
		& $\widehat{\theta}^{GP}$ & 0.12 & 0.38 & 0.51 & 0.94 & 1.26 & 2.45\tabularnewline
		\hline 
		\hline 
		\multirow{4}{*}{$0.006$} & $\widehat{\theta}^{T,CI}$ & 0.13 & 0.40 & 2.66 & 0.15 & 0.25 & 2.92\tabularnewline
		\cline{2-8} 
		& $\widehat{\theta}^{T}$ & 0.12 & 0.41 & 2.45 & 0.19 & 0.30 & 2.83\tabularnewline
		\cline{2-8} 
		& $\widehat{\theta}^{NLS}$& 0.15 & 2.21 & 2.13 & 1.78 & 3.64 & 7.71\tabularnewline
		\cline{2-8} 
		& $\widehat{\theta}^{GP}$ & 0.17 & 0.44 & 0.70 & 1.12 & 1.49 & 2.87\tabularnewline
		\hline 
\end{tabular}
\end{table}

\begin{table}
	\caption{Scaled mean square error for misspecified $\alpha-$Pinene model. \label{tab:mse_a_pinene_misspe}}
\small
		\begin{tabular}{|c|c|c|c|c|c|c|c|}
			\hline 
			$\sigma_{c}^{2}$ & $\times10^{-2}$ & $\textrm{M}(\widehat{\theta_{1}})$ & $\textrm{M}(\widehat{\theta_{2}})$ & $\textrm{M}(\widehat{\theta_{3}})$ & $\textrm{M}(\widehat{\theta_{4}})$ & $\textrm{M}(\widehat{\theta_{5}})$ & \multicolumn{1}{c|}{$\textrm{M}(\widehat{\theta})$}\tabularnewline
			\hline 
			\hline 
			\multirow{4}{*}{$0.002$} & $\widehat{\theta}^{T,CI}$ & 0.10 & 0.29 & 1.20 & 0.11 & 0.18 & 1.47\tabularnewline
			\cline{2-8} 
			& $\widehat{\theta}^{T}$ & 0.09 & 0.29 & 1.32 & 0.13 & 0.22 & 1.71\tabularnewline
			\cline{2-8} 
			& $\widehat{\theta}^{NLS}$ & 0.06 & 0.64 & 0.87 & 0.87 & 1.62 & 3.22\tabularnewline
			\cline{2-8} 
			& $\widehat{\theta}^{GP}$ & 0.06 & 1.79 & 0.35 & 12.9 & 16.1 & 31.3\tabularnewline
			\hline 
			\hline 
			\multirow{4}{*}{$0.004$} & $\widehat{\theta}^{T,CI}$ & 0.14 & 0.45 & 2.17 & 0.15 & 0.27 & 2.55\tabularnewline
			\cline{2-8} 
			& $\widehat{\theta}^{T}$ & 0.15 & 0.43 & 2.17 & 0.18 & 0.29 & 2.52\tabularnewline
			\cline{2-8} 
			& $\widehat{\theta}^{NLS}$ & 0.12 & 1.78 & 2.00 & 1.52 & 2.94 & 6.97\tabularnewline
			\cline{2-8} 
			& $\widehat{\theta}^{GP}$ & 0.65 & 1.82 & 0.53 & 13.6 & 16.8 & 32.7\tabularnewline
			\hline 
			\hline 
			\multirow{4}{*}{$0.006$} & $\widehat{\theta}^{T,CI}$ & 0.14 & 0.40 & 2.74 & 0.15 & 0.25 & 3.02\tabularnewline
			\cline{2-8} 
			& $\widehat{\theta}^{T}$ & 0.16 & 0.42 & 2.67 & 0.17 & 0.31 & 2.97\tabularnewline
			\cline{2-8} 
			& $\widehat{\theta}^{NLS}$& 0.15 & 2.24 & 2.17 & 1.82 & 3.74 & 7.93\tabularnewline
			\cline{2-8} 
			& $\widehat{\theta}^{GP}$ & 0.63 & 1.99 & 0.71 & 14.0 & 17.4 & 33.7\tabularnewline
			\hline 

	\end{tabular}
\end{table}
 Here, the approximated methods manage to efficiently reduce the variance.
However, only the TE estimators maintain a reasonable bias. 

\subsection{\label{subsec:FitzHugh-Nagumo}FitzHugh-Nagumo}

We now consider the FitzHugh-Nagumo model
\begin{equation}
\left\{ \begin{array}{lll}
\dot{V} & = & c\left(V-\frac{V^{3}}{3}+R\right)\\
\dot{R} & = & -\frac{1}{c}\left(V-a+bR\right),
\end{array}\right.\label{eq:FHNODE}
\end{equation}
which is a nonlinear ODE introduced to study the membrane potential
evolution of neurons \cite{Fitzhugh1961}. Here $V$ is the neuron
membrane potential, $R$ the synaptic conductance and we consider
the partial observation framework where only $V$ is observed on $\left[0,\,T\right]=\left[0,\,20\right]$.
We take the parameter and initial condition values $a^{*}=b^{*}=0.2$,
$c^{*}=3$ and $x_{0}^{*}=(V_{0}^{*},\,R_{0}^{*})=(-1,1)$ given by
\cite{Ramsay2007}. The related ODE solution exhibits a periodic behavior,
as seen in Figure \ref{fig:fhn_example}.
\begin{figure}
	\centering
\includegraphics[scale=0.5]{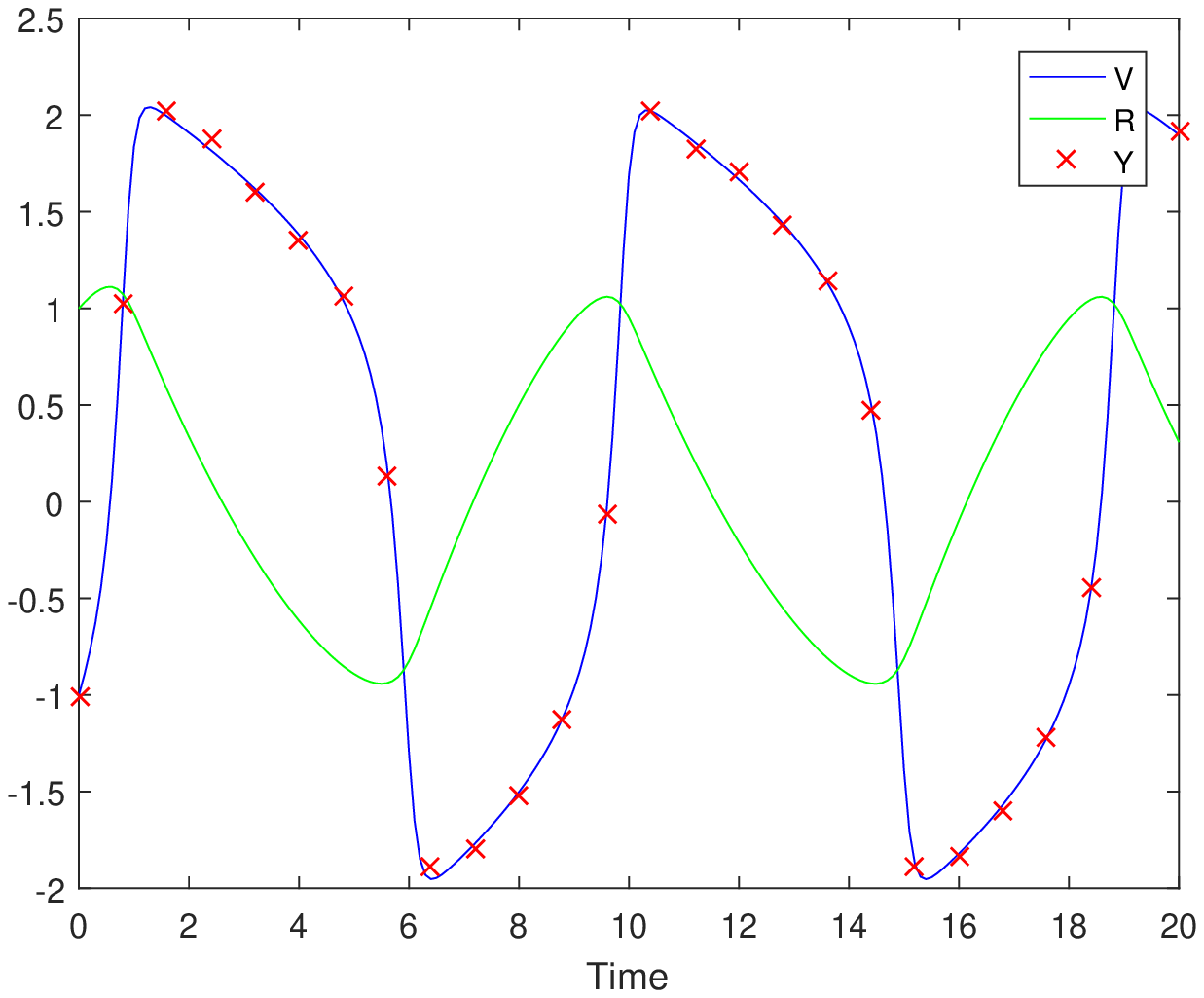}\caption{\label{fig:fhn_example}Solution of (\ref{eq:FHNODE}) and corresponding
			noisy observations for $\sigma=0.03$ (red).}
\end{figure}
Here, we choose
\[
A_{\theta}(x,t)=\left(\begin{array}{ccc}
c\left(1-V^{2}/3\right) & cR & 0\\
-\frac{1}{c} & -b & \frac{a}{c}\\
0 & 0 & 0
\end{array}\right)
\]
where a third constant state variable $Z=1$ is added to absorb the
exogenous term $a/c$. This model has been chosen as it contains a
parameter which is poorly identifiable under NLS, namely $b^{*}$.
As in the $\alpha$-Pinene case, we confirm this by eigendecomposition
of $\mathcal{I}_{n}\left(\theta^{*},x_{0}^{*}\right)$ for $n=25$,
which gives the eigenvectors 
\[
V^{*}\simeq\left(\begin{array}{ccccc}
0.04 & 0.08 & 0.56 & 0.12 & -0.82\\
-0.97 & 0.10 & -0.12 & -0.09 & -0.14\\
0.19 & 0.28 & -0.80 & 0.16 & -0.48\\
-0.04 & 0.76 & 0.20 & 0.54 & 0.29\\
-0.13 & -0.57 & -0.07 & 0.81 & 0.01
\end{array}\right)
\]
and the corresponding eigenvalues $\left\{ \mu_{l}\right\} _{1\leq l\leq5}=\left\{ 0.9,\,2.2,\,15.7,\,44.8,\,256.6\right\} $.
The lowest descend direction for the NLS criteria is nearly parallel
to the $b$ direction in the parameter space. We select $\left(k_{n},\lambda\right)$
among the set $\left\{ 50,\,60,\,70\right\} \times\left\{ 10^{i},\,5\times10^{i}\right\} _{-3\leq i\leq-1}.$

\paragraph{Influence of measurement noise}

We take $n=25$ and consider three levels of measurement noise ($\sigma=0.01$,
$\sigma=0.03$ and $\sigma=0.05$). Results are presented in Table
\ref{tab:var_fhn}.

\begin{table}
	\caption{Scaled variance and mean square error for FHN model. \label{tab:var_fhn}}
		\begin{tabular}{|c|c|c|c|c|c|c|c|c|c|c|}
			\hline 
			$\sigma$ & $\times10^{-2}$ & $V(\widehat{a})$ & $V(\widehat{b})$ & $V(\widehat{c})$ & $\left\Vert V\left(\theta\right)\right\Vert _{2}$ & 
			
			& $\textrm{M}(\widehat{a})$ & $\textrm{M}(\widehat{b})$ & $\textrm{M}(\widehat{c})$ & $\textrm{M}(\widehat{\theta})$ \tabularnewline
			\hline 
			\hline 
			\multirow{4}{*}{$0.01$} & $\widehat{\theta}^{T,CI}$& 0.03 & 0.70 & 0.01 & 0.70 & 
			
			& 0.03 & 0.79 & 0.98 & 1.77\tabularnewline
			\cline{2-11} 
			& $\widehat{\theta}^{T}$& 0.08 & 0.72 & 0.01 & 0.74 & %
		
			& 0.08 & 1.04 & 0.87 & 1.93\tabularnewline
			\cline{2-11} 
			& $\widehat{\theta}^{NLS}$ & 4.50 & 26.8 & 0.25 & 44.4 & %
			
			& 8.64 & 40.2 & 0.28 & 44.3\tabularnewline
			\cline{2-11} 
			& $\widehat{\theta}^{GP}$ & 0.01 & 0.52 & 0.01 & 0.52 & %
			
			& 1.58 & 49.9 & 2.90 & 54.3\tabularnewline
			\hline 
			\hline 
			\multirow{4}{*}{$0.03$} & $\widehat{\theta}^{T,CI}$& 0.45 & 12.2 & 0.19 & 12.2 & 
			& 0.46 & 12.3 & 0.93 & 13.1\tabularnewline
			\cline{2-11} 
			& $\widehat{\theta}^{T}$& 1.33 & 11.2 & 0.08 & 11.8 & 

			& 1.34 & 11.9 & 0.90 & 13.4\tabularnewline
			\cline{2-11} 
			& $\widehat{\theta}^{NLS}$& 5.25 & 42.4 & 0.25 & 42.7 & 

			& 11.4 & 58.4 & 0.30 & 64.9\tabularnewline
			\cline{2-11} 
			& $\widehat{\theta}^{GP}$ & 1.37 & 8.89 & 0.01 & 9.61 & 

			& 1.59 & 50.6 & 2.82 & 54.4\tabularnewline
			\hline 
			\hline 
			\multirow{4}{*}{$0.05$} & $\widehat{\theta}^{T,CI}$ & 1.35 & 15.6 & 0.15 & 15.6 & 
			& 1.35 & 17.4 & 1.10 & 18.3\tabularnewline
			\cline{2-11} 
			& $\widehat{\theta}^{T}$& 3.50 & 28.3 & 0.14 & 29.1 & 
	
			& 3.79 & 29.2 & 0.49 & 30.6\tabularnewline
			\cline{2-11} 
			& $\widehat{\theta}^{NLS}$ & 5.65 & 57.4 & 0.22 & 57.6 & 

			& 11.9 & 88.3 & 0.24 & 94.9\tabularnewline
			\cline{2-11} 
			& $\widehat{\theta}^{GP}$ & 2.11 & 12.3 & 0.06 & 33.7 & 
	
			& 2.17 & 30.6 & 2.53 & 33.7\tabularnewline
			\hline 
	\end{tabular}
\end{table}

 Once more, we see the difference between the approximate methods
and NLS for parameters facing practical identifiability issues. Here
GP always gives the smallest variance for $b^{*}$ and for $\theta^{*}$
in two case out of three. It is directly followed by our approaches.
In this finite sample case the TEs have a lower bias than NLS and
GP estimators which explains their lower mean square error.

\paragraph{Influence of model misspecification }

We choose sample size $n=25$ with variance $\sigma=0.03$ for the
measurement noise. The observations are now a realization of the hypoelliptic
stochastic differential equation: 
\begin{equation}
\left\{ \begin{array}{lll}
dV_{t} & = & c\left(V_{t}-\frac{V^{3}}{3}+R_{t}\right)dt\\
dR_{t} & = & -\frac{1}{c}\left(V_{t}-a+bR_{t}\right)dt+\sigma_{r}dW_{t}
\end{array}\right.\label{eq:FHNODE_misspe}
\end{equation}
with $W_{t}$ a Wiener process and $\sigma_{r}$ a diffusion parameter
but $\theta^{*}$ is still estimated by assuming the deterministic
model (\ref{eq:FHNODE_misspe}) is true. This model has been proposed
to include different sources of noise acting on $R_{t}$ (fluctuative
opening/closure of the ion channels within the membrane of the cell,
presynaptical currents etc.. \cite{Lindner1999}).We plot in Figure
\ref{fig:fhn_example_misspe} the solution of (\ref{eq:FHNODE}) and
one realization of (\ref{eq:FHNODE_misspe}) for the sake of comparison.
\begin{figure}
	\centering
	\includegraphics[scale=0.5]{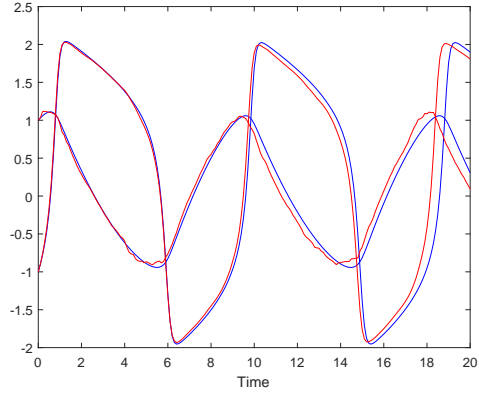}\caption{\label{fig:fhn_example_misspe}Solution of (\ref{eq:FHNODE}) (blue)	and a realization of (\ref{eq:FHNODE_misspe}) (red) for $\sigma_{r}^{2}=0.1$.}
\end{figure}
It has to be noted that dedicated methods for such hypoelliptic models
have been developed \cite{DitlevsenSamson2017,ClaironSamson2017},
but our point here is to show differences between approximate and
exact estimation methods for ODEs in the presence of misspecification.
We study the impact of misspecification by varying $\sigma_{r}^{2}$,
with results presented in Table \ref{tab:var_fhn_misspe}.

\begin{table}
	\caption{Scaled variance and mean square error for  misspecified FHN model. \label{tab:var_fhn_misspe}}
		\begin{tabular}{|c|c|c|c|c|c|c|c|c|c|c|}
			\hline 
		 $\sigma_{r}^{2}$ & $\times10^{-2}$ & $V(\widehat{a})$ & $V(\widehat{b})$ & $V(\widehat{c})$ & $\left\Vert V\left(\theta\right)\right\Vert _{2}$ & 
		
			& $\textrm{M}(\widehat{a})$ & $\textrm{M}(\widehat{b})$ & $\textrm{M}(\widehat{c})$ & \multicolumn{1}{c|}{$\textrm{M}(\widehat{\theta})$}\tabularnewline
			\hline 
			\hline 
			\multirow{4}{*}{$0.1$} & $\widehat{\theta}^{T,CI}$ & 3.63 & 10.3 & 0.07 & 10.5 & 

			& 3.98 & 10.4 & 1.00 & 11.9\tabularnewline
			\cline{2-11} 
			& $\widehat{\theta}^{T}$& 3.36 & 20.4 & 0.18 & 20.5 & 

			& 3.38 & 21.2 & 0.94 & 22.2\tabularnewline
			\cline{2-11} 
			& $\widehat{\theta}^{NLS}$& 6.77 & 40.9 & 0.27 & 43.9 & 

			& 10.1 & 76.1 & 0.27 & 82.4\tabularnewline
			\cline{2-11} 
			& $\widehat{\theta}^{GP}$& 6.11 & 28.7 & 0.48 & 29.7 & 

			& 7.62 & 43.4 & 0.49 & 45.8\tabularnewline
			\hline 
			\hline 
			\multirow{4}{*}{$0.15$} & $\widehat{\theta}^{T,CI}$& 4.56 & 16.4 & 0.09 & 16.4 & 

			& 4.63 & 16.8 & 0.96 & 17.7\tabularnewline
			\cline{2-11} 
			& $\widehat{\theta}^{T}$& 8.20 & 30.3 & 0.22 & 30.5 & 

			& 8.33 & 33.7 & 0.59 & 34.4\tabularnewline
			\cline{2-11} 
			& $\widehat{\theta}^{NLS}$& 8.9 & 98.5 & 0.48 & 1.02 &

			& 10.7 & 109 & 0.51 & 115\tabularnewline
			\cline{2-11} 
			& $\widehat{\theta}^{GP}$& 5.88 & 67.3 & 0.39 & 40.0 & 

			& 5.88 & 67.3 & 0.39 & 67.6\tabularnewline
			\hline 
			\hline 
			\multirow{4}{*}{$0.2$} & $\widehat{\theta}^{T,CI}$ & 3.90 & 16.1 & 0.39 & 16.1 & 
		
			& 3.92 & 16.1 & 1.23 & 16.9 \tabularnewline
			\cline{2-11} 
			& $\widehat{\theta}^{T}$& 5.10 & 32.5 & 0.23 & 33.2 & 

			& 6.36 & 39.6 & 0.26 & 41.6\tabularnewline
			\cline{2-11} 
			& $\widehat{\theta}^{NLS}$ & 6.10 & 64.2 & 0.35 & 64.8 & 

			& 8.06 & 79.5 & 0.37 & 82.1\tabularnewline
			\cline{2-11} 
			& $\widehat{\theta}^{GP}$ & 10.4 & 32.4 & 0.52  & 34.1 & 
		
			& 10.6 & 37.6 & 0.68 & 39.3\tabularnewline
			\hline 
	\end{tabular}
\end{table}

This example illustrate the benefit of using approximated methods
in the presence of model error as TEs and GP have a lower variance
than NLS and also a lower mean square error. By comparing $\widehat{\theta}^{T,CI}$
and $\widehat{\theta}^{T}$, we also notice the benefit of profiling
on CI.

\subsection{Repressilator}

We present the Repressilator model proposed in \cite{Elowitz2000}
for the study of a genetic regulation network. It is made of a feedback
loop of $3$ couples (mRNA, protein), denoted $\left(r_{i},\,p_{i}\right)_{1\leq i\leq3}$,
in which each protein inhibits the next gene transcription in the
loop:
\begin{equation}
\left\{ \begin{array}{lll}
\dot{r_{i}} & = & \frac{v_{i}k_{i,\left[i+1\right]}^{n}}{p_{\left[i+1\right]}^{n}+k_{i,\left[i+1\right]}^{n}}-k_{i}^{g}r_{i}\\
\dot{p_{i}} & = & k_{i}r_{i}-k_{i}^{p}p_{i}.
\end{array}\right.\label{eq:repressilator_ode_model}
\end{equation}
In this model, we aim to estimate $\theta^{*}=\left(v_{1}^{*},\,v_{2}^{*},\,v_{3}^{*},\,k_{1,2}^{*},\,k_{2,3}^{*},\,k_{1}^{g*},\,k_{2}^{g},^{*}\,k_{3}^{g*},\,k_{1}^{p*},\,k_{2}^{p*},\,k_{3}^{p*},\,n^{*}\right)=\left(50,\,100,\,80,\,50,\,30,\,1,\,1,\,1,\,1,\,2,\,3,\,3\right)$
with true initial conditions $\left(r_{1,0}^{*},\,r_{2,0}^{*},\,r_{3,0}^{*},\,p_{1,0}^{*},\,p_{2,0}^{*},\,p_{3,0}^{*}\right)=\left(60,\,20,\,6,18,\,27,\,1\right)$.
In order to reflect a real case observation framework, we consider
that only the mRNA concentrations are measured on $\left[0,\,T\right]=\left[0,\,20\right]$
and for structural identifiability reasons we set $\left(k_{3,1},\,k_{1},\,k_{2},\,k_{3}\right)=\left(40,\,5,\,6,\,7\right)$.
We plot in Figure \ref{fig:rep_example} the solution of (\ref{eq:repressilator_ode_model})
corresponding to $\theta^{*}$. 

\begin{figure}
	\centering
\includegraphics[scale=0.5]{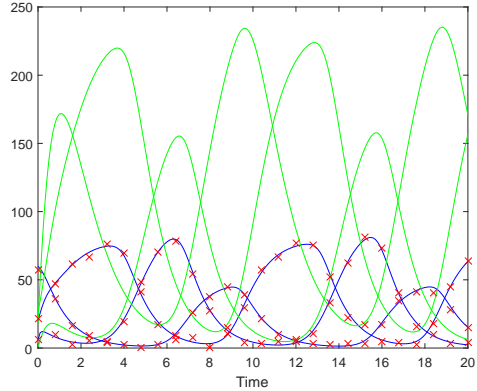}\caption{\label{fig:rep_example}Solution of (\ref{eq:repressilator_ode_model})
			with proteins in green, mRNAs in blue and corresponding noisy observations	for $\sigma=9$ in red.}
\end{figure}
Here, we choose:
\[
A_{\theta}(r,p,t)=\left(\begin{array}{ccccccc}
-k_{1}^{g} & 0 & 0 & 0 & 0 & 0 & \frac{v_{1}k_{1,2}^{n}}{p_{2}^{n}+k_{1,2}^{n}}\\
0 & -k_{2}^{g} & 0 & 0 & 0 & 0 & \frac{v_{3}k_{2,3}^{n}}{p_{3}^{n}+k_{2,3}^{n}}\\
0 & 0 & -k_{3}^{g} & 0 & 0 & 0 & \frac{v_{1}k_{3,1}^{n}}{p_{1}^{n}+k_{3,1}^{n}}\\
k_{1} & 0 & 0 & -k_{1}^{p} & 0 & 0 & 0\\
0 & k_{2} & 0 & 0 & -k_{2}^{p} & 0 & 0\\
0 & 0 & k_{3} & 0 & 0 & -k_{3}^{p} & 0\\
0 & 0 & 0 & 0 & 0 & 0 & 0
\end{array}\right)
\]
where, again, a constant artificial state variable $Z=1$ has been
added. This model has been identified as sloppy in \cite{Gutenkunst2007}
and for our particular experimental design, the eigendecomposition
of $\mathcal{I}_{n}\left(\theta^{*},x_{0}^{*}\right)$ for $n=25$
indicates the subset of parameters $\left(v_{1}^{*},\,v_{2}^{*},\,v_{3}^{*},\,k_{1,2}^{*},\,k_{2,3}^{*}\right)$
corresponds to the lowest eigenvalues. Henceforth, we separate $\theta^{*}$
into $\theta_{1}^{*}=\left(v_{1}^{*},\,v_{2}^{*},\,v_{3}^{*},\,k_{1,2}^{*},\,k_{2,3}^{*}\right)$
and $\theta_{2}^{*}=\left(k_{1}^{g*},\,k_{2}^{g*},\,k_{3}^{g*},\,k_{1}^{p*},\,k_{2}^{p*},\,k_{3}^{p*},\,n^{*}\right)$
for presenting the estimation results and in particular analyze how
the different methods behave with the poorly identifiable parameter
set $\theta_{1}^{*}$. In the following, $V(\widehat{\theta_{1}})$,
$V(\widehat{\theta_{2}}),$ $\textrm{M}(\widehat{\theta_{1}})$
and $\textrm{M}(\widehat{\theta_{2}})$ will denote the sum of the
variance (resp. mean square error) of $\theta_{1}$ and $\theta_{2}$ components.
We select $\left(k_{n},\lambda\right)$ among the set $\left\{ 20,\,25,\,30\right\} \times\left\{ 10^{i},\,5\times10^{i}\right\} _{-3\leq i\leq1}.$

\paragraph{Influence of measurement noise}

We take $n=25$ and consider three levels of measurement noise ($\sigma=9$,
$\sigma=10$ and $\sigma=11$). Results are presented in Table \ref{tab:var_rep}.

\begin{table}
	\caption{Scaled variance and mean square error for Repressilator model. \label{tab:var_rep}}
		\begin{tabular}{|c|c|c|c|c|c|c|c|c|}
			\hline 
			$\sigma$ &
		
			& $V(\widehat{\theta}_{1})$ & $V(\widehat{\theta}_{2})$ & $\left\Vert V\left(\widehat{\theta}\right)\right\Vert _{2}$ & 
	
			& $\textrm{M}(\widehat{\theta_{1}})$ & $\textrm{M}(\widehat{\theta_{2}})$ & \multicolumn{1}{c|}{$\textrm{M}(\widehat{\theta})$}\tabularnewline
			\hline 
			\hline 
			\multirow{3}{*}{$9$} & $\widehat{\theta}^{T,CI}$ & 0.16 & 0.12 & 0.16 & 
			& 0.21 & 0.15 & 0.24\tabularnewline
			\cline{2-9} 
			& $\widehat{\theta}^{T}$ & 0.19 & 0.20 & 0.21 & 

			& 0.21 & 0.22 & 0.25\tabularnewline
			\cline{2-9} 
			& $\widehat{\theta}^{NLS}$ & 0.24 & 0.17 & 0.26 & \
		
			& 0.24 & 0.17 & 0.27\tabularnewline
			\hline 
			\hline 
		   \multirow{3}{*}{$10$}& $\widehat{\theta}^{T,CI}$ & 0.20 & 0.20 & 0.22 & 
			& 0.25 & 0.22 & 0.30\tabularnewline
			\cline{2-9} 
			& $\widehat{\theta}^{T}$& 0.17 & 0.23 & 0.26 & 
	
			& 0.19 & 0.28 & 0.33\tabularnewline
			\cline{2-9} 
			& $\widehat{\theta}^{NLS}$ & 0.36 & 0.22 & 0.39 & 
		
			& 0.37 & 0.22 & 0.41\tabularnewline
			\hline 
			\hline 
			 \multirow{3}{*}{$11$}  & $\widehat{\theta}^{T,CI}$& 0.26 & 0.29 & 0.28 & 
			& 0.30 & 0.31 & 0.36\tabularnewline
			\cline{2-9} 
			& $\widehat{\theta}^{T}$& 0.22 & 0.36 & 0.38 & 

			& 0.23 & 0.39 & 0.41\tabularnewline
			\cline{2-9} 
			& $\widehat{\theta}^{NLS}$ & 0.39 & 0.27 & 0.42 & 
	
			& 0.42 & 0.28 & 0.45\tabularnewline
			\hline 
	\end{tabular}
\end{table}

We were unable to obtain results for GP because
of an important number of algorithmic failures during simulations
(almost $80\%$ of the runs) due to practical identifiability issues.
Indeed, GP requires the introduction of nuisance parameters $\beta$
needed for obtaining a smooth curve estimator $\tilde{X}_{\theta}^{\lambda}$
which can lead to overfitting with diverging parameter estimates.
In a partially observed framework, even for a $\widehat{\theta^{GP}}$
value far from $\theta^{*}$, the observed part of the smooth curve
$\tilde{X}_{\widehat{\theta^{GP}}}^{\lambda}$ can remain close to
the observations because the parameters $\widehat{\beta}_{\lambda}$
can counteract the effects of $\widehat{\theta^{GP}}$. Clearly, our
method improves the estimation of the subset of sloppy parameters.
Even though it is at the expense of $\theta_{2}$ estimation for the
highest noise level, our method globally improves the committed error
when all parameters are simultaneously estimated, which is the recommended
procedure in sloppy models \cite{Gutenkunst2007}.

\paragraph{Influence of model misspecification }

For the population size and the measurement noise variance, we set
$\left(n,\,\sigma\right)=\left(25,\,9\right)$. As in the $\alpha-$Pinene
case, the observations are now generated by a stochastically perturbed
version of the original ODE:
\begin{equation}
\left\{ \begin{array}{lll}
dr_{i} & = & \left(\frac{v_{i}k_{i,\left[i+1\right]}^{n}}{p_{\left[i+1\right]}^{n}+k_{i,\left[i+1\right]}^{n}}-k_{i}^{g}r_{i}\right)dt+c_{t}r_{i}dt\\
dp_{i} & = & \left(k_{i}r_{i}-k_{i}^{p}p_{i}\right)dt+c_{t}p_{i}dt
\end{array}\right.\label{eq:repressilator_ode_model_misspe}
\end{equation}
with $c_{t}\sim N(0,\,\sigma_{c}^{2})$. We plot in Figure \ref{fig:rep_example_misspe}
the solution of (\ref{eq:repressilator_ode_model}) and one realization
of (\ref{eq:repressilator_ode_model_misspe}) for the sake of comparison.
\begin{figure}
	\centering
\includegraphics[scale=0.5]{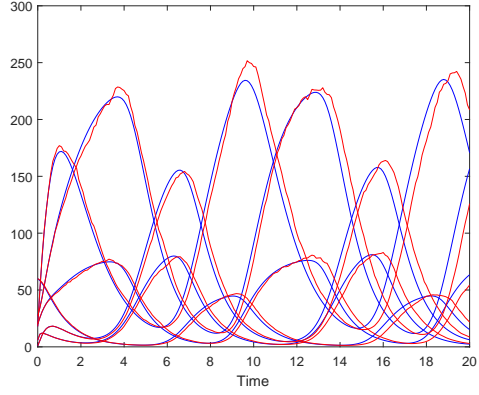}\caption{\label{fig:rep_example_misspe}Solution of (\ref{eq:repressilator_ode_model})	(blue) and a realization of (\ref{eq:repressilator_ode_model_misspe})  (red) for $\sigma_{c}^{2}=0.5$.}
\end{figure}
Results are presented in Table \ref{tab:var_rep_misspe}

\begin{table}
	\caption{Scaled variance and mean square error for misspecified Repressilator model. \label{tab:var_rep_misspe}}
	\begin{tabular}{|c|c|c|c|c|c|c|c|c|}
		\hline 
		$\sigma_{c}^{2}$& 
	
		& $V(\widehat{\theta}_{1})$ & $V(\widehat{\theta}_{2})$ & $\left\Vert V\left(\widehat{\theta}\right)\right\Vert _{2}$ &

		& $\textrm{M}(\widehat{\theta_{1}})$ & $\textrm{M}(\widehat{\theta_{2}})$ & \multicolumn{1}{c|}{$\textrm{M}(\widehat{\theta})$}\tabularnewline
		\hline 
		\hline 
		\multirow{3}{*}{$\sigma_{c}^{2}=0.5$} & $\widehat{\theta}^{T,CI}$ & 0.20 & 0.19 & 0.21 & 

		& 0.25 & 0.21 & 0.29\tabularnewline
		\cline{2-9} 
		& $\widehat{\theta}^{T}$ & 0.18 & 0.24 & 0.25 & 

		& 0.18 & 0.26 & 0.27\tabularnewline
		\cline{2-9} 
		& $\widehat{\theta}^{NLS}$ & 0.27 & 0.23 & 0.29 & 

		& 0.28 & 0.23 & 0.31\tabularnewline
		\hline 
		\hline 
		\multirow{3}{*}{$\sigma_{c}^{2}=1$} & $\widehat{\theta}^{T,CI}$ & 0.21 & 0.25 & 0.24 & 

		& 0.26 & 0.28 & 0.32\tabularnewline
		\cline{2-9} 
		& $\widehat{\theta}^{T}$& 0.18 & 0.24 & 0.25 & 

		& 0.19 & 0.29 & 0.32\tabularnewline
		\cline{2-9} 
		& $\widehat{\theta}^{NLS}$ & 0.26 & 0.28 & 0.33 & 

		& 0.27 & 0.29 & 0.35\tabularnewline
		\hline 
		\hline 
		\multirow{3}{*}{$\sigma_{c}^{2}=1.5$} & $\widehat{\theta}^{T,CI}$ & 0.22 & 0.27 & 0.27 & 

		& 0.26 & 0.30 & 0.35\tabularnewline
		\cline{2-9} 
		& $\widehat{\theta}^{T}$ & 0.24 & 0.25 & 0.28 & 

		& 0.25 & 0.29 & 0.33\tabularnewline
		\cline{2-9} 
		& $\widehat{\theta}^{NLS}$& 0.41 & 0.21 & 0.37 &

		& 0.44 & 0.21 & 0.40\tabularnewline
		\hline 
\end{tabular}
\end{table}
 They confirm once again the adantages of using an estimation method
based on a relaxation of the original model in the presence of model
error.

\subsection{FitzHugh-Nagumo with a functional parameter}

We resume the experimental design presented in Section \ref{subsec:FitzHugh-Nagumo}
but the true constant parameter $a^{*}$ is replaced by the function
$a^{*}(t)=0.2\left(1+\sin(\frac{t}{5})\right)$. We plot in Figure
\ref{fig:fhn_example_misspe-1} the solutions of (\ref{eq:FHNODE})
for $a^{*}=0.2$ and $a^{*}=0.2\left(1+\sin(\frac{t}{5})\right)$.
\begin{figure}
\includegraphics[scale=0.4]{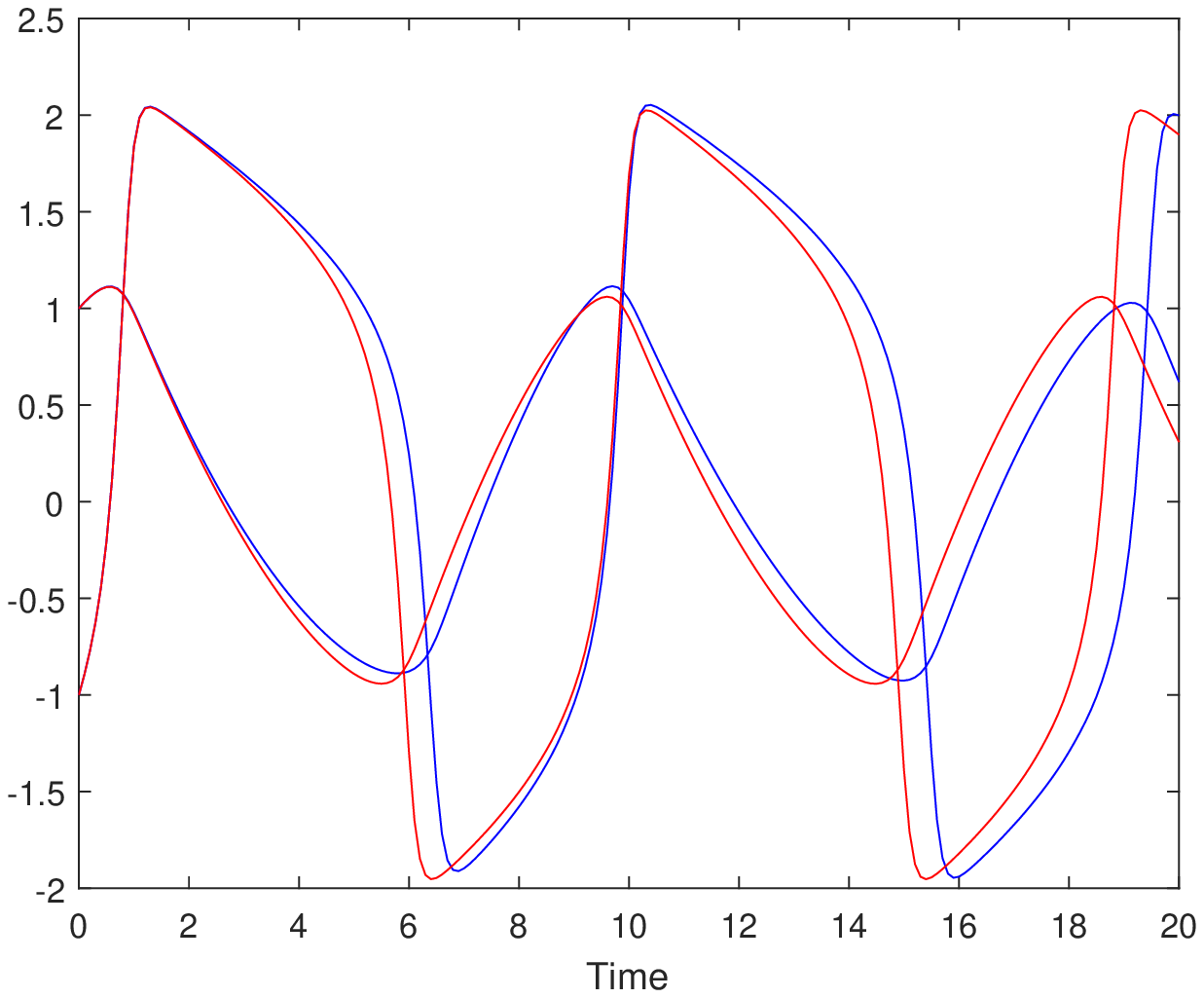}\includegraphics[scale=0.4]{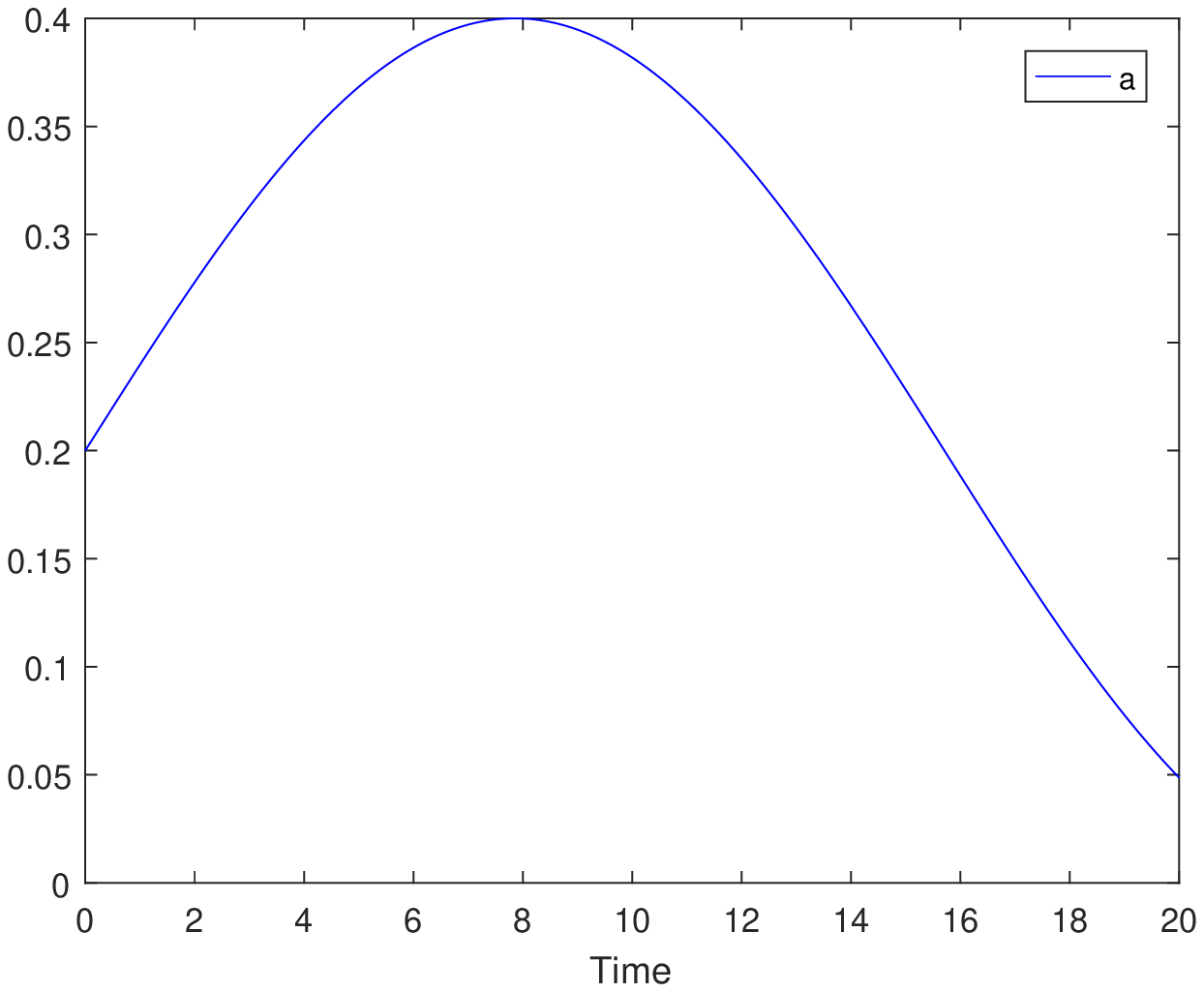}\caption{\label{fig:fhn_example_misspe-1}Left: Solution of (\ref{eq:FHNODE}) for $a^{*}$ constant (blue) and $a^{*}$ time-varying (red). Right: Plot of time-varying parameter $a^{*}$ .}
\end{figure}
Here, we compare the variational approach presented in Section \ref{subsec:Semi-parametric-estimation}
with a classic basis decomposition method for the simultaneous estimation
of $\theta^{*}=\left(b^{*},\,c^{*}\right)$ and $\vartheta^{*}=a^{*}$.
Our criteria is based on the cost (\ref{eq:semi_param_cost_with_discrepancy_model})
and thus requires the selection of three hyperparameters $k_{n}$
, $\lambda_{1}$ and $\lambda_{2}$. We retain the value for $\left(k_{n},\,\lambda_{1},\,\lambda_{2}\right)$
which minimizes the prediction error EP among the trial $\left\{ 20,\,25\right\} \times\left\{ 10^{i}\right\} _{-4\leq i\leq-2}\times\left\{ 10^{i}\right\} _{-4\leq i\leq-1}.$
To estimate $a^{*}$ with NLS and GP we use a finite basis decomposition,
i.e. we approximate $a^{*}$ with $\widehat{a}(t)\simeq\sum_{i=1}^{K_{r}}\beta_{r,K_{r}}p_{i}(t)$,
here $\left\{ p_{i}\right\} _{i}$ is a B-Spline basis with a uniform
knot sequence. The additional $K_{r}$ parameters $\left(\beta_{r,1},\,\ldots,\,\beta_{r,K_{r}}\right)$
are estimated by introducing the extended set $\theta^{ext}=\left(\theta,\,\beta_{r,1},\,\ldots,\,\beta_{r,K_{r}}\right)$
and $K_{r}$ is selected by minimizing the Akaike Information: $\textrm{AIC}(\widehat{\theta}^{ext})=n\log\left(\frac{\sum_{i}(y_{i}-CX_{\widehat{\theta}^{ext},\widehat{x_{0}}}(t_{i}))^{2}}{n}\right)+2K_{r}$,
where $\widehat{x_{0}}$ is the standart initial condition estimator
for NLS and for GP we choose $\widehat{x_{0}}=\tilde{X}_{\hat{\theta}_{\lambda}^{GP}}^{\lambda}(0)$.
For a given estimator $\widehat{a}$, we quantify its accuracy by
Monte-Carlo estimator of the integrated variance $V^{f}(\widehat{a})=\int_{0}^{T}\left(\mathbb{E}\left[\widehat{a}^{2}(t)\right]-\mathbb{E}\left[\widehat{a}(t)\right]{}^{2}\right)dt$
and mean square error $\textrm{M}^{f}(\widehat{a})=\int_{0}^{T}\mathbb{E}\left[\left(\widehat{a}(t)-a^{*}(t)\right)^{2}\right]dt$. 

\paragraph{Influence of measurement noise}

We take $n=50$ and consider two levels of measurement noise ($\sigma=0.03$
and $\sigma=0.05$). Results are presented in Table \ref{tab:var_fhn_fun}.
\begin{table}
	\caption{Scaled variance and mean square error for  FHN model with functional parameter. \label{tab:var_fhn_fun}}
	\hspace*{-1cm}
		\begin{tabular}{|c|c|c|c|c|c|c|c|c|c|c|}
			\hline 
			$\sigma$ & $\times10^{-2}$ & $V(\widehat{b})$ & $V(\widehat{c})$ & $\left\Vert V\left(\theta\right)\right\Vert _{2}$ & $V^{f}(\widehat{a})$ & 
		
			& $\textrm{M}(\widehat{b})$ & $\textrm{M}(\widehat{c})$ & $\textrm{M}(\widehat{\theta})$ & \multicolumn{1}{c|}{$\textrm{M}^{f}(\widehat{a})$}\tabularnewline
			\hline 
			\hline 
			\multirow{4}{*}{$0.03$} & $\widehat{\theta}^{T,CI}$& 2.50 & 0.05 & 2.50 & 1.06 & 
			
			& 2.87 & 0.07 & 2.88 & 1.77\tabularnewline
			\cline{2-11} 
			& $\widehat{\theta}^{T}$ & 3.50 & 0.06 & 3.52 & 0.77 & 
		
			& 4.25 & 0.06 & 4.27 & 1.08\tabularnewline
			\cline{2-11} 
			& $\widehat{\theta}^{NLS}$ & 5.32 & 3.55 & 7.94 & 2.26 & 
		
			& 5.91 & 3.72 & 8.71 & 2.41\tabularnewline
			\cline{2-11} 
			& $\widehat{\theta}^{GP}$ & 2.84 & 0.01 & 2.84 & 2.41 & 
			
			& 3.17 & 0.20 & 3.36 & 6.86\tabularnewline
			\hline 
			\hline 
			\multirow{4}{*}{$0.05$} & $\widehat{\theta}^{T,CI}$ & 11.7 & 0.17 & 11.8 & 3.04 & 
			& 12.1 & 0.27 & 12.2 & 4.00\tabularnewline
			\cline{2-11} 
			& $\widehat{\theta}^{T}$ & 8.17 & 0.15 & 8.21 & 2.61 &
		
			& 8.19 & 0.22 & 8.30 & 2.61\tabularnewline
			\cline{2-11} 
			& $\widehat{\theta}^{NLS}$& 14.1 & 3.77 & 15.5 & 10.1 & 
		
			& 14.3 & 4.06 & 25.4 & 11.4\tabularnewline
			\cline{2-11} 
			& $\widehat{\theta}^{GP}$ & 13.9 & 0.04 & 13.6 & 16.7 & 
		
			& 13.9 & 0.17 & 14.0 & 18.6\tabularnewline
			\hline 
	\end{tabular}
\end{table}

 Our estimators outperform GP and NLS both for parametric and functional
estimation accuracy. The finite basis decomposition used to replace
$a$ leads to use an approximated version of the original model for
estimation purposes. This induced misspecification can explain the
drop in accuracy for the NLS and GP estimators. Moreover, as pointed
out in \cite{BrunelClairon_Pontryagin2017}, the selection of a proper
basis and knot location for semiparametric estimation is complicated
and model-specific. In our case, the extension of the parametric estimation
method to the semi parametric framework is straighforward for hyperparameters
selection.

\paragraph{Influence of model misspecification }

We choose sample size $n=50$ with variance $\sigma=0.03$ for the
measurement noise. The observations are now realization of the hypoelliptic
stochastic differential equation (\ref{eq:FHNODE_misspe}) for two
values of $\sigma_{r}^{2}$. Results are presented in Table \ref{tab:var_fhn_fun_misspe}.

\begin{table}
	\caption{Scaled variance and mean square error for  misspecified FHN model with functional parameter. \label{tab:var_fhn_fun_misspe}}	

	\begin{tabular}{|c|c|c|c|c|c|c|c|c|c|c|}
		\hline 
		$\sigma_{r}^{2}$ & $\times10^{-2}$ & $V(\widehat{b})$ & $V(\widehat{c})$ & $\left\Vert V\left(\theta\right)\right\Vert _{2}$ & $V^{f}(\widehat{a})$ &
	
		& $\textrm{M}(\widehat{b})$ & $\textrm{M}(\widehat{c})$ & $\textrm{M}(\widehat{\theta})$ & \multicolumn{1}{c|}{$\textrm{M}^{f}(\widehat{a})$}\tabularnewline
		\hline 
		\hline 
		\multirow{4}{*}{$0.1$} & $\widehat{\theta}^{T,CI}$& 11.7 & 0.08 & 11.7 & 8.97 & 
	
		& 11.7 & 0.08 & 11.7 & 9.33\tabularnewline
		\cline{2-11} 
		& $\widehat{\theta}^{T}$ & 7.58 & 0.08 & 7.59 & 5.26 & 
	
		& 7.64 & 0.09 & 7.65 & 6.04\tabularnewline
		\cline{2-11} 
		& $\widehat{\theta}^{NLS}$ & 17.5 & 0.07 & 17.5 & 9.7 & 
		
		& 20.6 & 0.12 & 20.6 & 10.7\tabularnewline
		\cline{2-11} 
		& $\widehat{\theta}^{GP}$ & 65.4 & 0.11 & 65.4 & 42.8 & 
		
		& 65.8 & 0.14 & 65.8 & 43.9\tabularnewline
		\hline 
		\hline 
		\multirow{4}{*}{$0.15$}  & $\widehat{\theta}^{T,CI}$ & 13.7 & 0.10 & 13.7 & 9.8 & 
		
		& 13.7 & 0.12 & 13.7 & 10.5\tabularnewline
		\cline{2-11} 
		& $\widehat{\theta}^{T}$ & 9.50 & 0.09 & 9.50 & 6.75 & 
	
		& 9.56 & 0.10 & 9.58 & 7.43\tabularnewline
		\cline{2-11} 
		& $\widehat{\theta}^{NLS}$ & 32.3 & 16.7 & 43.2 & 36.9 & 
	
		& 35.8 & 21.3 & 51.3 & 47.7\tabularnewline
		\cline{2-11} 
		& $\widehat{\theta}^{GP}$ & 73.7 & 0.15 & 73.7 & 58.7 & 
	
		& 73.7 & 0.25 & 73.8 & 61.3\tabularnewline
		\hline 
\end{tabular}
\end{table}
 Once again, our methods give better results than NLS and GP. The
difference is even more striking here, possibly due to the accumulated
effect of the different source of misspecifications on GP and NLS.
For our approaches, the term $\lambda_{1}u_{1,j}{}^{T}u_{1,j}$ present
in cost (\ref{eq:semi_param_cost_with_discrepancy_model}) takes into
account model discrepancy and is expected to mitigate the effect of
misspecification on estimation \cite{Brynjarsdottir2014,kirk2016reverse}.

\section{\label{sec:Real-data-case}Real data case analysis}

We focus on a model discussed by \cite{Stein2013} to study the impact
on a microbiota ecosystem of the interaction between an antibiotic
treatment and a pathogen inoculation.The model used by the authors
is an 11-dimensional Generalized Lotka-Volterra ODE:
\begin{equation}
\dot{x}_{i}=\mu_{i}x_{i}+x_{i}\sum_{j=1}^{11}M_{i,j}x_{j}+x_{i}s_{i}v(t)\label{eq:original_microbial_ode}
\end{equation}
for $i=1,\dots,11$. Each state variable $x_{i}$ quantifies the presence
of one microbial species and $t\mapsto v(t)$ describes the perturbation
due to antibiotic administration (here, clindamycin). Regarding the
parameter set $\left(\mu_{i},\,M_{i,j},\,s_{i}\right)_{1\leq i,j\leq11}$,
$\mu_{i}$ is the growth term for $x_{i}$, $M_{i,j}$ the interaction
effect of $x_{j}$ on $x_{i}$ and $s_{i}$ the susceptibility of
$x_{i}$ to the antibiotic treatment. The names of the microbial species
as well as the values of $\left(\mu_{i},\,M_{i,j},\,s_{i}\right)_{1\leq i,j\leq11}$
are provided by \cite{Stein2013} (Figure 2). Regarding the acquired
data in \cite{Stein2013}, they are divided in three groups of three
subjects. Group 1 was exposed to the pathogen (here, \textit{C. difficile}),
Group 2 received a single dose of clindamycin and Group 3 received
clindamycin and was exposed to \textit{C. difficile} the day after.
We focus on Group 3 for which the perturbation is the impulse function
$v(t)=\mathrm{1}_{\left\{ t=0\right\} }.$ In this restricted case,
some microbiotal species have limited impacts on the whole ecosystem
evolution. For parameter estimation we restrict ourselves to the simplified
model: 
\begin{equation}
\dot{x_{i}^{s}}=\mu_{i}^{s}x_{i}^{s}+x_{i}^{s}\sum_{j=1}^{7}M_{i,j}^{s}x_{i}^{s}+x_{i}^{s}s_{i}^{s}v(t)\label{eq:restricted_microbial_ode}
\end{equation}
where $x_{i}^{s}=x_{i}$ for $1\leq i\leq5$, $x_{6}^{s}=x_{9}$ and
$x_{7}^{s}=x_{11}$ and the parameters $\left\{ \mu_{i}^{s},\,M_{i,j}^{s},\,s_{i}^{s}\right\} $
are defined and linked to the previous parametrization accordingly.
For comparison, we plot in Figure \ref{fig:microb_compare_ode} the
solution of (\ref{eq:original_microbial_ode}) and (\ref{eq:restricted_microbial_ode})
for three initial condition values corresponding to the three subjects
in Group 3.

\begin{figure}
\includegraphics[scale=0.32]{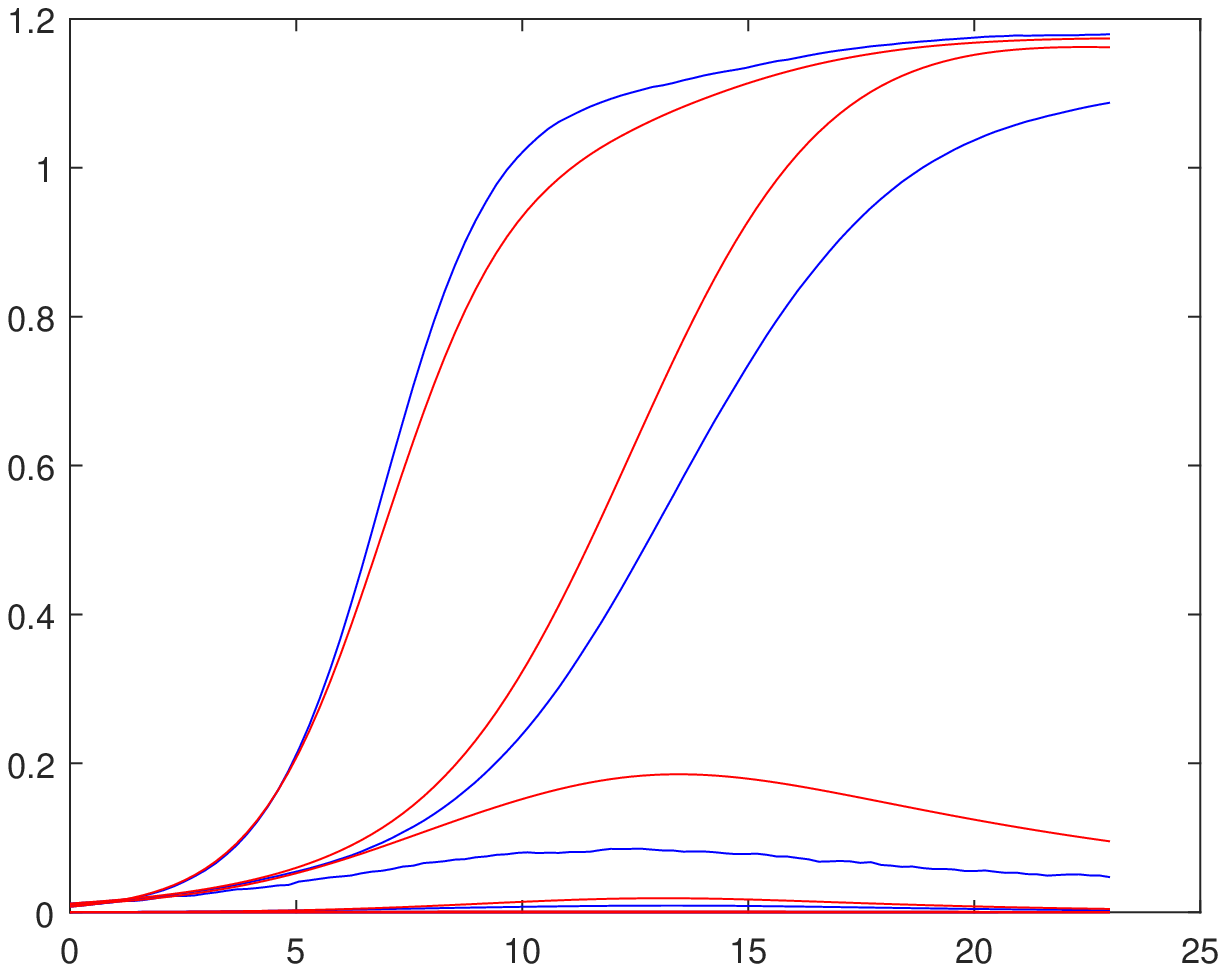}\includegraphics[scale=0.32]{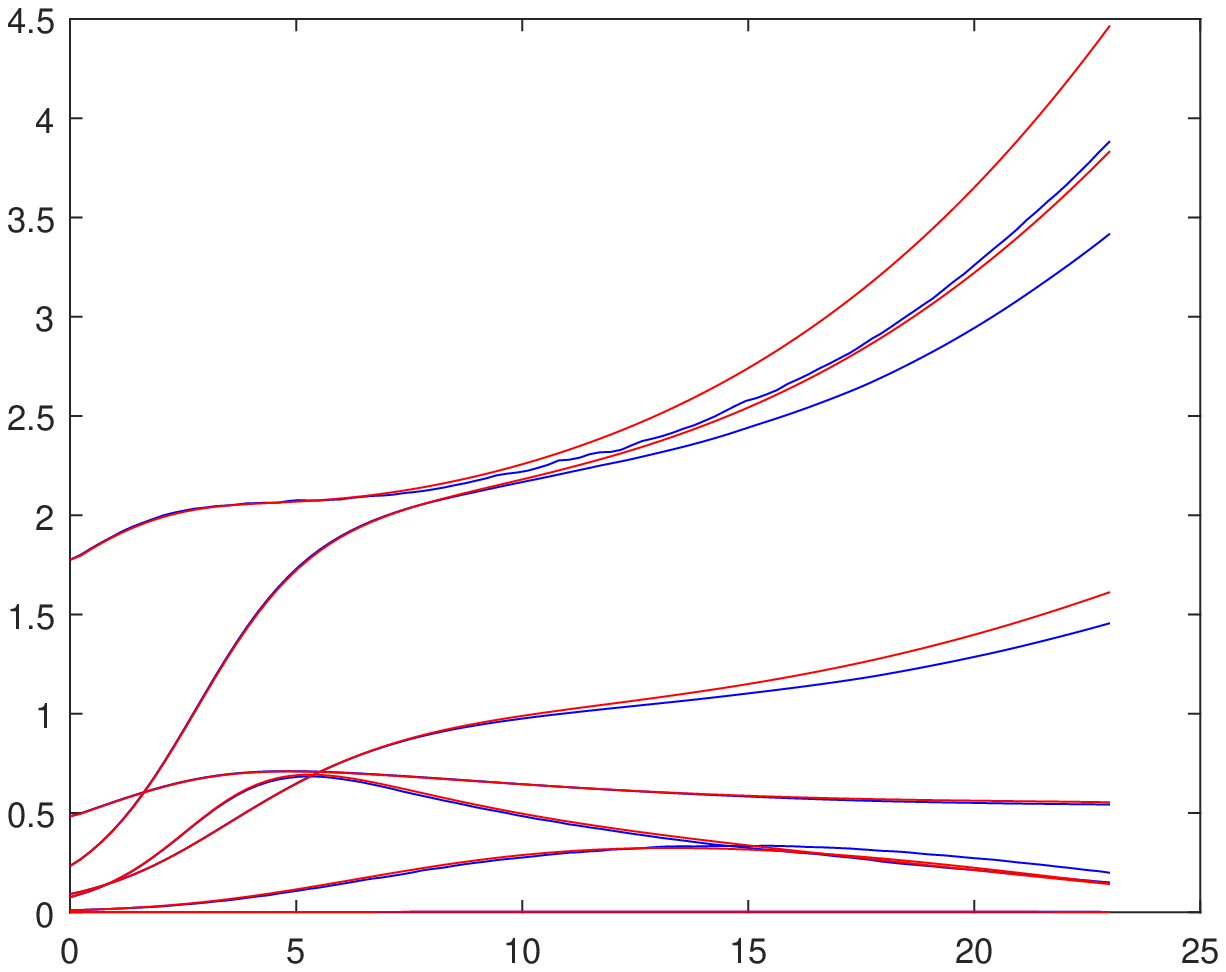}\includegraphics[scale=0.32]{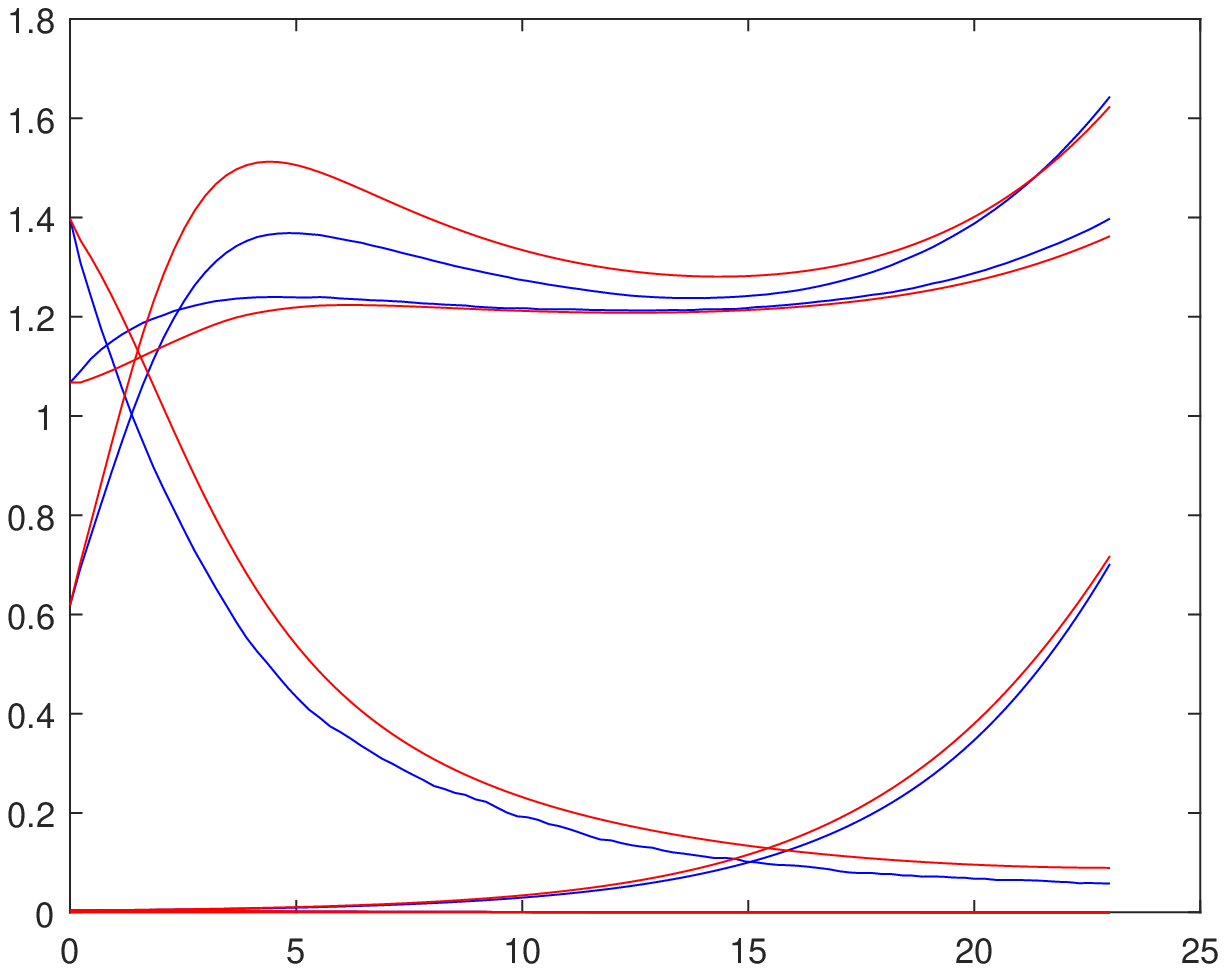}\caption{\label{fig:microb_compare_ode}Solution of (\ref{eq:original_microbial_ode}) (blue) and solution of (\ref{eq:restricted_microbial_ode}) (red) for different initial conditions.}
\end{figure}
 Here, we focus on the $M_{i,j}$ estimation for they give the nature
of interaction between species. The different $\left(\mu_{i},s_{i}\right)$
are considered as known and we estimate the parameter set of dimension
$p=31$: $\theta=(\left\{ M_{i,1}\right\} _{i\in\left\{ 1,3,4,5\right\} },$$\left\{ M_{i,2}\right\} _{i\in\left\{ 2,3,4,5\right\} },$ $\left\{ M_{i,3}\right\} _{i\in\left\{ 1,2,3\right\} },$
$\left\{ M_{i,4}\right\} _{i\in\left\{ 1,2,4\right\} },$ $\left\{ M_{i,5}\right\} _{i\in\left\{ 2,3,4,5\right\} },$$\left\{ M_{i,9}\right\} _{i\in\left\{ 3,5\right\} },$ $\left\{ M_{i,11}\right\} _{i\in\left\{ 1,2,3,4,5\right\} },$ $\left\{ M_{11,i}\right\} _{i\in\left\{ 1,2,3,5,9,11\right\} })$.
The other interaction terms are unidentifiable in practice when we
use only data coming from Group 3. Before our real data analysis we
perform first a simulation analysis in order to compare our approach
with NLS. 

\subsection{Preliminary results on simulated data}

To mimic the real case analysis, estimation is based on the observation
of three individuals with different initial conditions $x_{1,0}^{*},\,x_{2,0}^{*}$
and $x_{3,0}^{*}$. Both NLS and our approach requiring estimation
of initial conditions have $52$ unknown parameters. Hence, to save
computational time, we restrict here to $\widehat{\theta}^{NLS}$
and $\widehat{\theta}^{T,CI}$. Because of the parameter dimension
and similarly as in the Repressilator case, we split $\theta$ into
two subgroups $\theta_{1}$ and $\theta_{2}$ according to the difficulty
encountered by NLS to estimate them. To identify $\theta_{2}$, we
rely once again on the eigendecomposition of $V^{*}D^{*}(V^{*})^{-1}=\mathcal{I}_{n}\left(\theta^{*},\,x_{1,0}^{*},\,x_{2,0}^{*},\,x_{3,0}^{*}\right)$
for $n=25$ where $D^{*}=diag(\mu_{1},\ldots,\mu_{52})$ is the matrix
composed of the 52 eigenvalues $\mu_{i}$ sorted in increasing order
and $V^{*}$ is the matrix where each column $V_{.,i}^{*}$ is the
eigenvector related to $\mu_{i}$. The associated condition number
is equal to $\kappa(\mathcal{I}_{n})\simeq8\times10^{-10}$ which
clearly indicates an ill-posed problem for NLS. Moreover, we have
$\frac{\mu_{25}}{\mu_{52}}=2\times10^{-6}$, thus the first 25 eigenvectors
correspond to directions of weak change for the NLS criteria. For
each parameter $\theta^{j}$ in $\theta$, we compute $F(\theta^{j})=\frac{\sum_{i=1}^{25}\left(V_{j,i}^{*}\right)^{2}}{\sum_{i=1}^{52}\left(V_{j,i}^{*}\right)^{2}}$
to quantify the impact of $\theta^{j}$ on NLS criteria. By doing
so we identify 12 parameters such that $F(\theta_{j})>0.63$ which
will compose $\theta_{2}$, the set of parameters poorly estimated
by NLS. The choice of threshold for the eigenvalue rank and $F(\theta^{j})$
value is somewhat arbitrary, but we will see in simulations that the
variance and mean square error for $\widehat{\theta}^{NLS}$ come mainly from estimation
of $\theta_{2}$. As in the Repressilator case, $V(\widehat{\theta_{1}})$,
$\textrm{M}(\widehat{\theta_{1}})$ and $V(\widehat{\theta_{2}}),$
$\textrm{M}(\widehat{\theta_{2}})$ will have to be understood as
the sum of the variance and mean square error of the $\theta_{1}$ and $\theta_{2}$
elements. In the presented results, we also compare the ability of
the different estimators to reconstruct the orientation of the interaction
graph i.e. we estimate $\textrm{I}(\widehat{\theta})=\frac{1}{p}\mathbb{E}_{\theta^{*}}\left[\sum_{i=1}^{p}1_{\left\{ sign(\widehat{\theta_{i}})=sign(\theta_{i}^{*})\right\} }\right]$
that is, the expected fraction of correctly retrieved interaction. 

\paragraph{Influence of measurement noise}

We take $n=25$ and consider three levels of measurement noise ($\sigma=0.01$,
$\sigma=0.02$ and $\sigma=0.03$). Results are presented in Table
\ref{tab:var_microb}. 
\begin{table}
	\caption{Scaled variance, mean square error and I for Microbiota model. \label{tab:var_microb}}
\hspace{-2.5cm}
\begin{tabular}{|c|c|c|c|c|c|c|c|c|c|c|c|c|}
	\hline 
	$\sigma$ & 

	& $V(\widehat{\theta}_{1})$ & $V(\widehat{\theta}_{2})$ & $\left\Vert V\left(\widehat{\theta}\right)\right\Vert _{2}$ & 

	& $\textrm{M}(\widehat{\theta_{1}})$ & $\textrm{M}(\widehat{\theta_{2}})$ & $\textrm{M}(\widehat{\theta})$ & 

	& $I(\widehat{\theta_{1}})$ & $I(\widehat{\theta_{2}})$ & $I(\widehat{\theta})$\tabularnewline
	\hline 
	\hline 
		\multirow{2}{*}{$0.01$} & $\widehat{\theta}^{T,CI}$ & 0.11 & 0.18 & 0.11 & 
	
	& 0.18 & 0.22 & 0.22 & 
	
	& 1 & 1 & 1\tabularnewline
	\cline{2-13} 
	& $\widehat{\theta}^{NLS}$ & 0.67 & 1.86 & 2.04 & 

	& 0.68 & 1.89 & 2.09 & 

	& 0.99 & 0.98 & 0.99\tabularnewline
	\hline 
	\hline 
	\multirow{2}{*}{$0.02$} & $\widehat{\theta}^{T,CI}$ & 0.22 & 0.51 & 0.28 & 

	& 0.29 & 0.54 & 0.38 &

	& 1 & 1 & 1\tabularnewline
	\cline{2-13} 
	& $\widehat{\theta}^{NLS}$ & 2.85 & 7.94 & 8.29 & 

	& 2.89 & 8.06 & 8.46 & 

	& 0.98 & 0.93 & 0.96\tabularnewline
	\hline 
	\hline 
	\multirow{2}{*}{$0.03$} & $\widehat{\theta}^{T,CI}$& 0.27 & 0.64 & 0.29 & 

	& 0.34 & 0.70 & 0.43 & 

	& 1 & 1 & 1\tabularnewline
	\cline{2-13} 
	& $\widehat{\theta}^{NLS}$  & 3.03 & 9.06 & 6.43 & 

	& 3.31 & 9.80 & 7.46 & 

	& 0.98 & 0.90 & 0.95\tabularnewline
	\hline 
\end{tabular}
\end{table}
For both $\theta_{1}$ and $\theta_{2}$ our approach
outperforms NLS. In particular, we retrieve the right orientation
for the interaction graph, whereas the ability of $\widehat{\theta}^{NLS}$
to do so decreases with noise, especially for $\theta_{2}$.

\paragraph{Influence of model misspecification }

We choose the same sample size and measurement error levels as before
but the observations are now generated by using model (\ref{eq:original_microbial_ode}).
We are interested in quantifying robustness of the different estimators
with respect to misspecification due to neglected interactions, a
common feature in the study of biological networks. In particular,
we want to measure the ability of our estimator to retrieve the true
interactions between two state variables despite the presence of unmeasured
coufounders. Results are presented in Table \ref{tab:var_microb_misspe}.
\begin{table}
	\caption{Scaled variance, mean square error and I for misspecified Microbiota model\label{tab:var_microb_misspe}}	
	\hspace{-2cm}
		\begin{tabular}{|c|c|c|c|c|c|c|c|c|c|c|c|c|}
			\hline 
			$\sigma$& 

			& $V(\widehat{\theta}_{1})$ & $V(\widehat{\theta}_{2})$ & $\left\Vert V\left(\widehat{\theta}\right)\right\Vert _{2}$ & 

			& $\textrm{M}(\widehat{\theta_{1}})$ & $\textrm{M}(\widehat{\theta_{2}})$ & $\textrm{M}(\widehat{\theta})$ & 
		
			& $I(\widehat{\theta_{1}})$ & $I(\widehat{\theta_{2}})$ &$I(\widehat{\theta})$\tabularnewline
			\hline 
			\hline 
			\multirow{2}{*}{$0.01$} & $\widehat{\theta}^{T,CI}$ & 0.30 & 2.56 & 1.61 & 
			& 0.69 & 3.10 & 2.54 & 
			
			& 1 & 0.98 & 1\tabularnewline
			\cline{2-13} 
			& $\widehat{\theta}^{NLS}$ & 0.63 & 2.04 & 1.85 & 
	
			& 1.47 & 5.73 & 6.40 & 
			& 0.99 & 0.90 & 0.96\tabularnewline
			\hline 
			\hline 
			 \multirow{2}{*}{$0.02$} & $\widehat{\theta}^{T,CI}$ & 0.33 & 2.67 & 1.59 & 
			& 0.74 & 3.23 & 2.57 &

			& 1 & 0.98 & 0.99\tabularnewline
			\cline{2-13} 
			& $\widehat{\theta}^{NLS}$ & 1.98 & 6.35 & 4.78 & 
		
			& 2.90 & 9.97 & 9.32 & 

			& 0.98 & 0.88 & 0.95\tabularnewline
			\hline 
			\hline 
			\multirow{2}{*}{$0.03$}  & $\widehat{\theta}^{T,CI}$ & 0.35 & 2.76 & 1.63 & 
			& 0.77 & 3.37 & 2.59 & 
		
			& 0.99 & 0.98 & 0.99\tabularnewline
			\cline{2-13} 
			& $\widehat{\theta}^{NLS}$& 5.43 & 16.3 & 10.6 & 

			& 6.25 & 19.6 & 14.7 & 

			& 0.95 & 0.83 & 0.91\tabularnewline
			\hline 
	\end{tabular}
\end{table}
The situation is quite similar to the well-specified
case but with the additional feature that the capacity
of $\widehat{\theta}^{NLS}$ to retrieve the true interaction graph
is more affected by model misspecification than $\widehat{\theta}^{T,CI}$. 
\subsection{Real data analysis}

In this section, we proceed to $\theta$ estimation in model (\ref{eq:restricted_microbial_ode})
starting from the real data available in \cite{Stein2013} for the
Group 3. The data are presented in Figure \ref{fig:microbiotal_raw_data}.
\begin{figure}
\includegraphics[scale=0.32]{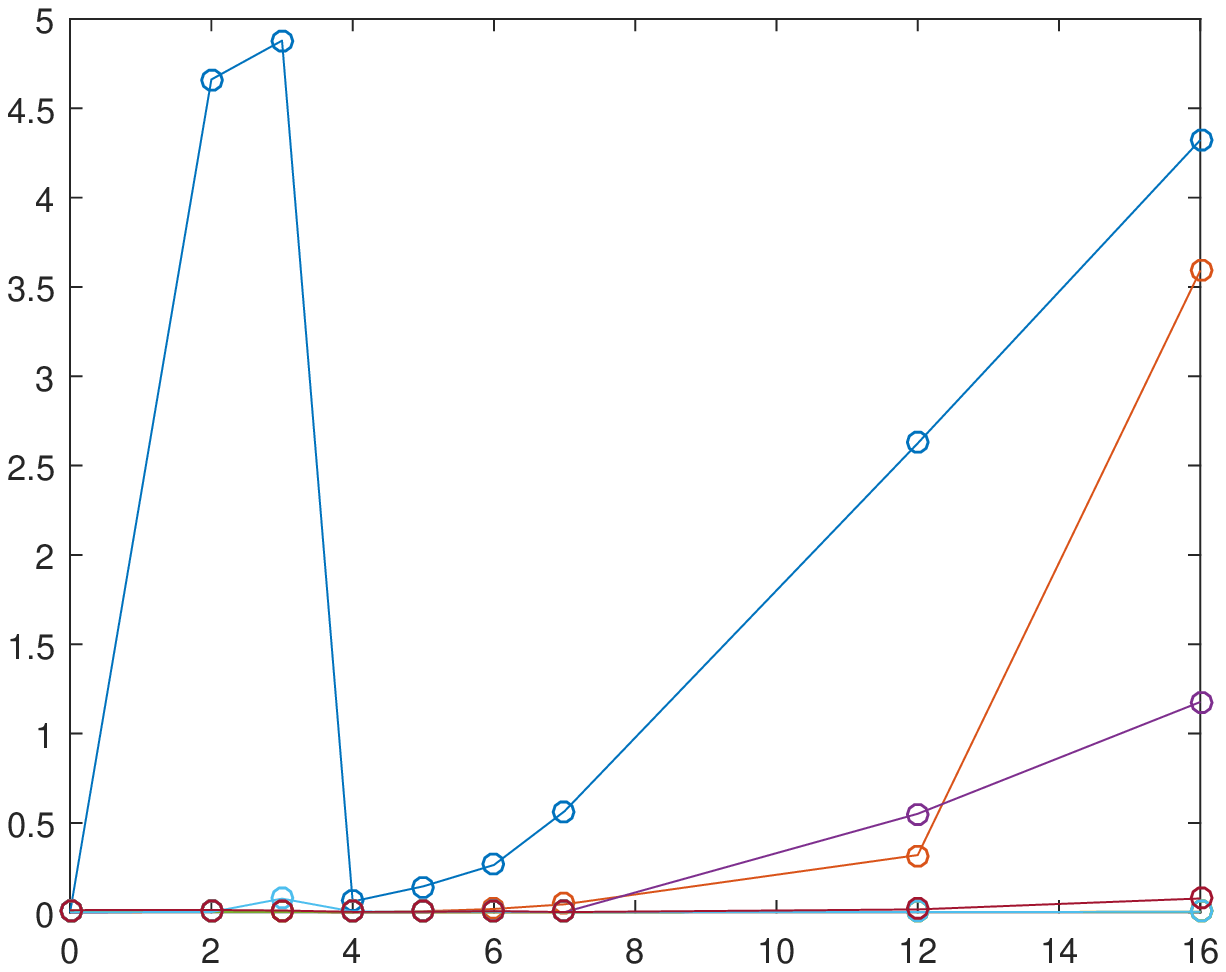}\includegraphics[scale=0.32]{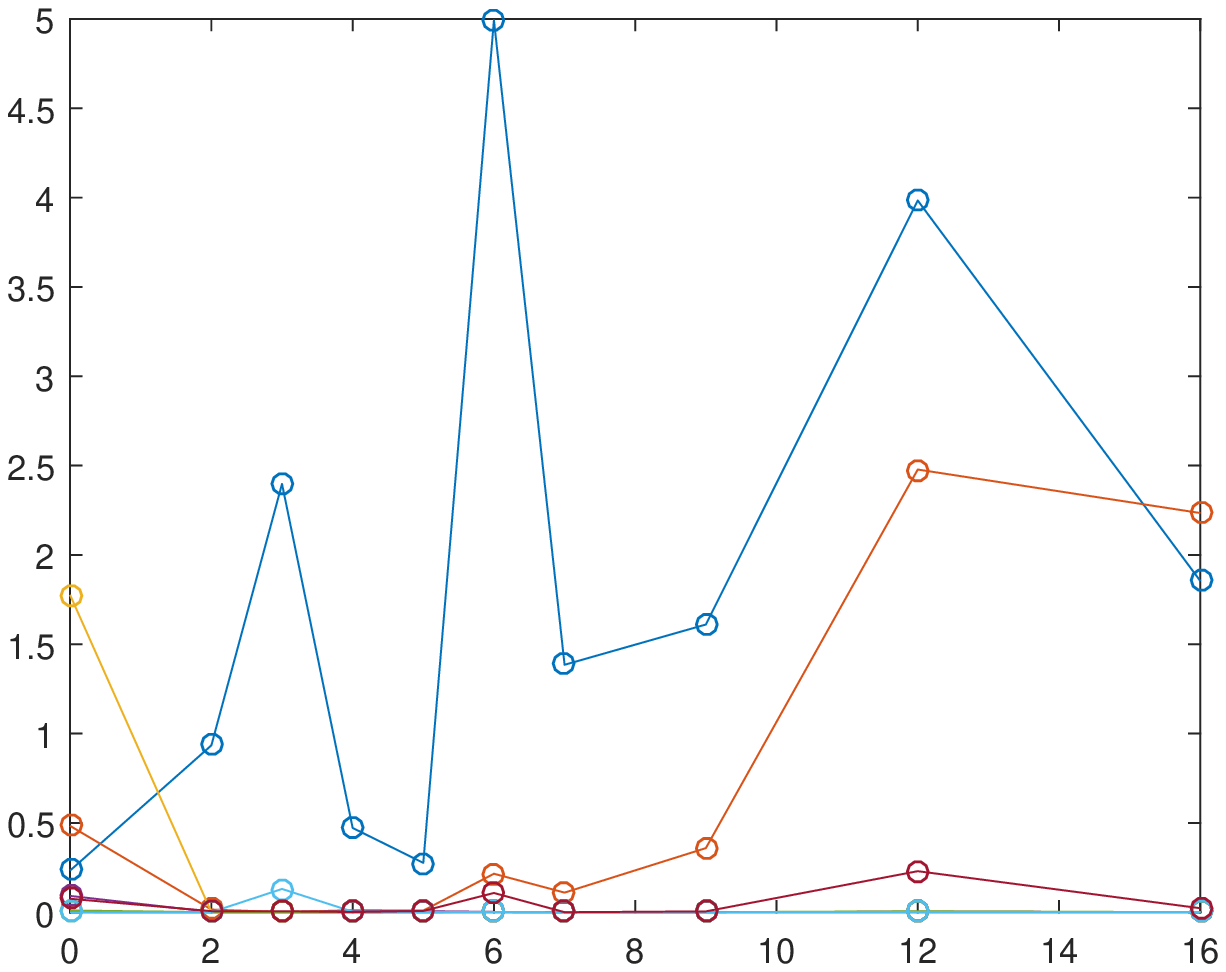}\includegraphics[scale=0.32]{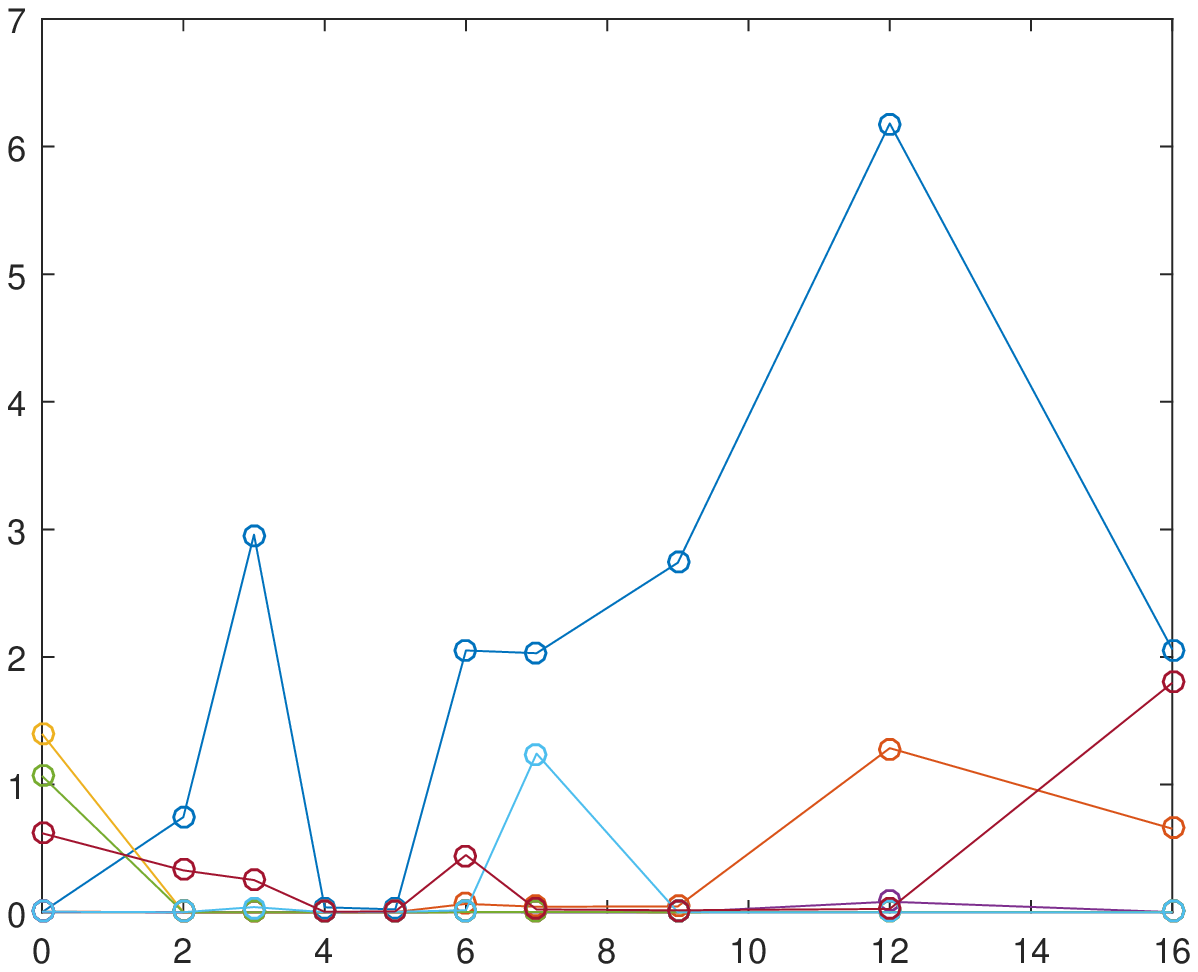}

\caption{Observed data for the three subjects belonging to group 3. Each line represents a microbiota species. \label{fig:microbiotal_raw_data}}
\end{figure}

 The original observation interval was $\left[0,\,23\right]$ but
here we restrict to $\left[0,\,16\right]$ since no data are available
on $\left]16,\,23\right[$ and a first estimation attempt on the whole
observation interval $\left[0,\,23\right]$ lead to poor data fitting
of the optimal trajectories $\overline{X_{\widehat{\theta}}^{d}}$ for any
estimator $\widehat{\theta}$ obtained for any $\left(k_{n},\,\lambda\right)$
selected values. After some trial and error, we selected $k_{n}=20$ and $\lambda=5\times10^{-6}$.
This gave the mesh size small enough to accurately estimate the ODE
perturbed model, an estimator $\widehat{\theta}^{T,CI}$ close to
the one obtained by \cite{Stein2013}, and allowance for the possibility
of model discrepancy. 

Despite the use of a simplified model and the presence of outliers
which render difficult a good data fitting of $\overline{X_{\widehat{\theta}^{T,CI}}^{d}}$
(see Figure \ref{fig:microbiotal_fitted_curve}), we obtain an estimator
consistent with the values obtained by \cite{Stein2013} when using
a more accurate model. We obtained a graph orientation close to the
one obtained in \cite{Stein2013} with only $4$ out of the $31$
estimated interaction parameters having a different sign (see Table
\ref{tab:est_param_microb_real_case}). This confirms the benefit
of using the approximated methods for real data analysis, where model
uncertainty presence is the rule rather than the exception. 
\begin{figure}
\includegraphics[scale=0.32]{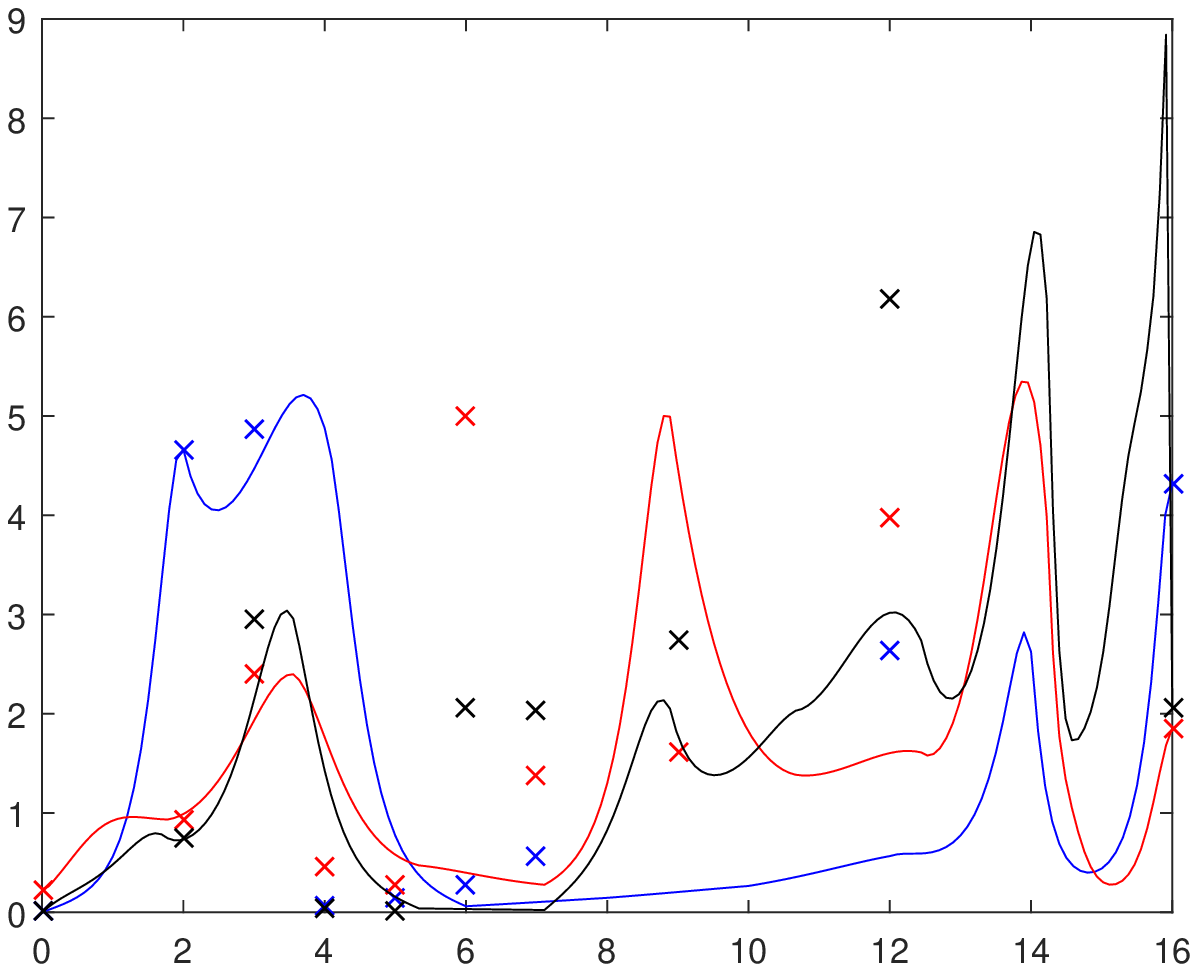}\includegraphics[scale=0.32]{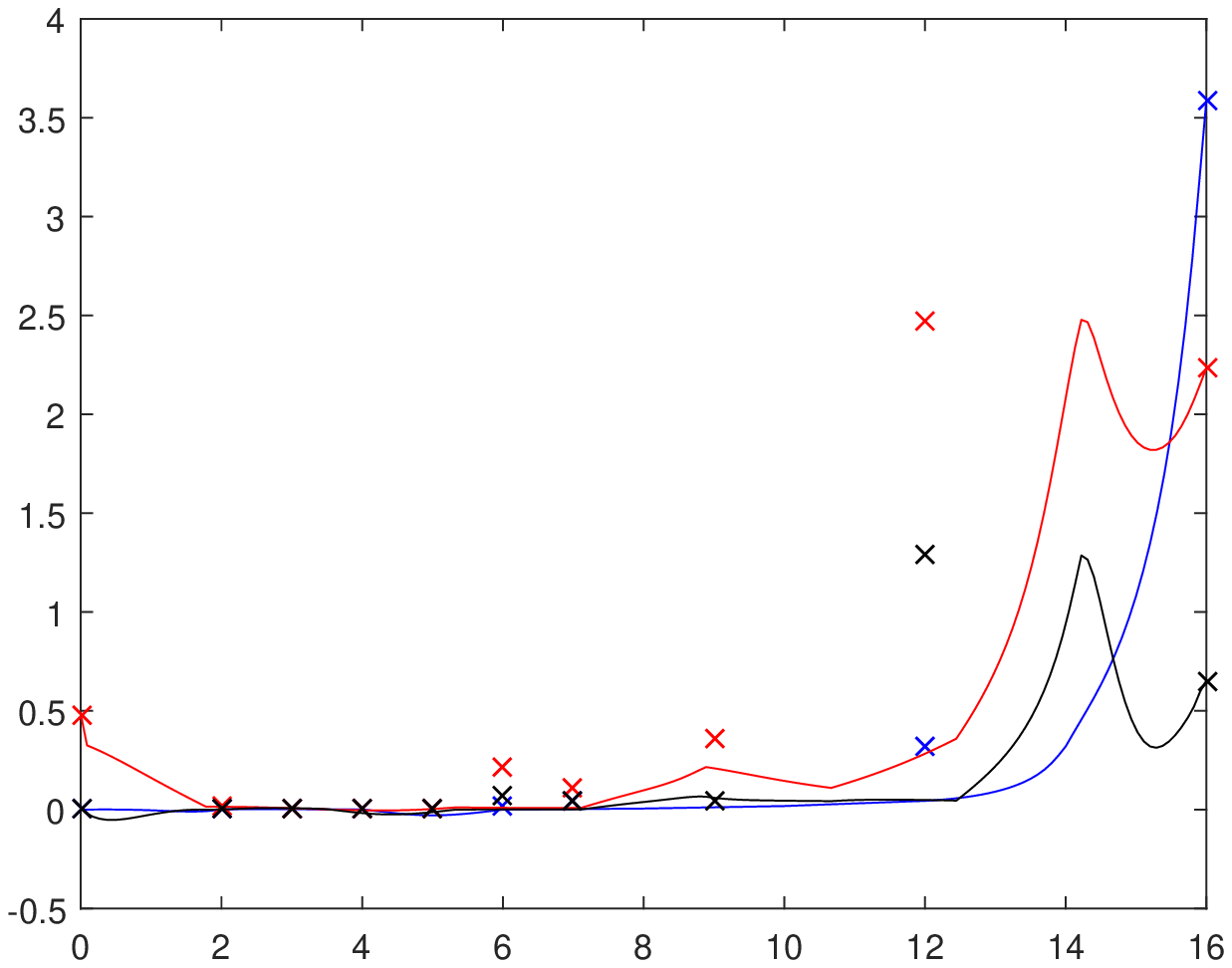}\includegraphics[scale=0.32]{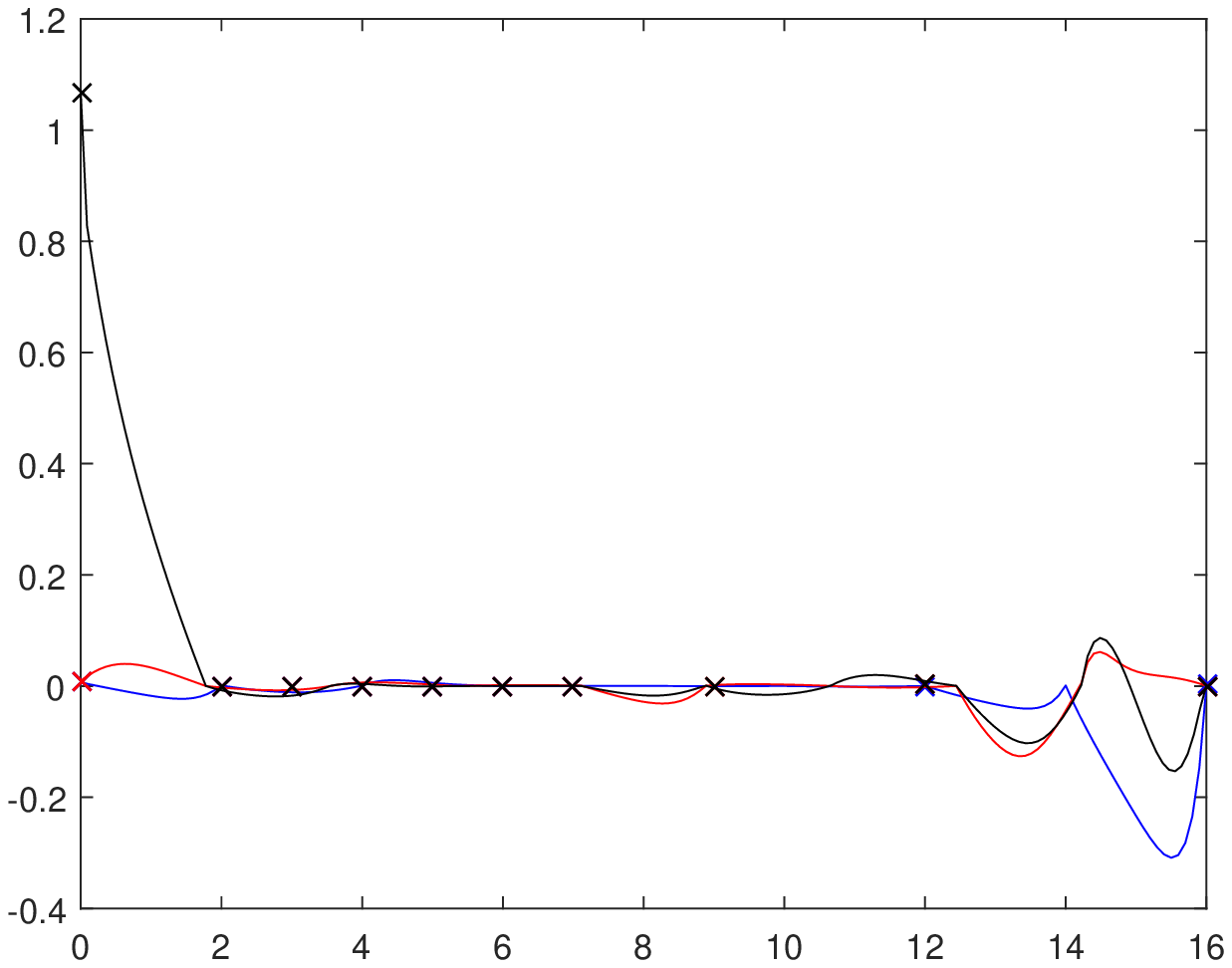}
\caption{First, second and fifth component of $\overline{X_{\widehat{\theta}^{T,CI}}^{d}}$
and corresponding raw data for subject 1 (blue), 2 (red), 3 (black) in Group 3 \label{fig:microbiotal_fitted_curve}}
\end{figure}

\begin{table}
	\caption{Scaled (by $10^{2}$) values for \cite{Stein2013} estimator $\widehat{\theta}^{Stein}$ and ours $\widehat{\theta}^{T,CI}$ \label{tab:est_param_microb_real_case}}
		\begin{tabular}{|c|c|c|c|c|c|c|c|c|c|c|c|}
			\hline 
			$\widehat{\theta}^{Stein}$ & -1 & -7 & 18 & -29 & 1 & -30 & 6 & -8 & 98 & -10  &\tabularnewline
			\hline 
			$\widehat{\theta}^{T,CI}$ & -31 & -28 & -31 & -20 & -10 & -31 & 7 & -6 & 235 & -7 &
			
			\tabularnewline
			\hline 
			\hline 
			$\widehat{\theta}^{Stein}$ & -19 & 36 & -10 &-23 & -55 & -16 & -20 & -79 & 10 & 9 &  \tabularnewline
			\hline 
			$\widehat{\theta}^{T,CI}$  & -59 & -176 & -6 & -86 & -47 & -30 & -21 & -80 & 30 & 18 &\tabularnewline

			\hline 
			\hline 
			$\widehat{\theta}^{Stein}$ & -191 & -28 & -16 & -50 & 10 & 48 & 69 & 37 & 37 & 67 & -49 \tabularnewline
			\hline 
			$\widehat{\theta}^{T,CI}$  & -458 & -52 & -33 & -44 & 1 & 27 & 64 & 38 & 28 & 36& 10   \tabularnewline
			\hline 
	\end{tabular}
\end{table}

Our methodology copes with potential model misspecification by limiting
its effect on estimation. However, our approach may also be useful
for investigating the possibility of misspecification presence by
analyzing the optimal control $\overline{u_{\widehat{\theta}^{T,CI}}^{d}}$
values, which can be seen as residuals quantifying the discrepancy
between the model and the actual system dynamic. We present in Figure
\ref{fig:microbiotal_optimal_control} for each subject the components
of $\overline{u_{\widehat{\theta}^{T,CI}}^{d}}$ corresponding to
the $\overline{X_{\widehat{\theta}^{T,CI}}^{d}}$ components presented
in Figure \ref{fig:microbiotal_fitted_curve}. 
\begin{figure} 
\includegraphics[scale=0.32]{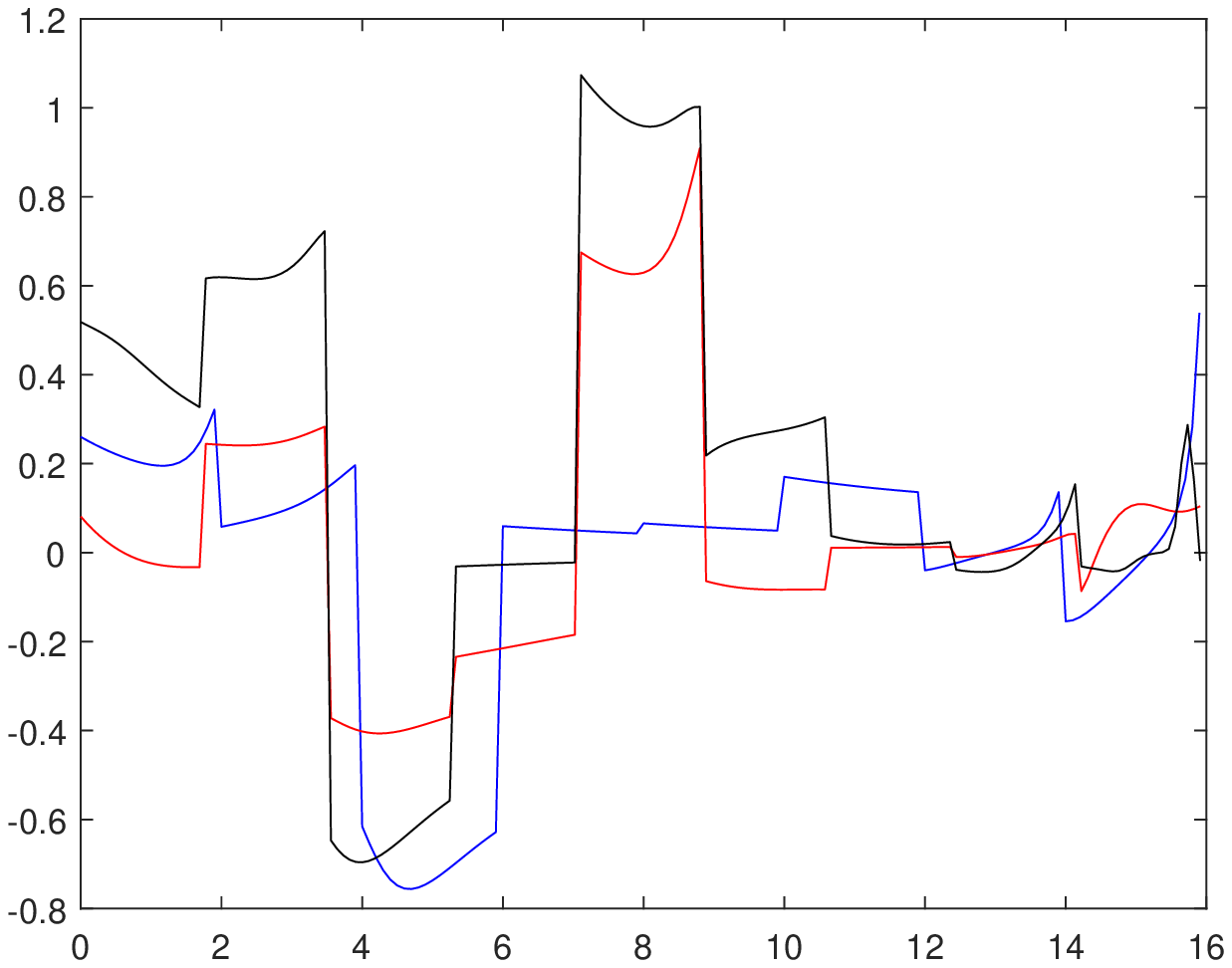}\includegraphics[scale=0.32]{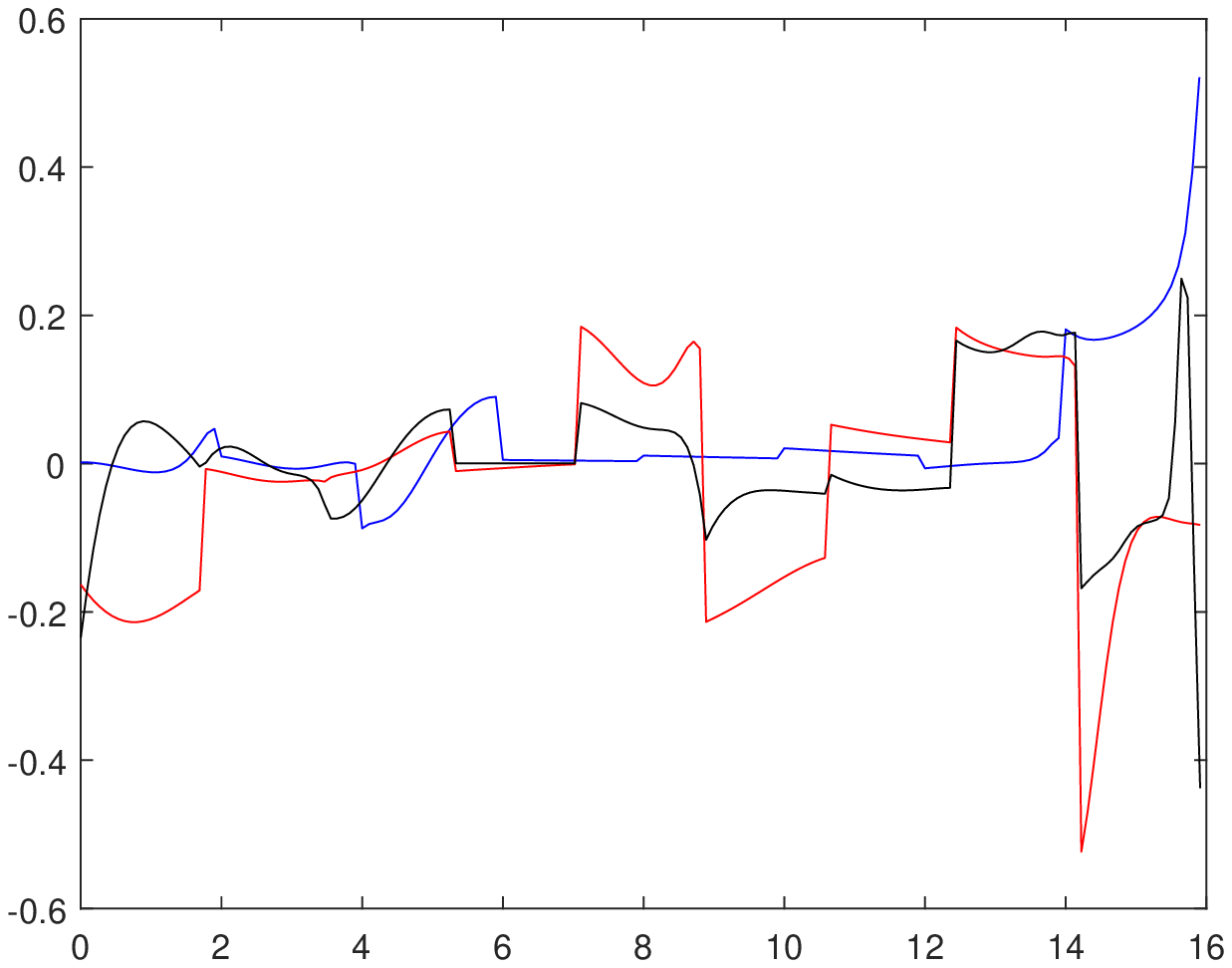}\includegraphics[scale=0.32]{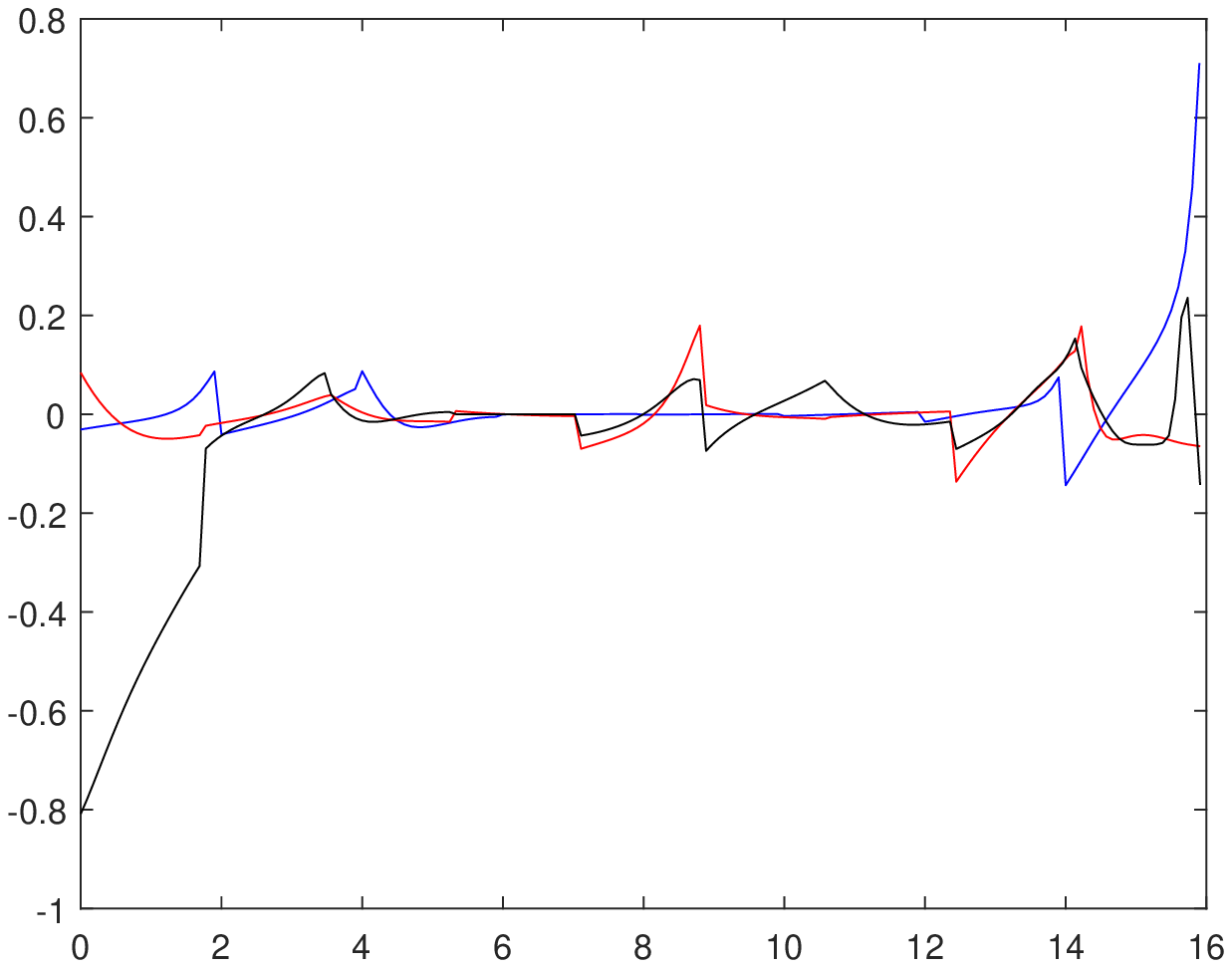}

\caption{First, second and fifth component of $\overline{u_{\widehat{\theta}^{T,CI}}^{d}}$ for subject 1 (blue), 2 (red), 3 (black) in group 3. \label{fig:microbiotal_optimal_control}}
\end{figure}
 One can see clearly there are shared patterns, for example in the
first graph on $\left[3,\,6\right]$, where the optimal controls present
the same behavior for all subjects. Such features indicate that some
deterministic elements of the actual system dynamic have been missed
by the assumed model. 

\section{Conclusion}

The advantages of using control theory to propose a method which regularizes
the estimation problem in ODEs when facing practical identifiability
issues and model misspecification have been pointed out in \cite{BrunelClairon_Kalman2014,BrunelClairon_Pontryagin2017,BrunelClairon_Riccati2014}.
However, the existing procedures based on continuous control theory
can be time consuming for high dimensional models and require a sufficient
number of observations, given that these methods rely on a non parametric
estimator $\widehat{Y}$ of the observed curve. The methodological
and theoretical sections of this work demonstrate we can construct
a consistent estimator with a parametric convergence rate based on
discrete control theory which overcomes these two problems. They also
expose how we can easily profile on the initial conditions to avoid
the estimation of additional nuisance parameters. The experimental
and real data analyses confirm the good performance of our method
in comparison with non-linear least squares and generalized profiling
for small sample case, where the asymptotic analysis results do not
hold. 

An under-exploited feature of the method so far is the obtained optimal
control. Here, we only use it for a qualitative based analysis in
the real data case, but we suspect that a full analysis of $\overline{u_{\widehat{\theta}^{T,CI}}^{d}}$
maybe be useful to construct a statistical test of misspecification
at the derivative level, which is more relevant for such models than
the test based on residuals \cite{HookerEllner2013}. This is a subject
for further work. A second point worth exploring in the future is
the extension to mixed effect model in which several subjects are
observed and despite that they present different trajectories it is
assumed their dynamics are ruled by the same evolution law. It means
each subject $i$ follows the equation $\dot{X}=f(t,X,\theta_{i})$
where $f$ is common to the whole population but $\theta_{i}$ is
an individual parameter defined as the realization of a random variable
following a law $p$ depending to a population parameter $\theta$
i.e. $\theta_{i}\sim p(\theta)$. For these models, dedicated methods
are necessary to incorporate inter-patient correlation in the estimation
process \cite{Raftery2010,Donnet2006,Lavielle2011,Pragues2013,Wang2014}.
For our method, it would be interesting to consider mixed-effect on
the estimated optimal controls $\overline{u_{i}^{d}}$ to take into
account correlations on the commited model error among the individuals.

\section*{Acknowledgements}
Quentin Clairon was supported by EPSRC award EP/M015637/1.

 \bibliographystyle{rss}
\bibliography{biblio_OCA_iterative_approach}

\end{document}


\section{Introduction}
\subsection{Model and objectives}

We are interested in the estimation of the parameter $\theta^{*}$
from data $y_{1},\dots,y_{n}$, observed on the interval $\left[0,\,T\right]$,
that are realizations of the observation process for $i=0,\dots,n$:
\begin{equation}
Y_{i}=CX^{*}(t_{i})+\epsilon_{i}\label{eq:ObservationProcess}
\end{equation}
where $X^{*}$ is the solution to the Initial Value Problem (IVP):
\begin{equation}
\left\{ \begin{array}{l}
\dot{x}(t)=A_{\theta}(x(t),t)x(t)\\
x(0)=x_{0}
\end{array}\right.\label{eq:ODEmodel}
\end{equation}
for $\left(\theta,x_{0}\right)=\left(\theta^{*},x_{0}^{*}\right).$

Here, we assume the mesh-size $\triangle_{i}$ between points is uniform
i.e $\triangle_{i}=\triangle=\frac{T}{n}$ and $\epsilon_{i}\sim N(0,\sigma^{2}I_{d'})$
are i.i.d. For the purpose of estimation, we introduce the sequence
of discretized costs $C_{T}^{d,l}$ at the observation points: 
\begin{equation}
\begin{array}{lll}
C_{T}^{d,l}(Y;\theta,x_{0},u) & = & \sum_{i=0}^{n}\triangle\left\Vert CX_{\theta,x_{0},u}^{d}(t_{i})-Y_{i}\right\Vert _{2}^{2}+\sum_{i=0}^{n-1}\triangle u_{i}{}^{T}Uu_{i}\\
& = & \sum_{i=0}^{n}X_{\theta,x_{0},u}^{d,l,e}(t_{i})^{T}\triangle W_{i}X_{\theta,x_{0},u}^{d,l,e}(t_{i})+\sum_{i=0}^{n-1}u_{i}{}^{T}\triangle Uu_{i}
\end{array}\label{eq:Discrete_cost}
\end{equation}
with $W_{i}=\left(\begin{array}{cc}
C^{T}C & -C^{T}Y_{i}\\
-Y_{i}{}^{T}C & Y_{i}^{T}Y_{i}
\end{array}\right)$ and $U$ a positive definite matrix. Here $X_{\theta,x_{0},u}^{d,l,e}$
is ruled by the finite difference equation: 
\begin{equation}
\left\{ \begin{array}{l}
X_{\theta,x_{0},u}^{d,l,e}(t_{i+1})=A_{G,\theta,i}^{d,l,e}X_{\theta,x_{0},u}^{d,l,e}(t_{i})+B_{1}\triangle u_{i}\\
X_{\theta,x_{0},u}^{d,l,e}(0)=X^{e}(0)=\left(x_{0},1\right)
\end{array}\right.\label{eq:extended_finitediffmodel}
\end{equation}
with $A_{G,\theta,i}^{d,l,e}=I_{d+1}+\Delta A_{G,\theta}(\overline{X_{\theta,x_{0}}^{d,l-1,e}}(t_{i}),t_{i})$,
$A_{G,\theta}(x,t)=\left(\begin{array}{ll}
A_{\theta}(x,t) & 0_{d,1}\\
0_{1,d} & 0
\end{array}\right)$ and $B_{1}=\left(\begin{array}{c}
B\\
0_{1,d_{u}}
\end{array}\right).$ In this case it is known (see \citet{Sontag1998} for example) the
cost $C_{T}^{d,l}(Y;\theta,x_{0},u)$ has a global minimum in $u$
, denoted $\overline{u_{\theta}^{d,l}}$ and equal to:
\begin{equation}
\overline{u_{\theta,i}^{d,l}}=-(U+\triangle B_{1}^{T}E_{\theta,i+1}^{d,l}(Y)B_{1})^{-1}B_{1}^{T}E_{\theta,i+1}^{d,l}(Y)A_{G,\theta,i}^{d,l,e}\overline{X_{\theta,x_{0}}^{d,l,e}}(t_{i})\label{eq:Optimalcontrol_sequences}
\end{equation}
with $E_{\theta,i}^{d,l}(Y)$ the solution of the discrete Riccati
equation:
\begin{equation}
\left\{ \begin{array}{l}
E_{\theta,i}^{d,l}(Y)=\left(A_{G,\theta,i}^{d,l,e}\right)^{T}E_{\theta,i+1}^{d,l}(Y)A_{G,\theta,i}^{d,l,e}+\triangle W_{i}\\
-\left(A_{G,\theta,i}^{d,l,e}\right)^{T}E_{\theta,i+1}^{d,l}(Y)\triangle B_{1}[U+\triangle B_{1}^{T}E_{\theta,i+1}^{d,l}(Y)B_{1}]^{-1}B_{1}^{T}E_{\theta,i+1}^{d,l}(Y)A_{G,\theta,i}^{d,l,e}\\
E_{\theta,n}^{d,l}(Y)=\triangle W_{n}
\end{array}\right.\label{eq:discrete_Riccati_equation}
\end{equation}
and $\overline{X_{\theta,x_{0}}^{d,l,e}}$ the extended optimal trajectory
i.e the solution of (\ref{eq:extended_finitediffmodel}) for the control
$\overline{u_{\theta}^{d,l}}$. Moreover, the minimum cost value is
equal to: 
\[
S_{n}^{l}(Y;\theta,x_{0})\,:=\inf_{u}C_{T}^{d,l}(Y;\theta,x_{0},u)=X^{e}(0)^{T}E_{\theta,0}^{d,l}(Y)X^{e}(0).
\]
From this, it is also easy to see that $\inf_{u}C_{T}^{d,l}(Y;\theta,x_{0},u)$
is a quadratic form w.r.t $x_{0}$ and then we can profile $\inf_{u}C_{T}^{d,l}(Y;\theta,x_{0},u)$
on $x_{0}$ without complexifying too much the computational problem.
We also introduce the cost sequentially profiled on $u$ then $x_{0}$:
\[
S_{n}^{CI,l}(Y;\theta)\,:=\inf_{x_{0}}\inf_{u}C_{T}^{d,l}(Y;\theta,x_{0},u)=\inf_{x_{0}}X^{e}(0)^{T}E_{\theta,0}^{d,l}(Y)X^{e}(0).
\]
Let us now assume we have access to the true continuous signal: $t\longrightarrow Y^{*}(t)=CX_{\theta^{*},x_{0}^{*}}(t)$
then we can define the sequence of continuous costs: 
\begin{equation}
\begin{array}{lll}
C_{T}^{l}(\theta,x_{0},u) & = & d^{'}\sigma^{2}+\int_{0}^{T}\left(\left\Vert CX_{\theta,x_{0},u}^{l}(t)-Y^{*}(t)\right\Vert _{2}^{2}+u(t)^{T}Uu(t)\right)dt\\
& = & \int_{0}^{T}\left(X_{\theta,x_{0},u}^{l,e}(t)^{T}W(t)X_{\theta,x_{0},u}^{l,e}(t)+u(t)^{T}Uu(t)\right)dt.
\end{array}\label{eq:Continous_cost}
\end{equation}
with $W(t)=\left(\begin{array}{cc}
C^{T}C & -C^{T}Y^{*}(t)\\
-Y^{*}(t)^{T}C & Y^{*}(t)^{T}Y^{*}(t)+d^{'}\sigma^{2}
\end{array}\right)$ and $X_{\theta,x_{0},u}^{l,e}$ the solution of the extended ODE: 

\begin{equation}
\left\{ \begin{array}{l}
\dot{X_{\theta,x_{0},u}^{l,e}}(t)=A_{G,\theta}(\overline{X_{\theta,x_{0}}^{l-1,e}}(t),t)X_{\theta,x_{0},u}^{l,e}(t)+B_{1}u(t)\\
X_{\theta,x_{0},u}^{l,e}(0)=X^{e}(0)=\left(x_{0},1\right).
\end{array}\right.\label{eq:extended_ODEmodel}
\end{equation}
As in the discrete case, the cost (\ref{eq:Continous_cost}) has
a unique minimum on $u$ , denoted $\overline{u_{\theta}^{l}}$ and
equal to: $\overline{u_{\theta}^{l}}(t)=-U^{-1}B_{1}^{T}E_{\theta}^{l}(t)\overline{X_{\theta,x_{0}}^{l,e}}(t)$
with $E_{\theta}^{l}$ the solution of the continuous Riccati equation:
\begin{equation}
\left\{ \begin{array}{l}
\dot{E_{\theta}^{l}}(t)=-W(t)-A_{G,\theta}(\overline{X_{\theta,x_{0}}^{l-1,e}}(t),t)^{T}E_{\theta}^{l}(t)-E_{\theta}^{l}(t)A_{G,\theta}(\overline{X_{\theta,x_{0}}^{l-1,e}}(t),t)\\
+E_{\theta}^{l}(t)B_{1}U^{-1}B_{1}^{T}E_{\theta}^{l}(t)\\
E_{\theta}^{l}(T)=0_{d+1,d+1}
\end{array}\right.\label{eq:continuous_Riccati_equation}
\end{equation}
and $\overline{X_{\theta,x_{0}}^{l,e}}$ the extended optimal trajectory
i.e. the solution corresponding of (\ref{eq:extended_finitediffmodel})
for the control $u:=\overline{u_{\theta}^{l}}$. We can also express
the profiled cost on $u$ then $x_{0}$ in a similar way as in the
discrete case:
\[
\left\{ \begin{array}{l}
S^{l}(\theta,x_{0})\,:=\inf_{u}C_{T}^{l}(\theta,x_{0},u)=X^{e}(0)^{T}E_{\theta}^{l}(0)X^{e}(0)\\
S^{CI,l}(\theta)\,:=\inf_{x_{0}}\inf_{u}C_{T}^{l}(\theta,x_{0},u)=\inf_{x_{0}}X^{e}(0)^{T}E_{\theta}^{l}(0)X^{e}(0).
\end{array}\right.
\]
Finally, we introduce the asymptotic cost: 
\begin{equation}
\begin{array}{lll}
C_{T}^{\infty}(\theta,x_{0},u) & = & d^{'}\sigma^{2}+\int_{0}^{T}\left(\left\Vert CX_{\theta,x_{0},u}^{\infty}(t)-Y^{*}(t)\right\Vert _{2}^{2}+u(t)^{T}Uu(t)\right)dt\\
& = & \int_{0}^{T}\left(X_{\theta,x_{0},u}^{\infty,e}(t)^{T}W(t)X_{\theta,x_{0},u}^{\infty,e}(t)+u(t)^{T}Uu(t)\right)dt
\end{array}\label{eq:as_Continous_cost}
\end{equation}
associated to the ODE:
\begin{equation}
\left\{ \begin{array}{l}
\dot{X_{\theta,x_{0},u}^{\infty,e}}=A_{G,\theta}(\overline{X_{\theta,x_{0}}^{\infty,e}}(t),t)X_{\theta,x_{0},u}^{\infty,e}+B_{1}u(t)\\
X_{\theta,x_{0},u}^{\infty,e}(0)=X^{e}(0)=\left(x_{0},\,1\right).
\end{array}\right.\label{eq:asymptotic_ODEmodel}
\end{equation}
Again, for a given couple $\left(\theta,x_{0}\right)$, the optimal
control denoted $\overline{u_{\theta}^{\infty}}$ is equal to: 
\[
\overline{u_{\theta}^{\infty}}(t)=-U^{-1}B_{1}^{T}E_{\theta}^{\infty}(t)\overline{X_{\theta,x_{0}}^{\infty,e}}(t)
\]
by introducing the continuous Riccati equation: 
\begin{equation}
\left\{ \begin{array}{l}
\dot{E_{\theta}^{\infty}}(t)=-W(t)-A_{G,\theta}(\overline{X_{\theta,x_{0}}^{\infty,e}}(t),t)^{T}E_{\theta}^{\infty}(t)-E_{\theta}^{\infty}(t)A_{G,\theta}(\overline{X_{\theta,x_{0}}^{\infty,e}}(t),t)\\
+E_{\theta}^{\infty}(t)B_{1}U^{-1}B_{1}^{T}E_{\theta}^{\infty}(t)\\
E_{\theta}^{\infty}(T)=0_{d+1,d+1}.
\end{array}\right.\label{eq:asymptotic_continuous_Riccati_equation}
\end{equation}
and we can express the profiled cost values $S^{\infty}(\theta,x_{0})\,:=\inf_{u}C_{T}^{\infty}(\theta,x_{0},u)=X^{e}(0)^{T}E_{\theta}^{\infty}(0)X^{e}(0)$
and $S^{CI,\infty}(\theta)\,:=\inf_{x_{0}}X^{e}(0)^{T}E_{\theta}^{\infty}(0)X^{e}(0)$
. In the different costs definition, we dropped the dependance in
$U$ because no asymptotic behavior conditions are required on it
for the next proofs. 

In the linear case since $A_{\theta}$ does not depends on $x$ anymore,
we have $C_{T}^{d,l_{1}}=C_{T}^{d,l_{2}}$ and $C_{T}^{l_{1}}=C_{T}^{l_{2}}=C_{T}^{\infty}$
for all $l_{1}$ and $l_{2}$ belonging to $\mathbb{N}$. Thus, there
is no need to consider the asymptotic in $l$ in this case and we
drop the dependance in $l$ in all quantities of interest i.e. $C_{T}^{d}:=C_{T}^{d,l}$,
$C_{T}:=C_{T}^{l}$ and $S_{n}^{CI}:=S_{n}^{CI,l}$ and for the solution
of the Riccati equation.

\subsection{Hypothesis \& Notations}

\subsubsection{Hypothesis for $\widehat{\theta}^{T}$ in the general case.}
\begin{description}
	\item [{Condition}] 1: For all $t\in\left[0,\,T\right]$ and for all $\theta\in\Theta$,
	$x\longmapsto A_{\theta}(x,t)$ has a compact support $\varLambda$.
\end{description}
\begin{description}
	\item [{Condition}] 2: For all $x\in\varLambda$, $\theta\longmapsto A_{\theta}(x,.)$
	is continuous on $\varTheta$ and $\forall\theta\in\varTheta$, $(x,t)\longmapsto A_{\theta}(x,t)$
	is continuous on $\varLambda\times\left[0,\,T\right]$.
\end{description}
\begin{description}
	\item [{Condition}] 3: Matrix $B$ has independent columns.
\end{description}
\begin{description}
	\item [{Condition}] 4: The true parameters $\left(\theta^{*},x_{0}^{*}\right)$
	belong to the interior of $\varTheta\times\chi$.
\end{description}
\begin{description}
	\item [{Condition}] 5: The solution $X_{\theta,x_{0}}$ of (\ref{eq:ODEmodel})
	is such that if $CX_{\theta,x_{0}}(t)=CX_{\theta^{*},x_{0^{*}}}(t)$
	for all $t\in\left[0,\,T\right]$ then $\left(\theta,x_{0}\right)=\left(\theta^{*},x_{0^{*}}\right)$.
\end{description}
\begin{description}
	\item [{Condition}] 6: For all $x\in\varLambda$, $\theta\longmapsto A_{\theta}(x,.)$
	is differentiable on $\varTheta$, for all $\theta\in\varTheta$,
	$(x,t)\longmapsto\frac{\partial A_{\theta}(x,t)}{\partial\theta}$
	is continuous on $\varLambda\times\left[0,\,T\right]$.
\end{description}
\begin{description}
	\item [{Condition}] 7: For all $x\in\varLambda$, $\theta\longmapsto A_{\theta}(x,.)$
	is twice differentiable on $\varTheta$, for all $\theta\in\varTheta$,
	$(x,t)\longmapsto\frac{\partial^{2}A_{\theta}(x,t)}{\partial^{2}\theta}$
	is continuous on $\varLambda\times\left[0,\,T\right]$.
\end{description}
\begin{description}
	\item [{Condition}] 8: The asymptotic hessian matrix $\frac{\partial^{2}S^{\infty}(\theta^{*},x_{0}^{*})}{\partial^{2}\left(\theta,x_{0}\right)}$
	is nonsingular.
\end{description}

\subsubsection{Hypothesis for $\widehat{\theta}^{T,CI}$ in the linear case}
\begin{description}
	\item [{Condition}] L1: For all $\theta\in\varTheta$, $t\longmapsto A_{\theta}(t)$
	is differentiable on $\left[0,\,T\right]$.
\end{description}
\begin{description}
	\item [{Condition}] L2: $\theta\longmapsto A_{\theta}$ is continuous on
	$\varTheta$.
\end{description}
\begin{description}
	\item [{Condition}] L3: For all $\theta\in\varTheta$, $R_{\theta}(0)$
	is nonsingular, where $R_{\theta}$ is defined by ODE (\ref{eq:continous_accurate_Riccati_equation}).
\end{description}
\begin{description}
	\item [{Condition}] L4: The true parameter $\theta^{*}$ belongs to the
	interior of $\varTheta$.
\end{description}
\begin{description}
	\item [{Condition}] L5: The solution $X_{\theta,x_{0}}$ of (\ref{eq:ODEmodel})
	is such that if $CX_{\theta,x_{0}}(t)=CX_{\theta^{*},x_{0^{*}}}(t)$
	for all $t\in\left[0,\,T\right]$ then $\left(\theta,x_{0}\right)=\left(\theta^{*},x_{0^{*}}\right)$.
\end{description}
\begin{description}
	\item [{Condition}] L6: $\theta\longmapsto A_{\theta}$ is $C^{2}$ on
	$\varTheta$.
\end{description}
\begin{description}
	\item [{Condition}] L7: The asymptotic hessian matrix $\frac{\partial^{2}S^{CI}(\theta^{*})}{\partial^{2}\theta}$
	is a nonsingular matrix.
\end{description}

\subsection{Notation}

We denote:
\begin{enumerate}
	\item $\overline{A}=\sup_{\left(\theta,x,t\right)\in\varTheta\times\varLambda\times\left[0,\,T\right]}\left\Vert A_{\theta}(x,t)\right\Vert _{2}$
	.
	\item $\overline{\partial A}=\sup_{\left(\theta,x,t\right)\in\varTheta\times\varLambda\times\left[0,\,T\right]}\left\Vert \frac{\partial A_{\theta}(x,t)}{\partial\theta}\right\Vert _{2}$
	.
	\item $\overline{\partial^{2}A}=\sup_{\left(\theta,x,t\right)\in\varTheta\times\varLambda\times\left[0,\,T\right]}\left\Vert \frac{\partial^{2}A_{\theta}(x,t)}{\partial^{2}\theta}\right\Vert _{2}$
	.
	\item $o_{l}(1)$ an arbitrary function $g$ (possibly vector or matrix
	valued) such that $\lim_{l\longmapsto\infty}g(l)=0$.
	\item $o_{n}(f(\triangle))$ (resp $O_{n}(f(\triangle))$ ) an arbitrary
	function $g$ (possibly vector or matrix valued) such that $\lim_{n\longmapsto\infty}\frac{g(\triangle)}{f(\triangle)}=\lim_{\triangle\longmapsto0}\frac{g(\triangle)}{f(\triangle)}=0.$
	(resp $\lim_{n\longmapsto\infty}\frac{g(\triangle)}{f(\triangle)}=L$
	with $L$ constant and finite).
	\item $o_{p,n}(f(\triangle))$ (resp $O_{p,n}(f(\triangle))$ ) a random
	variable $g$ such that $\frac{g(\triangle)}{f(\triangle)}$ tends
	to $0$ in probability when $n\longrightarrow+\infty$ ( resp $\frac{g(\triangle)}{f(\triangle)}$
	is bounded in probability when $n\longrightarrow+\infty$).
	\item For the proof of estimator consistency and asymptotic normality in
	the nonlinear case, we often use the notation $\upsilon=\left(\theta,\,x_{0}\right)$
	as well as $\varUpsilon=\Theta\times\chi$. 
\end{enumerate}
For the sake of notation we denote $\widehat{\theta}$ for both $\widehat{\theta}^{T}$
and $\widehat{\theta}^{T,CI}$ , the ambiguity being clarified by
the context.

\section{Profiled costs: alternative representations, well-definednesses and
	regularities}

In this section, we derive the expressions of $S_{n}^{l},\,S^{l}$
and $S^{\infty}$ (resp. $S_{n}^{CI,l},\,S^{CI,l}$ and $S^{CI,\infty}$)
w.r.t to the solutions of simplified versions of the original Riccati
equations (\ref{eq:discrete_Riccati_equation}, \ref{eq:continuous_Riccati_equation}
and \ref{eq:asymptotic_continuous_Riccati_equation}). The interest
is twofold, first it reduces the computational burden for $S_{n}^{l}$
and $S_{n}^{CI,l}$ optimization and make the theoretical asymptotic
analysis of our estimator easier by specifying precisely where the
measurement errors intervene in $S_{n}^{l}$ and $S_{n}^{CI,l}$.
In this section, we denote $\overline{X_{\upsilon}^{l}}$ and $\overline{X_{\upsilon}^{d,l}}$
the discrete and continous optimal trajectories whether or not we
profile on the initial conditions.
\begin{prop}
	\label{prop:continuous_Riccati_accurate_representation}We have: 
	\[
	\left\{ \begin{array}{lll}
	S^{l}(\theta,x_{0}) & = &  x_{0}^{T}R_{\theta}^{l}(0)x_{0}+2x_{0}^{T}h_{\theta}^{l}(0)\\
 &+ &\int_{0}^{T}\left(Y^{*}(t)^{T}Y^{*}(t)+d^{'}\sigma^{2}-h_{\theta}^{l}(t)^{T}BU^{-1}B^{T}h_{\theta}^{l}(t)\right)dt\\
	S^{CI,l}(\theta)  &= & -h_{\theta}^{l}(0)^{T}R_{\theta}^{l}(0)^{-1}h_{\theta}^{l}(0)\\
 &+ &\int_{0}^{T}\left(Y^{*}(t)^{T}Y^{*}(t)+d^{'}\sigma^{2}-h_{\theta}^{l}(t)^{T}BU^{-1}B^{T}h_{\theta}^{l}(t)\right)dt
	\end{array}\right.
	\]
	with:
	\begin{equation}
	\left\{ \begin{array}{l}
	\dot{R_{\theta}^{l}}(t)=-C^{T}C-A_{\theta}(\overline{X_{\upsilon}^{l-1}}(t),t)^{T}R_{\theta}^{l}(t)-R_{\theta}^{l}(t)A_{\theta}(\overline{X_{\upsilon}^{l-1}}(t),t)
\\+R_{\theta}^{l}(t)BU^{-1}B^{T}R_{\theta}^{l}(t)\\
	\dot{h_{\theta}^{l}}(t)=C^{T}Y^{*}(t)-A_{\theta}(\overline{X_{\upsilon}^{l-1}}(t),t)^{T}h_{\theta}^{l}(t)+R_{\theta}^{l}(t)BU^{-1}B^{T}h_{\theta}^{l}(t)\\
	\dot{\overline{X_{\upsilon}^{l}}(}t)=A_{\theta}(\overline{X_{\upsilon}^{l-1}}(t),t)\overline{X_{\upsilon}^{l}}(t)-BU^{-1}B^{T}(R_{\theta}^{l}(t)\overline{X_{\upsilon}^{l}}(t)+h_{\theta}^{l}(t))\\
	\left(R_{\theta}^{l}(T),\,h_{\theta}^{l}(T)\right)=\left(0_{d,d},\,0_{d,1}\right)
	\end{array}\right.\label{eq:continous_accurate_Riccati_equation}
	\end{equation}
	where $\overline{X_{\upsilon}^{l}}(0)=x_{0}$ for $S^{l}(\theta,x_{0})$
	computation and $\overline{X_{\upsilon}^{l}}(0)=-R_{\theta}^{l}(0)^{-1}h_{\theta}^{l}(0)$
	for $S^{CI,l}(\theta).$ Similarly, we have the expressions:
	\[
	\left\{ \begin{array}{lll}
	S^{\infty}(\theta,x_{0}) & = & x_{0}^{T}R_{\theta}^{\infty}(0)x_{0}+2x_{0}^{T}h_{\theta}^{\infty}(0)
\\&+&\int_{0}^{T}\left(Y^{*}(t)^{T}Y^{*}(t)+d^{'}\sigma^{2}-h_{\theta}^{\infty}(t)^{T}BU^{-1}B^{T}h_{\theta}^{\infty}(t)\right)dt\\
	S^{CI,\infty}(\theta) & = & -h_{\theta}^{\infty}(0)^{T}R_{\theta}^{\infty}(0)^{-\infty}h_{\theta}^{\infty}(0)\\
&+&\int_{0}^{T}\left(Y^{*}(t)^{T}Y^{*}(t)+d^{'}\sigma^{2}-h_{\theta}^{\infty}(t)^{T}BU^{-1}B^{T}h_{\theta}^{\infty}(t)\right)dt
	\end{array}\right.
	\]
	with:
	\begin{equation}
	\left\{ \begin{array}{l}
	\dot{R_{\theta}^{\infty}}(t)=-C^{T}C-A_{\theta}(\overline{X_{\upsilon}^{\infty}}(t),t)^{T}R_{\theta}^{\infty}(t)-R_{\theta}^{\infty}(t)A_{\theta}(\overline{X_{\upsilon}^{\infty}}(t),t)
\\+R_{\theta}^{\infty}(t)BU^{-1}B^{T}R_{\theta}^{\infty}(t)\\
	\dot{h_{\theta}^{\infty}}(t)=C^{T}Y^{*}(t)-A_{\theta}(\overline{X_{\upsilon}^{\infty}}(t),t)^{T}h_{\theta}^{\infty}(t)+R_{\theta}^{\infty}(t)BU^{-1}B^{T}h_{\theta}^{\infty}(t)\\
	\dot{\overline{X_{\upsilon}^{\infty}}(}t)=A_{\theta}(\overline{X_{\upsilon}^{\infty}}(t),t)\overline{X_{\upsilon}^{\infty}}(t)-BU^{-1}B^{T}(R_{\theta}^{\infty}(t)\overline{X_{\upsilon}^{\infty}}(t)+h_{\theta}^{\infty}(t))\\
	\left(R_{\upsilon}^{\infty}(T),\,h_{\upsilon}^{\infty}(T)\right)=\left(0_{d,d},\,0_{d,1}\right)
	\end{array}\right.\label{eq:asymptotic_continous_accurate_Riccati_equation}
	\end{equation}
	where $\overline{X_{\upsilon}^{\infty}}(0)=x_{0}$ for $S^{\infty}(\theta,x_{0})$
	and $\overline{X_{\upsilon}^{\infty}}(0)=-R_{\theta}^{\infty}(0)^{-1}h_{\theta}^{\infty}(0)$
	for $S^{CI,\infty}(\theta)$ .
	\begin{proof}
		It is easy to verify that $E_{\theta}^{l}$ is symmetric and thus
		can be decomposed under the form $E_{\theta}^{l}(t)=\left(\begin{array}{cc}
		R_{\theta}^{l}(t) & h_{\theta}^{l}(t)\\
		h_{\theta}^{l}(t)^{T} & \alpha_{\theta}^{l}(t)
		\end{array}\right)$. We now re-inject this expression into the ODE (\ref{eq:continuous_Riccati_equation})
		to obtain:{\small{}
			\[
			\begin{array}{l}
			\dot{E_{\theta}^{l}}(t)=-\left(\begin{array}{cc}
			C^{T}C & -C^{T}Y^{*}(t)\\
			-Y^{*}(t)^{T}C & Y^{*}(t)^{T}Y^{*}(t)+d^{'}\sigma^{2}
			\end{array}\right)\\
			-\left(\begin{array}{ll}
			A_{\theta}(\overline{X_{\upsilon}^{l-1}}(t),t) & 0_{d,1}\\
			0_{1,d} & 0
			\end{array}\right)^{T}\left(\begin{array}{cc}
			R_{\theta}^{l}(t) & h_{\theta}^{l}(t)\\
			h_{\theta}^{l}(t)^{T} & \alpha_{\theta}^{l}(t)
			\end{array}\right)
			\\-\left(\begin{array}{cc}
			R_{\theta}^{l}(t) & h_{\theta}^{l}(t)\\
			h_{\theta}^{l}(t)^{T} & \alpha_{\theta}^{l}(t)
			\end{array}\right)\left(\begin{array}{ll}
			A_{\theta}(\overline{X_{\upsilon}^{l-1}}(t),t) & 0_{d,1}\\
			0_{1,d} & 0
			\end{array}\right)\\
			+\left(\begin{array}{cc}
			R_{\theta}^{l}(t) & h_{\theta}^{l}(t)\\
			h_{\theta}^{l}(t)^{T} & \alpha_{\theta}^{l}(t)
			\end{array}\right)\left(\begin{array}{c}
			B\\
			0_{1,d_{u}}
			\end{array}\right)U^{-1}\left(\begin{array}{c}
			B\\
			0_{1,d_{u}}
			\end{array}\right)^{T}\left(\begin{array}{cc}
			R_{\theta}^{l}(t) & h_{\theta}^{l}(t)\\
			h_{\theta}^{l}(t)^{T} & \alpha_{\theta}^{l}(t)
			\end{array}\right)\\
			\end{array}
			\]
			\[
			\begin{array}{l}
			=-\left(\begin{array}{cc}
			C^{T}C & -C^{T}Y^{*}(t)\\
			-Y^{*}(t)^{T}C & Y^{*}(t)^{T}Y^{*}(t)+d^{'}\sigma^{2}
			\end{array}\right)\\
			-\left(\begin{array}{ll}
			A_{\theta}(\overline{X_{\upsilon}^{l-1}}(t),t)^{T}R_{\theta}^{l}(t) & A_{\theta}(\overline{X_{\upsilon}^{l-1}}(t),t)^{T}h_{\theta}^{l}(t)\\
			0_{1,d} & 0
			\end{array}\right)
			\\-\left(\begin{array}{cc}
			R_{\theta}^{l}(t)A_{\theta}(\overline{X_{\upsilon}^{l-1}}(t),t) & 0_{d,1}\\
			h_{\theta}^{l}(t)^{T}A_{\theta}(\overline{X_{\upsilon}^{l-1}}(t),t) & 0
			\end{array}\right)\\
			+\left(\begin{array}{cc}
			R_{\theta}^{l}(t) & h_{\theta}^{l}(t)\\
			h_{\theta}^{l}(t)^{T} & \alpha_{\theta}^{l}(t)
			\end{array}\right)\left(\begin{array}{cc}
			BU^{-1}B^{T} & 0_{d,1}\\
			0_{1,d} & 0
			\end{array}\right)\left(\begin{array}{cc}
			R_{\theta}^{l}(t) & h_{\theta}^{l}(t)\\
			h_{\theta}^{l}(t)^{T} & \alpha_{\theta}^{l}(t)
			\end{array}\right)\\
			=-\left(\begin{array}{cc}
			C^{T}C & -C^{T}Y^{*}(t)\\
			-Y^{*}(t)^{T}C & Y^{*}(t)^{T}Y^{*}(t)+d^{'}\sigma^{2}
			\end{array}\right)\\
			-\left(\begin{array}{ll}
			A_{\theta}(\overline{X_{\upsilon}^{l-1}}(t),t)^{T}R_{\theta}^{l}(t)+R_{\theta}^{l}(t)A_{\theta}(\overline{X_{\upsilon}^{l-1}}(t),t) & A_{\theta}(\overline{X_{\upsilon}^{l-1}}(t),t)^{T}h_{\theta}^{l}(t)\\
			h_{\theta}^{l}(t)^{T}A_{\theta}(\overline{X_{\upsilon}^{l-1}}(t),t) & 0
			\end{array}\right)\\
			+\left(\begin{array}{cc}
			R_{\theta}^{l}(t)BU^{-1}B^{T}R_{\theta}^{l}(t) & R_{\theta}^{l}(t)BU^{-1}B^{T}h_{\theta}^{l}(t)\\
			h_{\theta}^{l}(t)^{T}BU^{-1}B^{T}R_{\theta}^{l}(t) & h_{\theta}(t)^{T}BU^{-1}B^{T}h_{\theta}^{l}(t)
			\end{array}\right).
			\end{array}
			\]
		}From this we can derive the ODE followed by each element of $E_{\theta}^{l}$,
		{\small{}
			\[
			\left\{ \begin{array}{l}
			\dot{R_{\theta}^{l}}(t)=-C^{T}C-A_{\theta}(\overline{X_{\upsilon}^{l-1}}(t),t)^{T}R_{\theta}^{l}(t)-R_{\theta}^{l}(t)A_{\theta}(\overline{X_{\upsilon}^{l-1}}(t),t)
			\\+R_{\theta}^{l}(t)BU^{-1}B^{T}R_{\theta}^{l}(t)\\
			\dot{h_{\theta}^{l}}(t)=C^{T}Y^{*}(t)-A_{\theta}(\overline{X_{\upsilon}^{l-1}}(t),t)^{T}h_{\theta}^{l}(t)+R_{\theta}^{l}(t)BU^{-1}B^{T}h_{\theta}^{l}(t)\\
			\dot{\alpha_{\theta}^{l}}(t)=-Y^{*}(t)^{T}Y^{*}(t)-d^{'}\sigma^{2}+h_{\theta}^{l}(t)^{T}BU^{-1}B^{T}h_{\theta}^{l}(t)
			\end{array}\right.
			\]
		}in particular, $\alpha_{\theta}^{l}(0)=\int_{0}^{T}\left(Y^{*}(t)^{T}Y^{*}(t)+d^{'}\sigma^{2}-h_{\theta}^{l}(t)^{T}BU^{-1}B^{T}h_{\theta}^{l}(t)\right)dt$
		and:{\small{}
			\[
			\begin{array}{lll}
			\inf_{u}C_{T}^{l}(Y;\theta,x_{0},u) & = & X^{e}(0)^{T}E_{\theta}^{l}(0)X^{e}(0)\\
			& = & \left(\begin{array}{cc}
			x_{0}^{T} & 1\end{array}\right)\left(\begin{array}{cc}
			R_{\theta}^{l}(0) & h_{\theta}^{l}(0)\\
			h_{\theta}^{l}(0)^{T} & \alpha_{\theta}^{l}(0)
			\end{array}\right)\left(\begin{array}{c}
			x_{0}\\
			1
			\end{array}\right)\\
			& = & x_{0}^{T}R_{\theta}^{l}(0)x_{0}+2x_{0}^{T}h_{\theta}^{l}(0)+\alpha_{\theta}^{l}(0)\\
			& = & x_{0}^{T}R_{\theta}^{l}(0)x_{0}+2x_{0}^{T}h_{\theta}^{l}(0)
\\ &+ &\int_{0}^{T}\left(Y^{*}(t)^{T}Y^{*}(t)+d^{'}\sigma^{2}-h_{\theta}^{l}(t)^{T}BU^{-1}B^{T}h_{\theta}^{l}(t)\right)dt
			\end{array}
			\]
			hence the expression for $S^{l}(\theta,x_{0}$}). One can see $\inf_{u}C_{T}^{l}(\theta,x_{0},u)$
		is quadratic w.r.t to $x_{0}$ with the minimum reached for $\overline{x_{\theta,0}^{l}}=-R_{\theta}^{l}(0)^{-1}h_{\theta}^{l}(0)$,
		from this we easily obtain:{\small{}
			\[
			\begin{array}{lll}
			S^{CI,l}(\theta) & = & \inf_{x_{0}}\inf_{u}C_{T}^{l}(Y;\theta,x_{0},u)\\
			& = & -h_{\theta}^{l}(0)^{T}R_{\theta}^{l}(0)^{-1}h_{\theta}^{l}(0)
			\\&+&\int_{0}^{T}\left(Y^{*}(t)^{T}Y^{*}(t)+d^{'}\sigma^{2}-h_{\theta}^{l}(t)^{T}BU^{-1}B^{T}h_{\theta}^{l}(t)\right)dt.
			\end{array}
			\]
		}The optimal cost become{\small{}s
			\[
			\begin{array}{lll}
			\overline{u_{\theta}^{l}}(t) & = & -U^{-1}B_{1}^{T}E_{\theta}^{l}(t)\overline{X_{\upsilon}^{l,e}}(t)=-U^{-1}B^{T}(R_{\theta}^{l}(t)\overline{X_{\upsilon}^{l}}(t)+h_{\theta}^{l}(t))\end{array}
			\]
		}and from this we derive the ODE followed by $\overline{X_{\upsilon}^{l}}$.
		For $S^{\infty}$(resp. $S^{CI,\infty}$), the formal computation
		is almost the same as in $S^{l}$ (resp. $S^{CI,l}$) case and thus
		is omitted.
	\end{proof}
\end{prop}

\begin{prop}
	\label{prop:discrete_Riccati_accurate_representation}For all $(l,i)\in\mathbb{N}\times\left\llbracket 0,\,n\right\rrbracket $,
	$E_{\theta,i}^{d,l}(Y)$ is symmetric and can be written $E_{\theta,i}^{d,l}(Y)=\left(\begin{array}{cc}
	R_{\theta,i}^{d,l} & h_{\theta,i}^{d,l}(Y)\\
	h_{\theta,i}^{d,l}(Y)^{T} & \alpha_{\theta,i}^{d,l}(Y)
	\end{array}\right)$. Moreover, each element is ruled by the finite difference equations:
	\begin{equation}
	\begin{array}{lll}
	R_{\theta,i}^{d,l} & = & R_{\theta,i+1}^{d,l}+\triangle C^{T}C+\Delta\left(R_{\theta,i+1}^{d,l}A_{\theta}(\overline{X_{\upsilon}^{d,l-1}}(t_{i}),t_{i})+A_{\theta}(\overline{X_{\upsilon}^{d,l-1}}(t_{i}),t_{i})^{T}R_{\theta,i+1}^{d,l}\right)\\
	& + & \Delta^{2}A_{\theta}(\overline{X_{\upsilon}^{d,l-1}}(t_{i}),t_{i})^{T}R_{\theta,i+1}^{d,l}A_{\theta}(\overline{X_{\upsilon}^{d,l-1}}(t_{i}),t_{i})\\
	& - & \triangle R_{\theta,i+1}^{d,l}BG(R_{\theta,i+1}^{d,l})B^{T}R_{\theta,i+1}^{d,l}(I_{d}+\Delta A_{\theta}(\overline{X_{\upsilon}^{d,l-1}}(t_{i}),t_{i}))\\
	& - &\Delta^{2} A_{\theta}(\overline{X_{\upsilon}^{d,l-1}}(t_{i}),t_{i})^{T}R_{\theta,i+1}^{d,l}BG(R_{\theta,i+1}^{d,l})B^{T}R_{\theta,i+1}^{d,l}(I_{d}+\Delta A_{\theta}(\overline{X_{\upsilon}^{d,l-1}}(t_{i}),t_{i}))\\

	h_{\theta,i}^{d,l}(Y) & = & h_{\theta,i+1}^{d,l}(Y)-\triangle C^{T}Y_{i}+\Delta A_{\theta}(\overline{X_{\upsilon}^{d,l-1}}(t_{i}),t_{i})^{T}h_{\theta,i+1}^{d,l}(Y)\\
	& - & \triangle(I_{d}+\Delta A_{\theta}(\overline{X_{\upsilon}^{d,l-1}}(t_{i}),t_{i})^{T})R_{\theta,i+1}^{d,l}BG(R_{\theta,i+1}^{d,l})B^{T}h_{\theta,i+1}^{d,l}(Y)\\
	\overline{X_{\upsilon}^{d,l}}(t_{i+1}) & = & \left(I_{d}+\Delta A_{\theta}(\overline{X_{\upsilon}^{d,l-1}}(t_{i}),t_{i})\right)\overline{X_{\upsilon}^{d,l}}(t_{i})\\
	& - & \triangle BG(R_{\theta,i+1}^{d,l})B^{T}\left(R_{\theta,i+1}^{d,l}\left(I_{d}+\Delta A_{\theta}(\overline{X_{\upsilon}^{d,l-1}}(t_{i}),t_{i})\right)\overline{X_{\upsilon}^{d,l}}(t_{i})+h_{\theta,i+1}^{d,l}(Y)\right)
	\end{array}\label{eq:discrete_accurate_Riccati_equation}
	\end{equation}
	with final condition $(R_{\theta,n}^{d,l},\,h_{\theta,n}^{d,l}(Y))  =  \left(\triangle C^{T}C\,,\,-\triangle C^{T}Y_{n}\right)$ and $G(R_{\theta,i+1}^{d,l})\,:=\left[U+\triangle B^{T}R_{\theta,i+1}^{d,l}B\right]^{-1}=U^{-1}+O_{p,n}(\triangle)$,
	\textup{$\overline{X_{\upsilon}^{d,l}}(0)=x_{0}$} for $S_{n}^{l}(Y;\theta,x_{0})$
	and \textup{$\overline{X_{\upsilon}^{d,l}}(0)=-\left(R_{\theta,0}^{l}\right)^{-1}h_{\theta,0}^{l}$}
	for $S_{n}^{CI,l}(Y;\theta)$. Moreover, the profiled cost can be
	expressed as
	\[
	\left\{ \begin{array}{l}
	S_{n}^{l}(Y;\theta,x_{0})=x_{0}^{T}R_{\theta,0}^{d,l}x_{0}+2h_{\theta,0}^{d,l}(Y)^{T}x_{0}+\triangle Y_{n}^{T}Y_{n}\\
	+\triangle\sum_{i=0}^{n-1}\left(Y_{i}^{T}Y_{i}-h_{\theta,i+1}^{d,l}(Y)^{T}BG(R_{\theta,i+1}^{d,l})B^{T}h_{\theta,i+1}^{d,l}(Y)\right)\\
	S_{n}^{CI,l}(Y;\theta)=-h_{\theta,0}^{d,l}(Y)^{T}\left(R_{\theta,0}^{d,l}\right)^{-1}h_{\theta,0}^{d,l}(Y)+\triangle Y_{n}^{T}Y_{n}\\
	+\triangle\sum_{i=0}^{n-1}\left(Y_{i}^{T}Y_{i}-h_{\theta,i+1}^{d,l}(Y)^{T}BG(R_{\theta,i+1}^{d,l})B^{T}h_{\theta,i+1}^{d,l}(Y)\right).
	\end{array}\right.
	\]
\end{prop}
\begin{proof}
	We prove that by using the finite difference equation (\ref{eq:discrete_Riccati_equation})
	and reversed time induction. It is obvious the property holds for
	$E_{\theta,n}^{d,l}(Y)$. Now let us assume it holds for $i+1$ i.e.
	that: $E_{\theta,i+1}^{d,l}(Y)=\left(\begin{array}{cc}
	R_{\theta,i+1}^{d,l} & h_{\theta,i+1}^{d,l}(Y)\\
	h_{\theta,i+1}^{d,l}(Y)^{T} & \alpha_{\theta,i+1}^{d,l}(Y)
	\end{array}\right).$ Now let us compute $\left(A_{G,\theta,i}^{d,l,e}\right)^{T}E_{\theta,i+1}^{d,l}(Y)A_{G,\theta,i}^{d,l,e}$,
	we have:{\small{}
		\[
		\begin{array}{l}
		\left(A_{G,\theta,i}^{d,l,e}\right)^{T}E_{\theta,i+1}^{d,l}(Y)A_{G,\theta,i}^{d,l,e}\\
		
		=\left(\begin{array}{cc}
		R_{\theta,i+1}^{d,l} & h_{\theta,i+1}^{d,l}(Y)\\
		h_{\theta,i+1}^{d,l}(Y)^{T} & \alpha_{\theta,i+1}^{d,l}(Y)
		\end{array}\right)+\Delta\left(\begin{array}{ll}
		R_{\theta,i+1}^{d,l}A_{\theta}(\overline{X_{\upsilon}^{d,l-1}}(t_{i}),t_{i}) & 0_{d,1}\\
		h_{\theta,i+1}^{d,l}(Y)^{T}A_{\theta}(\overline{X_{\upsilon}^{d,l-1}}(t_{i}),t_{i}) & 0
		\end{array}\right)\\
		+\Delta\left(\begin{array}{cc}
		A_{\theta}(\overline{X_{\upsilon}^{d,l-1}}(t_{i}),t_{i})^{T}R_{\theta,i+1}^{d,l} & A_{\theta}(\overline{X_{\upsilon}^{d,l-1}}(t_{i}),t_{i})^{T}h_{\theta,i+1}^{d,l}(Y)\\
		0_{1,d} & 0
		\end{array}\right)\\
		+\Delta^{2}\left(\begin{array}{ll}
		A_{\theta}(\overline{X_{\upsilon}^{d,l-1}}(t_{i}),t_{i})^{T}R_{\theta,i+1}^{d}A_{\theta}(t_{i}) & 0_{d,1}\\
		0_{1,d} & 0
		\end{array}\right)\\
		=F(E_{\theta,i+1}^{d,l}(Y)):=\left(\begin{array}{cc}
		F_{1}(R_{\theta,i+1}^{d,l}) & F_{2}(h_{\theta,i+1}^{d,l}(Y\text{)})\\
		F_{2}(h_{\theta,i+1}^{d,l}(Y\text{)})^{T} & F_{3}(\alpha_{\theta,i+1}^{d,l}(Y))
		\end{array}\right)
		\end{array}
		\]
	}where the symmetric matrix $F$ is easily derivable by identification
	and the $d$ dimensional square matrix $F_{1}$ only depends on $R_{\theta,i+1}^{d,l}$.
	Now let us compute the expression in $[U+\triangle B_{1}^{T}E_{\theta,i+1}^{d,l}(Y)B_{1}]^{-1}$,
	we obtain: {\small{}
		\[
		\begin{array}{lll}

		[U+\triangle B_{1}^{T}E_{\theta,i+1}^{d,l}(Y)B_{1}]^{-1} & = & \left[U+\triangle\left(\begin{array}{cc}
		B^{T}R_{\theta,i+1}^{d,l} & B^{T}h_{\theta,i+1}^{d,l}(Y)\end{array}\right)\left(\begin{array}{c}
		B\\
		0_{1,d_{u}}
		\end{array}\right)\right]^{-1}\\
		& = & \left[U+\triangle\left(\begin{array}{cc}
		B^{T}R_{\theta,i+1}^{d,l} & B^{T}h_{\theta,i+1}^{d,l}(Y)\end{array}\right)\left(\begin{array}{c}
		B\\
		0_{1,d_{u}}
		\end{array}\right)\right]^{-1}\\
		& = & \left[U+\triangle B^{T}R_{\theta,i+1}^{d,l}B\right]^{-1}:=G(R_{\theta,i+1}^{d,l})
		\end{array}
		\]
	}with $G$ a $d_{u}$ symmetric (as the inverse of a symmetric matrix
	square matrix depending only of $R_{\theta,i+1}^{d,l}$. Moreover,
	we have:{\small{}
		\[
		\begin{array}{l}
		\begin{array}{l}
		\left(A_{G,\theta,i}^{d,l,e}\right)^{T}E_{\theta,i+1}^{d,l}(Y)B_{1}\\
		=\left(I_{d+1}+\Delta\left(\begin{array}{ll}
		A_{\theta}(\overline{X_{\upsilon}^{d,l-1}}(t_{i}),t_{i}) & 0_{d,1}\\
		0_{1,d} & 0
		\end{array}\right)\right)^{T}\left(\begin{array}{cc}
		R_{\theta,i+1}^{d,l} & h_{\theta,i+1}^{d,l}(Y)\\
		h_{\theta,i+1}^{d,l}(Y)^{T} & \alpha_{\theta,i+1}^{d,l}(Y)
		\end{array}\right)\left(\begin{array}{c}
		B\\
		0_{1,d_{u}}
		\end{array}\right)\\
		=\left(I_{d+1}+\Delta\left(\begin{array}{ll}
		A_{\theta}(\overline{X_{\upsilon}^{d,l-1}}(t_{i}),t_{i}) & 0_{d,1}\\
		0_{1,d} & 0
		\end{array}\right)\right)^{T}\left(\begin{array}{c}
		R_{\theta,i+1}^{d,l}B\\
		h_{\theta,i+1}^{d,l}(Y)^{T}B
		\end{array}\right)\\
		=\left(\begin{array}{c}
		R_{\theta,i+1}^{d,l}B+\Delta A_{\theta}(\overline{X_{\upsilon}^{d,l-1}}(t_{i}),t_{i})^{T}R_{\theta,i+1}^{d,l}B\\
		h_{\theta,i+1}^{d,l}(Y)^{T}B
		\end{array}\right)
		\end{array}\end{array}
		\]
	}so we can compute:{\scriptsize
		\[
		\begin{array}{l}
		\left(A_{G,\theta,i}^{d,l,e}\right)^{T}E_{\theta,i+1}^{d,l}(Y)B_{1}G(R_{\theta,i+1}^{d,l})B_{1}^{T}E_{\theta,i+1}^{d,l}(Y)A_{G,\theta,i}^{d,l,e}\\
		=\left(\begin{array}{c}
		R_{\theta,i+1}^{d,l}B+\Delta A_{\theta}(\overline{X_{\upsilon}^{d,l-1}}(t_{i}),t_{i})^{T}R_{\theta,i+1}^{d,l}B\\
		h_{\theta,i+1}^{d,l}(Y)^{T}B
		\end{array}\right)G(R_{\theta,i+1}^{d,l})\left(\begin{array}{c}
		R_{\theta,i+1}^{d,l}B+\Delta A_{\theta}(\overline{X_{\upsilon}^{d,l-1}}(t_{i}),t_{i})^{T}R_{\theta,i+1}^{d,l}B\\
		h_{\theta,i+1}^{d,l}(Y)^{T}B
		\end{array}\right)^{T}\\
		=\left(\begin{array}{c}
		(I_{d}+\Delta A_{\theta}(\overline{X_{\upsilon}^{d,l-1}}(t_{i}),t_{i})^{T})R_{\theta,i+1}^{d,l}B\\
		h_{\theta,i+1}^{d,l}(Y)^{T}B
		\end{array}\right)G(R_{\theta,i+1}^{d,l})\left(\begin{array}{cc}
		B^{T}R_{\theta,i+1}^{d,l}(I_{d}+\Delta A_{\theta}(\overline{X_{\upsilon}^{d,l-1}}(t_{i}),t_{i})) & B^{T}h_{\theta,i+1}^{d,l}(Y)\end{array}\right)\\
		:=\left(\begin{array}{cc}
		H_{1}(R_{\theta,i+1}^{d,l}) & H_{2}(R_{\theta,i+1}^{d,l},h_{\theta,i+1}^{d,l})\\
		H_{2}(R_{\theta,i+1}^{d,l},h_{\theta,i+1}^{d,l})^{T} & H_{3}(R_{\theta,i+1}^{d,l},h_{\theta,i+1}^{d,l})
		\end{array}\right).
		\end{array}
		\]
	}By re-injecting all the derived expression in (\ref{eq:discrete_Riccati_equation}),
	we obtain:{\small{}
		\[
		\begin{array}{l}
		E_{\theta,i}^{d,l}(Y)=\left(\begin{array}{cc}
		F_{1}(R_{\theta,i+1}^{d,l}) & F_{2}(h_{\theta,i+1}^{d,l}(Y\text{)})\\
		F_{2}(h_{\theta,i+1}^{d,l}(Y\text{)})^{T} & F_{3}(\alpha_{\theta,i+1}^{d,l}(Y))
		\end{array}\right)+\triangle\left(\begin{array}{cc}
		C^{T}C & -C^{T}Y_{i}\\
		-Y_{i}{}^{T}C & Y_{i}^{T}Y_{i}
		\end{array}\right)\\
		-\triangle\left(\begin{array}{cc}
		H_{1}(R_{\theta,i+1}^{d,l}) & H_{2}(R_{\theta,i+1}^{d,l},h_{\theta,i+1}^{d,l})\\
		H_{2}(R_{\theta,i+1}^{d,l},h_{\theta,i+1}^{d})^{T} & H_{3}(R_{\theta,i+1}^{d,l},h_{\theta,i+1}^{d,l})
		\end{array}\right)
		\end{array}
		\]
	}and $E_{\theta,i}^{d,l}(Y)$ is symmetric and has indeed the required
	form, hence the recursion. We also obtain the following finite difference
	equation:{\small{}
		\[
		\begin{array}{lll}
		R_{\theta,i}^{d,l} & = & R_{\theta,i+1}^{d,l}+\triangle C^{T}C+\Delta\left(R_{\theta,i+1}^{d,l}A_{\theta}(\overline{X_{\upsilon}^{d,l-1}}(t_{i}),t_{i})+A_{\theta}(\overline{X_{\upsilon}^{d,l-1}}(t_{i}),t_{i})^{T}R_{\theta,i+1}^{d,l}\right)\\
		& + & \Delta^{2}A_{\theta}(\overline{X_{\upsilon}^{d,l-1}}(t_{i}),t_{i})^{T}R_{\theta,i+1}^{d,l}A_{\theta}(\overline{X_{\upsilon}^{d,l-1}}(t_{i}),t_{i})\\
		& - & \triangle(I_{d}+\Delta A_{\theta}(\overline{X_{\upsilon}^{d,l-1}}(t_{i}),t_{i})^{T})R_{\theta,i+1}^{d,l}BG(R_{\theta,,i+1}^{d,l})B^{T}R_{\theta,i+1}^{d,l}(I_{d}+\Delta A_{\theta}(\overline{X_{\upsilon}^{d,l-1}}(t_{i}),t_{i}))\\
		h_{\theta,i}^{d,l}(Y) & = & h_{\theta,i+1}^{d,l}(Y)-\triangle C^{T}Y_{i}+\Delta A_{\theta}(\overline{X_{\upsilon}^{d,l-1}}(t_{i}),t_{i})^{T}h_{\theta,i+1}^{d,l}(Y)\\
		& - & \triangle(I_{d}+\Delta A_{\theta}(\overline{X_{\upsilon}^{d,l-1}}(t_{i}),t_{i})^{T})R_{\theta,i+1}^{d,l}BG(R_{\theta,i+1}^{d,l})B^{T}h_{\theta,i+1}^{d,l}(Y)\\
		\alpha_{\theta,i}^{d,l}(Y) & = & \alpha_{\theta,i+1}^{d,l}(Y)+\triangle Y_{i}^{T}Y_{i}-\triangle h_{\theta,i+1}^{d,l}(Y)^{T}BG(R_{\theta,i+1}^{d,l})B^{T}h_{\theta,i+1}^{d,l}(Y).
		\end{array}
		\]
	}From these equation, we easily derive that{\small{}:
		\[
		\begin{array}{l}
		S_{n}^{l}(Y;\theta,x_{0})\,:=X^{e}(0)^{T}E_{\theta,0}^{d,l}(Y)X^{e}(0)\\
		=\left(\begin{array}{cc}
		x_{0}^{T} & 1\end{array}\right)\left(\begin{array}{cc}
		R_{\theta,0}^{d,l} & h_{\theta,0}^{d,l}(Y)\\
		h_{\theta,0}^{d,l}(Y)^{T} & \alpha_{\theta,0}^{d,l}(Y)
		\end{array}\right)\left(\begin{array}{l}
		x_{0}\\
		1
		\end{array}\right)\\
		=\left(\begin{array}{cc}
		x_{0}^{T} & 1\end{array}\right)\left(\begin{array}{c}
		R_{\theta,0}^{d,l}x_{0}+h_{\theta,0}^{d,l}(Y)\\
		h_{\theta,0}^{d,l}(Y)^{T}x_{0}+\alpha_{\theta,0}^{d,l}(Y)
		\end{array}\right)\\
		=x_{0}^{T}R_{\theta,0}^{d,l}x_{0}+2h_{\theta,0}^{d,l}(Y)^{T}x_{0}+\alpha_{\theta,0}^{d,l}(Y)\\
		=x_{0}^{T}R_{\theta,0}^{d,l}x_{0}+2h_{\theta,0}^{d,l}(Y)^{T}x_{0}+\triangle Y_{n}^{T}Y_{n}\\
		+\triangle\sum_{i=0}^{n-1}\left(Y_{i}^{T}Y_{i}-h_{\theta,i+1}^{d,l}(Y)^{T}BG(R_{\theta,i+1}^{d,l})B^{T}h_{\theta,i+1}^{d,l}(Y)\right).
		\end{array}
		\]
	}As in the continuous case, $\inf_{u}C_{T}^{d,l}(Y;\theta,x_{0},u)$
	is quadratic w.r.t to $x_{0}$ with the minimum reached for $\overline{x_{\theta,0}^{d,l}}=-\left(R_{\theta,0}^{d,l}\right)^{-1}h_{\theta,0}^{l}$
	and $S_{n}^{CI,l}(Y;\theta)$ becomes:
	\[
\begin{array}{l}
	S_{n}^{CI,l}(Y;\theta)=-h_{\theta,0}^{d,l}(Y)^{T}\left(R_{\theta,0}^{d,l}\right)^{-1}h_{\theta,0}^{d,l}(Y)+\triangle Y_{n}^{T}Y_{n}
\\+\triangle\sum_{i=0}^{n-1}\left(Y_{i}^{T}Y_{i}-h_{\theta,i+1}^{d,l}(Y)^{T}BG(R_{\theta,i+1}^{d,l})B^{T}h_{\theta,i+1}^{d,l}(Y)\right).
\end{array}
	\]
	The optimal control becomes:
	\[
	\overline{u_{\theta,i}^{d,l}}=-G(R_{\theta,i+1}^{d,l})B^{T}\left(R_{\theta,i+1}^{d,l}\left(I_{d}+\Delta A_{\theta}(\overline{X_{\upsilon}^{d,l-1}}(t_{i}),t_{i})\right)\overline{X_{\upsilon}^{d,l}}(t_{i})+h_{\theta,i+1}^{d,l}(Y)\right)
	\]
	and from this we derive the finite difference equation followed by
	$\overline{X_{\upsilon}^{d,l}}$. 
\end{proof}

\begin{prop}
	\label{prop:continuous_R0_invertibility_criteria} Given $\theta\in\Theta$
	and $l\in\mathbb{N}$, $R_{\theta}^{l}(0)$ is invertible if and only
	if:
	
	1) the matrix
	\begin{equation}
	O_{\theta}^{l}(T)=\int_{0}^{T}\left(C\Phi_{\theta}^{l}(t,0)\right)^{T}C\Phi_{\theta}^{l}(t,0)dt\label{eq:continuous_invertibility_criteria}
	\end{equation}
	is invertible, where $\Phi_{\theta}^{l}$ is the resolvant of (\ref{eq:extended_ODEmodel}),
	
	2) the following implication holds:
	\begin{equation}
	\left\Vert CX_{\theta,x_{0}^{1}}^{l}-CX_{\theta,x_{0}^{2}}^{l}\right\Vert _{L^{2}}^{2}=0\Longrightarrow x_{0}^{1}=x_{0}^{2}.\label{eq:second_continuous_invertibility_criteria}
	\end{equation}
\end{prop}
\begin{proof}
	Similarly as in the discrete case, we have $\min_{u}\widetilde{C}_{T}^{l}(\theta,z_{0},u)=z_{0}^{T}\widetilde{E}_{\theta,0}^{l}(Y)z_{0}$
	with: 
	\[
	\widetilde{C}_{T}^{l}(\theta,z_{0},u)=\int_{0}^{T}Z_{\theta,z_{0},u}^{l}(t)^{T}\left(\begin{array}{cc}
	C^{T}C & 0_{d,1}\\
	0_{1,d} & 0
	\end{array}\right)Z_{\theta,z_{0},u}^{l}(t)+u(t)^{T}Uu(t)dt
	\]
	linked to the finite difference equation: 
	\begin{equation}
	\left\{ \begin{array}{l}
	\dot{Z^{l}}_{\theta,z_{0},u}(t)=A_{G,\theta}(\overline{X_{\upsilon}^{l-1}}(t),t)Z_{\theta,z_{0},u}^{l}(t)+B_{1}u(t)\\
	Z_{\theta,z_{0},u}^{l}(0)=z_{0}
	\end{array}\right.\label{eq:continuous_invertibility_criteria_ode_model}
	\end{equation}
	here $\widetilde{E}_{\theta}^{l}$ is the solution of the Riccati
	equation (\ref{eq:continuous_Riccati_equation}) where the weight
	function $t\longmapsto W(t)$ has been replaced by $\left(\begin{array}{cc}
	C^{T}C & 0_{d,1}\\
	0_{1,d} & 0
	\end{array}\right),$ henceforth $\widetilde{E}_{\theta}^{l}$ and $E_{\theta}^{l}$ share
	the same component $R_{\theta}^{l}$. We define $\widetilde{x_{0}}$
	s.t $\widetilde{x_{0}}^{T}R_{\theta}^{l}(0)\widetilde{x_{0}}=0$ and
	$\widetilde{z_{0}}=\left(\begin{array}{c}
	\widetilde{x_{0}}\\
	0
	\end{array}\right)$. From this, we derive $\min_{u}\widetilde{C}_{T}^{l}(\theta,\widetilde{z_{0}},u)=0$
	which implies $\overline{u_{\theta}^{l}}=0$. The ODE corresponding
	to the optimal trajectory becomes:
	\[
	\begin{array}{l}
	\dot{Z^{l}}_{\theta,z_{0}}(t)=A_{G,\theta}(\overline{X_{\upsilon}^{l-1}}(t),t)Z_{\theta,z_{0}}^{l}(t)\end{array}
	\]
	Hence the optimal trajectory is $Z_{\theta,\widetilde{z_{0}}}^{l}(t)=\left(\begin{array}{c}
	\Phi_{\theta}^{l}\left(t,0\right)\widetilde{x_{0}}\\
	0
	\end{array}\right)$. Thus, the minimal cost has the simpler expression
\[
 \min_{u}\widetilde{C}_{T}^{l}(\theta,\widetilde{z_{0}},u)=\int_{0}^{T}\widetilde{x_{0}}^{T}\Phi_{\theta}^{l}\left(t,0\right)^{T}\left(C^{T}C\right)\Phi_{\theta}^{l}\left(t,0\right)\widetilde{x_{0}}dt
\]
	so we have: $\widetilde{x_{0}}^{T}R_{\theta}^{l}(0)\widetilde{x_{0}}=\min_{u}\widetilde{C}_{T}^{l}(\theta,\widetilde{z_{0}},u)=\widetilde{x_{0}}^{T}O_{\theta}^{l}(T).\widetilde{x_{0}}=0$ 

	and we can conclude.
	
	We now demonstrate ($\ref{eq:second_continuous_invertibility_criteria})\Longrightarrow(\ref{eq:continuous_invertibility_criteria}$),
	we choose an arbitrary initial condition $\widetilde{x_{0}}$ respecting
	$\widetilde{x_{0}}^{T}O_{\theta}^{l}(T)\widetilde{x_{0}}=\widetilde{x_{0}}^{T}\left(\int_{0}^{T}\Phi_{\theta}^{l}\left(t,0\right)^{T}C^{T}C\Phi_{\theta}^{l}\left(t,0\right)dt\right)\widetilde{x_{0}}=0$.
	We denote $X_{\theta,x_{0}^{1}}^{l}$ and $X_{\theta,x_{0}^{2}}^{l}$
	the solutions of $\dot{X}=A(\overline{X_{\upsilon}^{l-1}}(t),t)X$
	respectively with initial conditions $x_{0}^{1}$ and $x_{0}^{2}\,:=\widetilde{x_{0}}+x_{0}^{1}$.
	By superposition principle for linear ODE solution we have $\Phi_{\theta}^{l}\left(t,0\right)\widetilde{x_{0}}=X_{\theta,x_{0}^{1}}^{l}(t)-X_{\theta,x_{0}^{2}}^{l}(t)$,
	hence (\ref{eq:second_continuous_invertibility_criteria}) imposes
	that $X_{\theta,x_{0}^{1}}^{l}(0)=X_{\theta,x_{0}^{2}}^{l}(0)$ and
	so $\widetilde{x_{0}}=0$.
	
	To demonstrate the implication ($\ref{eq:continuous_invertibility_criteria})\Longrightarrow(\ref{eq:second_continuous_invertibility_criteria}$),
	we take two functions $X_{\theta,x_{0}^{1}}^{l}$ and $X_{\theta,x_{0}^{2}}^{l}$
	assumed to verify $\left\Vert CX_{\theta,x_{0}^{1}}^{l}-CX_{\theta,x_{0}^{2}}^{l}\right\Vert _{L^{2}}^{2}=0$
	solutions of $\dot{X}=A(\overline{X_{\upsilon}^{l-1}}(t),t)X$ with
	respective initial conditions $x_{0}^{1}$ and $x_{0}^{2}$. We know
	the function difference $X_{\theta,\widetilde{x_{0}}}^{0}(t):=X_{\theta,x_{0}^{1}}^{l}(t)-X_{\theta,x_{0}^{2}}^{l}(t)$
	is equal to $\Phi_{\theta}^{l}\left(t,0\right)\widetilde{x_{0}}$
	and so $\widetilde{x_{0}}^{T}O_{\theta}^{l}(T)\widetilde{x_{0}}=0$,
	$O_{\theta}^{l}(T)$ invertibility gives us $x_{0}^{1}-x_{0}^{2}=0$.
\end{proof}
\begin{prop}
	\label{prop:E_bounded_probability}Under conditions C1-C2, we have
	$\left\Vert E_{\theta,i}^{d,l}(Y)\right\Vert _{2}=O_{p,n}(1)$ and
	$\left\Vert E_{\theta}^{l}(t)\right\Vert _{2}=O_{n}(1)$ respectively
	uniformly on $\mathbb{N}\times\left\llbracket 0,\,n\right\rrbracket \times\Theta$
	and on $\mathbb{N}\times\left[0,\,T\right]\times\Theta$.
\end{prop}
\begin{proof}
	By property of the solution of Riccati equation, we have 
	
	\[
	z_{k}^{T}E_{\theta,k}^{d,l}(Y)z_{k}=\min_{u=\left\{ u_{k},\ldots u_{n-1}\right\} }\left\{ \Delta\sum_{i=k}^{n}Z_{\theta,z_{k},u}^{d,l}(t_{i})^{T}W_{i}Z_{\theta,z_{k},u}^{d,l}(t_{i})+\Delta\sum_{i=k}^{n-1}u_{i}{}^{T}Uu_{i}\right\} \geq0
	\]
	where $Z_{\theta,z_{k},u}^{d,l}$ is the solution of 
	\[
	\left\{ \begin{array}{l}
	Z_{\theta,z_{k},u}^{d,l}(t_{i+1})=A_{G,\theta,i}^{d,l,e}Z_{\theta,z_{k},u}^{d,l}(t_{i})+B_{1}\triangle\overline{u_{i}}\\
	Z_{\theta,z_{k},u}^{d,l}(t_{k})=z_{k}
	\end{array}\right.
	\]
	where $\overline{u}=\left\{ \overline{u_{i}}\right\} _{i\in\left\llbracket k,\,n-1\right\rrbracket }$
	is the optimal control. This holds for every possible values $z_{k}$.
	In particular, we have the bound $z_{k}^{T}E_{\theta,k}^{d,l}(Y)z_{k}\leq\Delta\sum_{i=k}^{n}Z_{\theta,z_{k}}^{d,l}(t_{i})^{T}W_{i}Z_{\theta,z_{k}}^{d,l}(t_{i})$
	with $Z_{\theta,z_{k}}^{d,l}$ solution of:
	\[
	\left\{ \begin{array}{l}
	Z_{\theta,z_{k}}^{d,l}(t_{i+1})=A_{G,\theta,i}^{d,l,e}Z_{\theta,z_{k}}^{d,l}(t_{i})\\
	Z_{\theta,z_{k}}^{d,l}(t_{k})=z_{k}.
	\end{array}\right.
	\]
	According to the discrete Gronwall lemma \ref{lem:disc_Gronwall_lemma},
	and conditions 1-2, we have the bound $Z_{\theta,z_{k}}^{d,l}(t_{i})\leq e^{T\overline{A}}z_{k}$
	for $i\in\left\llbracket k,\,n-1\right\rrbracket $, thus $z_{k}^{T}E_{\theta,k}^{d,l}(Y)z_{k}\leq\Delta\sum_{i=k}^{n}\left(e^{T\overline{A}}z_{k}\right)^{T}W_{i}e^{T\overline{A}}z_{k}\leq\Delta\sum_{i=0}^{n}\left(e^{T\overline{A}}z_{k}\right)^{T}W_{i}e^{T\overline{A}}z_{k}$.
	Using Cauchy-Schwarz inequality we have $z_{k}^{T}E_{\theta,k}^{d,l}(Y)z_{k}\leq\left\Vert z_{k}\right\Vert _{2}^{2}e^{2T\overline{A}}\Delta\sum_{i=0}^{n}\left\Vert W_{i}\right\Vert _{2}$
	for all $z_{k}$. Since {\small{}
		\[
		\begin{array}{lll}
		\triangle\sum_{i=1}^{n}Y_{i} & = & \triangle\sum_{i=1}^{n}\left(Y_{i}-CX_{\theta^{*},x_{0}^{*}}(t_{i})+CX_{\theta^{*},x_{0}^{*}}(t_{i})\right)\\
		& = & \frac{1}{n}\sum_{i=1}^{n}\left(Y_{i}-CX_{\theta^{*},x_{0}^{*}}(t_{i})\right)+\triangle\sum_{i=1}^{n}CX_{\theta^{*},x_{0}^{*}}(t_{i})\\
		& \longrightarrow & \mathbb{E}_{\left(\theta^{*},x_{0}^{*}\right)}\left[\epsilon_{1}\right]+\int_{0}^{T}CX_{\theta^{*},x_{0}^{*}}(t)dt=\int_{0}^{T}CX_{\theta^{*},x_{0}^{*}}(t)dt
		\end{array}
		\]
	}almost surely (The first term is obtained by using the law of large
	number, the second as the limit of a Riemann sum) and {\small{}
		\[
		\begin{array}{lll}
		\triangle\sum_{i=1}^{n}\left\Vert Y_{i}\right\Vert _{2}^{2} & = & \triangle\sum_{i=1}^{n}\left\Vert Y_{i}-CX_{\theta^{*},x_{0}^{*}}(t_{i})\right\Vert _{2}^{2}+\triangle\sum_{i=1}^{n}\left\Vert CX_{\theta^{*},x_{0}^{*}}(t_{i})\right\Vert _{2}^{2}\\
		& + & 2\triangle\sum_{i=1}^{n}\left(Y_{i}-CX_{\theta^{*},x_{0}^{*}}(t_{i})\right)^{T}CX_{\theta^{*},x_{0}^{*}}(t_{i})\\
		& = & \frac{1}{n}\sum_{i=1}^{n}\left\Vert \epsilon_{i}\right\Vert _{2}^{2}+\triangle\sum_{i=1}^{n}\left\Vert CX_{\theta^{*},x_{0}^{*}}(t_{i})\right\Vert _{2}^{2}+2\frac{1}{n}\sum_{i=1}^{n}\epsilon_{i}^{T}CX_{\theta^{*},x_{0}^{*}}(t_{i})\\
		& \longrightarrow & d^{'}\sigma^{2}+\int_{0}^{T}\left\Vert CX_{\theta^{*},x_{0}^{*}}(t)\right\Vert _{2}^{2}dt
		\end{array}
		\]
	}by using Prohorov's theorem (\citet{Vaart1998} theorem 2.4), we
	know that $\Delta\sum_{i=0}^{n}\left\Vert W_{i}\right\Vert _{2}=O_{p,n}(1)$.
	Henceforth $\left\Vert E_{\theta,k}^{d,l}(Y)\right\Vert _{2}=O_{p,n}(1)$
	uniformly on $\mathbb{N}\times\left\llbracket 0,\,n\right\rrbracket \times\Theta$.
	Similarly as in the discrete case, we derive the following bound in
	the continuous case:
	\[
	z_{t}^{T}E_{\theta}^{l}(t)z_{t}\leq Z_{\theta,z_{t}}^{l}(T)^{T}QZ_{\theta,z_{t}}^{l}(T)+\int_{t}^{T}Z_{\theta,z_{t}}^{l}(s)^{T}W(s)Z_{\theta,z_{t}}^{l}(s)ds
	\]
	with 
	\[
	\left\{ \begin{array}{l}
	\dot{Z_{\theta,z_{t}}^{l}}(s)=A_{G,\theta}(\overline{X_{\theta}^{l-1}}(s),s)Z_{\theta,z_{t}}^{l}(s)\\
	Z_{\theta,z_{t}}^{l}(t)=z_{t}
	\end{array}\right.
	\]
	for all $z_{t}$. Using the continuous version of the Gronwall lemma
	we obtain $Z_{\theta,z_{t}}^{l}(s)\leq e^{T\overline{A}}z_{t}$ and
	so $z_{t}^{T}E_{\theta}^{l}(t)z_{t}\leq\left\Vert z_{t}\right\Vert _{2}^{2}e^{2T\overline{A}}(\left\Vert Q\right\Vert _{2}+\int_{0}^{T}\left\Vert W(s)\right\Vert _{2}ds)$.
	This allows us to conclude $\left\Vert E_{\theta}^{l}(t)\right\Vert _{2}=O_{n}(1)$
	uniformly on $\mathbb{N}\times\left[0,\,T\right]\times\Theta$.
\end{proof}

\section{Well-posedness nature of control problem (\ref{eq:as_Continous_cost})-(\ref{eq:asymptotic_ODEmodel})}

In order to derive asymptotic properties of $\widehat{\theta}^{T}$
and $\widehat{\theta}^{T,CI}$, we need to ensure the well-posedness
nature of the optimal control problems defining our estimators. That
is, the existence of an admissible solution for the problem (\ref{eq:as_Continous_cost})-(\ref{eq:asymptotic_ODEmodel})
for each possible value $\left(\theta,x_{0}\right)$. For linear models
this is a classic outcome of Linear-Quadratic theory. In the non-linear
case however it depends on the cost function and the vector field
regularities w.r.t state and control. That is why we prove here the
existence of $\min_{u\in\mathbb{U}_{\theta,x_{0}}}C_{T}^{\infty}(\theta,x_{0},u)$
where $\mathbb{U}_{\theta,x_{0}}$ is the set of feedback controls:
\[
\mathbb{U}_{\theta,x_{0}}=\left\{ u_{\upsilon}(t)=U^{-1}B^{T}(R_{\upsilon}(t)X_{\upsilon,u_{\upsilon}}(t)+h_{\upsilon}(t)),\,\left(R_{\upsilon},h_{\upsilon}\right)\in L^{2}\left(\left[0,\,T\right],\,\mathbb{R}^{d\times d}\times\mathbb{R}^{d}\right)\right\} .
\]
The proof is almost similar as the one presented in \citet{BrunelClairon_Pontryagin2017}
but with the additional requirement that $\overline{u}_{\theta,x_{0}}$
belongs to $\mathbb{U}_{\theta,x_{0}}$ which in turn calls for the
introduction of C3. 
\begin{thm}
	\label{thm:existence_solution_as_oca}Under conditions C1-C2-C3 for
	all signals $Y\in L^{2}(\left[0,\,T\right],\mathbb{R}^{d^{'}})$ and
	for all $\theta\in\Theta$, the asymptotic control problem (\ref{eq:as_Continous_cost})-(\ref{eq:asymptotic_ODEmodel})
	admits at least one solution. It exists a control $\overline{u}_{\theta}$
	belonging to $\mathbb{U}_{\theta,x_{0}}$ that minimizes the cost,
	i.e. $C_{T}^{\infty}(\theta,x_{0},\overline{u}_{\theta,x_{0}})=\min_{u\in\mathbb{U}_{\theta,x_{0}}}C_{T}^{\infty}(\theta,x_{0},u)$. 
\end{thm}
\begin{proof}
	Thanks to model regularity conditions, we know it exists admissible
	controls $u_{\upsilon}\in L^{2}(\left[0,\,T\right],\,\mathbb{R}^{d_{u}})$.
	Thus, we can consider an admissible minimizing sequence $\left\{ u_{\upsilon}^{i}\right\} _{i\in\mathbb{N}}$.
	Since we have $\lambda\left\Vert u_{\upsilon}^{i}\right\Vert _{L^{2}}^{2}\leq C_{T}^{\infty}(\upsilon,u_{\upsilon}^{i})$
	with $\lambda$ the maximum eigenvalue of $U$, the sequence $\left\{ u_{\upsilon}^{i}\right\} $
	is uniformly bounded in $L^{2}(\left[0,\,T\right],\mathbb{R}^{d})$,
	a reflexive Banach space, according to \textsc{theorem} III.27 in
	\citet{Brezis1983}, its exists a subsequence converging weakly to
	a limit $\overline{u}_{\theta,x_{0}}$. Using H{\"o}lder inequality $\left\Vert fg\right\Vert _{L^{1}}\leq\left\Vert f\right\Vert _{L^{2}}\left\Vert g\right\Vert _{L^{2}}$
	, we derive the subsequence boundedness in $L^{1}(\left[0,\,T\right],\mathbb{R}^{d})$.
	(For the sake of notation, we still denote the subsequence by $\left\{ u_{\upsilon}^{i}\right\} $).
	For the following, $\widetilde{u}_{\upsilon}$ denotes the upper
	bound of the sequence $\left\{ u_{\upsilon}^{i}\right\} $.
	
	Knowing that:
	\[
	\begin{array}{lll}
	\left\Vert \dot{X}_{\upsilon,u_{\upsilon}^{i}}(t)-\dot{X}_{\upsilon}(t)\right\Vert _{2} & \leq & \left\Vert A_{\theta}(X_{\upsilon,u_{\upsilon}^{i}}(t),t)X_{\upsilon,u_{\upsilon}^{i}}(t)-A_{\theta}(X_{\upsilon}(t),t)X_{\upsilon}(t)\right\Vert _{2}\\
	& + & \left\Vert Bu_{\upsilon}^{i}(t)\right\Vert _{2}\\
	& \leq & \left(\overline{A}+\overline{A}_{x}\right)\left\Vert X_{\upsilon,u_{\upsilon}^{i}}(t)-X_{\upsilon}(t)\right\Vert _{2}+\left\Vert B\right\Vert _{2}\left\Vert u_{\upsilon}^{i}(t)\right\Vert _{2}
	\end{array}
	\]
	here $\overline{A}_{x}<\infty$ is the Lipschitz constant of $A$
	w.r.t state, which existence is ensured by C2 as a continuous function
	on a compact subset. Gronwall's lemma gives us:
	\[
	\begin{array}{lll}
	\left\Vert X_{\theta,x_{0},u_{\upsilon}^{i}}(t)-X_{\upsilon}(t)\right\Vert _{2} & \leq & \left\Vert B\right\Vert _{2}\int_{0}^{t}e^{\left(\overline{A}+\overline{A}_{x}\right)\left(t-s\right)}\left\Vert u_{\upsilon}^{i}(s)\right\Vert _{2}ds\end{array}
	\]
	and so: 
	\[
	\begin{array}{lll}
	\left\Vert X_{\upsilon,u_{\upsilon}^{i}}(t)\right\Vert _{2} & \leq & \left\Vert X_{\upsilon,u_{\upsilon}^{i}}(t)-X_{\upsilon}(t)\right\Vert _{2}+\left\Vert X_{\upsilon}(t)\right\Vert _{2}\\
	& \leq & \left\Vert B\right\Vert _{2}e^{\left(\overline{A}+\overline{A}_{x}\right)t}\int_{0}^{t}\left\Vert u_{\upsilon}^{i}(s)\right\Vert _{2}ds+\sup_{t\in\left[0,\,T\right]}\left\Vert X_{\upsilon}(t)\right\Vert _{2}.
	\end{array}
	\]
	The control $u_{\upsilon}^{i}$ being bounded in $L^{1}(\left[0,\,T\right],\mathbb{R}^{d})$,
	$X_{\upsilon,u_{\upsilon}^{i}}$ (modulo a subsequence) is
	uniformly bounded on $\left[0,\,T\right]$ and since: 
	\[
	\begin{array}{lll}
	\left\Vert \dot{X}_{\upsilon,u_{\upsilon}^{i}}(t)\right\Vert _{2} & \leq & \left\Vert \dot{X}_{\upsilon,u_{\upsilon}^{i}}(t)-\dot{X}_{\upsilon}(t)\right\Vert _{2}+\left\Vert \dot{X}_{\upsilon}(t)\right\Vert _{2}\\
	& \leq & \left\Vert A_{\theta}(X_{\upsilon,u_{\upsilon}^{i}}(t),t)X_{\upsilon,u_{\upsilon}^{i}}(t)-A_{\theta}(X_{\upsilon}(t),t)X_{\upsilon}(t)\right\Vert _{2}\\
	& + & \left\Vert Bu_{\upsilon}^{i}(t)\right\Vert _{2}+\left\Vert A_{\theta}(X_{\upsilon}(t),t)X_{\upsilon}(t)\right\Vert _{2}\\
	& \leq & \left(\overline{A}+\overline{A}_{x}\right)\left\Vert X_{\upsilon,u_{\upsilon}^{i}}(t)-X_{\upsilon}(t)\right\Vert _{2}+\left\Vert Bu_{\upsilon}^{i}(t)\right\Vert _{2}\\
	& + & \overline{A}\sup_{t\in\left[0,\,T\right]}\left\Vert X_{\upsilon}(t)\right\Vert _{2}
	\end{array}
	\]
	$\dot{X}_{\upsilon,u_{\upsilon}^{i}}$ is also bounded in
	$L^{2}(\left[0,\,T\right],\mathbb{R}^{d})$, hence (again modulo a
	subsequence) $\dot{X}_{\theta,x_{0},u_{\upsilon}^{i}}$ converges
	weakly to a limit $\overline{\dot{X}}_{\upsilon}$. 
	
	Since the sequence $X_{\upsilon,u_{\upsilon}^{i}}$ is equicontinuous
	because $\left\Vert X_{\upsilon,u_{\upsilon}^{i}}(t)-X_{\upsilon,u_{\upsilon}^{i}}(t')\right\Vert _{2}\leq\overline{A}\left|t-t^{'}\right|+\left\Vert B\right\Vert _{2}\widetilde{u}\sqrt{\left|t-t^{'}\right|}$,
	we can invoke Arzela-Ascoli theorem to obtain the uniform convergence
	(modulo a subsequence) of $X_{\upsilon,u_{\upsilon}^{i}}$
	toward a continuous function $\overline{X_{\upsilon}}$ on $\left[0,\,T\right]$.
	Using the identity, $X_{\upsilon,u_{\upsilon}^{i}}(t)=x_{0}+\int_{0}^{t}\dot{X}_{\upsilon,u_{\upsilon}^{i}}(s)ds$
	and by taking the limit we know $\overline{X_{\upsilon}}$ is
	an absolutely continuous function with $\overline{\dot{X}}_{\upsilon}(t)=\dot{\overline{X_{\upsilon}}}(t)$
	a.e. 
	
	Since $B$ has independant columns then $\left(U^{-1}B^{T}\right)^{+}=\left(BU^{-2}B^{T}\right)^{-1}BU^{-1}$
	is the Moore-Penrose inverse of $U^{-1}B^{T}$ and is a left inverse
	of $U^{-1}B^{T}$, and reminding that the sequences $\left\{ X_{\upsilon,u_{\upsilon}^{i}}\right\} _{i\in\mathbb{N}}$
	and $\left\{ u_{\upsilon}^{i}\right\} _{i\in\mathbb{N}}$ are
	bounded and convergent, we can construct a bounded sequence $\left\{ R_{\upsilon}^{i},\,h_{\upsilon}^{i}\right\} $
	which converges to a limit $\left\{ \overline{R_{\upsilon}},\,\overline{h_{\upsilon}}\right\} $
	such that it respect the relation $R_{\upsilon}^{i}(t)X_{\upsilon,u_{\upsilon}^{i}}(t)+h_{\upsilon}^{i}(t)=\left(U^{-1}B^{T}\right)^{+}u_{\upsilon}^{i}(t)$
	and $\overline{R_{\upsilon}}(t)\overline{X_{\upsilon}}(t)+\overline{h_{\upsilon}}(t)=\left(U^{-1}B^{T}\right)^{+}\overline{u}_{\upsilon}(t)$.
	For example, we could take for $R_{\upsilon}^{i}$ a continuous
	bounded function and $h_{\upsilon}^{i}(t)=-R_{\upsilon}^{i}(t)X_{\upsilon,u_{\upsilon}^{i}}(t)+\left(U^{-1}B^{T}\right)^{+}u_{\upsilon}^{i}(t)$
	and thus, $\overline{u}_{\upsilon}\in\mathbb{U}_{\upsilon}$.
	
	We respect the hypothesis of \textsc{theorem} 6.38 in \citet{clarke2013variationalcalculus}
	and we derive from that: 
	\[
	C_{T}^{\infty}(\upsilon,\overline{u}_{\upsilon})\leq\liminf_{i\longrightarrow\infty}C_{T}^{\infty}(\upsilon,u_{\upsilon}^{i})=\inf_{u}C_{T}^{\infty}(\upsilon,u).
	\]
	We now demonstrate $\overline{u}_{\\upsilon}$ is an admissible
	process (thus the infimum is reached). Using uniform convergence we
	have $\overline{X}_{\upsilon}(0)=x_{0}$. The last thing left
	to show is that $\overline{X}_{\upsilon}$ is a trajectory corresponding
	to $\overline{u}_{\upsilon}$, thus $\overline{X}_{\upsilon}=X_{\upsilon,\overline{u}_{\upsilon}}$.
	For any measurable subset $S$ of $\left[0,\,T\right]$ we have: 
	\[
	\int_{S}\left(\dot{X}_{\upsilon,u_{\upsilon}^{i}}(t)-A_{\theta}(X_{\upsilon,u_{\upsilon}^{i}}(t),t)X_{\upsilon,u_{\upsilon}^{i}}(t)-Bu_{\upsilon}^{i}(t)\right)dt=0
	\]
	by weak convergence we directly obtain $\int_{S}\dot{X}_{\upsilon,u_{\upsilon}^{i}}(t)dt\longrightarrow\int_{S}\dot{\overline{X_{\upsilon}}}dt$
	and $\int_{S}Bu_{\upsilon}^{i}(t)dt\longrightarrow\int_{S}B\overline{u}_{\upsilon}(t)dt.$
	Using continuity of the vector field on a compact and invoking dominated
	convergence theorem: $\int_{S}A_{\theta}(X_{\upsilon,u_{\upsilon}^{i}}(t),t)X_{\upsilon,u_{\upsilon}^{i}}(t)dt\longrightarrow\int_{S}A_{\theta}(\overline{X_{\upsilon}}(t),t)\overline{X_{\upsilon}}(t)$.
	By taking the limit we obtain: 
	\[
	\int_{S}\left(\dot{\overline{X}}_{\upsilon}(t)-A_{\theta}(\overline{X_{\upsilon}}(t),t)\overline{X_{\upsilon}}(t)-B\overline{u}_{\upsilon}(t)\right)dt=0.
	\]
	Hence, we have indeed demonstrate $\overline{u}_{\upsilon}\in L^{2}(\left[0,\,T\right],\mathbb{R}^{d})$
	and
	\[
	\left\{ \begin{array}{l}
	\dot{\overline{X}}_{\upsilon}(t)=A(\overline{X_{\upsilon}}(t),t)\overline{X_{\upsilon}}(t)+B\overline{u}_{\upsilon}(t)\,\textrm{a.e}\,\textrm{on}\,\left[0,\,T\right]\\
	\overline{X_{\upsilon}}(0)=x_{0}
	\end{array}\right.
	\]
	which finishes the proof.
\end{proof}

\section{Consistency}

\subsection{$\widehat{\theta^{T}}$, general case}
\begin{thm}
	\label{thm:Consistency_theorem}Under conditions C1 to C5, we have
	\[
	\left(\widehat{\theta},\,\widehat{x_{0}}\right)\longrightarrow\left(\theta^{*},\,x_{0}^{*}\right)
	\]
	in probability when $\left(l,n\right)\longrightarrow\infty$.
\end{thm}
\begin{proof}
	First, we decompose the difference $S^{\infty}(\upsilon)-S_{n}^{l}(Y;\upsilon)$
	in two terms we can analyse separately:
	\[
	\begin{array}{lll}
	S^{\infty}(\upsilon)-S_{n}^{l}(Y;\upsilon) & = & S^{\infty}(\upsilon)-S^{l}(\upsilon)+S^{l}(\upsilon)-S_{n}^{l}(Y;\upsilon).\end{array}
	\]
	By using the continuous profiled costs formula given by proposition
	\ref{prop:continuous_Riccati_accurate_representation}, we obtain
	for the first term{\small{}
		\[
		\begin{array}{l}
		S^{\infty}(\upsilon)-S^{l}(\upsilon)\\
		=x_{0}^{T}R_{\upsilon}^{\infty}(0)x_{0}+2x_{0}^{T}h_{\upsilon}^{\infty}(0)+\int_{0}^{T}\left(Y^{*}(t)^{T}Y^{*}(t)+d^{'}\sigma^{2}-h_{\upsilon}^{\infty}(t)^{T}BU^{-1}B^{T}h_{\upsilon}^{\infty}(t)\right)dt\\
		-x_{0}^{T}R_{\upsilon}^{l}(0)x_{0}-2x_{0}^{T}h_{\upsilon}^{l}(0)-\int_{0}^{T}\left(Y^{*}(t)^{T}Y^{*}(t)+d^{'}\sigma^{2}-h_{\upsilon}^{l}(t)^{T}BU^{-1}B^{T}h_{\upsilon}^{l}(t)\right)dt\\
		=x_{0}^{T}\left(R_{\upsilon}^{\infty}(0)-R_{\upsilon}^{l}(0)\right)x_{0}+2x_{0}^{T}\left(h_{\upsilon}^{\infty}(0)-h_{\upsilon}^{l}(0)\right)\\
		+\int_{0}^{T}\left(h_{\upsilon}^{l}(t)^{T}BU^{-1}B^{T}\left(h_{\upsilon}^{l}(t)-h_{\upsilon}^{\infty}(t)\right)+\left(h_{\upsilon}^{l}(t)-h_{\upsilon}^{\infty}(t)\right)^{T}BU^{-1}B^{T}h_{\upsilon}^{\infty}(t)\right)dt.
		\end{array}
		\]
	}By using lemma \ref{lem:Riccati_optimal_traj_uniform_convergence},
	we conclude that $\sup_{\upsilon\in\varUpsilon}\left|S^{\infty}(\upsilon)-S^{l}(\upsilon)\right|=o_{l}(1)$.
	Now, we control the difference $S^{l}(\upsilon)-S_{n}^{l}(Y;\upsilon)$
	by using the form given by proposition \ref{prop:continuous_Riccati_accurate_representation}
	and \ref{prop:discrete_Riccati_accurate_representation} : {\small{}
		\[
		\begin{array}{l}
		S^{l}(\upsilon)-S_{n}^{l}(Y;\upsilon) = \\ x_{0}^{T}\left(R_{\upsilon}^{l}(0)-R_{\upsilon,0}^{d,l}\right)x_{0}+2x_{0}^{T}\left(h_{\upsilon}^{l}(0)-h_{\upsilon,0}^{d,l}(Y)\right)\\
		 +  \int_{0}^{T}\left(Y^{*}(t)^{T}Y^{*}(t)+d^{'}\sigma^{2}\right)dt-\triangle\sum_{i=0}^{n-1}Y_{i}^{T}Y_{i}\\
		 -  \left(\int_{0}^{T}h_{\upsilon}^{l}(t)^{T}BU^{-1}B^{T}h_{\upsilon}^{l}(t)dt-\triangle\sum_{i=0}^{n-1}h_{\upsilon,i+1}^{d,l}(Y)^{T}BG(R_{\upsilon,i+1}^{d,l})B^{T}h_{\upsilon,i+1}^{d,l}(Y)\right).
		\end{array}
		\]
	}Lemma \ref{lem:h_R_continuity_limits}, gives us $\sup_{\upsilon\in\varUpsilon}\left\Vert R_{\upsilon,i}^{d,l}-R_{\upsilon}^{l}(t_{i})\right\Vert _{2}=o_{p,n}(1)$,
	$\sup_{\upsilon\in\varUpsilon}\left\Vert h_{\upsilon,i}^{d,l}-h_{\upsilon}^{l}(t_{i})\right\Vert _{2}=o_{p,n}(1)$,
	so $x_{0}^{T}\left(R_{\upsilon}^{l}(0)-R_{\upsilon,0}^{d,l}\right)x_{0}+2x_{0}^{T}\left(h_{\upsilon}^{l}(0)-h_{\upsilon,0}^{d,l}(Y)\right)=o_{p,n}(1)$
	uniformly on $\varUpsilon$. Regarding the second term, we have: {\small{}
		\[
		\begin{array}{l}
		\int_{0}^{T}\left(Y^{*}(t)^{T}Y^{*}(t)+d^{'}\sigma^{2}\right)dt-\triangle\sum_{i=0}^{n}Y_{i}^{T}Y_{i}\\
		=\triangle\sum_{i=1}^{n}\left(Y^{*}(t_{i})^{T}Y^{*}(t_{i})+d^{'}\sigma^{2}\right)+O_{p,n}(\triangle^{2})\\
		-\triangle\sum_{i=0}^{n}\left(Y^{*}(t_{i})^{T}Y^{*}(t_{i})+2Y^{*}(t_{i})^{T}\varepsilon_{i}+\varepsilon_{i}^{T}\varepsilon_{i}\right)\\
		=-\triangle\sum_{i=1}^{n}Y^{*}(t_{i})^{T}\varepsilon_{i}-\triangle\sum_{i=0}^{n}(d^{'}\sigma^{2}-\varepsilon_{i}^{T}\varepsilon_{i})+O_{p,n}(\triangle^{2}).
		\end{array}
		\]
	}with $\triangle\sum_{i=1}^{n}Y^{*}(t_{i})^{T}\varepsilon_{i}\sim N(0,\sigma^{2}\triangle^{2}\sum_{i=1}^{n}Y^{*}(t_{i})^{T}Y^{*}(t_{i}))=o_{p,n}(1)$
	and $\triangle\sum_{i=0}^{n}(d^{'}\sigma^{2}-\varepsilon_{i}^{T}\varepsilon_{i})=o_{p,n}(1)$
	by using the strong law of large numbers. Thus, 
	\[
	\int_{0}^{T}\left(Y^{*}(t)^{T}Y^{*}(t)+d^{'}\sigma^{2}\right)dt-\triangle\sum_{i=0}^{n}Y_{i}^{T}Y_{i}=o_{p,n}(1).
	\]
	Regarding the third term, we have:{\scriptsize
		\[
		\begin{array}{l}
		\int_{0}^{T}h_{\upsilon}^{l}(t)^{T}BU^{-1}B^{T}h_{\upsilon}^{l}(t)dt-\triangle\sum_{i=0}^{n-1}h_{\upsilon,i+1}^{d,l}(Y)^{T}BG(R_{\upsilon,i+1}^{d,l})B^{T}h_{\upsilon,i+1}^{d,l}(Y)\\
		=\triangle\sum_{i=0}^{n-1}\left(h_{\upsilon}^{l}(t_{i+1})^{T}BU^{-1}B^{T}h_{\upsilon}^{l}(t_{i+1})-h_{\upsilon,i+1}^{d,l}(Y)^{T}BG(R_{\upsilon,i+1}^{d,l})B^{T}h_{\upsilon,i+1}^{d,l}(Y)\right)+O_{n}(\triangle^{2})\\
		=\triangle\sum_{i=0}^{n-1}\left(\left(h_{\upsilon,i+1}^{d,l}(Y)+o_{p,n}(1)\right)^{T}BU^{-1}B^{T}\left(h_{\upsilon,i+1}^{d,l}(Y)+o_{p,n}(1)\right)-h_{\upsilon,i+1}^{d,l}(Y)^{T}BG(R_{\upsilon,i+1}^{d})B^{T}h_{\upsilon,i+1}^{d}(Y)\right)\\
		+O_{n}(\triangle^{2})\\
		=\triangle\sum_{i=0}^{n-1}\left(\left(h_{\upsilon,i+1}^{d,l}(Y)+o_{p,n}(1)\right)^{T}BU^{-1}B^{T}\left(h_{\upsilon,i+1}^{d,l}(Y)+o_{p,n}(1)\right)-h_{\upsilon,i+1}^{d,l}(Y)^{T}BU^{-1}B^{T}h_{\upsilon,i+1}^{d,l}(Y)\right)+O_{p,n}(\triangle)\\
		=\triangle\sum_{i=0}^{n-1}o_{p,n}(1)+O_{p,n}(\triangle)=o_{p,n}(1)
		\end{array}
		\]
	}uniformly on $\Gamma$. The third equality has been obtained by using
	lemma \ref{lem:h_R_continuity_limits} and the fifth one by using
	proposition \ref{prop:E_bounded_probability}. From this, we have
	$\sup_{\upsilon\in\varUpsilon}\left|S^{l}(\upsilon)-S_{n}^{l}(Y;\upsilon)\right|=o_{p,n}(1)$
	and triangular inequality gives us: 
	\[
	\sup_{\upsilon\in\varUpsilon}\left|S^{\infty}(\upsilon)-S_{n}^{l}(Y;\upsilon)\right|=o_{l}(1)+o_{p,n}(1).
	\]
	Application of proposition \ref{prop:identifiability_condition} and
	the fact that $\upsilon\longmapsto S^{\infty}(\upsilon)$ is continuous
	on the compact $\varUpsilon$ gives us the identifiability criteria
	required to apply theorem 5.7 in \citet{Vaart1998}.
\end{proof}
\begin{prop}
	\label{prop:identifiability_condition}Under conditions C1 to C5 $\upsilon^{*}$
	is the unique global minimizer of $S^{\infty}(\upsilon)$ on $\varUpsilon$. 
\end{prop}
We have shown in theorem \ref{thm:existence_solution_as_oca} the
control problem defining our estimator is well posed i.e. it exists
a control $\overline{u_{v}^{\infty}}$ such that $\overline{u_{v}^{\infty}}=\arg\min_{u\in\mathbb{U}_{\upsilon}}C_{T}^{\infty}(\upsilon,u)$
where $\mathbb{U}_{\upsilon}$ is the set of feedback controls:
\[
\mathbb{U}_{\upsilon}=\left\{ u_{\upsilon}(t)=U^{-1}B^{T}(R_{\upsilon}(t)X_{\upsilon,u_{\upsilon}}(t)+h_{\upsilon}(t)),\,\left(R_{\upsilon},h_{\upsilon}\right)\in L^{2}\left(\left[0,\,T\right],\,\mathbb{R}^{d\times d}\times\mathbb{R}^{d}\right)\right\} .
\]
By using theorem 2 in \citet{CimenBanks2004}, we derive $\overline{u_{v}^{l}}=\arg\min_{u\in\mathbb{U}_{\upsilon}}C_{T}^{l}(\upsilon,u)$
uniformly converges to $\overline{u_{v}^{\infty}}$ , hence:
\[
S^{\infty}(\upsilon)=d^{'}\sigma^{2}+\inf_{u\in\mathbb{U}_{\upsilon}}\left\{ \int_{0}^{T}\left(\left\Vert CX_{\upsilon,u}^{\infty}(t)-Y^{*}(t)\right\Vert _{2}^{2}+u(t)^{T}Uu(t)\right)dt\right\} .
\]
from which we derive the lower bound $S^{\infty}(\upsilon)\geq d^{'}\sigma^{2}$
for all $\upsilon\in\varUpsilon$. We now demonstrate this lower bound
$d^{'}\sigma^{2}$ can only be attained by $\upsilon^{*}$. The associated
ODE to the profiled cost $S^{\infty}(\upsilon)$ is: 
\[
\left\{ \begin{array}{l}
\dot{X_{v,\overline{u_{v}^{\infty}}}^{\infty}}=A_{\theta}(X_{v,\overline{u_{v}^{\infty}}}^{\infty}(t),t)X_{v,\overline{u_{v}^{\infty}}}^{\infty}+B\overline{u_{v}^{\infty}}\\
X_{v,\overline{u_{v}^{\infty}}}^{\infty}(0)=x_{0}.
\end{array}\right.
\]
For all $\upsilon\in\varUpsilon$, let us consider the fonctions $R_{\upsilon}^{0}$
and $h_{\upsilon}^{0}$ such that $R_{\upsilon}^{0}(t)X_{v,\overline{u_{v}^{\infty}}}^{\infty}(t)+h_{\upsilon}^{0}(t)=0$
for all $t\in\left[0,\,T\right]$, in this case the corresponding
control is $u_{\upsilon}^{0}(t)=0$ for all $t\in\left[0,\,T\right]$
and it belongs to $\mathbb{U}_{\upsilon}$. The solution $X_{\upsilon}^{\infty}$
of the corresponding ODE:
\[
\left\{ \begin{array}{l}
\dot{X}_{\upsilon}^{\infty}=A_{\theta}(X_{\upsilon,\overline{u_{\upsilon}^{\infty}}}^{\infty}(t),t)X_{\upsilon}^{\infty}\\
X_{\upsilon}^{\infty}(0)=x_{0}
\end{array}\right.
\]
gives us the upper bound $S^{\infty}(\upsilon)\leq d^{'}\sigma^{2}+\int_{0}^{T}\left\Vert CX_{\upsilon}^{\infty}(t)-Y^{*}(t)\right\Vert _{2}^{2}dt$.
From this, it is obvious that for $v^{*}$, $u_{v^{*}}^{0}$ is also
the optimal control since in this case the ODE becomes:
\[
\left\{ \begin{array}{l}
\dot{X}_{v^{*}}^{\infty}=A_{\theta}(X_{v^{*}}^{\infty}(t),t)X_{v^{*}}^{\infty}\\
X_{v^{*}}^{\infty}(0)=x_{0}^{*}
\end{array}\right.
\]
which is the true model and so $S^{\infty}(v^{*})\leq d^{'}\sigma^{2}+\int_{0}^{T}\left\Vert CX_{v^{*}}^{\infty}(t)-Y^{*}(t)\right\Vert _{2}^{2}dt=d^{'}\sigma^{2}$.
Moreover thanks to the identifiability condition, we know that $\int_{0}^{T}\left\Vert CX_{v}(t)-CX_{v^{*}}(t)\right\Vert _{2}^{2}dt=0$
if and only if $v=v^{*}$. Thus $S^{\infty}(v)=d^{'}\sigma^{2}$ if
and only if $v=v^{*}$.

\subsection{$\widehat{\theta}^{T,CI}$, linear case}
\begin{thm}
	\label{thm:Consistency_theorem_prof_ci}Under conditions LC1-LC2-LC3-LC4-LC5,
	we have $\widehat{\theta}\longrightarrow\theta^{*}$ in probability
	when $n\longrightarrow\infty$.
\end{thm}
\begin{proof}
	As in the nonlinear case, we derive:{\small{}
		\[
		\begin{array}{l}
		S^{CI}(\theta)-S_{n}^{CI}(Y;\theta)   \\ =h_{\theta,0}^{d}(Y)^{T}\left(R_{\theta,0}^{d}\right)^{-1}\left(h_{\theta,0}^{d}(Y)-h_{\theta}(0)\right)+h_{\theta,0}^{d}(Y)^{T}\left(\left(R_{\theta,0}^{d}\right)^{-1}-R_{\theta}(0)^{-1}\right)h_{\theta,0}^{d}(Y)\\
		 +  (h_{\theta,0}^{d}(Y)-h_{\theta}(0))\left(R_{\theta,0}^{d}\right)^{-1}h_{\theta}(0)\\
		 +  \int_{0}^{T}\left(Y^{*}(t)^{T}Y^{*}(t)+d^{'}\sigma^{2}\right)dt-\triangle\sum_{i=0}^{n}Y_{i}^{T}Y_{i}\\
		 -  \left(\int_{0}^{T}h_{\theta}(t)^{T}BU^{-1}B^{T}h_{\theta}(t)dt-\triangle\sum_{i=0}^{n-1}h_{\theta,i+1}^{d}(Y)^{T}BG(R_{\theta,i+1}^{d})B^{T}h_{\theta,i+1}^{d}(Y)\right)
		\end{array}
		\]
	}by using propositions \ref{prop:continuous_Riccati_accurate_representation}
	and \ref{prop:discrete_Riccati_accurate_representation}. Lemma \ref{lem:h_R_continuity_limits}
	gives us $\sup_{\theta\in\Theta}\left\Vert R_{\theta,i}^{d}-R_{\theta}(t_{i})\right\Vert _{2}=o_{p,n}(1)$
	and $\sup_{\theta\in\Theta}\left\Vert h_{\theta,i}^{d}-h_{\theta}(t_{i})\right\Vert _{2}=o_{p,n}(1)$.
	Since $R_{\theta,0}^{d}=R_{\theta}(0)+o_{p,n}(1)$ and LC3 holds,
	it exists $n'\in\mathbb{N}$ such that LC3disc holds i.e. $R_{\theta,0}^{d}$
	is invertible for all $n\geq n'$ and from equation (189) in the Matrix
	Cookbook, we derive $\sup_{\theta\in\Theta}\left\Vert \left(R_{\theta,0}^{d}\right)^{-1}-R_{\theta}(0)^{-1}\right\Vert _{2}=o_{p,n}(1)$.
	So $h_{\theta,0}^{d}(Y)^{T}\left(R_{\theta,0}^{d}\right)^{-1}\left(h_{\theta,0}^{d}(Y)-h_{\theta}(0)\right)=o_{p,n}(1)$,
	$(h_{\theta,0}^{d}(Y)-h_{\theta}(0))\left(R_{\theta,0}^{d}\right)^{-1}h_{\theta}(0)=o_{p,n}(1)$
	and $h_{\theta,0}^{d}(Y)^{T}\left(\left(R_{\theta,0}^{d}\right)^{-1}-R_{\theta}(0)^{-1}\right)h_{\theta,0}^{d}(Y)=o_{p,n}(1)$
	uniformly on $\Theta$. In theorem \ref{thm:Consistency_theorem},
	we already derived: 
	\[
	\int_{0}^{T}\left(Y^{*}(t)^{T}Y^{*}(t)+d^{'}\sigma^{2}\right)dt-\triangle\sum_{i=0}^{n}Y_{i}^{T}Y_{i}=o_{p,n}(1)
	\]
	and{\small{}
		\[
		\int_{0}^{T}h_{\theta}(t)^{T}BU^{-1}B^{T}h_{\theta}(t)dt-\triangle\sum_{i=0}^{n-1}h_{\theta,i+1}^{d}(Y)^{T}BG(R_{\theta,i+1}^{d})B^{T}h_{\theta,i+1}^{d}(Y)=o_{p,n}(1)
		\]
	}uniformly on $\Theta$ and so $\sup_{\theta}\left|S^{CI}(\theta)-S_{n}^{CI}(Y;\theta)\right|=o_{p,n}(1).$
   Similarly as in theorem \ref{thm:Consistency_theorem}, application
	of proposition \ref{prop:identifiability_condition_prof_ci} and $\theta\longmapsto S^{CI}(\theta)$
	continuity on $\Theta$ gives us the identifiability criteria required
	to apply theorem 5.7 in \citet{Vaart1998}.
\end{proof}
\begin{prop}
	\label{prop:identifiability_condition_prof_ci}Under conditions LC1-LC2-LC3-LC4-LC5,
	$\theta^{*}$ is the unique global minimizer of $S^{CI}(\theta)$
	on $\Theta$. 
\end{prop}
By definition, $S^{CI}(\theta)=d^{'}\sigma^{2}+\inf_{x_{0}}\inf_{u}\left\{ \int_{0}^{T}\left(\left\Vert CX_{\theta,x_{0},u}(t)-Y^{*}(t)\right\Vert _{2}^{2}+u(t)^{T}Uu(t)\right)dt\right\} $,
so $S^{CI}(\theta)\geq d^{'}\sigma^{2}$ for all $\theta\in\Theta$.
As in proposition \ref{prop:identifiability_condition}, we now demonstrate
this lower bound $d^{'}\sigma^{2}$ can only be attained by $\theta^{*}$.
The null control $u_{\theta}^{0}(t)=0$ for all $t\in\left[0,\,T\right]$
gives us the upper bound $S^{CI}(\theta)\leq d^{'}\sigma^{2}+\inf_{x_{0}\in\chi}\left\{ \int_{0}^{T}\left\Vert CX_{\theta,x_{0}}(t)-Y^{*}(t)\right\Vert _{2}^{2}dt\right\} $.
Since $\int_{0}^{T}\left\Vert CX_{\theta^{*},x_{0}^{*}}(t)-Y^{*}(t)\right\Vert _{2}^{2}dt=\int_{0}^{T}\left\Vert CX_{\theta^{*},x_{0}^{*}}(t)-CX_{\theta^{*},x_{0}^{*}}(t)\right\Vert _{2}^{2}dt=0$,
we can further refine the previous upper bound to $S^{CI}(\theta^{*})\leq d^{'}\sigma^{2}$
and conclude that $S^{CI}(\theta^{*})=d^{'}\sigma^{2}$. Again, by
using the identifiability condition we derive $S^{CI}(\theta)=d^{'}\sigma^{2}$
only if $\theta=\theta^{*}$.

\section{Asymptotic normality}

In this section the notation $Y^{d*}$ will often appear; it denotes
the set of discrete and perfectly measured observations i.e without
measurement noise $Y^{d*}:=\left\{ CX^{*}(t_{0}),\,\ldots,CX^{*}(t_{n})\right\} $.
We also introduce $X^{*l}:=\overline{X_{\theta^{*}}^{l}}$, $R^{*l}:=R_{\theta^{*}}^{l}$,
$h^{*l}:=h_{\theta^{*}}^{l}$, (resp. $X^{*d,l}:=\overline{X_{\theta^{*}}^{d,l}}$,
$R^{*d,l}:=R_{\theta^{*}}^{d,l}$, $h^{*d,l}:=h_{\theta^{*}}^{d,l}(Y^{d*})$)
the solution of the ODE (resp. the finite difference equation) evaluated
along the noiseless continuous signal $Y^{*}$ (resp. discrete signal
$Y^{d*}$).

\subsection{$\widehat{\theta}^{T}$, general case}
\begin{thm}
	\label{thm:Asymptotic_Normality}Under conditions C1 to C8 and if
	$l$ is such that $l=O_{n}(\sqrt{\Delta})$, then $(\widehat{\theta},\widehat{x_{0}})$
	is asymptotically normal and $(\widehat{\theta},\widehat{x_{0}})-\left(\theta^{*},x_{0}^{*}\right)=o_{p,n}(n^{-\frac{1}{2}})$.
\end{thm}
\begin{proof}
	By merging the proposition \ref{prop:gradient_S_true_param_as_representation}
	and proposition \ref{prop:param_as_representation}, we obtain the
	following asymptotic representation between $\widehat{\upsilon}$
	and $\upsilon^{*}$: {\small{}
		\[
		\begin{array}{lll}
		(\frac{\partial^{2}S^{\infty}(\upsilon^{*})}{\partial^{2}\upsilon}+o_{l}(1)+o_{p,n}(1))\left(\widehat{\upsilon}-\upsilon^{*}\right) & = & \left(\triangle\sum_{j=0}^{n}\epsilon_{j}^{T}\right)(K_{\upsilon}^{l}+o_{n}(1))\\
		& + & L\left(\triangle\sum_{j=0}^{n}\epsilon_{j}\right)+o_{p,n}(\sqrt{\Delta})+o_{l}(1)
		\end{array}
		\]
	}with $K_{\upsilon}^{l}$ and $L$ defined in proposition \ref{prop:param_as_representation}.
	So if we choose $l$ such that $l=O_{n}(\sqrt{\Delta})$ and use condition
	8 which ensures the matrix $\frac{\partial^{2}S^{\infty}(\theta^{*})}{\partial^{2}\theta}+o_{l}(1)+o_{p,n}(1)$
	tends to a nonsingular one with probability 1, we can use the central
	limit theorem to conclude.
\end{proof}
\begin{prop}
	\label{prop:gradient_S_true_param_as_representation}Under conditions
	C1 to C7, we have
	\[
	\begin{array}{l}
	-\nabla_{\upsilon}S_{n}^{l}(Y;\upsilon^{*})\\
	=\left(\triangle\sum_{j=0}^{n}\epsilon_{j}^{T}\right)(K_{\upsilon}^{l}+o_{n}(1))+L\left(\triangle\sum_{j=0}^{n}\epsilon_{j}\right)+o_{p,n}(\sqrt{\Delta})+o_{l}(1)
	\end{array}
	\]
	with $K_{\upsilon}^{l}=2CBU^{-1}B^{T}\int_{0}^{T}\frac{\partial h^{*l}(t)}{\partial\upsilon}dt$
	and $L=\left(\begin{array}{c}
	0_{p,d^{'}}\\
	-2C^{T}
	\end{array}\right).$
\end{prop}
\begin{proof}
	First of all, we use the following decomposition:
	\[
	\begin{array}{lll}
	-\nabla_{\upsilon}S_{n}^{l}(Y;\upsilon^{*}) & = & \nabla_{\upsilon}S_{n}^{l}(Y^{d*};\upsilon^{*})-\nabla_{\upsilon}S_{n}^{l}(Y;\upsilon^{*})\\
	& + & \nabla_{\upsilon}S^{l}(\upsilon^{*})-\nabla_{\upsilon}S_{n}^{l}(Y^{d*};\upsilon^{*})+\nabla_{\upsilon}S^{\infty}(\upsilon^{*})-\nabla_{\upsilon}S^{l}(\upsilon^{*})
	\end{array}
	\]
	since first order conditions imposes $\nabla_{\upsilon}S^{\infty}(\upsilon^{*})=0$.
	Moreover lemma \ref{lem:Continuous_S_Gradient_Hessian_convergence}
	and \ref{lem:deterministic_gradient_behavior} gives us respectively
	$\nabla_{\upsilon}S^{\infty}(\upsilon^{*})-\nabla_{\upsilon}S^{l}(\upsilon^{*})=o_{l}(1)$
	and $\nabla_{\upsilon}S^{l}(\upsilon^{*})-\nabla_{\upsilon}S_{n}^{l}(Y^{d*};\upsilon^{*})=O_{p,n}(\triangle)$,
	so the previous asymptotic decomposition becomes:
	\[
	\begin{array}{lll}
	-\nabla_{\upsilon}S_{n}^{l}(Y;\upsilon^{*}) & = & \nabla_{\upsilon}S_{n}^{l}(Y^{d*};\upsilon^{*})-\nabla_{\upsilon}S_{n}^{l}(Y;\upsilon^{*})+o_{l}(1)+O_{p,n}(\triangle).\end{array}
	\]
	Now, we analyze the asymptotic behavior of $\nabla_{\upsilon}S_{n}^{l}(Y^{d*};\upsilon^{*})-\nabla_{\upsilon}S_{n}^{l}(Y;\upsilon^{*})$.
	We denote $R_{i}^{d,l}:=R_{\upsilon^{*},i}^{d,l}$, $h_{i}^{d,l}:=h_{\upsilon^{*},i}^{d,l}(Y)$
	and $X^{d,l}:=\overline{X_{\upsilon^{*}}^{d,l}}$, the differences
	$R_{i}^{*d,l}-R_{i}^{d,l}$, $h_{i}^{*d,l}-h_{i}^{d,l}$ and $X^{*d,l}(t_{i+1})-X^{d,l}(t_{i+1})$
	respect the equations:
	
	{\small{}
		\[
		\begin{array}{lll}
		R_{i}^{*d,l}-R_{i}^{d,l} & = & \left(I_{d}+\Delta A_{\theta^{*}}(X^{*d,l-1}(t_{i}),t_{i})-\triangle R_{i+1}^{d,l}BU^{-1}B^{T}\right)\left(R_{i+1}^{*d,l}-R_{i+1}^{d,l}\right)\\
		& + & \triangle\left(R_{i+1}^{*d,l}-R_{i+1}^{d,l}\right)\left(BU^{-1}B^{T}R_{i+1}^{*d,l}+A_{\theta^{*}}(X^{d,l-1}(t_{i}),t_{i})\right)\\
		& + & \Delta\left(A_{\theta}(X^{*d,l-1}(t_{i}),t_{i})-A_{\theta^{*}}(X^{d,l-1}(t_{i}),t_{i})\right)^{T}R_{i+1}^{d,l}\\
		& + & \Delta R_{i+1}^{*d,l}\left(A_{\theta^{*}}(X^{*d,l-1}(t_{i}),t_{i})-A_{\theta^{*}}(X^{d,l-1}(t_{i}),t_{i})\right)+O_{p,n}(\triangle^{2})
		\end{array}
		\]
		\[
		\begin{array}{lll}
		h_{i}^{*d,l}-h_{i}^{d,l} & = & \left(I_{d}+\Delta A_{\theta^{*}}(X^{*d,l-1}(t_{i}),t_{i})^{T}+\triangle R_{i+1}^{d,l}BU^{-1}B^{T}\right)\left(h_{i+1}^{*d,l}-h_{i+1}^{d,l}\right)-\triangle C^{T}\epsilon_{i}\\
		& + & \Delta\left(A_{\theta^{*}}(X^{*d,l-1}(t_{i}),t_{i})-A_{\theta^{*}}(X^{d,l-1}(t_{i}),t_{i})\right)^{T}h_{i+1}^{d,l}\\
		& + & \triangle\left(R_{i+1}^{d,l}-R_{i+1}^{*d,l}\right)BU^{-1}B^{T}h_{i+1}^{*d,l}+O_{p,n}(\triangle^{2})
		\end{array}
		\]
		\[
		\begin{array}{l}
		X^{*d,l}(t_{i+1})-X^{d,l}(t_{i+1})\\
		  =  \left(I_{d}+\Delta A_{\theta^{*}}(X^{*d,l-1}(t_{i}),t_{i})+\triangle BU^{-1}B^{T}R^{*d,l}(t_{i+1})\right)\left(X^{*d,l}(t_{i})-X^{d,l}(t_{i})\right)\\
		 +  \triangle BU^{-1}B^{T}\left(\left(R^{*d,l}(t_{i+1})-R^{d,l}(t_{i+1})\right)X^{d,l}(t_{i})+h^{*d,l}(t_{i+1})-h^{d,l}(t_{i+1})\right)\\
		 +  \Delta\left(A_{\theta^{*}}(X^{*d,l-1}(t_{i}),t_{i})-A_{\theta^{*}}(X^{d,l-1}(t_{i}),t_{i})\right)X^{d,l}(t_{i})+O_{p,n}(\triangle^{2}).
		\end{array}
		\]
	}From these equation, we prove by induction $h_{i}^{*d,l}-h_{i}^{d,l}=-\triangle C^{T}\sum_{j=i}^{n}\epsilon_{j}+O_{p,n}(\triangle)$
	for all $l\in\mathbb{N}$, indeed $h_{n}^{*d,l}-h_{n}^{d,l}=-\triangle C^{T}\epsilon_{n}$,
	and by assuming $h_{i+1}^{*d,l}-h_{i+1}^{d,l}=-\triangle C^{T}\sum_{j=i+1}^{n}\epsilon_{j}+O_{p,n}(\triangle)$
	we have:{\small{}
		\[
		\begin{array}{l}
		h_{i}^{*d,l}-h_{i}^{d,l}\\
		  = \left(I_{d}+O_{p,n}(\Delta)\right)(-\triangle C^{T}\sum_{j=i+1}^{n}\epsilon_{j}+O_{p,n}(\triangle))+O_{p,n}(\triangle)-\triangle C^{T}\epsilon_{i}+O_{p,n}(\triangle)\\
		 =  -\triangle C^{T}\sum_{j=i+1}^{n}\epsilon_{j}+O_{p,n}(\triangle)+O_{p,n}(\Delta)O_{p,n}(1)+O_{p,n}(\Delta)O_{p,n}(\triangle)-\triangle C^{T}\epsilon_{i}+O_{p,n}(\triangle)\\
		 =  -\triangle C^{T}\sum_{j=i}^{n}\epsilon_{j}+O_{p,n}(\triangle).
		\end{array}
		\]
	}Again, we can prove by induction on $i\in\left\llbracket 0,\,n\right\rrbracket $
	that for each $l\in\mathbb{N}$, $R_{i}^{*d,l}-R_{i}^{d,l}=O_{p,n}(\triangle)$.
	By differentiating these previous equations, we obtain: {\small
		\[
		\begin{array}{l}
		 \frac{\partial R_{i}^{*d,l}}{\partial\upsilon}-\frac{\partial R_{i}^{d,l}}{\partial\upsilon}\\
		=  \frac{\partial R_{i+1}^{*d,l}}{\partial\upsilon}-\frac{\partial R_{i+1}^{d,l}}{\partial\upsilon}+\Delta\left(\frac{\partial R_{i+1}^{*d,l}}{\partial\upsilon}A_{\theta^{*}}(X^{*d,l-1}(t_{i}),t_{i})-\frac{\partial R_{i+1}^{d,l}}{\partial\upsilon}A_{\theta^{*}}(X^{d,l-1}(t_{i}),t_{i})\right)\\
		+  \Delta\left(R_{i+1}^{*d,l}\frac{\partial A_{\theta^{*}}}{\partial\upsilon}(X^{*d,l-1}(t_{i}),t_{i})-R_{i+1}^{d,l}\frac{\partial A_{\theta^{*}}}{\partial\upsilon}(X^{d,l-1}(t_{i}),t_{i})\right)\\
		+  \Delta\left(R_{i+1}^{*d,l}\frac{\partial A_{\theta^{*}}}{\partial x}(X^{*d,l-1}(t_{i}),t_{i})\frac{\partial X^{*d,l-1}(t_{i})}{\partial\upsilon}-R_{i+1}^{d,l}\frac{\partial A_{\theta^{*}}}{\partial x}(X^{d,l-1}(t_{i}),t_{i})\frac{\partial X^{d,l-1}(t_{i})}{\partial\upsilon}\right)\\
		+  \Delta\left(\frac{\partial A_{\theta}}{\partial\upsilon}(X^{*d,l-1}(t_{i}),t_{i})^{T}R_{i+1}^{*d,l}-\frac{\partial A_{\theta}}{\partial\upsilon}(X^{d,l-1}(t_{i}),t_{i})^{T}R_{i+1}^{d,l}\right)\\
		+  \Delta\left(\frac{\partial A_{\theta^{*}}}{\partial x}(X^{*d,l-1}(t_{i}),t_{i})^{T}\frac{\partial X^{*d,l-1}(t_{i})}{\partial\upsilon}R_{i+1}^{*d,l}-\frac{\partial A_{\theta^{*}}}{\partial x}(X^{d,l-1}(t_{i}),t_{i})^{T}\frac{\partial X^{d,l-1}(t_{i})}{\partial\upsilon}R_{i+1}^{d,l}\right)\\
		+  \Delta\left(A_{\upsilon}(X^{*d,l-1}(t_{i}),t_{i})^{T}\frac{\partial R_{i+1}^{*d,l}}{\partial\upsilon}-A_{\theta}(X^{d,l-1}(t_{i}),t_{i})^{T}\frac{\partial R_{i+1}^{d,l}}{\partial\upsilon}\right)\\
		-  \Delta\left(\frac{\partial R_{i+1}^{*d,l}}{\partial\upsilon}BU^{-1}B^{T}R_{i+1}^{*d,l}-\frac{\partial R_{i+1}^{d,l}}{\partial\upsilon}BU^{-1}B^{T}R_{i+1}^{d,l}\right)\\
		-  \Delta\left(R_{i+1}^{*d,l}BU^{-1}B^{T}\frac{\partial R_{i+1}^{*d,l}}{\partial\upsilon}-R_{i+1}^{d,l}BU^{-1}B^{T}\frac{\partial R_{i+1}^{d,l}}{\partial\upsilon}\right)+O_{p,n}(\triangle^{2})
		\end{array}
		\]
		\[
		\begin{array}{l}
		 \frac{\partial h_{i}^{*d,l}}{\partial\upsilon}-\frac{\partial h_{i}^{d,l}}{\partial\upsilon}\\
		=  \frac{\partial h_{i+1}^{*d,l}}{\partial\upsilon}-\frac{\partial h_{i+1}^{d,l}}{\partial\upsilon}+\Delta\frac{\partial A_{\theta^{*}}}{\partial\upsilon}(X^{*d,l-1}(t_{i}),t_{i})^{T}h_{i+1}^{*d,l}-\Delta\frac{\partial A_{\theta^{*}}}{\partial\upsilon}(X^{d,l-1}(t_{i}),t_{i})^{T}h_{i+1}^{d,l}\\
		+  \Delta\frac{\partial A_{\theta^{*}}}{\partial x}(X^{*d,l-1}(t_{i}),t_{i})^{T}\frac{\partial X^{*d,l-1}(t_{i})}{\partial\upsilon}h_{i+1}^{*d,l}-\Delta\frac{\partial A_{\theta^{*}}}{\partial x}(X^{d,l-1}(t_{i}),t_{i})^{T}\frac{\partial X^{d,l-1}(t_{i})}{\partial\upsilon}h_{i+1}^{d,l}\\
		+  \Delta A_{\theta^{*}}(X^{*d,l-1}(t_{i}),t_{i})^{T}\frac{\partial h_{i+1}^{*d,l}}{\partial\upsilon}-\Delta A_{\theta^{*}}(X^{d,l-1}(t_{i}),t_{i})^{T}\frac{\partial h_{i+1}^{d,l}}{\partial\upsilon}\\
		- \triangle\left(\frac{\partial R_{i+1}^{*d,l}}{\partial\upsilon}BU^{-1}B^{T}h_{i+1}^{*d,l}-\frac{\partial R_{i+1}^{d,l}}{\partial\upsilon}BU^{-1}B^{T}h_{i+1}^{d,l}\right)\\
			- \triangle\left(R_{i+1}^{*d,l}BU^{-1}B^{T}\frac{\partial h_{i+1}^{*d,l}}{\partial\upsilon}-R_{i+1}^{d,l}BU^{-1}B^{T}\frac{\partial h_{i+1}^{d,l}}{\partial\upsilon}\right)+O_{p,n}(\triangle^{2})
		\end{array}
		\]
	}From this, we can prove by induction $\frac{\partial R_{i}^{*d,l}}{\partial\upsilon}-\frac{\partial R_{i}^{d,l}}{\partial\upsilon}=O_{p,n}(\Delta)$,
	$\frac{\partial h_{i}^{*d,l}}{\partial\upsilon}-\frac{\partial h_{i}^{d,l}}{\partial\upsilon}=O_{p,n}(\Delta).$
	 For notation clarity we treat the case $d=1$, by using these approximations,
	we obtain for $\nabla_{\theta}S_{n}^{l}(Y^{d*};v^{*})-\nabla_{\theta}S_{n}^{l}(Y;v^{*})$:
	\[
	\begin{array}{l}
	\nabla_{\theta}S_{n}^{l}(Y^{d*};v^{*})-\nabla_{\theta}S_{n}^{l}(Y;v^{*})\\
	=\left(x_{0}^{*}\right)^{T}\left(\frac{\partial R_{0}^{*d,l}}{\partial\theta}-\frac{\partial R_{0}^{d,l}}{\partial\theta}\right)x_{0}^{*}+2\left(x_{0}^{*}\right)^{T}\left(\frac{\partial h_{0}^{*d,l}}{\partial\theta}-\frac{\partial h_{0}^{d,l}}{\partial\theta}\right)\\
	-2\triangle\sum_{i=0}^{n-1}\left(\left(h_{i+1}^{*d,l}\right){}^{T}BU^{-1}B^{T}\frac{\partial h_{i+1}^{*d,l}}{\partial\theta}-\left(h_{i+1}^{d,l}\right){}^{T}BU^{-1}B^{T}\frac{\partial h_{i+1}^{d,l}}{\partial\theta}\right)+O_{p,n}(\triangle)\\
	=\left(x_{0}^{*}\right)^{T}\left(\frac{\partial R_{0}^{*d,l}}{\partial\theta}-\frac{\partial R_{0}^{d,l}}{\partial\theta}\right)x_{0}^{*}+2\left(x_{0}^{*}\right)^{T}\left(\frac{\partial h_{0}^{*d,l}}{\partial\theta}-\frac{\partial h_{0}^{d,l}}{\partial\theta}\right)\\
	-2\triangle\sum_{i=0}^{n-1}\left(h_{i+1}^{*d,l}-h_{i+1}^{d,l}\right){}^{T}BU^{-1}B^{T}\frac{\partial h_{i+1}^{*d,l}}{\partial\theta}\\
	-2\triangle\sum_{i=0}^{n-1}\left(h_{i+1}^{d,l}\right){}^{T}BU^{-1}B^{T}\left(\frac{\partial h_{i+1}^{*d,l}}{\partial\theta}-\frac{\partial h_{i+1}^{d,l}}{\partial\theta}\right)+O_{p,n}(\triangle)\\
	=-2\triangle\sum_{i=0}^{n-1}\left(h_{i+1}^{*d,l}-h_{i+1}^{d,l}\right){}^{T}BU^{-1}B^{T}\frac{\partial h_{i+1}^{*d,l}}{\partial\theta}+O_{p,n}(\triangle)\\
	=-2\triangle\sum_{i=0}^{n-1}\left(-\triangle C^{T}\sum_{j=i+1}^{n-1}\epsilon_{j}+O_{p,n}(\triangle)\right){}^{T}BU^{-1}B^{T}\frac{\partial h_{i+1}^{*d,l}}{\partial\theta}+O_{p,n}(\triangle)\\
	=2\triangle^{2}\sum_{i=0}^{n-1}\left(\sum_{j=i+1}^{n-1}\epsilon_{j}^{T}\right)CBU^{-1}B^{T}\frac{\partial h_{i+1}^{*d,l}}{\partial\theta}+O_{p,n}(\triangle).
	\end{array}
	\]
	Let us denote $K_{i}^{l}(\beta)=2CBU^{-1}B^{T}\frac{\partial h_{i+1}^{*l}}{\partial\beta}$
	and decompose the right hand side term{\small{}:
		\[
		\begin{array}{lll}
		\triangle^{2}\sum_{i=0}^{n-1}\left(\sum_{j=i+1}^{n}\epsilon_{j}^{T}\right)K_{i}^{l}(\theta) & = & \triangle^{2}\sum_{i=0}^{n-1}\left(\sum_{j=0}^{n}\epsilon_{j}^{T}-\sum_{j=0}^{i}\epsilon_{j}^{T}\right)K_{i}^{l}(\theta)\\
		& = & \triangle^{2}\sum_{i=0}^{n-1}\left(\sum_{j=0}^{n}\epsilon_{j}^{T}\right)K_{i}^{l}(\theta)\\
		& - & \triangle^{2}\sum_{i=0}^{n-1}\left(\sum_{j=0}^{i}\epsilon_{j}^{T}\right)K_{i}^{l}(\theta).
		\end{array}
		\]
	}By definition of $K_{i}^{l}(\theta)$, $\triangle\sum_{i=0}^{n-1}K_{i}^{l}(\theta)$
	converges to the limit $K^{l}(\theta)=2CBU^{-1}B^{T}\int_{0}^{T}\frac{\partial h^{*l}(t)}{\partial\theta}dt$
	as a Riemann sum and so $\triangle^{2}\sum_{j=0}^{n}\epsilon_{j}^{T}\sum_{i=0}^{n-1}K_{i}^{l}(\theta)=\triangle\sum_{j=0}^{n}\epsilon_{j}^{T}(K^{l}(\theta)+o_{n}(1))$.
	Moreover{\small{},
		\[
		\begin{array}{lll}
		\triangle^{2}\sum_{i=0}^{n-1}\left(\sum_{j=0}^{i}\epsilon_{j}^{T}\right)K_{i}^{l}(\theta) & = & \triangle^{\frac{3}{2}}\sum_{i=0}^{n-1}\sqrt{\Delta}\left(\sum_{j=0}^{i}\epsilon_{j}^{T}\right)K_{i}^{l}(\theta)\\
		& = & \triangle^{\frac{3}{2}}\sum_{i=0}^{n-1}K_{i}^{l}(\theta)o_{p,n}(1)\\
		& = & \triangle\sum_{i=0}^{n-1}K_{i}^{l}(\theta)o_{p,n}(\sqrt{\Delta})=o_{p,n}(\sqrt{\Delta})
		\end{array}
		\]
	}since $\sqrt{\Delta}\sum_{j=0}^{i}\epsilon_{j}=o_{p,n}(1)$ for each
	$i$. Thus:{\small{}
		\[
		\begin{array}{lll}
		\triangle^{2}\sum_{i=0}^{n-1}\left(\sum_{j=i+1}^{n}\epsilon_{j}^{T}\right)K_{i}^{l}(\theta) & = & \left(\triangle\sum_{j=0}^{n}\epsilon_{j}^{T}\right)(K^{l}(\theta)+o_{n}(1))+o_{p,n}(\sqrt{\Delta}).\end{array}
		\]
		and we derive from that }
	\[
	\begin{array}{l}
	\nabla_{\theta}S_{n}^{l}(Y^{d*};v^{*})-\nabla_{\theta}S_{n}^{l}(Y;v^{*})\\
	=\left(\triangle\sum_{j=0}^{n}\epsilon_{j}^{T}\right)(K^{l}(\theta)+o_{n}(1))+o_{p,n}(\sqrt{\Delta}).
	\end{array}
	\]
	Now let us focus on $\nabla_{x_{0}}S_{n}^{l}(Y^{d*};v^{*})-\nabla_{x_{0}}S_{n}^{l}(Y;v^{*})$:
	\[
	\begin{array}{l}
	\nabla_{x_{0}}S_{n}^{l}(Y^{d*};v^{*})-\nabla_{x_{0}}S_{n}^{l}(Y;v^{*})\\
	=2\left(R_{0}^{*d,l}-R_{0}^{d,l}\right)x_{0}^{*}+\left(x_{0}^{*}\right)^{T}\left(\frac{\partial R_{0}^{*d,l}}{\partial x_{0}}-\frac{\partial R_{0}^{d,l}}{\partial x_{0}}\right)x_{0}^{*}\\
	+2\left(h_{0}^{*d,l}-h_{0}^{d,l}\right)+2\left(x_{0}^{*}\right)^{T}\left(\frac{\partial h_{0}^{*d,l}}{\partial x_{0}}-\frac{\partial h_{0}^{d,l}}{\partial x_{0}}\right)\\
	-2\triangle\sum_{i=0}^{n-1}\left(h_{i+1}^{*d,l}-h_{i+1}^{d,l}\right){}^{T}BU^{-1}B^{T}\frac{\partial h_{i+1}^{*d,l}}{\partial x_{0}}\\
	-2\triangle\sum_{i=0}^{n-1}\left(h_{i+1}^{d,l}\right){}^{T}BU^{-1}B^{T}\left(\frac{\partial h_{i+1}^{*d,l}}{\partial\theta}-\frac{\partial h_{i+1}^{d,l}}{\partial x_{0}}\right)+O_{p,n}(\triangle)\\
	=2\left(h_{0}^{*d,l}-h_{0}^{d,l}\right)-2\triangle\sum_{i=0}^{n-1}\left(h_{i+1}^{*d,l}-h_{i+1}^{d,l}\right){}^{T}BU^{-1}B^{T}\frac{\partial h_{i+1}^{*d,l}}{\partial x_{0}}+O_{p,n}(\triangle)\\
	=2\left(-\triangle C^{T}\sum_{j=0}^{n}\epsilon_{j}+O_{p,n}(\triangle)\right)\\
	-2\triangle\sum_{i=0}^{n-1}\left(-\triangle C^{T}\sum_{j=i+1}^{n-1}\epsilon_{j}+O_{p,n}(\triangle)\right){}^{T}BU^{-1}B^{T}\frac{\partial h_{i+1}^{*d,l}}{\partial x_{0}}+O_{p,n}(\triangle)\\
	=-2C^{T}\left(\triangle\sum_{j=0}^{n}\epsilon_{j}\right)+2\triangle^{2}\sum_{i=0}^{n-1}\left(\sum_{j=i+1}^{n-1}\epsilon_{j}\right)CBU^{-1}B^{T}\frac{\partial h_{i+1}^{*d,l}}{\partial x_{0}}+O_{p,n}(\triangle).
	\end{array}
	\]
	Similarly as in $\nabla_{\theta}S_{n}^{l}(Y^{d*};v^{*})-\nabla_{\theta}S_{n}^{l}(Y;v^{*})$
	case, we reformulate this expression as: 
	\[
	\begin{array}{l}
	\nabla_{x_{0}}S_{n}^{l}(Y^{d*};v^{*})-\nabla_{x_{0}}S_{n}^{l}(Y;v^{*})\\
	=-2C^{T}\left(\triangle\sum_{j=0}^{n}\epsilon_{j}\right)+\left(\triangle\sum_{j=0}^{n}\epsilon_{j}^{T}\right)(K^{l}(x_{0})+o_{n}(1))+o_{p,n}(\sqrt{\Delta}).
	\end{array}
	\]
\end{proof}
\begin{prop}
	\label{prop:param_as_representation}Under conditions C1 to C7, we
	have $-\nabla_{\upsilon}S_{n}^{l}(Y;\upsilon^{*})=(\frac{\partial^{2}S^{\infty}(\upsilon^{*})}{\partial^{2}\upsilon}+o_{p,n}(1)+o_{l}(1))\left(\widehat{\upsilon}-\upsilon^{*}\right).$
\end{prop}
\begin{proof}
	For notation clarity we treat the case $d=1$, if $\theta\longmapsto A_{\theta}$
	is $C^{1}$ on $\Theta$, then $\upsilon\longmapsto S_{n}^{l}(Y;\upsilon)$
	is $C^{1}$ as well {\small{}with}
	\[
	\begin{array}{l}
	\nabla_{\theta}S_{n}^{l}(Y;\upsilon)\,=x_{0}^{T}\frac{\partial R_{\upsilon,0}^{d,l}}{\partial\theta}x_{0}+2x_{0}^{T}\frac{\partial h_{\upsilon,0}^{d,l}}{\partial\theta}(Y)\\
	-\frac{\partial}{\partial\theta}\left(\triangle\sum_{i=0}^{n-1}h_{\upsilon,i+1}^{d,l}(Y)^{T}BG(R_{\upsilon,i+1}^{d,l})B^{T}h_{\upsilon,i+1}^{d,l}(Y)\right)\\
	=x_{0}^{T}\frac{\partial R_{\upsilon,0}^{d,l}}{\partial\theta}x_{0}+2x_{0}^{T}\frac{\partial h_{\upsilon,0}^{d,l}}{\partial\theta}(Y)\\
	-2\triangle\sum_{i=0}^{n-1}h_{\upsilon,i+1}^{d,l}(Y)^{T}BG(R_{\upsilon,i+1}^{d,l})B^{T}\frac{\partial h_{\upsilon,i+1}^{d,l}}{\partial\theta}(Y)\\
	-\triangle\sum_{i=0}^{n-1}h_{\upsilon,i+1}^{d,l}(Y)^{T}B\frac{\partial G}{\partial\theta}(R_{\upsilon,i+1}^{d,l})B^{T}h_{\upsilon,i+1}^{d,l}(Y)
	\end{array}
	\]
	and
	\[
	\begin{array}{l}
	\nabla_{x_{0}}S_{n}^{l}(Y;v)\,=2R_{\upsilon,0}^{d,l}x_{0}+x_{0}^{T}\frac{\partial R_{\upsilon,0}^{d,l}}{\partial x_{0}}x_{0}+2h_{\upsilon,0}^{d,l}+2x_{0}^{T}\frac{\partial h_{\upsilon,0}^{d,l}}{\partial x_{0}}\\
	-2\triangle\sum_{i=0}^{n-1}h_{\upsilon,i+1}^{d,l}(Y)^{T}BG(R_{\upsilon,i+1}^{d,l})B^{T}\frac{\partial h_{\upsilon,i+1}^{d,l}}{\partial x_{0}}(Y)\\
	-\triangle\sum_{i=0}^{n-1}h_{\upsilon,i+1}^{d,l}(Y)^{T}B\frac{\partial G}{\partial x_{0}}(R_{\upsilon,i+1}^{d,l})B^{T}h_{\upsilon,i+1}^{d,l}(Y).
	\end{array}
	\]
	If now $\theta\longmapsto A_{\theta}$ is $C^{2}$ on $\Theta$,
	from proposition \ref{lem:h_R_continuity_limits} we derive $\upsilon\longmapsto S_{n}^{l}(Y;\upsilon)$
	is $C^{2}$ as well with components equal to: 
	\[
	\begin{array}{l}
	\frac{\partial^{2}S_{n}^{l}(Y,v)}{\partial^{2}\theta}\,=x_{0}^{T}\frac{\partial^{2}R_{v,0}^{d,l}}{\partial^{2}\theta}x_{0}+2x_{0}^{T}\frac{\partial^{2}h_{v,0}^{d,l}}{\partial^{2}\theta}(Y)\\
	-2\triangle\sum_{i=0}^{n-1}\frac{\partial}{\partial\theta}\left(h_{v,i+1}^{d,l}(Y)^{T}BG(R_{v,i+1}^{d,l})B^{T}\frac{\partial h_{v,i+1}^{d,l}}{\partial\theta}(Y)\right)\\
	-\triangle\sum_{i=0}^{n-1}\frac{\partial}{\partial\theta}\left(h_{v,i+1}^{d,l}(Y)^{T}B\frac{\partial G}{\partial\theta}(R_{v,i+1}^{d,l})B^{T}h_{v,i+1}^{d,l}(Y)\right)\\
	=x_{0}^{T}\frac{\partial^{2}R_{v,0}^{d,l}}{\partial^{2}\theta}x_{0}+2x_{0}^{T}\frac{\partial^{2}h_{v,0}^{d,l}}{\partial^{2}\theta}(Y)\\
	-2\triangle\sum_{i=0}^{n-1}\frac{\partial h_{v,i+1}^{d,l}(Y)^{T}}{\partial\theta}BG(R_{v,i+1}^{d,l})B^{T}\frac{\partial h_{v,i+1}^{d,l}}{\partial\theta}(Y)\\
	-2\triangle\sum_{i=0}^{n-1}h_{v,i+1}^{d,l}(Y)^{T}B\frac{\partial G}{\partial\theta}(R_{v,i+1}^{d,l})B^{T}\frac{\partial h_{v,i+1}^{d,l}}{\partial\theta}(Y)\\
	-2\triangle\sum_{i=0}^{n-1}h_{v,i+1}^{d,l}(Y)^{T}BG(R_{v,i+1}^{d,l})B^{T}\frac{\partial^{2}h_{v,i+1}^{d,l}}{\partial^{2}\theta}(Y)\\
	-\triangle\sum_{i=0}^{n-1}\frac{\partial h_{v,i+1}^{d,l}(Y)^{T}}{\partial\theta}B\frac{\partial G}{\partial\theta}(R_{v,i+1}^{d,l})B^{T}h_{v,i+1}^{d,l}(Y)\\
	-\triangle\sum_{i=0}^{n-1}h_{v,i+1}^{d,l}(Y)^{T}B\frac{\partial^{2}G}{\partial^{2}\theta}(R_{v,i+1}^{d,l})B^{T}h_{v,i+1}^{d,l}(Y)\\
	-\triangle\sum_{i=0}^{n-1}h_{v,i+1}^{d,l}(Y)^{T}B\frac{\partial G}{\partial\theta}(R_{v,i+1}^{d,l})B^{T}\frac{\partial h_{v,i+1}^{d,l}(Y)}{\partial\theta}.\\
	=x_{0}^{T}\frac{\partial^{2}R_{v,0}^{d,l}}{\partial^{2}\theta}x_{0}+2x_{0}^{T}\frac{\partial^{2}h_{v,0}^{d,l}}{\partial^{2}\theta}(Y)\\
	-2\triangle\sum_{i=0}^{n-1}\frac{\partial h_{v,i+1}^{d,l}(Y)^{T}}{\partial\theta}BU^{-1}B^{T}\frac{\partial h_{v,i+1}^{d,l}}{\partial\theta}(Y)\\
	-2\triangle\sum_{i=0}^{n-1}h_{v,i+1}^{d,l}(Y)^{T}BU^{-1}B^{T}\frac{\partial^{2}h_{v,i+1}^{d,l}}{\partial^{2}\theta}(Y)+O_{p,n}(\Delta).
	\end{array}
	\]
	\[
	\begin{array}{l}
	\frac{\partial^{2}S_{n}^{l}(Y,v)}{\partial^{2}x_{0}}\,=2\left(\frac{\partial R_{\upsilon,0}^{d,l}}{\partial x_{0}}x_{0}+R_{\upsilon,0}^{d,l}\right)+\left(2\frac{\partial R_{\upsilon,0}^{d,l}}{\partial x_{0}}x_{0}+x_{0}^{T}\frac{\partial^{2}R_{\upsilon,0}^{d,l}}{\partial^{2}x_{0}}x_{0}\right)+4\frac{\partial h_{\upsilon,0}^{d,l}}{\partial x_{0}}+2x_{0}^{T}\frac{\partial^{2}h_{\upsilon,0}^{d,l}}{\partial^{2}x_{0}}\\
	-2\triangle\sum_{i=0}^{n-1}\frac{\partial}{\partial x_{0}}\left(h_{\upsilon,i+1}^{d,l}(Y)^{T}BG(R_{\upsilon,i+1}^{d,l})B^{T}\frac{\partial h_{\upsilon,i+1}^{d,l}}{\partial x_{0}}(Y)\right)\\
	-\triangle\sum_{i=0}^{n-1}\frac{\partial}{\partial x_{0}}\left(h_{\upsilon,i+1}^{d,l}(Y)^{T}B\frac{\partial G}{\partial x_{0}}(R_{\upsilon,i+1}^{d,l})B^{T}h_{\upsilon,i+1}^{d,l}(Y)\right)\\
	=2\left(\frac{\partial R_{\upsilon,0}^{d,l}}{\partial x_{0}}x_{0}+R_{\upsilon,0}^{d,l}\right)+\left(2\frac{\partial R_{\upsilon,0}^{d,l}}{\partial x_{0}}x_{0}+x_{0}^{T}\frac{\partial^{2}R_{\upsilon,0}^{d,l}}{\partial^{2}x_{0}}x_{0}\right)+4\frac{\partial h_{\upsilon,0}^{d,l}}{\partial x_{0}}+2x_{0}^{T}\frac{\partial^{2}h_{\upsilon,0}^{d,l}}{\partial^{2}x_{0}}\\
	-2\triangle\sum_{i=0}^{n-1}\frac{\partial h_{v,i+1}^{d,l}(Y)^{T}}{\partial x_{0}}BU^{-1}B^{T}\frac{\partial h_{v,i+1}^{d,l}}{\partial x_{0}}(Y)\\
	-2\triangle\sum_{i=0}^{n-1}h_{v,i+1}^{d,l}(Y)^{T}BU^{-1}B^{T}\frac{\partial^{2}h_{v,i+1}^{d,l}}{\partial^{2}x_{0}}(Y)+O_{p,n}(\Delta).
	\end{array}
	\]
	\[
	\begin{array}{l}
	\frac{\partial^{2}S_{n}^{l}(Y,v)}{\partial\theta\partial x_{0}}=\frac{\partial}{\partial x_{0}}\left(x_{0}^{T}\frac{\partial R_{\upsilon,0}^{d,l}}{\partial\theta}x_{0}+2x_{0}^{T}\frac{\partial h_{\upsilon,0}^{d,l}}{\partial\theta}(Y)\right)\\
	-2\triangle\sum_{i=0}^{n-1}\frac{\partial}{\partial x_{0}}\left(h_{\upsilon,i+1}^{d,l}(Y)^{T}BG(R_{\upsilon,i+1}^{d,l})B^{T}\frac{\partial h_{\upsilon,i+1}^{d,l}}{\partial\theta}(Y)\right)\\
	-\triangle\sum_{i=0}^{n-1}\frac{\partial}{\partial x_{0}}\left(h_{\upsilon,i+1}^{d,l}(Y)^{T}B\frac{\partial G}{\partial\theta}(R_{\upsilon,i+1}^{d,l})B^{T}h_{\upsilon,i+1}^{d,l}(Y)\right)\\
	=2\frac{\partial R_{\upsilon,0}^{d,l}}{\partial\theta}x_{0}+x_{0}^{T}\frac{\partial^{2}R_{\upsilon,0}^{d,l}}{\partial\theta\partial x_{0}}x_{0}+2\frac{\partial h_{\upsilon,0}^{d,l}}{\partial\theta}(Y)+2x_{0}^{T}\frac{\partial^{2}h_{\upsilon,0}^{d,l}}{\partial\theta\partial x_{0}}(Y)\\
	-2\triangle\sum_{i=0}^{n-1}\left(\frac{\partial h_{\upsilon,i+1}^{d,l}(Y)}{\partial x_{0}}\right)^{T}BU^{-1}B^{T}\frac{\partial h_{\upsilon,i+1}^{d,l}}{\partial\theta}(Y)\\
	-2\triangle\sum_{i=0}^{n-1}h_{\upsilon,i+1}^{d,l}(Y)^{T}BU^{-1}B^{T}\frac{\partial^{2}h_{\upsilon,i+1}^{d,l}}{\partial\theta\partial x_{0}}(Y)+O_{p,n}(\Delta).
	\end{array}
	\]
	By using Taylor's theorem we derive that it exists $\widetilde{v}$
	on the line segment between $\widehat{v}$ and $v^{*}$ such that{\small{}:
		\[
		\nabla_{v}S_{n}^{l}(Y;\widehat{v})-\nabla_{v}S_{n}^{l}(Y;v^{*})=-\nabla_{v}S_{n}^{l}(Y;v^{*})=\frac{\partial^{2}S_{n}^{l}(Y,\widetilde{v})}{\partial^{2}v}^{T}\left(\widehat{v}-v^{*}\right)
		\]
	}since first order optimality condition imposes $\nabla_{v}S_{n}^{l}(Y;\widehat{v})=0$.
	If $\theta\longmapsto A_{\theta}$ is $C^{2}$ on $\Theta$, from
	lemma \ref{lem:h_R_continuity_limits} we derive that  $v\longmapsto S^{l}(v)$
	is $C^{2}$ and we have:
	\[
	\begin{array}{l}
	\frac{\partial^{2}S^{l}(v)}{\partial^{2}\theta}\,=x_{0}^{T}\frac{\partial^{2}R_{v}^{l}(0)}{\partial^{2}\theta}x_{0}+2x_{0}^{T}\frac{\partial^{2}h_{v}^{l}(0)}{\partial^{2}\theta}-2\int_{0}^{T}\frac{\partial}{\partial\theta}\left(h_{v}^{l}(t)^{T}BU^{-1}B^{T}\frac{\partial h_{v}^{l}(t)}{\partial\theta}\right)dt\\
	=x_{0}^{T}\frac{\partial^{2}R_{v}^{l}(0)}{\partial^{2}\theta}x_{0}+2x_{0}^{T}\frac{\partial^{2}h_{v}^{l}(0)}{\partial^{2}\theta}
	\\-2\int_{0}^{T}\left(\frac{\partial h_{v}^{l}}{\partial\theta}(t)^{T}BU^{-1}B^{T}\frac{\partial h_{v}^{l}(t)}{\partial\theta}+h_{v}^{l}(t)^{T}BU^{-1}B^{T}\frac{\partial^{2}h_{v}^{l}(t)}{\partial^{2}\theta}\right)dt
	\end{array}
	\]
	which gives us the difference:
	\[
	\begin{array}{l}
	\frac{\partial^{2}S^{l}(\upsilon)}{\partial^{2}\theta}-\frac{\partial^{2}S_{n}^{l}(Y,\upsilon)}{\partial^{2}\theta}\\
	=x_{0}^{T}\left(\frac{\partial^{2}R_{\upsilon}^{l}(0)}{\partial^{2}\theta}-\frac{\partial^{2}R_{\upsilon,0}^{d,l}}{\partial^{2}\theta}\right)x_{0}+2x_{0}^{T}\left(\frac{\partial^{2}h_{\upsilon}^{l}(0)}{\partial^{2}\theta}-\frac{\partial^{2}h_{\upsilon,0}^{d,l}}{\partial^{2}\theta}(Y)\right)\\
	-2\left(\int_{0}^{T}\frac{\partial h_{\upsilon}^{l}}{\partial\theta}(t)^{T}BU^{-1}B^{T}\frac{\partial h_{\upsilon}^{l}(t)}{\partial\theta}dt-\triangle\sum_{i=0}^{n-1}\frac{\partial h_{\upsilon,i+1}^{d,l}(Y)^{T}}{\partial\theta}BU^{-1}B^{T}\frac{\partial h_{\upsilon,i+1}^{d,l}(Y)}{\partial\theta}\right)\\
	-2\left(\int_{0}^{T}h_{\upsilon}^{l}(t)^{T}BU^{-1}B^{T}\frac{\partial^{2}h_{\upsilon}^{l}(t)}{\partial^{2}\theta}dt-\triangle\sum_{i=0}^{n-1}h_{\upsilon,i+1}^{d,l}(Y)^{T}BU^{-1}\frac{\partial^{2}h_{\upsilon,i+1}^{d,l}(Y)}{\partial^{2}\theta}\right)+O_{p,n}(\Delta)\\
	=x_{0}^{T}\left(\frac{\partial^{2}R_{\upsilon}^{l}(0)}{\partial^{2}\theta}-\frac{\partial^{2}R_{\upsilon,0}^{d,l}}{\partial^{2}\theta}\right)x_{0}+2x_{0}^{T}\left(\frac{\partial^{2}h_{\upsilon}^{l}(0)}{\partial^{2}\theta}-\frac{\partial^{2}h_{\upsilon,0}^{d,l}}{\partial^{2}\theta}(Y)\right)\\
	-2\triangle\left(\sum_{i=0}^{n-1}\frac{\partial h_{\upsilon}^{l}(t_{i+1})^{T}}{\partial\theta}BU^{-1}B^{T}\frac{\partial h_{\upsilon}^{l}(t_{i+1})}{\partial\theta}-\frac{\partial h_{\upsilon,i+1}^{d}(Y)^{T}}{\partial\theta}BU^{-1}B^{T}\frac{\partial h_{\upsilon,i+1}^{dl}(Y)}{\partial\theta}\right)\\
	-2\triangle\left(\sum_{i=0}^{n-1}h_{\upsilon}^{l}(t_{i+1})^{T}BU^{-1}B^{T}\frac{\partial^{2}h_{\upsilon}^{l}(t_{i+1})}{\partial^{2}\theta}-h_{\upsilon,i+1}^{d,l}(Y)^{T}BU^{-1}B^{T}\frac{\partial^{2}h_{\upsilon,i+1}^{d,l}(Y)}{\partial^{2}\theta}\right)
	\\+O_{p,n}(\Delta)
	\end{array}
	\]
	again, from proposition \ref{lem:h_R_continuity_limits}, we know
	\[
	\begin{array}{l}
	\frac{\partial h_{\upsilon}^{l}(t_{i+1})}{\partial\theta}^{T}BU^{-1}B^{T}\frac{\partial h_{\upsilon}^{l}(t_{i+1})}{\partial\theta}-\frac{\partial h_{\upsilon,i+1}^{d,l}(Y)^{T}}{\partial\theta}BU^{-1}B^{T}\frac{\partial h_{\upsilon,i+1}^{d,l}(Y)}{\partial\theta}\\
	=(\frac{\partial h_{\upsilon,i+1}^{d,l}(Y)}{\partial\theta}+o_{p,n}(1))^{T}BU^{-1}B^{T}(\frac{\partial h_{\upsilon,i+1}^{d,l}(Y)}{\partial\theta}+o_{p,n}(1))-\frac{\partial h_{\upsilon,i+1}^{d,l}(Y)^{T}}{\partial\theta}BU^{-1}B^{T}\frac{\partial h_{\upsilon,i+1}^{d,l}(Y)}{\partial\theta}\\
	=o_{p,n}(1)
	\end{array}
	\]
	and:
	\[
	\begin{array}{l}
	h_{\upsilon}^{l}(t_{i+1})^{T}BU^{-1}B^{T}\frac{\partial^{2}h_{\upsilon}^{l}(t_{i+1})}{\partial^{2}\theta}-h_{\upsilon,i+1}^{d,l}(Y)^{T}BU^{-1}B^{T}\frac{\partial^{2}h_{\upsilon,i+1}^{d,l}(Y)}{\partial^{2}\theta}\\
	=(\frac{\partial^{2}h_{\upsilon,i+1}^{d,l}(Y)}{\partial^{2}\theta}+o_{p,n}(1))^{T}BU^{-1}B^{T}(\frac{\partial^{2}h_{\upsilon,i+1}^{d,l}(Y)}{\partial^{2}\theta}+o_{p,n}(1))-h_{\upsilon,i+1}^{d,l}(Y)^{T}BU^{-1}B^{T}\frac{\partial^{2}h_{\upsilon,i+1}^{d,l}(Y)}{\partial^{2}\theta}\\
	=o_{p,n}(1)
	\end{array}
	\]
	from this we can derive that $\frac{\partial^{2}S^{l}(\upsilon)}{\partial^{2}\theta}-\frac{\partial^{2}S_{n}^{l}(Y,\upsilon)}{\partial^{2}\theta}=o_{p,n}(1)$,
	in the same way we obtain $\frac{\partial^{2}S^{l}(\upsilon)}{\partial^{2}x_{0}}-\frac{\partial^{2}S_{n}^{l}(Y,\upsilon)}{\partial^{2}x_{0}}=o_{p,n}(1)$
	and $\frac{\partial^{2}S^{l}(\upsilon)}{\partial\theta\partial x_{0}}-\frac{\partial^{2}S_{n}^{l}(Y,\upsilon)}{\partial\theta\partial x_{0}}=o_{p,n}(1)$.
	Moreover, lemma \ref{lem:Continuous_S_Gradient_Hessian_convergence}
	gives us $\frac{\partial^{2}S^{l}(v^{*})}{\partial^{2}v}=\frac{\partial^{2}S^{\infty}(v^{*})}{\partial^{2}v}+o_{l}(1)$,
	so by consistency of $\widehat{v}$ (and thus of $\widetilde{v}$)
	and by using the continuous mapping theorem, we derive:
	\[
	\frac{\partial^{2}S_{n}^{l}(Y,\widetilde{v})}{\partial^{2}v}=\frac{\partial^{2}S^{\infty}(\upsilon^{*})}{\partial^{2}v}+o_{p,n}(1)+o_{l}(1)
	\]
	and thus conclude the proof.
\end{proof}

\subsection{$\widehat{\theta}^{T,CI}$, linear case}
\begin{thm}
	\label{thm:Asymptotic_Normality_prof_ci}Under conditions LC1-LC2-LC3-LC4-LC5-LC6-LC7,
	$\widehat{\theta}$ is asymptotically normal and $\widehat{\theta}-\theta^{*}=o_{p,n}(n^{-\frac{1}{2}})$.
\end{thm}
\begin{proof}
	Since $R_{\theta,0}^{d}=R_{\theta}(0)+o_{p,n}(1)$ and LC3 holds,
	it exists $n'\in\mathbb{N}$ such that LC3disc holds i.e. $R_{\theta,0}^{d}$
	is invertible for all $n\geq n'$. By merging the proposition \ref{prop:gradient_S_true_param_as_representation_prof_ci}
	and proposition \ref{prop:param_as_representation_prof_ci}, we obtain
	the following asymptotic representation between $\widehat{\theta}$
	and $\theta^{*}$: {\small{}
		\[
		\begin{array}{lll}
		(\frac{\partial^{2}S^{CI}(\theta^{*})}{\partial^{2}\theta}+o_{p,n}(1))\left(\widehat{\theta}-\theta^{*}\right) & = & \triangle^{2}\sum_{i=0}^{n-1}\left(\sum_{j=i+1}^{n}\epsilon_{j}^{T}\right)K_{i}\\
		& + & \left(H+o_{n}(1)\right)\left(\triangle\sum_{j=1}^{n}\epsilon_{j}\right)\\
		& + & \left(\triangle\sum_{j=1}^{n}\epsilon_{j}\right)^{T}\left(J^{l}+o_{n}(1)\right)+O_{p,n}(\Delta)
		\end{array}
		\]
	}with $K_{i}$, $H$ and $J$ defined in proposition \ref{prop:param_as_representation_prof_ci}.
	Let us decompose the first right hand side term:{\small{}
		\[
		\begin{array}{lll}
		\triangle^{2}\sum_{i=0}^{n-1}\left(\sum_{j=i+1}^{n}\epsilon_{j}^{T}\right)K_{i} & = & \triangle^{2}\sum_{i=0}^{n-1}\left(\sum_{j=1}^{n}\epsilon_{j}^{T}-\sum_{j=1}^{i}\epsilon_{j}^{T}\right)K_{i}\\
		& = & \triangle^{2}\sum_{i=0}^{n-1}\left(\sum_{j=1}^{n}\epsilon_{j}^{T}\right)K_{i}\\
		& - & \triangle^{2}\sum_{i=0}^{n-1}\left(\sum_{j=1}^{i}\epsilon_{j}^{T}\right)K_{i}.
		\end{array}
		\]
	}By definition of $K_{i}$, $\triangle\sum_{i=0}^{n-1}K_{i}$ converges
	to a limit $K$ as a Riemann sum and so $\triangle^{2}\sum_{j=1}^{n}\epsilon_{j}^{T}\sum_{i=0}^{n-1}K_{i}=\triangle\sum_{j=1}^{n}\epsilon_{j}^{T}(K+o_{n}(1))$.
	Similarly as in proposition \ref{prop:gradient_S_true_param_as_representation},
	we derive{\small{}
		\[
		\begin{array}{lll}
		\triangle^{2}\sum_{i=0}^{n-1}\left(\sum_{j=1}^{i}\epsilon_{j}^{T}\right)K_{i} & = & o_{p,n}(\sqrt{\Delta})\end{array}
		\]
	}Thus{\small{}:
		\[
		\begin{array}{lll}
		\triangle^{2}\sum_{i=0}^{n-1}\left(\sum_{j=i+1}^{n}\epsilon_{j}^{T}\right)K_{i} & = & \left(\triangle\sum_{j=1}^{n}\epsilon_{j}^{T}\right)(K+o_{n}(1))+o_{p,n}(\sqrt{\Delta}).\end{array}
		\]
	}We can now reformulate the asymptotic representation:{\small{}
		\[
		\begin{array}{lll}
		(\frac{\partial^{2}S^{CI}(\theta^{*})}{\partial^{2}\theta}+o_{p,n}(1))\left(\widehat{\theta}-\theta^{*}\right) & = & \left(\triangle\sum_{j=1}^{n}\epsilon_{j}^{T}\right)(K+o_{n}(1))\\
		& + & \left(H+o_{n}(1)\right)\left(\triangle\sum_{j=1}^{n}\epsilon_{j}^{T}\right)\\
		& + & \left(\triangle\sum_{j=1}^{n}\epsilon_{j}\right)^{T}\left(J+o_{n}(1)\right)+o_{p,n}(\sqrt{\Delta}).
		\end{array}
		\]
	}By using L7 which ensures $\frac{\partial^{2}S^{CI}(\theta^{*})}{\partial^{2}\theta}+o_{p,n}(1)$
	tends to a nonsingular matrix with probability 1, we can use the central
	limit theorem to conclude.
\end{proof}
\begin{prop}
	\label{prop:gradient_S_true_param_as_representation_prof_ci}Under
	conditions LC1-LC2-LC3disc-LC3b-LC4-LC5-LC6, we have
	\[
	\begin{array}{l}
	-\nabla_{\theta}S_{n}^{CI}(Y;\theta^{*})=\left(H+o_{n}(1)\right)\left(\triangle\sum_{j=1}^{n}\epsilon_{j}\right)\\
	+\left(\triangle\sum_{j=1}^{n}\epsilon_{j}\right)^{T}\left(J+o_{n}(1)\right)+\triangle^{2}\sum_{i=0}^{n-1}\left(\sum_{j=i+1}^{n}\epsilon_{j}^{T}\right)K_{i}+O_{p,n}(\Delta)
	\end{array}
	\]
	with $H=-\left(h^{*}(0)^{T}R^{*}(0)^{-1}\frac{\partial R^{*}(0)}{\partial\theta}R^{*}(0)^{-1}+2\left(\frac{\partial h^{*}(0)}{\partial\theta}\right)^{T}R^{*}(0)^{-1}\right)$,
	$J=CR^{*}(0)^{-1}\frac{\partial R^{*}(0)}{\partial\theta}R^{*}(0)^{-1}h^{*}(0)$
	and $K_{i}=2CBU^{-1}B^{T}\frac{\partial h_{i+1}^{*}}{\partial\theta}.$
\end{prop}
\begin{proof}
	We consider the decomposition:
	\[
	\begin{array}{lll}
	-\nabla_{\theta}S_{n}^{CI}(Y;\theta^{*}) & = & \nabla_{\theta}S_{n}^{CI}(Y^{d*};\theta^{*})-\nabla_{\theta}S_{n}^{CI}(Y;\theta^{*})+\nabla_{\theta}S^{CI}(\theta^{*})-\nabla_{\theta}S_{n}^{CI}(Y^{d*};\theta^{*})\end{array}
	\]
	since:{\small{}
		\[
		\begin{array}{l}
		\nabla_{\theta}S_{n}^{CI}(Y^{d*};\theta^{*})-\nabla_{\theta}S^{CI}(\theta^{*})\\
		=\left(h_{0}^{*d}\right)^{T}\left(R_{0}^{*d}\right)^{-1}\frac{\partial R_{0}^{*d}}{\partial\theta}\left(R_{0}^{*d}\right)^{-1}h_{0}^{*d}-h^{*}(0)R^{*}(0)^{-1}\frac{\partial(R^{*}(0))}{\partial\theta}R^{*}(0)^{-1}h^{*}(0)\\
		+2\left(\frac{\partial h_{0}^{*d}}{\partial\theta}\right)^{T}\left(R_{0}^{*d}\right)^{-1}h_{0}^{*d}-2\left(\frac{\partial h_{0}^{*d}}{\partial\theta}\right)^{T}\left(R_{0}^{*d}\right)^{-1}h_{0}^{*d}\\
		-2\sum_{i=0}^{n-1}\left(\int_{t_{i}}^{t_{i+1}}h^{*}(t)^{T}BU^{-1}B^{T}\frac{\partial h^{*}(t)}{\partial\theta}dt-\triangle\left(h_{i+1}^{*d}\right)^{T}BU^{-1}B^{T}\frac{\partial h_{i+1}^{*d}}{\partial\theta}\right)+O_{n}(\triangle)\\
		=-2\sum_{i=0}^{n-1}\left(\int_{t_{i}}^{t_{i+1}}h^{*}(t)^{T}BU^{-1}B^{T}\frac{\partial h^{*}(t)}{\partial\theta}dt-\triangle\left(h_{i+1}^{*d}\right)^{T}BU^{-1}B^{T}\frac{\partial h_{i+1}^{*d}}{\partial\theta}\right)+O_{n}(\triangle)
		\end{array}
		\]
	}we derive as in lemma \ref{lem:deterministic_gradient_behavior}
	that $\nabla_{\theta}S^{CI}(\theta^{*})-\nabla_{\theta}S_{n}^{CI}(Y^{d*};\theta^{*})=O_{p,n}(\triangle)$,
	so the previous asymptotic decomposition becomes $
	-\nabla_{\theta}S_{n}^{CI}(Y;\theta^{*})  =  \nabla_{\theta}S_{n}^{CI}(Y^{d*};\theta^{*})-\nabla_{\theta}S_{n}^{CI}(Y;\theta^{*})+O_{p,n}(\triangle).$ We already derive in proposition \ref{prop:gradient_S_true_param_as_representation}
	that $h_{i}^{*d}-h_{i}^{d}=-\triangle C^{T}\sum_{j=i}^{n}\epsilon_{j}+O_{p,n}(\triangle)$,
	$R_{i}^{*d}-R_{i}^{d}=O_{p,n}(\triangle)$, $\frac{\partial R_{i}^{*d}}{\partial\theta}-\frac{\partial R_{i}^{d}}{\partial\theta}=O_{p,n}(\Delta)$
	and $\frac{\partial h_{i}^{*d}}{\partial\theta}-\frac{\partial h_{i}^{d}}{\partial\theta}=O_{p,n}(\Delta).$
	By using these approximations, we obtain for $\nabla_{\theta}S_{n}^{CI}(Y^{d*};\theta^{*})-\nabla_{\theta}S_{n}^{CI}(Y;\theta^{*})$:{\small{}
		\[
		\begin{array}{l}
		\nabla_{\theta}S_{n}^{CI}(Y^{d*};\theta^{*})-\nabla_{\theta}S_{n}^{CI}(Y;\theta^{*})\\
		=\left(h_{0}^{*d}\right)^{T}\left(R_{0}^{*d}\right)^{-1}\frac{\partial R_{\theta,0}^{*d}}{\partial\theta}\left(R_{0}^{*d}\right)^{-1}h_{0}^{*d}-\left(h_{0}^{d}\right)^{T}\left(R_{0}^{d}\right)^{-1}\frac{\partial R_{\theta,0}^{d}}{\partial\theta}\left(R_{0}^{d}\right)^{-1}h_{0}^{d}\\
		+2\left(\frac{\partial h_{0}^{d}}{\partial\theta}\right)^{T}\left(R_{0}^{d}\right)^{-1}h_{0}^{d}-2\left(\frac{\partial h_{0}^{*d}}{\partial\theta}\right)^{T}\left(R_{0}^{*d}\right)^{-1}h_{0}^{*d}\\
		-2\triangle\sum_{i=0}^{n-1}\left(\left(h_{i+1}^{*d}\right){}^{T}BU^{-1}B^{T}\frac{\partial h_{i+1}^{*d}}{\partial\theta}-\left(h_{i+1}^{d}\right){}^{T}BU^{-1}B^{T}\frac{\partial h_{i+1}^{d}}{\partial\theta}\right)+O_{p,n}(\triangle)\\
		=\left(h_{0}^{*d}\right)^{T}\left(R_{0}^{*d}\right)^{-1}\frac{\partial R_{0}^{*d}}{\partial\theta}\left(R_{0}^{*d}\right)^{-1}\left(h_{0}^{*d}-h_{0}^{d}\right)+\left(h_{0}^{*d}\right)^{T}\left(R_{0}^{*d}\right)^{-1}\frac{\partial R_{0}^{*d}}{\partial\theta}\left(\left(R_{0}^{*d}\right)^{-1}-\left(R_{0}^{d}\right)^{-1}\right)h_{0}^{d}\\
		+\left(h_{0}^{*d}\right)^{T}\left(R_{0}^{*d}\right)^{-1}\left(\frac{\partial R_{0}^{*d}}{\partial\theta}-\frac{\partial R_{0}^{d}}{\partial\theta}\right)\left(R_{0}^{d}\right)^{-1}h_{0}^{d}+\left(h_{0}^{*d}-h_{0}^{d}\right)^{T}\left(R_{0}^{d}\right)^{-1}\frac{\partial R_{0}^{d}}{\partial\theta}\left(R_{0}^{d}\right)^{-1}h_{0}^{d}\\
		+\left(h_{0}^{*d}\right)^{T}\left(\left(R_{0}^{*d}\right)^{-1}-\left(R_{0}^{d}\right)^{-1}\right)\frac{\partial R_{0}^{d}}{\partial\theta}\left(R_{0}^{d}\right)^{-1}h_{0}^{d}+2\left(\frac{\partial h_{0}^{d}}{\partial\theta}-\frac{\partial h_{0}^{*d}}{\partial\theta}\right)^{T}\left(R_{0}^{d}\right)^{-1}h_{0}^{d}\\
		+2\left(\frac{\partial h{}_{0}^{*d}}{\partial\theta}\right)^{T}\left(\left(R_{0}^{d}\right)^{-1}-\left(R_{0}^{*d}\right)^{-1}\right)h_{0}^{d}+2\left(\frac{\partial h_{0}^{*d}}{\partial\theta}\right)^{T}\left(R_{0}^{*d}\right)^{-1}\left(h_{0}^{d}-h_{0}^{*d}\right)\\
		-2\triangle\sum_{i=0}^{n-1}\left(h_{i+1}^{*d}-h_{i+1}^{d}\right){}^{T}BU^{-1}B^{T}\frac{\partial h_{i+1}^{*d}}{\partial\theta}\\
		-2\triangle\sum_{i=0}^{n-1}\left(h_{i+1}^{d}\right){}^{T}BU^{-1}B^{T}\left(\frac{\partial h_{i+1}^{*d}}{\partial\theta}-\frac{\partial h_{i+1}^{d}}{\partial\theta}\right)+O_{p,n}(\triangle)
		\end{array}
		\]
		\[
		\begin{array}{l}
		=\left(h_{0}^{*d}\right)^{T}\left(R_{0}^{*d}\right)^{-1}\frac{\partial R_{0}^{*d}}{\partial\theta}\left(R_{0}^{*d}\right)^{-1}\left(h_{0}^{*d}-h_{0}^{d}\right)+\left(h_{0}^{*d}-h_{0}^{d}\right)^{T}\left(R_{0}^{d}\right)^{-1}\frac{\partial R_{0}^{d}}{\partial\theta}\left(R_{0}^{d}\right)^{-1}h_{0}^{d}\\
		+2\left(\frac{\partial h_{0}^{*d}}{\partial\theta}\right)^{T}\left(R_{0}^{*d}\right)^{-1}\left(h_{0}^{d}-h_{0}^{*d}\right)-2\triangle\sum_{i=0}^{n-1}\left(h_{i+1}^{*d}-h_{i+1}^{d}\right){}^{T}BU^{-1}B^{T}\frac{\partial h_{i+1}^{*d}}{\partial\theta}+O_{p,n}(\triangle)\\
		=\left(\left(h_{0}^{*d}\right)^{T}\left(R_{0}^{*d}\right)^{-1}\frac{\partial R_{0}^{*d}}{\partial\theta}\left(R_{0}^{*d}\right)^{-1}+2\left(\frac{\partial h_{0}^{*d}}{\partial\theta}\right)^{T}\left(R_{0}^{*d}\right)^{-1}\right)\left(h_{0}^{*d}-h_{0}^{d}\right)\\
		+\left(h_{0}^{*d}-h_{0}^{d}\right)^{T}\left(R_{0}^{d}\right)^{-1}\frac{\partial R_{0}^{d}}{\partial\theta}\left(R_{0}^{d}\right)^{-1}h_{0}^{d}-2\triangle\sum_{i=0}^{n-1}\left(h_{i+1}^{*d}-h_{i+1}^{d}\right){}^{T}BU^{-1}B^{T}\frac{\partial h_{i+1}^{*d}}{\partial\theta}+O_{p,n}(\triangle)\\
		=\left(\left(h_{0}^{*d}\right)^{T}\left(R_{0}^{*d}\right)^{-1}\frac{\partial R_{0}^{*d}}{\partial\theta}\left(R_{0}^{*d}\right)^{-1}+2\left(\frac{\partial h_{0}^{*d}}{\partial\theta}\right)^{T}\left(R_{0}^{*d}\right)^{-1}\right)\left(h_{0}^{*d}-h_{0}^{d}\right)\\
		+\left(h_{0}^{*d}-h_{0}^{d}\right)^{T}\left(R_{0}^{d}\right)^{-1}\frac{\partial R_{0}^{d}}{\partial\theta}\left(R_{0}^{d}\right)^{-1}h_{0}^{d}+2\triangle^{2}\sum_{i=0}^{n-1}\left(\sum_{j=i+1}^{n}\epsilon_{j}^{T}\right)BU^{-1}B^{T}\frac{\partial h_{i+1}^{*d}}{\partial\theta}+O_{p,n}(\triangle).
		\end{array}
		\]
	}The first term is equal to {\small{}
		\[
		\begin{array}{l}
		\left(\left(h_{0}^{*d}\right)^{T}\left(R_{0}^{*d}\right)^{-1}\frac{\partial R_{0}^{*d}}{\partial\theta}\left(R_{0}^{*d}\right)^{-1}+2\left(\frac{\partial h_{0}^{*d}}{\partial\theta}\right)^{T}\left(R_{0}^{*d}\right)^{-1}\right)\left(h_{0}^{*d}-h_{0}^{d}\right)\\
		=-\left(\left(h_{0}^{*d}\right)^{T}\left(R_{0}^{*d}\right)^{-1}\frac{\partial R_{0}^{*d}}{\partial\theta}\left(R_{0}^{*d}\right)^{-1}+2\left(\frac{\partial h_{0}^{*d}}{\partial\theta}(Y)\right)^{T}\left(R_{0}^{*d}\right)^{-1}\right)C^{T}\left(\triangle\sum_{j=1}^{n}\epsilon_{j}\right)+O_{p,n}(\triangle)\\
		=\left(H+o_{n}(1)\right)\left(\triangle\sum_{j=1}^{n}\epsilon_{j}\right)+O_{p,n}(\triangle)
		\end{array}
		\]
	}thanks to lemma \ref{lem:h_R_continuity_limits}. For the second
	term, we have{\small{}:
		\[
		\begin{array}{l}
		\left(h_{0}^{*d}-h_{0}^{d}\right)^{T}\left(R_{0}^{d}\right)^{-1}\frac{\partial R_{0}^{d}}{\partial\theta}\left(R_{0}^{d}\right)^{-1}h_{0}^{d}\\
		=\left(\triangle\sum_{j=1}^{n}\epsilon_{j}\right)^{T}C\left(R_{0}^{d}\right)^{-1}\frac{\partial R_{0}^{d}}{\partial\theta}\left(R_{0}^{d}\right)^{-1}h_{0}^{d}+O_{p,n}(\triangle)\\
		=\left(\triangle\sum_{j=1}^{n}\epsilon_{j}\right)^{T}\left(J+o_{n}(1)\right)+O_{p,n}(\triangle)
		\end{array}
		\]
	}from which we derive the expression{\small{}:
		\[
		\begin{array}{l}
		\nabla_{\theta}S_{n}^{CI}(Y^{d*};\theta^{*})-\nabla_{\theta}S_{n}^{CI}(Y;\theta^{*})\\
		=\left(H+o_{n}(1)\right)\left(\triangle\sum_{j=1}^{n}\epsilon_{j}\right)+\left(\triangle\sum_{j=1}^{n}\epsilon_{j}\right)^{T}\left(J+o_{n}(1)\right)\\
		+2\triangle^{2}\sum_{i=0}^{n-1}\left(\sum_{j=i+1}^{n}\epsilon_{j}^{T}\right)BU^{-1}B\frac{\partial h_{i+1}^{*d}}{\partial\theta}+O_{p,n}(\triangle)
		\end{array}
		\]
	}and we can conclude the proof.
\end{proof}
\begin{prop}
	\label{prop:param_as_representation_prof_ci}Under conditions LC1-LC2-LC3disc-LC3-LC4-LC5-LC6,
	we have $-\nabla S_{n}^{CI}(Y;\theta^{*})=(\frac{\partial^{2}S^{CI}(\theta^{*})}{\partial^{2}\theta}+o_{p,n}(1))\left(\widehat{\theta}-\theta^{*}\right).$
\end{prop}
\begin{proof}
	For notation clarity we treat the case $d=1$, if $\theta\longmapsto A_{\theta}$
	is $C^{1}$ on $\Theta$, then $\theta\longmapsto S_{n}^{CI}(Y;\theta)$
	is $C^{1}$ as well an{\small{}d
		\[
		\begin{array}{l}
		\nabla_{\theta}S_{n}^{CI}(Y;\theta)\,=\left(h_{\theta,0}^{d}\right)^{T}\left(R_{\theta,0}^{d}\right)^{-1}\frac{\partial R_{\theta,0}^{d}}{\partial\theta}\left(R_{\theta,0}^{d}\right)^{-1}h_{\theta,0}^{d}-2\left(\frac{\partial h_{\theta,0}^{d}}{\partial\theta}(Y)\right)^{T}\left(R_{\theta,0}^{d}\right)^{-1}h_{\theta,0}^{d}\\
		-2\triangle\sum_{i=0}^{n-1}h_{\theta,i+1}^{d}(Y)^{T}BG(R_{\theta,i+1}^{d})B^{T}\frac{\partial h_{\theta,i+1}^{d}}{\partial\theta}(Y)
		\\-\triangle\sum_{i=0}^{n-1}h_{\theta,i+1}^{d}(Y)^{T}B\frac{\partial G}{\partial\theta}(R_{\theta,i+1}^{d})B^{T}h_{\theta,i+1}^{d}(Y).
		\end{array}
		\]
	}If now $\theta\longmapsto A_{\theta}$ is $C^{2}$ on $\Theta$,then
	$\theta\longmapsto S_{n}^{CI}(Y;\theta)$ is also $C^{2}$ and equal
	t{\small{}o
		\[
		\begin{array}{l}
		\frac{\partial^{2}S_{n}^{CI}(Y,\theta)}{\partial^{2}\theta}\,=\frac{\partial}{\partial\theta}\left(\left(h_{\theta,0}^{d}\right)^{T}\left(R_{\theta,0}^{d}\right)^{-1}\frac{\partial R_{\theta,0}^{d}}{\partial\theta}\left(R_{\theta,0}^{d}\right)^{-1}h_{\theta,0}^{d}\right)-2\frac{\partial}{\partial\theta}\left(\left(\frac{\partial h_{\theta,0}^{d}}{\partial\theta}(Y)\right)^{T}\left(R_{\theta,0}^{d}\right)^{-1}h_{\theta,0}^{d}\right)\\
		-2\triangle\sum_{i=0}^{n-1}\frac{\partial}{\partial\theta}\left(h_{\theta,i+1}^{d}(Y)^{T}BG(R_{\theta,i+1}^{d})B^{T}\frac{\partial h_{\theta,i+1}^{d}}{\partial\theta}(Y)\right)\\
		-\triangle\sum_{i=0}^{n-1}\frac{\partial}{\partial\theta}\left(h_{\theta,i+1}^{d}(Y)^{T}B\frac{\partial G}{\partial\theta}(R_{\theta,i+1}^{d})B^{T}h_{\theta,i+1}^{d}(Y)\right)\\
		=\frac{\partial}{\partial\theta}\left(\left(h_{\theta,0}^{d}\right)^{T}\left(R_{\theta,0}^{d}\right)^{-1}\frac{\partial R_{\theta,0}^{d}}{\partial\theta}\left(R_{\theta,0}^{d}\right)^{-1}h_{\theta,0}^{d}\right)-2\frac{\partial}{\partial\theta}\left(\left(\frac{\partial h_{\theta,0}^{d}}{\partial\theta}(Y)\right)^{T}\left(R_{\theta,0}^{d}\right)^{-1}h_{\theta,0}^{d}\right)\\
		-2\triangle\sum_{i=0}^{n-1}\frac{\partial h_{\theta,i+1}^{d}(Y)^{T}}{\partial\theta}BG(R_{\theta,i+1}^{d})B^{T}\frac{\partial h_{\theta,i+1}^{d}}{\partial\theta}(Y)\\
		-2\triangle\sum_{i=0}^{n-1}h_{\theta,i+1}^{d}(Y)^{T}B\frac{\partial G}{\partial\theta}(R_{\theta,i+1}^{d})B^{T}\frac{\partial h_{\theta,i+1}^{d}}{\partial\theta}(Y)\\
		-2\triangle\sum_{i=0}^{n-1}h_{\theta,i+1}^{d}(Y)^{T}BG(R_{\theta,i+1}^{d})B^{T}\frac{\partial^{2}h_{\theta,i+1}^{d}}{\partial^{2}\theta}(Y)\\
		-\triangle\sum_{i=0}^{n-1}\frac{\partial h_{\theta,i+1}^{d}(Y)^{T}}{\partial\theta}B\frac{\partial G}{\partial\theta}(R_{\theta,i+1}^{d})B^{T}h_{\theta,i+1}^{d}(Y)\\
		-\triangle\sum_{i=0}^{n-1}h_{\theta,i+1}^{d}(Y)^{T}B\frac{\partial^{2}G}{\partial^{2}\theta}(R_{\theta,i+1}^{d})B^{T}h_{\theta,i+1}^{d}(Y)\\
		-\triangle\sum_{i=0}^{n-1}h_{\theta,i+1}^{d}(Y)^{T}B\frac{\partial G}{\partial\theta}(R_{\theta,i+1}^{d})B^{T}\frac{\partial h_{\theta,i+1}^{d}(Y)}{\partial\theta}.\\
		=\frac{\partial}{\partial\theta}\left(\left(h_{\theta,0}^{d}\right)^{T}\left(R_{\theta,0}^{d}\right)^{-1}\frac{\partial R_{\theta,0}^{d}}{\partial\theta}\left(R_{\theta,0}^{d}\right)^{-1}h_{\theta,0}^{d}\right)-2\frac{\partial}{\partial\theta}\left(\left(\frac{\partial h_{\theta,0}^{d}}{\partial\theta}(Y)\right)^{T}\left(R_{\theta,0}^{d}\right)^{-1}h_{\theta,0}^{d}\right)\\
		-2\triangle\sum_{i=0}^{n-1}\frac{\partial h_{\theta,i+1}^{d}(Y)^{T}}{\partial\theta}BU^{-1}B^{T}\frac{\partial h_{\theta,i+1}^{d}}{\partial\theta}(Y)\\
		-2\triangle\sum_{i=0}^{n-1}h_{\theta,i+1}^{d}(Y)^{T}BU^{-1}B^{T}\frac{\partial^{2}h_{\theta,i+1}^{d}}{\partial^{2}\theta}(Y)+O_{p,n}(\Delta).
		\end{array}
		\]
	}By using Taylor's theorem we derive that it exists $\widetilde{\theta}$
	on the line segment between $\widehat{\theta}$ and $\theta^{*}$
	such that{\small{}:
		\[
		\nabla_{\theta}S_{n}^{CI}(Y;\widehat{\theta})-\nabla_{\theta}S_{n}^{CI}(Y;\theta^{*})=-\nabla_{\theta}S_{n}^{CI}(Y;\theta^{*})=\frac{\partial^{2}S_{n}^{CI}(Y,\widetilde{\theta})}{\partial^{2}\theta}^{T}\left(\widehat{\theta}-\theta^{*}\right)
		\]
	}since first order optimality condition imposes $\nabla_{\theta}S_{n}^{CI}(Y;\widehat{\theta})=0$.
	If $\theta\longmapsto A_{\theta}$ is $C^{2}$ on $\Theta$, from
	lemma \ref{lem:h_R_continuity_limits} we derive that  $\theta\longmapsto S^{CI}(\theta)$
	is $C^{2}$ an{\small{}d
		\[
		\begin{array}{l}
		\frac{\partial^{2}S^{CI}(\theta)}{\partial^{2}\theta}\,=\frac{\partial}{\partial\theta}\left(h_{\theta}(0)R_{\theta}(0)^{-1}\frac{\partial(R_{\theta}(0))}{\partial\theta}R_{\theta}(0)^{-1}h_{\theta}(0)\right)-2\frac{\partial}{\partial\theta}\left(\frac{\partial h_{\theta}(0)^{T}}{\partial\theta}R_{\theta}(0)^{-1}h_{\theta}(0)\right)\\
		-2\int_{0}^{T}\left(\frac{\partial h_{\theta}}{\partial\theta}(t)^{T}BU^{-1}B^{T}\frac{\partial h_{\theta}(t)}{\partial\theta}+h_{\theta}(t)^{T}BU^{-1}B^{T}\frac{\partial^{2}h_{\theta}(t)}{\partial^{2}\theta}\right)dt
		\end{array}
		\]
	}which gives us the difference{\small{}:
		\[
		\begin{array}{l}
		\frac{\partial^{2}S^{CI}(\theta)}{\partial^{2}\theta}-\frac{\partial^{2}S_{n}^{CI}(Y,\theta)}{\partial^{2}\theta}\\
		=\frac{\partial}{\partial\theta}\left(h_{\theta}(0)R_{\theta}(0)^{-1}\frac{\partial(R_{\theta}(0))}{\partial\theta}R_{\theta}(0)^{-1}h_{\theta}(0)\right)-\frac{\partial}{\partial\theta}\left(\left(h_{\theta,0}^{d}\right)^{T}\left(R_{\theta,0}^{d}\right)^{-1}\frac{\partial R_{\theta,0}^{d}}{\partial\theta}\left(R_{\theta,0}^{d}\right)^{-1}h_{\theta,0}^{d}\right)\\
		+2\frac{\partial}{\partial\theta}\left(\left(\frac{\partial h_{\theta,0}^{d}}{\partial\theta}(Y)\right)^{T}\left(R_{\theta,0}^{d}\right)^{-1}h_{\theta,0}^{d}\right)-2\frac{\partial}{\partial\theta}\left(\frac{\partial h_{\theta}(0)^{T}}{\partial\theta}R_{\theta}(0)^{-1}h_{\theta}(0)\right)\\
		-2\left(\int_{0}^{T}\frac{\partial h_{\theta}}{\partial\theta}(t)^{T}BU^{-1}B^{T}\frac{\partial h_{\theta}(t)}{\partial\theta}dt-\triangle\sum_{i=0}^{n-1}\frac{\partial h_{\theta,i+1}^{d}(Y)^{T}}{\partial\theta}BU^{-1}B^{T}\frac{\partial h_{\theta,i+1}^{d}(Y)}{\partial\theta}\right)\\
		-2\left(\int_{0}^{T}h_{\theta}(t)^{T}BU^{-1}B^{T}\frac{\partial^{2}h_{\theta}(t)}{\partial^{2}\theta}dt-\triangle\sum_{i=0}^{n-1}h_{\theta,i+1}^{d}(Y)^{T}BU^{-1}\frac{\partial^{2}h_{\theta,i+1}^{d}(Y)}{\partial^{2}\theta}\right)+O_{p,n}(\Delta)\\
		=-2\triangle\left(\sum_{i=0}^{n-1}\frac{\partial h_{\theta}(t_{i+1})^{T}}{\partial\theta}BU^{-1}B^{T}\frac{\partial h_{\theta}(t_{i+1})}{\partial\theta}-\frac{\partial h_{\theta,i+1}^{d}(Y)^{T}}{\partial\theta}BU^{-1}B^{T}\frac{\partial h_{\theta,i+1}^{d}(Y)}{\partial\theta}\right)\\
		-2\triangle\left(\sum_{i=0}^{n-1}h_{\theta}(t_{i+1})^{T}BU^{-1}B^{T}\frac{\partial^{2}h_{\theta}(t_{i+1})}{\partial^{2}\theta}-h_{\theta,i+1}^{d}(Y)^{T}BU^{-1}B^{T}\frac{\partial^{2}h_{\theta,i+1}^{d}(Y)}{\partial^{2}\theta}\right)+o_{p,n}(1)
		\end{array}
		\]
	}thanks to lemma \ref{lem:h_R_continuity_limits}. From this we derive
	that $\frac{\partial^{2}S^{CI}(\theta)}{\partial^{2}\theta}-\frac{\partial^{2}S_{n}^{CI}(Y,\theta)}{\partial^{2}\theta}=o_{p,n}(1)$
	similarly as in proposition \ref{prop:param_as_representation}. Since
	$\widetilde{\theta}$ is consistent, we can use the continuous mapping
	theorem to conclude.
\end{proof}

\section{Useful lemma}

\subsection{Discrete Gronwall Lemma}

Here, we just recall the discrete Gronwall lemma, in a form particularly
convenient for us. 
\begin{lem}
	\label{lem:disc_Gronwall_lemma}Let introduce the positive sequences
	$f_{n},\,\lambda_{n},\,\epsilon_{n}$ linked by the recursive inequality
	$f_{n}\leq(1+\lambda_{n-1})f_{n-1}+\epsilon_{n-1}$, then we have:
	\[
	f_{n}\leq e^{\sum_{i=1}^{n-1}\lambda_{i}}f_{0}+\sum_{i=1}^{n-1}e^{\sum_{j=i+1}^{n-1}\lambda_{j}}\epsilon_{i}
	\]
	in particular for $\lambda=\lambda_{1}=\ldots=\lambda_{n}$ we have
	$f_{n}\leq e^{n\lambda}f_{0}+\sum_{i=1}^{n-1}e^{(n-i)\lambda}\epsilon_{i}$.
\end{lem}
\begin{proof}
	Let us prove it recursively. For $n=1$, we have $f_{1}\leq(1+\lambda_{0})f_{0}+\epsilon_{0}$,
	by remembering that $e^{x}\geq1+x$, for all $x\geq0$, the initialization
	is easy to make. Now let us assume the property holds for $n$, we
	have: 
	\[
	\begin{array}{lll}
	f_{n+1} & \leq & (1+\lambda_{n})f_{n}+\epsilon_{n}\\
	& \leq & e^{\lambda_{n}}f_{n}+\epsilon_{n}\\
	& \leq & e^{\lambda_{n}}\left(e^{\sum_{i=1}^{n-1}\lambda_{i}}f_{0}+\sum_{i=1}^{n-1}e^{\sum_{j=i+1}^{n-1}\lambda_{j}}\epsilon_{i}\right)+\epsilon_{n}\\
	& \leq & e^{\sum_{i=1}^{n}\lambda_{i}}f_{0}+\sum_{i=1}^{n-1}e^{\lambda_{n}+\sum_{j=i+1}^{n-1}\lambda_{j}}\epsilon_{i}+\epsilon_{n}\\
	& = & e^{\sum_{i=1}^{n}\lambda_{i}}f_{0}+\sum_{i=1}^{n-1}e^{\sum_{j=i+1}^{n}\lambda_{j}}\epsilon_{i}
	\end{array}
	\]
	Hence the conclusion.
\end{proof}

\subsection{Consistency}
\begin{lem}
	\label{lem:h_R_continuity_limits}Under conditions C1 to C3 for each
	$l\in\mathbb{N}$, $\upsilon\longmapsto\left(\overline{X_{\upsilon}^{l}},R_{\upsilon}^{l},h_{\upsilon}^{l}\right)$
	is continuous on $\varUpsilon$ and $\sup_{\upsilon\in\varUpsilon}\left\Vert R_{\upsilon,i}^{d,l}-R_{\upsilon}^{l}(t_{i})\right\Vert _{2}=o_{p,n}(1)$,
	$\sup_{\upsilon\in\varUpsilon}\left\Vert h_{\upsilon,i}^{d,l}-h_{\upsilon}^{l}(t_{i})\right\Vert _{2}=o_{p,n}(1)$,
	$\sup_{\upsilon\in\varUpsilon}\left\Vert \overline{X_{\upsilon}^{l}}(t_{i})-\overline{X_{\upsilon}^{d,l}}(t_{i})\right\Vert _{2}=o_{p,n}(1).$ 
	
	Under conditions C1 to C3 and C6 for each $l\in\mathbb{N}$, $\upsilon\longmapsto\left(\overline{X_{\upsilon}^{l}},R_{\upsilon}^{l},h_{\upsilon}^{l}\right)$is
	$C^{1}$ on $\varUpsilon$ and $\sup_{\upsilon\in\varUpsilon}\left\Vert \frac{\partial R_{\upsilon}^{l}(t_{i})}{\partial\upsilon}-\frac{\partial R_{\upsilon,i}^{d,l}}{\partial\upsilon}\right\Vert _{2}=o_{p,n}(1)$,
	$\sup_{\upsilon\in\varUpsilon}\left\Vert \frac{\partial h_{\upsilon}^{l}(t_{i})}{\partial\upsilon}-\frac{\partial h_{\upsilon,i}^{d,l}(Y)}{\partial\upsilon}\right\Vert _{2}=o_{p,n}(1)$,
	$\sup_{\upsilon\in\varUpsilon}\left\Vert \frac{\partial\overline{X_{\upsilon}^{l}}(t_{i})}{\partial\upsilon}-\frac{\partial\overline{X_{\upsilon}^{d,l}}(t_{i})}{\partial\upsilon}\right\Vert _{2}=o_{p,n}(1)$.
	
	Under conditions C1 to C3 and C6-C7 , for each $l\in\mathbb{N}$,$\upsilon\longmapsto\left(\overline{X_{\upsilon}^{l}},R_{\upsilon}^{l},h_{\upsilon}^{l}\right)$
	is $C^{2}$ on $\varUpsilon$ and $\sup_{\upsilon\in\varUpsilon}\left\Vert \frac{\partial^{2}R_{\upsilon}^{l}(t_{i})}{\partial^{2}\upsilon}-\frac{\partial^{2}R_{\upsilon,i}^{d,l}}{\partial^{2}\upsilon}\right\Vert _{2}=o_{p,n}(1)$,
	$\sup_{\upsilon\in\varUpsilon}\left\Vert \frac{\partial^{2}h_{\upsilon}^{l}(t_{i})}{\partial^{2}\upsilon}-\frac{\partial^{2}h_{\upsilon,i}^{d,l}(Y)}{\partial^{2}\upsilon}\right\Vert _{2}=o_{p,n}(1)$
	and $\sup_{\upsilon\in\varUpsilon}\left\Vert \frac{\partial^{2}\overline{X_{\upsilon}^{l}}(t_{i})}{\partial^{2}\upsilon}-\frac{\partial^{2}\overline{X_{\upsilon}^{d,l}}(t_{i})}{\partial^{2}\upsilon}\right\Vert _{2}=o_{p,n}(1)$. 
\end{lem}
\begin{proof}
	By integrating equation (\ref{eq:continous_accurate_Riccati_equation}),
	we obtain{\small{}:
		\[
		\begin{array}{l}
		R_{\upsilon}^{l}(t_{i})  =  R_{\upsilon}^{l}(t_{i+1})\\
		+\int_{t_{i}}^{t_{i+1}}\left(C^{T}C+A_{\theta}(\overline{X_{\upsilon}^{l-1}}(t),t)^{T}R_{\upsilon}^{l}(t)+R_{\upsilon}^{l}(t)A_{\theta}(\overline{X_{\upsilon}^{l-1}}(t),t)-R_{\upsilon}^{l}(t)BU^{-1}B^{T}R_{\upsilon}^{l}(t)\right)dt\\
		 = R_{\upsilon}^{l}(t_{i+1})+\triangle C^{T}C+\triangle A_{\theta}(\overline{X_{\upsilon}^{l-1}}(t_{i+1}),t_{i+1})^{T}R_{\upsilon}^{l}(t_{i+1})+\triangle R_{\upsilon}^{l}(t_{i+1})A_{\theta}(\overline{X_{\upsilon}^{l-1}}(t_{i+1}),t_{i+1})\\
		 - \triangle R_{\upsilon}^{l}(t_{i+1})BU^{-1}B^{T}R_{\upsilon}^{l}(t_{i+1})+O_{n}(\triangle^{2})\\
		 =  R_{\upsilon}^{l}(t_{i+1})+\triangle C^{T}C+\triangle A_{\theta}(\overline{X_{\upsilon}^{l-1}}(t_{i}),t_{i})^{T}R_{\upsilon}^{l}(t_{i+1})+\triangle R_{\upsilon}^{l}(t_{i+1})A_{\theta}(\overline{X_{\upsilon}^{l-1}}(t_{i}),t_{i})\\
		 -  \triangle R_{\upsilon}^{l}(t_{i+1})BU^{-1}B^{T}R_{\upsilon}^{l}(t_{i+1})+O_{n}(\triangle^{2})\\
		h_{\theta}^{l}(t_{i})  =  h_{\upsilon}^{l}(t_{i+1})+\int_{t_{i}}^{t_{i+1}}\left(-C^{T}Y^{*}(t)+A_{\theta}(\overline{X_{\upsilon}^{l-1}}(t),t)^{T}h_{\upsilon}^{l}(t)-R_{\upsilon}^{l}(t)BU^{-1}B^{T}h_{\upsilon}^{l}(t)\right)dt\\
		=  h_{\upsilon}^{l}(t_{i+1})-\triangle C^{T}Y^{*}(t_{i})+\triangle A_{\theta}(\overline{X_{\upsilon}^{l-1}}(t_{i}),t_{i})^{T}h_{\upsilon}^{l}(t_{i+1})-R_{\upsilon}^{l}(t_{i+1})BU^{-1}B^{T}h_{\upsilon}^{l}(t_{i+1})+O_{n}(\triangle^{2}).
		\end{array}
			\]
			\[
			\begin{array}{l}
		\overline{X_{\upsilon}^{l}}(t_{i+1}) =  \overline{X_{\upsilon}^{l}}(t_{i})+\int_{t_{i}}^{t_{i+1}}\left(A_{\theta}(\overline{X_{\upsilon}^{l-1}}(t),t)\overline{X_{\upsilon}^{l}}(t)-BU^{-1}B^{T}\left(R_{\upsilon}^{l}(t)\overline{X_{\upsilon}^{l}}(t)+h_{\upsilon}^{l}(t)\right)\right)dt\\
		 =  \overline{X_{\upsilon}^{l}}(t_{i})+\triangle A_{\theta}(\overline{X_{\upsilon}^{l-1}}(t_{i}),t_{i})\overline{X_{\upsilon}^{l}}(t_{i})-\triangle BU^{-1}B^{T}\left(R_{\upsilon}^{l}(t_{i+1})\overline{X_{\upsilon}^{l}}(t_{i})+h_{\upsilon}^{l}(t_{i+1})\right)+O_{n}(\Delta^{2}).
		\end{array}
		\]
	}By using equation (\ref{eq:discrete_accurate_Riccati_equation}),
	we obtain the next equality where terms of order $O_{p,n}(\Delta^{2})$
	or higher have been neglected{\small{}:
		\[
		\begin{array}{l}
		R_{\upsilon,i}^{d,l}  = R_{\upsilon,i+1}^{d,l}+\triangle C^{T}C+\Delta\left(R_{\upsilon,i+1}^{d,l}A_{\theta}(\overline{X_{\upsilon}^{d,l-1}}(t_{i}),t_{i})+A_{\theta}(\overline{X_{\upsilon}^{d,l-1}}(t_{i}),t_{i})^{T}R_{\upsilon,i+1}^{d,l}\right)\\
		 +  \Delta^{2}A_{\theta}(\overline{X_{\upsilon}^{d,l-1}}(t_{i}),t_{i})^{T}R_{\upsilon,i+1}^{d,l}A_{\theta}(\overline{X_{\upsilon}^{d,l-1}}(t_{i}),t_{i})\\
		 -  \triangle(I_{d}+\Delta A_{\theta}(\overline{X_{\upsilon}^{d,l-1}}(t_{i}),t_{i})^{T})R_{\upsilon,i+1}^{d,l}BG(R_{\upsilon,i+1}^{d,l-1})B^{T}R_{\upsilon,i+1}^{d,l}(I_{d}+\Delta A_{\theta}(\overline{X_{\upsilon}^{d,l-1}}(t_{i}),t_{i}))\\
		 = R_{\upsilon,i+1}^{d,l}+\triangle C^{T}C+\Delta R_{\upsilon,i+1}^{d,l}A_{\theta}(\overline{X_{\upsilon}^{d,l-1}}(t_{i}),t_{i})+\Delta A_{\theta}(\overline{X_{\upsilon}^{d,l-1}}(t_{i}),t_{i})^{T}R_{\upsilon,i+1}^{d,l}\\
		 - \triangle R_{\upsilon,i+1}^{d,l}BU^{-1}B^{T}R_{\upsilon,i+1}^{d,l}+O_{p,n}(\Delta^{2})\\
		h_{\upsilon,i}^{d,l}(Y)  =  h_{\upsilon,i+1}^{d,l}(Y)-\triangle C^{T}Y(t_{i})+\Delta A_{\theta}(\overline{X_{\upsilon}^{d,l-1}}(t_{i}),t_{i})^{T}h_{\upsilon,i+1}^{d,l}(Y)\\
		 -  \triangle(I_{d}+\Delta A_{\theta}(\overline{X_{\upsilon}^{d,l-1}}(t_{i}),t_{i})^{T})R_{\upsilon,i+1}^{d,l}BG(R_{\upsilon,i+1}^{d,l})B^{T}h_{\upsilon,i+1}^{d,l}(Y)\\
		 =  h_{\upsilon,i+1}^{d,l}(Y)-\triangle C^{T}Y(t_{i})+\Delta A_{\theta}(\overline{X_{\upsilon}^{d,l-1}}(t_{i}),t_{i})^{T}h_{\upsilon,i+1}^{d,l}(Y)\\
		 -  \triangle R_{\upsilon,i+1}^{d,l}BU^{-1}B^{T}h_{\upsilon,i+1}^{d,l}(Y)+O_{p,n}(\Delta^{2}).\\
		\overline{X_{\upsilon}^{d,l}}(t_{i+1}) =  \overline{X_{\upsilon}^{d,l}}(t_{i})+\Delta A_{\theta}(\overline{X_{\upsilon}^{d,l-1}}(t_{i}),t_{i})\overline{X_{\upsilon}^{d,l}}(t_{i})\\
		-\triangle BU^{-1}B^{T}\left(R_{\upsilon}^{d,l}(t_{i+1})\overline{X_{\upsilon}^{d,l}}(t_{i})+h_{\upsilon}^{d,l}(t_{i+1})\right)+O_{p,n}(\triangle^{2}).
		\end{array}
		\]
	} All the previous approximations have been obtained by using $G(R_{\upsilon,i+1}^{d,l})=U^{-1}+O_{p,n}(\triangle)$,
	$A_{\theta}(\overline{X_{\upsilon}^{l-1}}(t_{i+1}),t_{i+1})=A_{\theta}(\overline{X_{\upsilon}^{l-1}}(t_{i}),t_{i})+O_{n}(\triangle)$,
	$A_{\theta}(\overline{X_{\upsilon}^{d,l-1}}(t_{i}),t_{i})=O_{p,n}(1)$,
	$\left\Vert E_{\upsilon,i}^{d,l-1}(Y)\right\Vert _{2}=O_{p,n}(1)$
	and $\left\Vert E_{\upsilon}^{l-1}(t)\right\Vert _{2}=O_{n}(1)$ (proposition
	\ref{prop:E_bounded_probability}). By making the subtraction between
	the two last finite difference equations, we obtain{\small{}:
		\[
		\begin{array}{l}
		R_{\upsilon}^{l}(t_{i})-R_{\upsilon,i}^{d,l} \\
		 = R_{\upsilon}^{l}(t_{i+1})+\triangle C^{T}C+\triangle A_{\theta}(\overline{X_{\upsilon}^{l-1}}(t_{i}),t_{i})^{T}R_{\upsilon}^{l}(t_{i+1})+\triangle R_{\upsilon}^{l}(t_{i+1})A_{\upsilon}(\overline{X_{v}^{l-1}}(t_{i}),t_{i})\\
		 -  \triangle R_{\upsilon}^{l}(t_{i+1})BU^{-1}B^{T}R_{\upsilon}^{l}(t_{i+1})\\
		 -  R_{\upsilon,i+1}^{d,l}-\triangle C^{T}C-\Delta A_{\theta}(\overline{X_{\upsilon}^{d,l-1}}(t_{i}),t_{i})^{T}R_{\upsilon,i+1}^{d,l}-\Delta R_{\upsilon,i+1}^{d,l}A_{\theta}(\overline{X_{\upsilon}^{d,l-1}}(t_{i}),t_{i})\\
		 + \triangle R_{\upsilon,,i+1}^{d,l}BU^{-1}B^{T}R_{\upsilon,i+1}^{d,l}+O_{p,n}(\Delta^{2})\\
		 =  R_{\upsilon}^{l}(t_{i+1})-R_{\upsilon,i+1}^{d,l}+\triangle A_{\theta}(\overline{X_{\upsilon}^{l-1}}(t_{i}),t_{i})^{T}R_{\upsilon}^{l}(t_{i+1})-\triangle A_{\theta}(\overline{X_{\upsilon}^{d,l-1}}(t_{i}),t_{i})^{T}R_{\upsilon}^{l}(t_{i+1})\\
		 + \triangle A_{\theta}(\overline{X_{\upsilon}^{d,l-1}}(t_{i}),t_{i})^{T}R_{\upsilon}^{l}(t_{i+1})-\Delta A_{\theta}(\overline{X_{\upsilon}^{d,l-1}}(t_{i}),t_{i})^{T}R_{\upsilon,i+1}^{d,l}\\
		 + \triangle R_{v}^{l}(t_{i+1})A_{\theta}(\overline{X_{\upsilon}^{l-1}}(t_{i}),t_{i})-\triangle R_{\upsilon,i+1}^{d,l}A_{\theta}(\overline{X_{\upsilon}^{l-1}}(t_{i}),t_{i})\\
		+ \triangle R_{\upsilon,i+1}^{d,l}A_{\theta}(\overline{X_{\upsilon}^{l-1}}(t_{i}),t_{i})-\Delta R_{\upsilon,i+1}^{d,l}A_{\theta}(\overline{X_{\upsilon}^{d,l-1}}(t_{i}),t_{i})\\
		 - \left(\triangle R_{\upsilon}^{l}(t_{i+1})BU^{-1}B^{T}R_{\upsilon}^{l}(t_{i+1})-\triangle R_{\upsilon,i+1}^{d,l}BU^{-1}B^{T}R_{\upsilon,i+1}^{d,l}\right)+O_{p,n}(\Delta^{2})\\
		 =  R_{\upsilon}^{l}(t_{i+1})-R_{\upsilon,i+1}^{d,l}+\triangle\left(A_{\theta}(\overline{X_{\upsilon}^{l-1}}(t_{i}),t_{i})-A_{\theta}(\overline{X_{\upsilon}^{d,l-1}}(t_{i}),t_{i})\right){}^{T}R_{\upsilon}^{l}(t_{i+1})\\
		 + \triangle A_{\theta}(\overline{X_{\upsilon}^{d,l-1}}(t_{i}),t_{i})^{T}\left(R_{\upsilon}^{l}(t_{i+1})-R_{\upsilon,i+1}^{d,l}\right)+\triangle\left(R_{\upsilon}^{l}(t_{i+1})-R_{\upsilon,i+1}^{d,l-1}\right)A_{\theta}(\overline{X_{\upsilon}^{l-1}}(t_{i}),t_{i})\\
		 +  \Delta R_{\upsilon,i+1}^{d,l}\left(A_{\theta}(\overline{X_{\upsilon}^{l-1}}(t_{i}),t_{i})-A_{\theta}(\overline{X_{\upsilon}^{d,l-1}}(t_{i}),t_{i})\right)\\
		 - \triangle\left(R_{\upsilon}^{l}(t_{i+1})-R_{\upsilon,i+1}^{d,l}\right)BU^{-1}B^{T}R_{\upsilon}^{l}(t_{i+1})-\triangle R_{\upsilon,i+1}^{d,l}BU^{-1}B^{T}\left(R_{\upsilon}^{l}(t_{i+1})-R_{\upsilon,i+1}^{d,l}\right)+O_{p,n}(\Delta^{2})
		\end{array}
		\]
	}an{\small{}d:
		\[
		\begin{array}{l}
		h_{\upsilon,i}^{d,l}(Y)-h_{\upsilon}^{l}(t_{i}) \\
		 =  h_{\upsilon,i+1}^{d,l}(Y)-h_{\upsilon}^{l}(t_{i+1})-\triangle C^{T}\epsilon_{i}+\Delta A_{\theta}(\overline{X_{\upsilon}^{d,l-1}}(t_{i}),t_{i})^{T}h_{\upsilon,i+1}^{d,l}(Y)\\
		 -  \triangle R_{\upsilon,i+1}^{d,l}BU^{-1}B^{T}h_{\upsilon,i+1}^{d,l}(Y)\\
		 - \triangle A_{\theta}(\overline{X_{\upsilon}^{l-1}}(t_{i}),t_{i})^{T}h_{\upsilon}^{l}(t_{i+1})+R_{\upsilon}^{l}(t_{i+1})BU^{-1}B^{T}h_{\upsilon}^{l}(t_{i+1})+O_{p,n}(\Delta^{2})\\
		 =  h_{\upsilon,i+1}^{d,l}(Y)-h_{\upsilon}^{l}(t_{i+1})+\Delta\left(A_{\theta}(\overline{X_{\upsilon}^{d,l-1}}(t_{i}),t_{i})-A_{\theta}(\overline{X_{\upsilon}^{l-1}}(t_{i}),t_{i})\right)^{T}h_{\upsilon,i+1}^{d,l}(Y)\\
		 + \Delta A_{\theta}(\overline{X_{\upsilon}^{l-1}}(t_{i}),t_{i})^{T}\left(h_{\upsilon,i+1}^{d,l}(Y)-h_{\upsilon}^{l}(t_{i+1})\right)-\triangle R_{\upsilon,i+1}^{d,l}BU^{-1}B^{T}(h_{\upsilon,i+1}^{d,l}(Y)-h_{\upsilon}^{l}(t_{i+1}))\\
		 - \triangle(R_{\upsilon,i+1}^{d,l}-R_{\upsilon}^{l}(t_{i+1}))BU^{-1}B^{T}h_{\upsilon}^{l}(t_{i+1})-\triangle C^{T}\epsilon_{i}+O_{p,n}(\Delta^{2}).
		\end{array}
		\]
	}Triangular inequality gives us{\small{}:
		\[
		\begin{array}{l}
		\left\Vert R_{\upsilon}^{l}(t_{i})-R_{\upsilon,i}^{d,l}\right\Vert _{2}
		\\  \leq  \left\Vert R_{\upsilon}^{l}(t_{i+1})-R_{\upsilon,i+1}^{d,l}\right\Vert _{2}+2\Delta\overline{A}\left\Vert R_{\upsilon}^{l}(t_{i+1})-R_{\upsilon,i+1}^{d,l}\right\Vert _{2}\\
		 +  \triangle\left(\left\Vert R_{\upsilon,i+1}^{d,l}\right\Vert _{2}+\left\Vert R_{\upsilon}^{l}(t_{i+1})\right\Vert _{2}\right)\left\Vert BU^{-1}B^{T}\right\Vert _{2}\left\Vert R_{\upsilon}^{l}(t_{i+1})-R_{\upsilon,i+1}^{d,l}\right\Vert _{2}\\
		 +  \triangle\left(\left\Vert R_{\upsilon}^{l}(t_{i+1})\right\Vert _{2}+\left\Vert R_{\upsilon,i+1}^{d,l}\right\Vert _{2}\right)\left\Vert A_{\theta}(\overline{X_{\upsilon}^{l-1}}(t_{i}),t_{i})-A_{\theta}(\overline{X_{\upsilon}^{d,l-1}}(t_{i}),t_{i})\right\Vert _{2}+O_{p,n}(\Delta^{2})\\
		 \leq \left(1+\Delta\left(2\overline{A}+\left(\left\Vert R_{\upsilon,i+1}^{d,l}\right\Vert _{2}+\left\Vert R_{\upsilon}^{l}(t_{i+1})\right\Vert _{2}\right)\left\Vert BU^{-1}B^{T}\right\Vert _{2}\right)\right)\left\Vert R_{\upsilon,i+1}^{d,l}-R_{\upsilon}^{l}(t_{i+1})\right\Vert _{2}\\
		 +  \triangle\left(\left\Vert R_{\upsilon}^{l}(t_{i+1})\right\Vert _{2}+\left\Vert R_{\upsilon,i+1}^{d,l}\right\Vert _{2}\right)\left\Vert A_{\theta}(\overline{X_{\upsilon}^{l-1}}(t_{i}),t_{i})-A_{\theta}(\overline{X_{\upsilon}^{d,l-1}}(t_{i}),t_{i})\right\Vert _{2}+O_{p,n}(\Delta^{2})\\
		 \leq  (1+\Delta\left(2\overline{A}+O_{n}(1)\right)\left\Vert R_{\upsilon,i+1}^{d,l}-R_{\upsilon}^{l}(t_{i+1})\right\Vert _{2}\\
		 +  O_{n}(\triangle)\left\Vert A_{\theta}(\overline{X_{\upsilon}^{l-1}}(t_{i}),t_{i})-A_{\theta}(\overline{X_{\upsilon}^{d,l-1}}(t_{i}),t_{i})\right\Vert _{2}+O_{p,n}(\Delta^{2})
		\end{array}
		\]
		and:
		\[
		\begin{array}{l}
		\left\Vert \overline{X_{\upsilon}^{l}}(t_{i+1})-\overline{X_{\upsilon}^{d,l}}(t_{i+1})\right\Vert _{2}\\
		  \leq  \left\Vert \overline{X_{\upsilon}^{l}}(t_{i})-\overline{X_{\upsilon}^{d,l}}(t_{i})\right\Vert _{2}+\Delta\left\Vert \overline{X_{\upsilon}^{d,l}}(t_{i})\right\Vert _{2}\left\Vert A_{\theta}(\overline{X_{\upsilon}^{l-1}}(t_{i}),t_{i})-A_{\theta}(\overline{X_{\upsilon}^{d,l-1}}(t_{i}),t_{i})\right\Vert _{2}\\
		 +  \Delta\left(\left\Vert A_{\theta}(\overline{X_{\upsilon}^{d,l-1}}(t_{i}),t_{i})\right\Vert _{2}+\left\Vert BU^{-1}B^{T}\right\Vert _{2}\left\Vert R_{\upsilon}^{l}(t_{i+1})\right\Vert _{2}\right)\left\Vert \overline{X_{\upsilon}^{l}}(t_{i})-\overline{X_{\upsilon}^{d,l}}(t_{i})\right\Vert _{2}\\
		 +  \triangle\left\Vert BU^{-1}B^{T}\right\Vert _{2}\left\Vert \overline{X_{\upsilon}^{d,l}}(t_{i})\right\Vert _{2}\left\Vert R_{\upsilon}^{l}(t_{i+1})-R_{\upsilon}^{d,l}(t_{i+1})\right\Vert _{2}\\
		 +  \triangle\left\Vert BU^{-1}B^{T}\right\Vert _{2}\left\Vert h_{\upsilon}^{l}(t_{i+1})-h_{\upsilon}^{dl}(t_{i+1})\right\Vert _{2}+O_{p,n}(\Delta^{2})\\
		 \leq  \left(1+O_{p,n}(\triangle)\right)\left\Vert \overline{X_{\upsilon}^{l}}(t_{i})-\overline{X_{\upsilon}^{d,l}}(t_{i})\right\Vert _{2}\\
		 +  \triangle\left\Vert BU^{-1}B^{T}\right\Vert _{2}\left(\left\Vert \overline{X_{\upsilon}^{d,l}}(t_{i})\right\Vert _{2}\left\Vert R_{\upsilon}^{l}(t_{i+1})-R_{\upsilon}^{d,l}(t_{i+1})\right\Vert _{2}+\left\Vert h_{\upsilon}^{l}(t_{i+1})-h_{\upsilon}^{d,l}(t_{i+1})\right\Vert _{2}\right)\\
		+  \Delta\left\Vert \overline{X_{\upsilon}^{d,l}}(t_{i})\right\Vert _{2}\left\Vert A_{\theta}(\overline{X_{\upsilon}^{l-1}}(t_{i}),t_{i})-A_{\theta}(\overline{X_{\upsilon}^{d,l-1}}(t_{i}),t_{i})\right\Vert _{2}+O_{p,n}(\Delta^{2}).
		\end{array}
		\]
	}From these last equations describing the behavior of $h_{\upsilon,i}^{d,l}(Y)-h_{\upsilon}^{l}(t_{i}),\,\left\Vert R_{\upsilon}^{l}(t_{i})-R_{\upsilon,i}^{d,l}\right\Vert _{2}$
	and $\left\Vert \overline{X_{\upsilon}^{l}}(t_{i+1})-\overline{X_{\upsilon}^{d,l}}(t_{i+1})\right\Vert _{2}$,
	we prove now by induction $\sup_{\upsilon\in\varUpsilon}\left\Vert R_{\upsilon,i}^{d,l}-R_{\upsilon}^{l}(t_{i})\right\Vert _{2} = o_{p,n}(1),$ $\sup_{\upsilon\in\varUpsilon}\left\Vert h_{\upsilon,i}^{d,l}-h_{\upsilon}^{l}(t_{i})\right\Vert _{2}=o_{p,n}(1)$
	and $\sup_{\upsilon\in\varUpsilon}\left\Vert \overline{X_{\upsilon}^{l}}(t_{i})-\overline{X_{\upsilon}^{d,l}}(t_{i})\right\Vert _{2}=o_{p,n}(1)$
	for each $l\in\mathbb{N}$. First of all, let us start the initialization
	with $l=1$, we already know $\sup_{\upsilon\in\varUpsilon}\left\Vert \overline{X_{\upsilon}^{0}}(t_{i})-\overline{X_{\upsilon}^{d,0}}(t_{i})\right\Vert _{2}=\left\Vert x_{0}-x_{0}\right\Vert _{2}=0$,
	the inequality respected by $\left\Vert R_{\upsilon}^{1}(t_{i})-R_{\upsilon,i}^{d,1}\right\Vert _{2}$
	becomes
		$\left\Vert R_{\upsilon}^{1}(t_{i})-R_{\upsilon,i}^{d,1}\right\Vert _{2} \leq (1+O_{n}(\Delta))\left\Vert R_{\upsilon,i+1}^{d,1}-R_{\upsilon}^{1}(t_{i+1})\right\Vert _{2}+O_{p,n}(\Delta^{2})$. Discrete Gronwall lemma \ref{lem:disc_Gronwall_lemma} gives us
	$\left\Vert R_{\upsilon}^{1}(t_{i})-R_{\upsilon,i}^{d,1}\right\Vert _{2}\leq\sum_{j=1}^{n-i}e^{O_{p,n}(1)}O_{p,n}(\Delta^{2})+O_{p,n}(\Delta)=O_{p,n}(\Delta)$
	since $\left\Vert R_{\upsilon}^{1}(t_{n})-R_{\upsilon,n}^{d,1}\right\Vert _{2}=O_{p,n}(\Delta)$,
	from this we can derive the uniform bound $\left\Vert R_{\upsilon,i}^{d,1}-R_{\upsilon}^{1}(t_{i})\right\Vert _{2}=o_{p,n}(1)$
	for all $i\in\left\llbracket 0,\,n\right\rrbracket $ and $\theta\in\Theta$.
	By using this and $R_{\upsilon}^{l}(t_{i+1})=O_{n}(1)$, $h_{\upsilon}^{l}(t_{i+1})=O_{n}(1)$,
	we can simplify the expression of $h_{\upsilon,i}^{d,1}(Y)-h_{\upsilon}^{1}(t_{i})${\small{}:
		\[
		\begin{array}{lll}
		h_{\upsilon,i}^{d,1}(Y)-h_{\upsilon}^{1}(t_{i}) & = & h_{\upsilon,i+1}^{d,1}(Y)-h_{\upsilon}^{1}(t_{i+1})+\Delta A_{\theta}(x_{0},t_{i})^{T}\left(h_{\upsilon,i+1}^{d,1}(Y)-h_{\upsilon}^{1}(t_{i+1})\right)\\
		& - & \triangle(R_{\upsilon,i+1}^{d,1}-R_{\upsilon}^{1}(t_{i+1}))BU^{-1}B^{T}h_{\upsilon}^{1}(t_{i+1})-\triangle C^{T}\epsilon_{i}+O_{p,n}(\Delta^{2})\\
		& = & \left(I_{d}+\Delta A_{\theta}(x_{0},t_{i})^{T}\right)\left(h_{\upsilon,i+1}^{d,1}(Y)-h_{\upsilon}^{1}(t_{i+1})\right)-\triangle C^{T}\epsilon_{i}+O_{p,n}(\Delta)\\
		& = & \left(I_{d}+O_{n}(\Delta)\right)\left(h_{\upsilon,i+1}^{d,1}(Y)-h_{\upsilon}^{1}(t_{i+1})\right)-\triangle C^{T}\epsilon_{i}+O_{p,n}(\Delta).
		\end{array}
		\]
	}Since $h_{\upsilon,n}^{d,l}(Y)-h_{\upsilon}^{l}(t_{n})=-\triangle C^{T}\epsilon_{n}$,
	we show by induction $h_{\upsilon,i}^{d,1}(Y)-h_{\upsilon}^{1}(t_{i})=-\triangle C^{T}\sum_{j=i}^{n}\epsilon_{j}+O_{p,n}(\Delta)$
	indeed:{\small{}
		\[
		\begin{array}{lll}
		h_{\upsilon,i}^{d}(Y)-h_{\upsilon}(t_{i}) & = & \left(I_{d}+O_{p,n}(\Delta)\right)\left(-\triangle C^{T}\sum_{j=i+1}^{n}\epsilon_{j}+O_{p,n}(\Delta)\right)-\triangle C^{T}\epsilon_{i}+O_{p,n}(\Delta)\\
		& = & -\triangle C^{T}\sum_{j=i+1}^{n}\epsilon_{j}+O_{p,n}(\Delta)-\triangle C^{T}\epsilon_{i}+O_{p,n}(\Delta)\\
		& = & -\triangle C^{T}\sum_{j=i}^{n}\epsilon_{j}+O_{p,n}(\Delta)
		\end{array}
		\]
	}from this we conclude $\left\Vert h_{\upsilon,i}^{d,1}(Y)-h_{\upsilon}^{1}(t_{i})\right\Vert _{2}=o_{p,n}(1)$
	for all $i\in\left\llbracket 0,\,n\right\rrbracket $ and $\upsilon\in\varUpsilon${\small{}.
	}From these inequalities, we derive: {\small{}
		\[
		\begin{array}{lll}
		\left\Vert \overline{X_{\upsilon}^{1}}(t_{i+1})-\overline{X_{\upsilon}^{d,1}}(t_{i+1})\right\Vert _{2} & \leq & \left(1+O_{n}(\triangle)\right)\left\Vert \overline{X_{\upsilon}^{1}}(t_{i})-\overline{X_{\upsilon}^{d,1}}(t_{i})\right\Vert _{2}\\
		& + & \triangle\left\Vert BU^{-1}B^{T}\right\Vert _{2}\left\Vert \overline{X_{\upsilon}^{d,1}}(t_{i})\right\Vert _{2}\left\Vert R_{\upsilon}^{1}(t_{i+1})-R_{\upsilon}^{d,1}(t_{i+1})\right\Vert _{2}\\
		& + & \triangle\left\Vert BU^{-1}B^{T}\right\Vert _{2}\left\Vert h_{\upsilon}^{1}(t_{i+1})-h_{\upsilon}^{d,1}(t_{i+1})\right\Vert _{2}+O_{p,n}(\Delta^{2})\\
		& \leq & \left(1+O_{n}(\triangle)\right)\left\Vert \overline{X_{\upsilon}^{1}}(t_{i})-\overline{X_{\upsilon}^{d,1}}(t_{i})\right\Vert _{2}+o_{p,n}(\triangle)
		\end{array}
		\]
	}we can use the discrete Gronwall lemma \ref{lem:disc_Gronwall_lemma}
	to obtain $\left\Vert \overline{X_{\upsilon}^{1}}(t_{i+1})-\overline{X_{\upsilon}^{d,1}}(t_{i+1})\right\Vert _{2}\leq\sum_{j=1}^{n-i}e^{O_{n}(1)}o_{p,n}(\triangle)=o_{p,n}(1).$
	Now let us assume the property holds for $l-1$, by using the induction
	hypothesis, we can simplify the inequality respected by $\left\Vert R_{\upsilon}^{l}(t_{i})-R_{\upsilon,i}^{d,l}\right\Vert _{2}$
	which becomes $ \left\Vert R_{\upsilon}^{l}(t_{i})-R_{\upsilon,i}^{d,l}\right\Vert _{2}  \leq  (1+O_{p,n}(\Delta))\left\Vert R_{\upsilon,i+1}^{d,l}-R_{\upsilon}^{l}(t_{i+1})\right\Vert _{2}+O_{p,n}(\Delta)$ again, by using the discrete Gronwall lemma we have $\left\Vert R_{\upsilon,i}^{d,l}-R_{\upsilon}^{l}(t_{i})\right\Vert _{2}=o_{p,n}(1)$,
	with the help of this result, we have{\small{}
		\[
		\begin{array}{lll}
		h_{\upsilon,i}^{d,l}(Y)-h_{\upsilon}^{l}(t_{i}) & = & \left(I_{d}+O_{p}(\Delta)\right)(h_{\upsilon,i+1}^{d,l}(Y)-h_{\upsilon}^{l}(t_{i+1}))-\triangle C^{T}\epsilon_{i}+o_{p,n}(\Delta)\end{array}
		\]
	}and by induction we derive the expression $
	h_{\upsilon,i}^{d}(Y)-h_{v}(t_{i})  =  -\triangle C^{T}\sum_{j=i}^{n}\epsilon_{j}+O_{p,n}(\Delta)$ and thus $\left\Vert h_{\upsilon,i}^{d,l}(Y)-h_{\upsilon}^{l}(t_{i})\right\Vert _{2}=o_{p,n}(1)$.
	By proceeding the same way as in the $l=1$ case, we derive $\left\Vert \overline{X_{\upsilon}^{l}}(t_{i})-\overline{X_{\upsilon}^{d,l}}(t_{i})\right\Vert _{2}=o_{p,n}(1)$
	uniformly on $\left\llbracket 0,\,n\right\rrbracket \times\varUpsilon$,
	which conclude the proof by induction.
	
	If $\theta\longmapsto A_{\theta}(x,.)$ is $C^{1}$ on $\Theta$,
	then $t\longrightarrow R_{\upsilon}(t)$ and $t\longrightarrow h_{\upsilon}(t)$
	are differentiable for all $\upsilon\in\varUpsilon$. By differentiating
	and integrating (\ref{eq:continous_accurate_Riccati_equation}), we
	have{\small{}:
		\[
		\begin{array}{l}
		\frac{\partial R_{\upsilon}^{l}(t_{i})}{\partial\upsilon}\\
		  =  \frac{\partial R_{\upsilon}^{l}(t_{i+1})}{\partial\upsilon}-\int_{t_{i}}^{t_{i+1}}\left(\frac{\partial R_{\upsilon}^{l}(t)}{\partial\upsilon}BU^{-1}B^{T}R_{\upsilon}^{l}(t)+R_{\upsilon}^{l}(t)BU^{-1}B^{T}\frac{\partial R_{\upsilon}^{l}(t)}{\partial\upsilon}\right)dt\\
		 +  \int_{t_{i}}^{t_{i+1}}\left(\frac{\partial R_{\upsilon}^{l}(t)}{\partial\upsilon}A_{\theta}(\overline{X_{\upsilon}^{l-1}}(t),t)+R_{\upsilon}^{l}(t)\frac{\partial A_{\theta}}{\partial\upsilon}(\overline{X_{\upsilon}^{l-1}}(t),t)+R_{\upsilon}^{l}(t)\frac{\partial A_{\theta}}{\partial x}(\overline{X_{\upsilon}^{l-1}}(t),t)\frac{\partial\overline{X_{\upsilon}^{l-1}}(t)}{\partial\upsilon}\right)dt\\
		 +  \int_{t_{i}}^{t_{i+1}}\left(\frac{\partial A_{\theta}}{\partial\upsilon}(\overline{X_{\upsilon}^{l-1}}(t),t)^{T}R_{\upsilon}^{l}(t)+\frac{\partial A_{\theta}}{\partial x}(\overline{X_{\upsilon}^{l-1}}(t),t)^{T}\frac{\partial\overline{X_{v}^{l-1}}(t)}{\partial\upsilon}R_{\upsilon}^{l}(t)+A_{\theta}(\overline{X_{\upsilon}^{l-1}}(t),t)^{T}\frac{\partial R_{v}^{l}(t)}{\partial\upsilon}\right)dt\\
		 =  \frac{\partial R_{\upsilon}^{l}(t_{i+1})}{\partial\upsilon}+\triangle\frac{\partial A_{\theta}}{\partial\upsilon}(\overline{X_{\upsilon}^{l-1}}(t_{i}),t_{i})^{T}R_{v}^{l}(t_{i+1})+\triangle\frac{\partial A_{\theta}}{\partial x}(\overline{X_{\upsilon}^{l-1}}(t_{i}),t_{i})^{T}\frac{\partial\overline{X_{\upsilon}^{l-1}}(t_{i})}{\partial\upsilon}R_{\upsilon}^{l}(t_{i+1})\\
		 +  \triangle A_{\theta}(\overline{X_{\upsilon}^{l-1}}(t_{i}),t_{i})^{T}\frac{\partial R_{\upsilon}^{l}(t_{i+1})}{\partial\upsilon}+\triangle\frac{\partial R_{\upsilon}^{l}(t_{i+1})}{\partial\upsilon}A_{\theta}(\overline{X_{\upsilon}^{l-1}}(t_{i}),t_{i})\\
		 +  \triangle R_{\upsilon}^{l}(t_{i+1})\frac{\partial A_{\theta}}{\partial\upsilon}(\overline{X_{\upsilon}^{l-1}}(t_{i}),t_{i})+\triangle R_{\upsilon}^{l}(t_{i+1})\frac{\partial A_{\theta}}{\partial x}(\overline{X_{\upsilon}^{l-1}}(t_{i}),t_{i})\frac{\partial\overline{X_{\upsilon}^{l-1}}(t_{i})}{\partial\upsilon}\\
		 -  \triangle\frac{\partial R_{\upsilon}^{l}(t_{i+1})}{\partial\upsilon}BU^{-1}B^{T}R_{\upsilon}^{l}(t_{i+1})-\triangle R_{\upsilon}^{l}(t_{i+1})BU^{-1}B^{T}\frac{\partial R_{\upsilon}^{l}(t_{i+1})}{\partial\upsilon}+O_{n}(\triangle^{2})
		\end{array}
		\]
		\[	
		\begin{array}{l}
		\frac{\partial h_{\upsilon}^{l}(t_{i})}{\partial\upsilon} \\
		 =  \frac{\partial h_{\upsilon}^{l}(t_{i+1})}{\partial\upsilon}-\int_{t_{i}}^{t_{i+1}}\left(\frac{\partial R_{\upsilon}^{l}(t)}{\partial\upsilon}BU^{-1}B^{T}h_{\upsilon}^{l}(t)+R_{\upsilon}^{l}(t)BU^{-1}B^{T}\frac{\partial h_{\upsilon}^{l}(t)}{\partial\upsilon}\right)dt\\
		 +  \int_{t_{i}}^{t_{i+1}}\left(\frac{\partial A_{\theta}}{\partial\upsilon}(\overline{X_{\upsilon}^{l-1}}(t),t)^{T}h_{\upsilon}^{l}(t)+\frac{\partial A_{\theta}}{\partial x}(\overline{X_{\upsilon}^{l-1}}(t),t)^{T}\frac{\partial\overline{X_{\upsilon}^{l-1}}(t)}{\partial\upsilon}h_{\upsilon}^{l}(t)+A_{\theta}(\overline{X_{\upsilon}^{l-1}}(t),t)^{T}\frac{\partial h_{\upsilon}^{l}(t)}{\partial\upsilon}\right)dt\\
		 =  \frac{\partial h_{\upsilon}^{l}(t_{i+1})}{\partial\upsilon}+\triangle\frac{\partial A_{\theta}}{\partial\upsilon}(\overline{X_{\upsilon}^{l-1}}(t_{i}),t_{i})^{T}h_{\upsilon}^{l}(t_{i+1})\\
		 + \triangle\frac{\partial A_{\theta}}{\partial x}(\overline{X_{\upsilon}^{l-1}}(t_{i}),t_{i})^{T}\frac{\partial\overline{X_{\upsilon}^{l-1}}(t_{i})}{\partial\upsilon}h_{\upsilon}^{l}(t_{i+1})+\triangle A_{\theta}(\overline{X_{\upsilon}^{l-1}}(t_{i}),t_{i})^{T}\frac{\partial h_{\upsilon}^{l}(t_{i+1})}{\partial\upsilon}\\
		 - \triangle\frac{\partial R_{\upsilon}^{l}(t_{i+1})}{\partial\upsilon}BU^{-1}B^{T}h_{\upsilon}^{l}(t_{i+1})-\triangle R_{\upsilon}^{l}(t_{i+1})BU^{-1}B^{T}\frac{\partial h_{\upsilon}^{l}(t_{i+1})}{\partial v}+O_{n}(\triangle^{2}).
		\end{array}
		\]
		\[
		\begin{array}{l}
		\frac{\partial\overline{X_{\upsilon}^{l}}(t_{i+1})}{\partial\upsilon}\\
		  =  \frac{\partial\overline{X_{\upsilon}^{l}}(t_{i})}{\partial\upsilon}-\int_{t_{i}}^{t_{i+1}}BU^{-1}B^{T}\left(\frac{\partial R_{\upsilon}^{l}(t)}{\partial\upsilon}\overline{X_{\upsilon}^{l}}(t)+R_{\upsilon}^{l}(t)\frac{\partial\overline{X_{\upsilon}^{l}}(t)}{\partial\upsilon}+\frac{\partial h_{\upsilon}^{l}(t)}{\partial\upsilon}\right)dt\\
		 +  \int_{t_{i}}^{t_{i+1}}\left(\frac{\partial A_{\theta}}{\partial\upsilon}(\overline{X_{\upsilon}^{l-1}}(t),t)\overline{X_{\upsilon}^{l}}(t)+\frac{\partial A_{\theta}}{\partial x}(\overline{X_{\upsilon}^{l-1}}(t),t)\frac{\partial\overline{X_{\upsilon}^{l-1}}(t)}{\partial\upsilon}\overline{X_{\upsilon}^{l}}(t)+A_{\theta}(\overline{X_{\upsilon}^{l-1}}(t),t)\frac{\partial\overline{X_{\upsilon}^{l}}(t)}{\partial\upsilon}\right)dt\\
		=  \frac{\partial\overline{X_{\upsilon}^{l}}(t_{i})}{\partial\upsilon}+\triangle\left(\frac{\partial A_{\theta}}{\partial\upsilon}(\overline{X_{\upsilon}^{l-1}}(t_{i}),t_{i})+\frac{\partial A_{\theta}}{\partial x}(\overline{X_{\upsilon}^{l-1}}(t_{i}),t_{i})\frac{\partial\overline{X_{\upsilon}^{l-1}}(t_{i})}{\partial\upsilon}\right)\overline{X_{\upsilon}^{l}}(t_{i})
		\\+\triangle A_{\theta}(\overline{X_{\upsilon}^{l-1}}(t_{i}),t_{i})\frac{\partial\overline{X_{\upsilon}^{l}}(t_{i})}{\partial\upsilon}\\
		 -  \triangle BU^{-1}B^{T}\left(\frac{\partial R_{\upsilon}^{l}(t_{i+1})}{\partial\upsilon}\overline{X_{\upsilon}^{l}}(t_{i})+R_{\upsilon}^{l}(t_{i})\frac{\partial\overline{X_{\upsilon}^{l}}(t_{i})}{\partial\upsilon}+\frac{\partial h_{\upsilon}^{l}(t_{i+1})}{\partial\upsilon}\right)+O_{n}(\triangle^{2}).
		\end{array}
		\]
	}Again, if $\theta\longmapsto A_{\theta}$ is $C^{1}$ on $\Theta$,
	then we derive by induction that $\upsilon\longmapsto R_{\upsilon,i}^{d}$
	and $\upsilon\longmapsto h_{\upsilon,i}^{d}(Y)$ are $C^{1}$ as well
	for all $i\in\left\llbracket 0,\,n\right\rrbracket $ and ruled by
	the following finite difference equations{\small{}:
		\[
		\begin{array}{l}
		\frac{\partial R_{\upsilon,i}^{d,l}}{\partial\upsilon}\\
		  =  \frac{\partial R_{\upsilon,i+1}^{d,l}}{\partial\upsilon}+\Delta\frac{\partial R_{\upsilon,i+1}^{d,l}}{\partial\upsilon}A_{\theta}(\overline{X_{\upsilon}^{d,l-1}}(t_{i}),t_{i})+\Delta R_{\upsilon,i+1}^{d,l}\frac{\partial A_{\theta}}{\partial\upsilon}(\overline{X_{\upsilon}^{d,l-1}}(t_{i}),t_{i})\\
		 +  \Delta R_{\upsilon,i+1}^{d,l}\frac{\partial A_{\theta}}{\partial x}(\overline{X_{\upsilon}^{d,l-1}}(t_{i}),t_{i})\frac{\partial\overline{X_{\upsilon}^{d,l-1}}(t_{i})}{\partial\upsilon}+\Delta\frac{\partial A_{\theta}}{\partial\upsilon}(\overline{X_{\upsilon}^{d,l-1}}(t_{i}),t_{i})^{T}R_{\upsilon,i+1}^{d,l}\\
		 +  \Delta\frac{\partial A_{\theta}}{\partial x}(\overline{X_{\upsilon}^{d,l-1}}(t_{i}),t_{i})^{T}\frac{\partial\overline{X_{\upsilon}^{d,l-1}}(t_{i})}{\partial\upsilon}R_{\upsilon,i+1}^{d,l}+\Delta A_{\theta}(\overline{X_{\upsilon}^{d,l-1}}(t_{i}),t_{i})^{T}\frac{\partial R_{\upsilon,i+1}^{d,l}}{\partial\upsilon}\\
		 -  \triangle^{2}\frac{\partial}{\partial\upsilon}(A_{\theta}(\overline{X_{\upsilon}^{d,l-1}}(t_{i}),t_{i})^{T})R_{\upsilon,i+1}^{d,l}BG(R_{\upsilon,i+1}^{d,l-1})B^{T}R_{\upsilon,i+1}^{d,l}(I_{d}+\Delta A_{\theta}(\overline{X_{\upsilon}^{d,l-1}}(t_{i}),t_{i}))\\
		 -  \triangle(I_{d}+\Delta A_{\theta}(\overline{X_{\upsilon}^{d,l-1}}(t_{i}),t_{i})^{T})\frac{\partial R_{\upsilon,i+1}^{d,l}}{\partial\upsilon}BG(R_{\upsilon,i+1}^{d})B^{T}R_{\upsilon,i+1}^{d,l}(I_{d}+\Delta A_{\theta}(\overline{X_{\upsilon}^{d,l-1}}(t_{i}),t_{i}))\\
		 -  \triangle(I_{d}+\Delta A_{\theta}(\overline{X_{\upsilon}^{d,l-1}}(t_{i}),t_{i})^{T})R_{\upsilon,i+1}^{d,l}B\frac{\partial G}{\partial\upsilon}(R_{\upsilon,i+1}^{d,l-1})B^{T}R_{\upsilon,i+1}^{d,l}(I_{d}+\Delta A_{\theta}(\overline{X_{\upsilon}^{d,l-1}}(t_{i}),t_{i}))\\
		 -  \triangle(I_{d}+\Delta A_{\theta}(\overline{X_{\upsilon}^{d,l-1}}(t_{i}),t_{i})^{T})R_{\upsilon,i+1}^{d,l}BG(R_{\upsilon,i+1}^{d,l-1})B^{T}\frac{\partial R_{\upsilon,i+1}^{d,l}}{\partial\upsilon}(I_{d}+\Delta A_{\theta}(\overline{X_{\upsilon}^{d,l-1}}(t_{i}),t_{i}))\\
		 -  \triangle^{2}(I_{d}+\Delta A_{\theta}(\overline{X_{\upsilon}^{d,l-1}}(t_{i}),t_{i})^{T})\frac{\partial R_{\upsilon,i+1}^{d,l}}{\partial\upsilon}BG(R_{\upsilon,i+1}^{d})B^{T}R_{\upsilon,i+1}^{d,l}\frac{\partial}{\partial\upsilon}(A_{\theta}(\overline{X_{\upsilon}^{d,l-1}}(t_{i}),t_{i}))\\
		 =  \frac{\partial R_{v,i+1}^{d,l}}{\partial\upsilon}+\Delta\left(\frac{\partial R_{\upsilon,i+1}^{d,l}}{\partial\upsilon}A_{\theta}(\overline{X_{\upsilon}^{d,l-1}}(t_{i}),t_{i})+R_{\upsilon,i+1}^{d,l}\frac{\partial A_{\theta}}{\partial\theta}(\overline{X_{\upsilon}^{d,l-1}}(t_{i}),t_{i})\right)\\
		 +  \Delta\left(R_{\upsilon,i+1}^{d,l}\frac{\partial A_{\theta}}{\partial x}(\overline{X_{\upsilon}^{d,l-1}}(t_{i}),t_{i})\frac{\partial\overline{X_{\upsilon}^{d,l-1}}(t_{i})}{\partial\upsilon}+\frac{\partial A_{\theta}}{\partial\upsilon}(\overline{X_{\upsilon}^{d,l-1}}(t_{i}),t_{i})^{T}R_{\upsilon,i+1}^{d,l}\right)\\
		 +  \Delta\left(\frac{\partial A_{\theta}}{\partial x}(\overline{X_{\upsilon}^{d,l-1}}(t_{i}),t_{i})^{T}\frac{\partial\overline{X_{\upsilon}^{d,l-1}}(t_{i})}{\partial\upsilon}R_{\upsilon,i+1}^{d,l}+A_{\theta}(\overline{X_{\upsilon}^{d,l-1}}(t_{i}),t_{i})^{T}\frac{\partial R_{\upsilon,i+1}^{d,l}}{\partial\upsilon}\right)\\
		 -  \Delta\frac{\partial R_{\upsilon,i+1}^{d,l}}{\partial\upsilon}BU^{-1}B^{T}R_{\upsilon,i+1}^{d,l}-\Delta R_{\upsilon,i+1}^{d,l}BU^{-1}B^{T}\frac{\partial R_{\upsilon,i+1}^{d,l}}{\partial\upsilon}+O_{p,n}(\triangle^{2})
		\end{array}
		\]
		\[
		\begin{array}{l}
		\frac{\partial h_{\upsilon,i}^{d,l}}{\partial\upsilon}(Y)\\
		 =  \frac{\partial h_{\upsilon,i+1}^{d,l}}{\partial\upsilon}(Y)+\Delta\frac{\partial A_{\theta}}{\partial\upsilon}(\overline{X_{\upsilon}^{d,l-1}}(t_{i}),t_{i})^{T}h_{v,i+1}^{d,l}(Y)\\
		+  \Delta\frac{\partial A_{\theta}}{\partial x}(\overline{X_{\upsilon}^{d,l-1}}(t_{i}),t_{i})^{T}\frac{\partial\overline{X_{v}^{d,l-1}}(t_{i})}{\partial\upsilon}h_{\upsilon,i+1}^{d,l}(Y)+\Delta A_{\theta}(\overline{X_{\upsilon}^{d,l-1}}(t_{i}),t_{i})^{T}\frac{\partial h_{\upsilon,i+1}^{d,l}}{\partial\upsilon}(Y)\\
		 -  \triangle^{2}\frac{\partial A_{\theta}}{\partial\upsilon}(\overline{X_{\upsilon}^{d,l-1}}(t_{i}),t_{i})^{T}R_{\upsilon,i+1}^{d,l}BG(R_{\upsilon,i+1}^{d,l})B^{T}h_{\upsilon,i+1}^{d,l}(Y)\\
		 -  \triangle(I_{d}+\Delta A_{\theta}(\overline{X_{\upsilon}^{d,l-1}}(t_{i}),t_{i})^{T})\frac{\partial R_{\upsilon,i+1}^{d,l}}{\partial\upsilon}BG(R_{\upsilon,i+1}^{d,l})B^{T}h_{\upsilon,i+1}^{d,l}(Y)\\
		 - \triangle(I_{d}+\Delta A_{\theta}(\overline{X_{\upsilon}^{d,l-1}}(t_{i}),t_{i})^{T})R_{\upsilon,i+1}^{d,l}B\frac{\partial G}{\partial\upsilon}(R_{\upsilon,i+1}^{d,l-1})B^{T}h_{\upsilon,i+1}^{d,l}(Y)\\
		 - (I_{d}+\Delta A_{\theta}(\overline{X_{\upsilon}^{d,l-1}}(t_{i}),t_{i})^{T})R_{\upsilon,i+1}^{d,l}BG(R_{\upsilon,i+1}^{d,l})B^{T}\frac{\partial h_{\upsilon,i+1}^{d,l}}{\partial\upsilon}(Y)\\
		 =  \frac{\partial h_{\upsilon,i+1}^{d,l}}{\partial\upsilon}(Y)+\Delta\frac{\partial A_{\theta}}{\partial\upsilon}(\overline{X_{\upsilon}^{d,l-1}}(t_{i}),t_{i})^{T}h_{\upsilon,i+1}^{d,l}(Y)
		 \\+\Delta\frac{\partial A_{\theta}}{\partial x}(\overline{X_{\upsilon}^{d,l-1}}(t_{i}),t_{i})^{T}\frac{\partial\overline{X_{\upsilon}^{d,l-1}}(t_{i})}{\partial\upsilon}h_{\upsilon,i+1}^{d,l}(Y)
		 +  \Delta A_{\theta}(\overline{X_{\upsilon}^{d,l-1}}(t_{i}),t_{i})^{T}\frac{\partial h_{\upsilon,i+1}^{d,l}}{\partial\upsilon}(Y)\\
		 -  \triangle\frac{\partial R_{\upsilon,i+1}^{d,l}}{\partial\upsilon}BGU^{-1}B^{T}h_{\upsilon,i+1}^{d,l}(Y)-\Delta R_{\upsilon,i+1}^{d,l}BU^{-1}B^{T}\frac{\partial h_{\upsilon,i+1}^{d,l}}{\partial\upsilon}(Y)+O_{p,n}(\triangle^{2})
		\end{array}
		\]
		\[
		\begin{array}{l}
		\frac{\partial\overline{X_{\upsilon}^{d,l}}(t_{i+1})}{\partial\upsilon}\\
		 =  \frac{\partial\overline{X_{\upsilon}^{d,l}}(t_{i})}{\partial\upsilon}+\Delta\frac{\partial A_{\theta}}{\partial\upsilon}(\overline{X_{\upsilon}^{d,l-1}}(t_{i}),t_{i})\overline{X_{\upsilon}^{d,l}}(t_{i})+\Delta\frac{\partial A_{\theta}}{\partial x}(\overline{X_{\upsilon}^{d,l-1}}(t_{i}),t_{i})\frac{\partial\overline{X_{\upsilon}^{d,l-1}}(t_{i})}{\partial\upsilon}\overline{X_{\upsilon}^{d,l}}(t_{i})\\
		 +  \Delta A_{\theta}(\overline{X_{\upsilon}^{d,l-1}}(t_{i}),t_{i})\frac{\partial\overline{X_{\upsilon}^{d,l}}(t_{i})}{\partial\upsilon}\\
		 - \triangle BU^{-1}B^{T}\left(\frac{\partial R_{\upsilon,i+1}^{d,l}}{\partial\upsilon}\overline{X_{\upsilon}^{d,l}}(t_{i})+R_{\upsilon}^{d,l}(t_{i+1})\frac{\partial\overline{X_{\upsilon}^{d,l}}(t_{i})}{\partial\upsilon}+\frac{\partial h_{\upsilon}^{d,l}(t_{i+1})}{\partial\upsilon}\right)+O_{p,n}(\triangle^{2}).
		\end{array}
		\]
	}with $\frac{\partial G}{\partial\upsilon}(R_{\upsilon,i+1}^{d})=-G(R_{\upsilon,i+1}^{d})\triangle B^{T}\frac{\partial R_{\upsilon,i+1}^{d}}{\partial\upsilon}BG(R_{\upsilon,i+1}^{d})=O_{p,n}(\triangle)$.
	By making the difference between these last equations, we obtain:{\small{}
		\[
		\begin{array}{l}
		\frac{\partial R_{\upsilon}^{l}(t_{i})}{\partial\upsilon}-\frac{\partial R_{\upsilon,i+1}^{d,l}}{\partial\upsilon}=\frac{\partial R_{\upsilon}^{l}(t_{i+1})}{\partial\upsilon}-\frac{\partial R_{\upsilon,i+1}^{d,l}}{\partial\upsilon}\\
		+\triangle\frac{\partial A_{\theta}}{\partial\upsilon}(\overline{X_{\upsilon}^{l-1}}(t_{i}),t_{i})^{T}\left(R_{\upsilon}^{l}(t_{i+1})-R_{\upsilon,i+1}^{d,l}\right)+\left(\frac{\partial A_{\theta}}{\partial\upsilon}(\overline{X_{\upsilon}^{l-1}}(t_{i}),t_{i})-\frac{\partial A_{\theta}}{\partial\upsilon}(\overline{X_{\upsilon}^{d,l-1}}(t_{i}),t_{i}\right){}^{T}R_{\upsilon,i+1}^{d,l}\\
		+\triangle\frac{\partial A_{\theta}}{\partial x}(\overline{X_{\upsilon}^{l-1}}(t_{i}),t_{i})^{T}\left(\frac{\partial\overline{X_{\upsilon}^{l-1}}(t_{i})}{\partial\upsilon}\left(R_{\upsilon}^{l}(t_{i+1})-R_{\upsilon,i+1}^{d,l}\right)+\left(\frac{\partial\overline{X_{\upsilon}^{l-1}}(t_{i})}{\partial\upsilon}-\frac{\partial\overline{X_{\upsilon}^{d,l-1}}(t_{i})}{\partial\upsilon}\right)R_{\upsilon,i+1}^{d,l}\right)\\
		+\triangle\left(\frac{\partial A_{\theta}}{\partial x}(\overline{X_{\upsilon}^{l-1}}(t_{i}),t_{i})-\frac{\partial A_{\theta}}{\partial x}(\overline{X_{\upsilon}^{d,l-1}}(t_{i}),t_{i})\right)^{T}\frac{\partial\overline{X_{\upsilon}^{d,l-1}}(t_{i})}{\partial\upsilon}R_{\upsilon,i+1}^{d,l}\\
		+\triangle A_{\theta}(\overline{X_{\upsilon}^{l-1}}(t_{i}),t_{i})^{T}\left(\frac{\partial R_{\upsilon}^{l}(t_{i+1})}{\partial\upsilon}-\frac{\partial R_{\upsilon,i+1}^{d,l}}{\partial\upsilon}\right)+\triangle\left(A_{\theta}(\overline{X_{\upsilon}^{l-1}}(t_{i}),t_{i})-A_{\theta}(\overline{X_{\upsilon}^{d,l-1}}(t_{i}),t_{i})\right)^{T}\frac{\partial R_{\upsilon,i+1}^{d,l}}{\partial\upsilon}\\
		+\triangle\frac{\partial R_{\upsilon}^{l}(t_{i+1})}{\partial\upsilon}\left(A_{\theta}(\overline{X_{\upsilon}^{l-1}}(t_{i}),t_{i})-A_{\theta}(\overline{X_{\upsilon}^{d,l-1}}(t_{i}),t_{i})\right)+\triangle\left(\frac{\partial R_{\upsilon}^{l}(t_{i+1})}{\partial\upsilon}-\frac{\partial R_{\upsilon,i+1}^{d,l}}{\partial\upsilon}\right)A_{\theta}(\overline{X_{\upsilon}^{d,l-1}}(t_{i}),t_{i})\\
		+\triangle R_{\upsilon}^{l}(t_{i+1})\left(\frac{\partial A_{\theta}}{\partial\upsilon}(\overline{X_{\upsilon}^{l-1}}(t_{i}),t_{i})-\frac{\partial A_{\theta}}{\partial\upsilon}(\overline{X_{\upsilon}^{d,l-1}}(t_{i}),t_{i})\right)+\triangle\left(R_{\upsilon}^{l}(t_{i+1})-R_{\upsilon,i+1}^{d,l}\right)\frac{\partial A_{\theta}}{\partial\upsilon}(\overline{X_{\upsilon}^{d,l-1}}(t_{i}),t_{i})\\

		+\triangle R_{\upsilon}^{l}(t_{i+1})\left(A_{\theta}(\overline{X_{\upsilon}^{l-1}}(t_{i}),t_{i})-A_{\theta}(\overline{X_{\upsilon}^{d,l-1}}(t_{i}),t_{i})\right)\frac{\partial\overline{X_{\upsilon}^{d,l-1}}(t_{i})}{\partial\theta}\\

		+\triangle\left(R_{\upsilon}^{l}(t_{i+1})-R_{\upsilon,i+1}^{d,l}\right)A_{\theta}(\overline{X_{\upsilon}^{d,l-1}}(t_{i}),t_{i})\frac{\partial\overline{X_{\upsilon}^{d,l-1}}(t_{i})}{\partial\theta}\\

		+\triangle\frac{\partial R_{\upsilon}^{l}(t_{i+1})}{\partial\upsilon}BU^{-1}B^{T}\left(R_{\upsilon,i+1}^{d,l}-R_{\upsilon}^{l}(t_{i+1})\right)+\Delta\left(\frac{\partial R_{\upsilon,i+1}^{d,l}}{\partial\upsilon}-\frac{\partial R_{\upsilon}^{l}(t_{i+1})}{\partial\upsilon}\right)BU^{-1}B^{T}R_{\upsilon,i+1}^{d,l}\\
		-\triangle R_{\upsilon}^{l}(t_{i+1})BU^{-1}B^{T}\left(\frac{\partial R_{\upsilon}^{l}(t_{i+1})}{\partial\upsilon}-\frac{\partial R_{\upsilon,i+1}^{d,l}}{\partial\upsilon}\right)+\Delta\left(R_{\upsilon,i+1}^{d,l}-R_{\upsilon}^{l}(t_{i+1})\right)BU^{-1}B^{T}\frac{\partial R_{\upsilon,i+1}^{d,l}}{\partial\upsilon}
\\+O_{p,n}(\triangle^{2})
		\end{array}
		\]
		\[
		\begin{array}{l}
		
		\frac{\partial h_{\upsilon}^{l}(t_{i})}{\partial\upsilon}-\frac{\partial h_{\upsilon,i}^{d,l}}{\partial\upsilon}(Y)=\frac{\partial h_{\upsilon}^{l}(t_{i+1})}{\partial\upsilon}-\frac{\partial h_{\upsilon,i+1}^{d,l}}{\partial\upsilon}(Y)\\
		+\triangle\frac{\partial A_{\theta}}{\partial\upsilon}(\overline{X_{\upsilon}^{l-1}}(t_{i}),t_{i})^{T}\left(h_{\upsilon}^{l}(t_{i+1})-h_{\upsilon,i+1}^{d,l}(Y)\right)
\\+\Delta\left(\frac{\partial A_{\theta}}{\partial\upsilon}(\overline{X_{\upsilon}^{l-1}}(t_{i}),t_{i})-\frac{\partial A_{\theta}}{\partial\upsilon}(\overline{X_{\upsilon}^{d,l-1}}(t_{i}),t_{i})\right)h_{\upsilon,i+1}^{d,l}(Y)\\
		+\triangle\frac{\partial A_{\theta}}{\partial x}(\overline{X_{\upsilon}^{l-1}}(t_{i}),t_{i})^{T}\left(\frac{\partial\overline{X_{\upsilon}^{l-1}}(t_{i})}{\partial\upsilon}\left(h_{\upsilon}^{l}(t_{i+1})-h_{\upsilon,i+1}^{d,l}(Y)\right)+\left(\frac{\partial\overline{X_{\upsilon}^{l-1}}(t_{i})}{\partial\upsilon}-\frac{\partial\overline{X_{\upsilon}^{d,l-1}}(t_{i})}{\partial\upsilon}\right)\right)h_{\upsilon,i+1}^{d,l}(Y)\\
		+\triangle\left(\frac{\partial A_{\theta}}{\partial x}(\overline{X_{\upsilon}^{l-1}}(t_{i}),t_{i})-\frac{\partial A_{\theta}}{\partial x}(\overline{X_{\upsilon}^{d,l-1}}(t_{i}),t_{i})\right)^{T}\frac{\partial\overline{X_{\upsilon}^{d,l-1}}(t_{i})}{\partial\upsilon}h_{\upsilon,i+1}^{d,l}(Y)\\
		+\triangle\left(A_{\theta}(\overline{X_{\upsilon}^{l-1}}(t_{i}),t_{i})-A_{\theta}(\overline{X_{\upsilon}^{d,l-1}}(t_{i}),t_{i})\right)\frac{\partial h_{\upsilon}^{l}(t_{i+1})}{\partial\upsilon}\\
+\Delta A_{\theta}(\overline{X_{\upsilon}^{d,l-1}}(t_{i}),t_{i})^{T}\left(\frac{\partial h_{\upsilon}^{l}(t_{i+1})}{\partial\upsilon}-\frac{\partial h_{\upsilon,i+1}^{d,l}(Y)}{\partial v}\right)\\
		+\triangle\frac{\partial R_{\upsilon}^{l}(t_{i+1})}{\partial\upsilon}BU^{-1}B^{T}\left(h_{\upsilon}^{l}(t_{i+1})-h_{\upsilon,i+1}^{d,l}(Y)\right)-\triangle\left(\frac{\partial R_{\upsilon,i+1}^{d,l}}{\partial\upsilon}-\frac{\partial R_{\upsilon}^{l}(t_{i+1})}{\partial\upsilon}\right)BU^{-1}B^{T}h_{\upsilon,i+1}^{d,l}(Y)\\
		+\triangle R_{\upsilon}^{l}(t_{i+1})BU^{-1}B^{T}\left(\frac{\partial h_{\upsilon}^{l}(t_{i+1})}{\partial\upsilon}-\frac{\partial h_{\upsilon,i+1}^{d,l}}{\partial\upsilon}(Y)\right)+\triangle\left(R_{\upsilon}^{l}(t_{i+1})-R_{\upsilon,i+1}^{d,l}\right)BU^{-1}B^{T}\frac{\partial h_{\upsilon,i+1}^{d,l}}{\partial\upsilon}(Y)\\
+O_{p,n}(\triangle^{2})
		\end{array}
		\]
		\[
		\begin{array}{l}
		\frac{\partial\overline{X_{\upsilon}^{l}}(t_{i+1})}{\partial\upsilon}-\frac{\partial\overline{X_{\upsilon}^{d,l}}(t_{i+1})}{\partial\upsilon}=\frac{\partial\overline{X_{\upsilon}^{l}}(t_{i})}{\partial\upsilon}-\frac{\partial\overline{X_{\upsilon}^{d,l}}(t_{i})}{\partial\upsilon}\\
		+\Delta\frac{\partial A_{\theta}}{\partial\upsilon}(\overline{X_{\upsilon}^{l-1}}(t_{i}),t_{i})\left(\overline{X_{\upsilon}^{l}}(t_{i})-\overline{X_{\upsilon}^{d,l}}(t_{i})\right)\\
+\left(\frac{\partial A_{\theta}}{\partial\upsilon}(\overline{X_{\upsilon}^{l-1}}(t_{i}),t_{i}))-\Delta\frac{\partial A_{\theta}}{\partial\upsilon}(\overline{X_{\upsilon}^{d,l-1}}(t_{i}),t_{i})\right)\overline{X_{\upsilon}^{d,l}}(t_{i})\\
		+\Delta\frac{\partial A_{\theta}}{\partial x}(\overline{X_{\upsilon}^{l-1}}(t_{i}),t_{i})\left(\frac{\partial\overline{X_{\upsilon}^{l-1}}(t_{i})}{\partial\upsilon}\left(\overline{X_{\upsilon}^{l}}(t_{i})-\overline{X_{\upsilon}^{d,l}}(t_{i})\right)+\left(\frac{\partial\overline{X_{\upsilon}^{l-1}}(t_{i})}{\partial\upsilon}-\frac{\partial\overline{X_{\upsilon}^{d,l-1}}(t_{i})}{\partial\upsilon}\right)\overline{X_{\upsilon}^{d,l}}(t_{i})\right)\\
		+\Delta\left(\frac{\partial A_{\theta}}{\partial x}(\overline{X_{\upsilon}^{l-1}}(t_{i}),t_{i})-\frac{\partial A_{\theta}}{\partial x}(\overline{X_{\upsilon}^{d,l-1}}(t_{i}),t_{i})\right)\frac{\partial\overline{X_{\upsilon}^{d,l-1}}(t_{i})}{\partial\upsilon}\overline{X_{\upsilon}^{d,l}}(t_{i})\\
		+\Delta\frac{\partial A_{\theta}}{\partial x}(\overline{X_{\upsilon}^{l-1}}(t_{i}),t_{i})\left(\frac{\partial\overline{X_{\upsilon}^{l}}(t_{i})}{\partial\upsilon}-\frac{\partial\overline{X_{\upsilon}^{d,l}}(t_{i})}{\partial v}\right)\\
+\Delta\left(A_{\theta}(\overline{X_{\upsilon}^{l-1}}(t_{i}),t_{i})-A_{\theta}(\overline{X_{\upsilon}^{d,l-1}}(t_{i}),t_{i})\right)\frac{\partial\overline{X_{\upsilon}^{d,l}}(t_{i})}{\partial\upsilon}\\
		-\triangle BU^{-1}B^{T}\left(\frac{\partial R_{\upsilon}^{l}(t_{i+1})}{\partial\upsilon}\left(\overline{X_{\upsilon}^{l}}(t_{i})-\overline{X_{\upsilon}^{d,l}}(t_{i}\right))+\left(\frac{\partial R_{\upsilon}^{l}(t_{i+1})}{\partial\upsilon}-\frac{\partial R_{\upsilon,i+1}^{d,l}}{\partial\upsilon}\right)\overline{X_{\upsilon}^{d,l}}(t_{i})\right)\\
		-\triangle BU^{-1}B^{T}\left(R_{\upsilon}^{l}(t_{i})\left(\frac{\partial\overline{X_{\upsilon}^{l}}(t_{i})}{\partial\upsilon}-\frac{\partial\overline{X_{\upsilon}^{d,l}}(t_{i})}{\partial\upsilon}\right)+\left(R_{\upsilon}^{l}(t_{i})-R_{\upsilon}^{d,l}(t_{i+1})\right)\frac{\partial\overline{X_{\upsilon}^{d,l}}(t_{i})}{\partial\upsilon}\right)\\
		-\triangle BU^{-1}B^{T}\left(\frac{\partial h_{\upsilon}^{l}(t_{i+1})}{\partial\upsilon}-\frac{\partial h_{\upsilon}^{d,l}(t_{i+1})}{\partial\upsilon}\right)+O_{p,n}(\triangle^{2}).
		\end{array}
		\]
	}From this, we can see that $\frac{\partial R_{\upsilon,i}^{d,l}}{\partial\upsilon}=O_{n}(1)$
	and $\frac{\partial h_{\upsilon,i}^{d,l}(Y)}{\partial\upsilon}=O_{_{p,n}}(1)$
	for all $i\in\left\llbracket 0,\,n\right\rrbracket $ and $\upsilon\in\varUpsilon$.
	By using that $\left\Vert R_{\upsilon,i}^{d,l}(Y)-R_{\upsilon}^{l}(t_{i})\right\Vert _{2}=o_{p_{,n}}(1)$,
	$\left\Vert h_{\upsilon,i}^{d,l}(Y)-h_{\upsilon}^{l}(t_{i})\right\Vert _{2}=o_{p,n}(1)$
	and $\left\Vert \overline{X_{\upsilon}^{l}}(t_{i})-\overline{X_{\upsilon}^{l}}(t_{i})\right\Vert _{2}=o_{p,n}(1)$,
	we can simplify the previous equations{\small{}:
		\[
		\begin{array}{l}
		\frac{\partial R_{\upsilon}^{l}(t_{i})}{\partial\upsilon}-\frac{\partial R_{\upsilon,i+1}^{d,l}}{\partial\upsilon}    \\
=\left(1+2\triangle A_{\upsilon}(\overline{X_{\upsilon}^{d,l-1}}(t_{i}),t_{i})-2\Delta R_{v}^{l}(t_{i+1})BU^{-1}B^{T}\right)\left(\frac{\partial R_{\upsilon}^{l}(t_{i+1})}{\partial\upsilon}-\frac{\partial R_{\upsilon,i+1}^{d,l}}{\partial\upsilon}\right)\\
		 +  \triangle A_{\theta}(\overline{X_{\upsilon}^{l-1}}(t_{i}),t_{i})^{T}\left(\frac{\partial\overline{X_{\upsilon}^{l-1}}(t_{i})}{\partial\upsilon}-\frac{\partial\overline{X_{\upsilon}^{d,l-1}}(t_{i})}{\partial\upsilon}\right)R_{\upsilon,i+1}^{d,l}+O_{p,n}(\triangle)
\end{array}
		\]
\[
		\begin{array}{l}
		\frac{\partial h_{\upsilon}^{l}(t_{i})}{\partial\upsilon}-\frac{\partial h_{\upsilon,i}^{d,l}(Y)}{\partial\upsilon} \\
= \left(I_{p}+\Delta\left(A_{\theta}(\overline{X_{\upsilon}^{d,l-1}}(t_{i}),t_{i})^{T}-R_{\upsilon}^{l}(t_{i+1})BU^{-1}B^{T}\right)\right)\left(\frac{\partial h_{\upsilon}^{l}(t_{i+1})}{\partial\upsilon}-\frac{\partial h_{\upsilon,i+1}^{d,l}(Y)}{\partial\upsilon}\right)\\
		 +  \triangle A_{\theta}(\overline{X_{\upsilon}^{l-1}}(t_{i}),t_{i})^{T}\left(\frac{\partial\overline{X_{\upsilon}^{l-1}}(t_{i})}{\partial\upsilon}-\frac{\partial\overline{X_{\upsilon}^{d,l-1}}(t_{i})}{\partial\upsilon}\right)h_{\upsilon,i+1}^{d,l}(Y)\\
		 -  \triangle\left(\frac{\partial R_{\upsilon,i+1}^{d,l}}{\partial\upsilon}-\frac{\partial R_{\upsilon}^{l}(t_{i+1})}{\partial\upsilon}\right)BU^{-1}B^{T}h_{\upsilon,i+1}^{d,l}(Y)+O_{p,n}(\triangle)\\
\end{array}
		\]
\[
		\begin{array}{l}
		\frac{\partial\overline{X_{\upsilon}^{l}}(t_{i+1})}{\partial\upsilon}-\frac{\partial\overline{X_{\upsilon}^{d,l}}(t_{i+1})}{\partial\upsilon}  \\

=  \left(1+2\Delta A_{\theta}(\overline{X_{\upsilon}^{l-1}}(t_{i}),t_{i})\overline{X_{\upsilon}^{d,l}}(t_{i})+\Delta BU^{-1}B^{T}R_{\upsilon}^{l}(t_{i})\right)\left(\frac{\partial\overline{X_{\upsilon}^{l}}(t_{i})}{\partial\upsilon}-\frac{\partial\overline{X_{\upsilon}^{d,l}}(t_{i})}{\partial\upsilon}\right)\\
		 - \triangle BU^{-1}B^{T}\left(\left(\frac{\partial R_{v}^{l}(t_{i+1})}{\partial\upsilon}-\frac{\partial R_{\upsilon,i+1}^{d,l}}{\partial\upsilon}\right)\overline{X_{\upsilon}^{d,l}}(t_{i})+\left(\frac{\partial h_{\upsilon}^{l}(t_{i+1})}{\partial\upsilon}-\frac{\partial h_{\upsilon,i+1}^{d,l}(Y)}{\partial\upsilon}\right)\right)+O_{p,n}(\triangle).
		\end{array}
		\]
	}and from this, derive the following inequalities{\small{}:
		\[
		\begin{array}{l}
		\left\Vert \frac{\partial R_{\upsilon}^{l}(t_{i})}{\partial\upsilon}-\frac{\partial R_{\upsilon,i+1}^{d,l}}{\partial v}\right\Vert _{2}  \leq  \left(1+\Delta2\overline{A}+\Delta\left\Vert R_{\upsilon}(t_{i+1})\right\Vert _{2}\left\Vert BU^{-1}B^{T}\right\Vert _{2}\right)\left\Vert \frac{\partial R_{v}^{l}(t_{i+1})}{\partial\upsilon}-\frac{\partial R_{\upsilon,i+1}^{d,l}}{\partial\upsilon}\right\Vert _{2}\\
		 + \triangle\overline{A}\left\Vert R_{\upsilon,i+1}^{d,l}\right\Vert _{2}\left\Vert \frac{\partial\overline{X_{\upsilon}^{l-1}}(t_{i})}{\partial\upsilon}-\frac{\partial\overline{X_{\upsilon}^{d,l-1}}(t_{i})}{\partial\upsilon}\right\Vert _{2}+O_{p,n}(\Delta)\\
		 =  \left(1+O_{p,n}(\triangle)\right)\left\Vert \frac{\partial R_{\upsilon}^{l}(t_{i+1})}{\partial\upsilon}-\frac{\partial R_{\upsilon,i+1}^{d,l}}{\partial\upsilon}\right\Vert _{2}+O_{p,n}(\triangle)\left\Vert \frac{\partial\overline{X_{\upsilon}^{l-1}}(t_{i})}{\partial\upsilon}-\frac{\partial\overline{X_{\upsilon}^{d,l-1}}(t_{i})}{\partial\upsilon}\right\Vert +O_{p,n}(\Delta)\\
		\end{array}
		\]
\[
\begin{array}{l}
		\left\Vert \frac{\partial h_{\upsilon}^{l}(t_{i})}{\partial\upsilon}-\frac{\partial h_{\upsilon,i}^{d,l}(Y)}{\partial\upsilon}\right\Vert _{2}  \leq  \left(1+\Delta\overline{A}+\Delta\left\Vert R_{\upsilon}^{l}(t_{i+1})\right\Vert _{2}\left\Vert BU^{-1}B^{T}\right\Vert _{2}\right)\left\Vert \frac{\partial h_{\upsilon}^{l}(t_{i+1})}{\partial\upsilon}-\frac{\partial h_{\upsilon,i+1}^{d,l}(Y)}{\partial\upsilon}\right\Vert _{2}\\
		 + \Delta\overline{A}\left\Vert h_{\upsilon,i+1}^{d,l}(Y)\right\Vert _{2}\left\Vert \frac{\partial\overline{X_{\upsilon}^{l-1}}(t_{i})}{\partial\upsilon}-\frac{\partial\overline{X_{\upsilon}^{d,l-1}}(t_{i})}{\partial\upsilon}\right\Vert _{2}+O_{p,n}(\triangle)\\
		 +  \Delta\left\Vert h_{\upsilon,i+1}^{d,l}(Y)\right\Vert _{2}\left\Vert BU^{-1}B^{T}\right\Vert _{2}\left\Vert \frac{\partial R_{\upsilon,i+1}^{d,l}}{\partial\upsilon}-\frac{\partial R_{\upsilon}^{l}(t_{i+1})}{\partial\upsilon}\right\Vert _{2}+O_{p,n}(\triangle)\\
		 =  (1+O_{p,n}(\triangle))\left\Vert \frac{\partial h_{\upsilon}^{l}(t_{i+1})}{\partial\upsilon}-\frac{\partial h_{\upsilon,i+1}^{d,l}(Y)}{\partial\upsilon}\right\Vert _{2}+O_{p,n}(\triangle)\left\Vert \frac{\partial\overline{X_{\upsilon}^{l-1}}(t_{i})}{\partial\upsilon}-\frac{\partial\overline{X_{\upsilon}^{d,l-1}}(t_{i})}{\partial\upsilon}\right\Vert _{2}\\
		 + O_{p,n}(\triangle)\left\Vert \frac{\partial R_{\upsilon,i+1}^{d,l}}{\partial\upsilon}-\frac{\partial R_{\upsilon}^{l}(t_{i+1})}{\partial\upsilon}\right\Vert +O_{p,n}(\triangle)\\
\end{array}
		\]
\[
\begin{array}{l}
		\left\Vert \frac{\partial\overline{X_{\upsilon}^{l}}(t_{i+1})}{\partial\upsilon}-\frac{\partial\overline{X_{\upsilon}^{d,l}}(t_{i+1})}{\partial\upsilon}\right\Vert _{2} \\

		 \leq  \left(1+\Delta2\overline{A}\left\Vert \overline{X_{\upsilon}^{d,l}}(t_{i})\right\Vert _{2}+\Delta\left\Vert R_{\upsilon}^{l}(t_{i})\right\Vert _{2}\left\Vert BU^{-1}B^{T}\right\Vert _{2}\right)\left\Vert \frac{\partial\overline{X_{\upsilon}^{l}}(t_{i})}{\partial\upsilon}-\frac{\partial\overline{X_{\upsilon}^{d,l}}(t_{i})}{\partial v}\right\Vert _{2}\\

		 +  \Delta\left\Vert BU^{-1}B^{T}\right\Vert _{2}\left\Vert \overline{X_{\upsilon}^{d,l}}(t_{i})\right\Vert _{2}\left\Vert \frac{\partial R_{\upsilon}^{l}(t_{i+1})}{\partial\upsilon}-\frac{\partial R_{\upsilon,i+1}^{d,l}}{\partial\upsilon}\right\Vert _{2}+O_{p,n}(\triangle)\\
		 +  \Delta\left\Vert BU^{-1}B^{T}\right\Vert _{2}\left\Vert \frac{\partial h_{\upsilon}^{l}(t_{i+1})}{\partial\upsilon}-\frac{\partial h_{\upsilon,i+1}^{d,l}(Y)}{\partial\upsilon}\right\Vert _{2}+O_{p,n}(\triangle)\\
		  =  (1+O_{p,n}(\triangle))\left\Vert \frac{\partial\overline{X_{\upsilon}^{l}}(t_{i})}{\partial\upsilon}-\frac{\partial\overline{X_{\upsilon}^{d,l}}(t_{i})}{\partial\upsilon}\right\Vert _{2}\\
		 +  O_{p,n}(\triangle)\left(\left\Vert \frac{\partial R_{\upsilon}^{l}(t_{i+1})}{\partial v}-\frac{\partial R_{\upsilon,i+1}^{d,l}}{\partial\upsilon}\right\Vert _{2}+\left\Vert \frac{\partial h_{\upsilon}^{l}(t_{i+1})}{\partial\upsilon}-\frac{\partial h_{\upsilon,i+1}^{d,l}(Y)}{\partial\upsilon}\right\Vert _{2}\right)+O_{p,n}(\triangle).
		\end{array}
		\]
	}Because $\left\Vert \frac{\partial\overline{X_{\upsilon}^{0}}(t_{i})}{\partial\upsilon}-\frac{\partial\overline{X_{\upsilon}^{d,0}}(t_{i})}{\partial\upsilon}\right\Vert _{2}=0$
	, we can prove again by induction, by using the previous inequalities
	and the discrete Gronwall lemma \ref{lem:disc_Gronwall_lemma} that
	$\left\Vert \frac{\partial R_{\upsilon}^{l}(t_{i})}{\partial\upsilon}-\frac{\partial R_{\upsilon,i}^{d,l}}{\partial\upsilon}\right\Vert _{2}=o_{p,n}(1)$,
	$\left\Vert \frac{\partial h_{\upsilon}^{l}(t_{i})}{\partial\upsilon}-\frac{\partial h_{\upsilon,i}^{d,l}(Y)}{\partial\upsilon}\right\Vert _{2}=o_{p,n}(1)$
	and $\left\Vert \frac{\partial\overline{X_{\upsilon}^{l}}(t_{i})}{\partial\upsilon}-\frac{\partial\overline{X_{\upsilon}^{d,l}}(t_{i})}{\partial\upsilon}\right\Vert _{2}=o_{p,n}(1)$.
	Despite the formal computation, there are no theoretical difficulties
	to derive under condition C8 that we can differentiate again the sensitivity
	equation and obtain $\left\Vert \frac{\partial^{2}R_{\upsilon}^{l}(t_{i})}{\partial^{2}\upsilon}-\frac{\partial^{2}R_{\upsilon,i}^{d,l}}{\partial^{2}\upsilon}\right\Vert _{2}=o_{p,n}(1),\left\Vert \frac{\partial^{2}h_{\upsilon}^{l}(t_{i})}{\partial^{2}\upsilon}-\frac{\partial^{2}h_{\upsilon,i}^{d,l}(Y)}{\partial^{2}\upsilon}\right\Vert _{2}=o_{p,n}(1)$
	and $\left\Vert \frac{\partial^{2}\overline{X_{\upsilon}^{l}}(t_{i})}{\partial^{2}\upsilon}-\frac{\partial^{2}\overline{X_{\upsilon}^{d,l}}(t_{i})}{\partial^{2}\upsilon}\right\Vert _{2}=o_{p,n}(1)$. 
\end{proof}
\begin{lem}
	\label{lem:Riccati_optimal_traj_uniform_convergence}Under conditions
	C1 to C3, the uniform convergence of $\overline{X_{\upsilon}^{l}}$
	for each $\upsilon\in\varUpsilon$ leads to $\sup_{\upsilon\in\varUpsilon}\left\Vert R_{\upsilon}^{l}-R_{\upsilon}^{\infty}\right\Vert _{L^{2}}\longrightarrow0$,
	$\sup_{\upsilon\in\varUpsilon}\left\Vert h_{\upsilon}^{l}-h_{\upsilon}^{\infty}\right\Vert _{L^{2}}\longrightarrow0$
	and $\sup_{\upsilon\in\varUpsilon}\left\Vert \overline{X_{\upsilon}^{l}}-\overline{X_{\upsilon}^{\infty}}\right\Vert _{L^{2}}\longrightarrow0$
	when $l\longrightarrow\infty$.
\end{lem}
\begin{proof}
	Let us consider $R_{\upsilon}^{l}$, $h_{\upsilon}^{l}$, $\overline{X_{\upsilon}^{l}}$
	, and $R_{\upsilon}^{l'}$, $h_{\upsilon}^{l'}$ $\overline{X_{\upsilon}^{l'}}$
	, respectively the solutions of (\ref{eq:continous_accurate_Riccati_equation})
	for a given $l\in\mathbb{N}$ and $l'\in\mathbb{N}$. By making the
	difference of the ODEs ruling the reversed time functions $\widetilde{R_{\upsilon}^{l}}:=R_{\upsilon}^{l}(T-.)$,
	$\widetilde{h_{\upsilon}^{l}}:=h_{\upsilon}^{l}(T-.)$, $\widetilde{X_{\upsilon}^{l}}:=\overline{X_{\upsilon}^{l}}(T-.)$
	and $\widetilde{R_{\upsilon}^{l'}}:=R_{\upsilon}^{l'}(T-.)$, $\widetilde{h_{\upsilon}^{l'}}:=h_{\upsilon}^{l'}(T-.)$,
	$\widetilde{X_{\upsilon}^{l'}}:=\overline{X_{\upsilon}^{l'}}(T-.)$
	we obtain:{\small{}
		\[
		\begin{array}{l}
		\dot{\widetilde{R_{\upsilon}^{l}}}(t)-\dot{\widetilde{R_{\upsilon}^{l^{'}}}}(t)\\
 =  \widetilde{A_{\theta}}(\widetilde{X_{\upsilon}^{l'-1}}(t),t)^{T}\left(\widetilde{R_{\upsilon}^{l'}}(t)-\widetilde{R_{\upsilon}^{l}}(t)\right)+\left(\widetilde{A_{\theta}}(\widetilde{X_{\upsilon}^{l'-1}}(t),t)-\widetilde{A_{\theta}}(\widetilde{X_{\upsilon}^{l-1}}(t),t)\right)^{T}\widetilde{R_{\upsilon}^{l}}(t)\\
	 +  \widetilde{R_{\upsilon}^{l'}}(t)\left(\widetilde{A_{\theta}}(\widetilde{X_{\upsilon}^{l'-1}}(t),t)-\widetilde{A_{\theta}}(\widetilde{X_{\upsilon}^{l-1}}(t),t)\right)+\left(\widetilde{R_{\upsilon}^{l'}}(t)-\widetilde{R_{\upsilon}^{l}}(t)\right)\widetilde{A_{\theta}}(\widetilde{X_{\upsilon}^{l'-1}}(t),t)\\
		 +  \widetilde{R_{\upsilon}^{l}}(t)BU^{-1}B^{T}\left(\widetilde{R_{\upsilon}^{l'}}(t)-\widetilde{R_{v}^{l}}(t)\right)+\left(\widetilde{R_{\upsilon}^{l'}}(t)-\widetilde{R_{\upsilon}^{l}}(t)\right)BU^{-1}B^{T}\widetilde{R_{\upsilon}^{l'}}(t)\\
\end{array}
		\]
\[
		\begin{array}{l}
		\dot{\widetilde{h_{\upsilon}^{l}}}(t)-\dot{\widetilde{h_{\upsilon}^{l'}}}(t)  \\
= \widetilde{A_{\theta}}(\widetilde{X_{\upsilon}^{l'-1}}(t),t)^{T}\left(\widetilde{h_{\upsilon}^{l'}}(t)-\widetilde{h_{\upsilon}^{l}}(t)\right)+\left(\widetilde{A_{\theta}}(\widetilde{X_{\upsilon}^{l'-1}}(t),t)-\widetilde{A_{\theta}}(\widetilde{X_{\upsilon}^{l-1}}(t),t)\right)^{T}\widetilde{h_{\upsilon}^{l}}(t)\\
		 +  \widetilde{R_{\upsilon}^{l'}}(t)BU^{-1}B^{T}\left(\widetilde{h_{\upsilon}^{l'}}(t)-\widetilde{h_{\upsilon}^{l}}(t)\right)+\left(\widetilde{R_{\upsilon}^{l'}}(t)-\widetilde{R_{\upsilon}^{l}}(t)\right)BU^{-1}B^{T}\widetilde{h_{\upsilon}^{l'}}(t)
		\end{array}
		\]
		with $\widetilde{R_{\upsilon}^{l}}(0)-\widetilde{R_{\upsilon}^{l'}}(0)=0_{d,d}$
	}and $\widetilde{h_{\upsilon}^{l}}(0)-\widetilde{h_{\upsilon}^{l^{'}}}(0)=0_{d,1}$.
	Here by taking the norm and by using proposition \ref{prop:E_bounded_probability},
	we know $\left\Vert R_{\upsilon}^{l}(t)\right\Vert _{2},\,\left\Vert R_{\upsilon}^{l'}(t)\right\Vert _{2},\,\left\Vert h_{\upsilon}^{l}(t)\right\Vert _{2}$,
	$\left\Vert h_{\upsilon}^{l'}(t)\right\Vert _{2}$ are uniformly bounded
	on $\mathbb{N}\times\left[0,\,T\right]\times\varUpsilon$ and we obtain:
	{\small{}
		\[
		\begin{array}{lll}
		\frac{d}{dt}\left\Vert \widetilde{R_{\upsilon}^{l'}}(t)-\widetilde{R_{\upsilon}^{l}}(t)\right\Vert _{2} & \leq & O_{n}(1)\left\Vert \widetilde{R_{\upsilon}^{l'}}(t)-\widetilde{R_{\upsilon}^{l}}(t)\right\Vert _{2}+O_{n}(1)\left\Vert \widetilde{A_{\theta}}(\widetilde{X_{\upsilon}^{l'-1}}(t),t)-\widetilde{A_{\theta}}(\widetilde{X_{\upsilon}^{l-1}}(t),t)\right\Vert _{2}\\
		\frac{d}{dt}\left\Vert \widetilde{h_{\upsilon}^{l'}}(t)-\widetilde{h_{\upsilon}^{l}}(t)\right\Vert _{2} & \leq & O_{n}(1)\left\Vert \widetilde{h_{\upsilon}^{l'}}(t)-\widetilde{h_{\upsilon}^{l}}(t)\right\Vert _{2}\\
		& + & O_{n}(1)\left\Vert \widetilde{A_{\theta}}(\widetilde{X_{\upsilon}^{l'-1}}(t),t)-\widetilde{A_{\theta}}(\widetilde{X_{\upsilon}^{l-1}}(t),t)\right\Vert _{2}+O_{n}(1)\left\Vert \widetilde{R_{\upsilon}^{l'}}(t)-\widetilde{R_{\upsilon}^{l}}(t)\right\Vert _{2}.
		\end{array}
		\]
	}By using the continuous Gronwall lemma, we easily obtain from the
	first inequality $\left\Vert R_{\upsilon}^{l'}(t)-R_{\upsilon}^{l}(t)\right\Vert _{2}\leq O_{n}(1)\int_{0}^{t}\left\Vert \widetilde{A_{\theta}}(\widetilde{X_{\upsilon}^{l'-1}}(t),t)-\widetilde{A_{\theta}}(\widetilde{X_{\upsilon}^{l-1}}(t),t)\right\Vert _{2}^{2}dt.$
	Since $\overline{X_{\upsilon}^{l-1}}-\overline{X_{\upsilon}^{l'-1}}\longrightarrow0$
	on $C\left(\left[0,\,T\right],\,\mathbb{R}^{d}\right)$ and $\forall\theta\in\varTheta$,
	$(x,t)\longmapsto A_{\theta}(x,t)$ is continuous on $\varLambda\times\left[0,\,T\right]$,
	we have 
\[
\sup_{\upsilon\in\varUpsilon}\int_{0}^{T}\left\Vert A_{\theta}(\overline{X_{\upsilon}^{l'-1}}(t),t)-A_{\theta}(\overline{X_{\upsilon}^{l-1}}(t),t)\right\Vert _{2}^{2}dt\longrightarrow0
\]
	when $\left(l,l^{'}\right)\longrightarrow+\infty$ and so $\sup_{\upsilon\in\varUpsilon}\left\Vert R_{\upsilon}^{l'}-R_{\upsilon}^{l}\right\Vert _{L^{2}}^{2}\longrightarrow0$.
	From this limit, we derive the uniform convergence of the sequence
	$\left\{ R_{\upsilon}^{l}\right\} _{l\in\mathbb{N}}$ , and $\left\{ h_{\upsilon}^{l}\right\} _{l\in\mathbb{N}}$
	as well. Now let us control the difference $\overline{X_{\upsilon}^{l}}-\overline{X_{\upsilon}^{l'}}$,
	by integrating and taking the norm, we obtain{\small{}:
		\[
		\begin{array}{lll}
		\left\Vert \overline{X_{\upsilon}^{l}}(t)-\overline{X_{\upsilon}^{l'}}(t)\right\Vert _{2} & \leq & O_{n}(1)\int_{0}^{t}\left\Vert \overline{X_{\upsilon}^{l}}(s)-\overline{X_{\upsilon}^{l'}}(s)\right\Vert _{2}ds\\
		& + & O_{n}(1)\int_{0}^{t}\left\Vert R_{v}^{l}(s)-R_{\upsilon}^{l'}(s)\right\Vert _{2}ds+O_{n}(1)\int_{0}^{t}\left\Vert h_{\upsilon}^{l}(s)-h_{\upsilon}^{l'}(s)\right\Vert _{2}ds
		\end{array}
		\]
	}by taking the norm, using Gronwall lemma and the limits $\sup_{\upsilon\in\varUpsilon}\left\Vert R_{\upsilon}^{l'}-R_{\upsilon}^{l}\right\Vert _{L^{2}}^{2}\longrightarrow0$
	, $\sup_{\upsilon\in\varUpsilon}\left\Vert h_{\upsilon}^{l'}-h_{\upsilon}^{l}\right\Vert _{L^{2}}^{2}\longrightarrow0$
	we  conclude $\sup_{\upsilon\in\varUpsilon}\left\Vert \overline{X_{\upsilon}^{l}}-\overline{X_{\upsilon}^{l'}}\right\Vert _{L^{2}}^{2}\longrightarrow0.$
	
\end{proof}
\subsection{Asymptotic normality}

\begin{lem}
	\label{lem:deterministic_gradient_behavior}Under conditions C1 to
	C6, we have $h_{i}^{*d,l}=h^{*l}(t_{i})+O_{n}(\Delta)$, $\frac{\partial h_{i}^{*d,l}}{\partial\upsilon}=\frac{\partial h^{*l}(t_{i})}{\partial\upsilon}+O_{n}(\Delta)$
	and $\nabla_{\upsilon}S_{n}^{l}(Y^{d*};\upsilon^{*})-\nabla_{\upsilon}S^{l}(\upsilon^{*})=O_{n}(\triangle)$.
\end{lem}
\begin{proof}
	By formal computations similar as in lemma \ref{lem:h_R_continuity_limits},
	we obtain the finite difference equations ruling $R_{i}^{*d,l}-R^{*l}(t_{i})$,
	$h_{i}^{*d,l}-h^{*l}(t_{i})$, $X^{*d,l}(t_{i+1})-X^{*l}(t_{i+1})$:{\small{}
		\[
		\begin{array}{l}
		R_{i}^{*d,l}-R^{*l}(t_{i}) \\
 =  R_{i+1}^{*d,l}-R^{*l}(t_{i+1})+\Delta A_{\theta^{*}}(X^{*d,l-1}(t_{i}),t_{i})^{T}\left(R_{i+1}^{*d,l}-R^{*l}(t_{i+1})\right)\\
		 +  \Delta\left(A_{\theta^{*}}(X^{*d,l-1}(t_{i}),t_{i})-A_{\theta^{*}}(X^{*l-1}(t_{i}),t_{i})\right)^{T}R^{*l}(t_{i+1})\\
		 +  \Delta R_{i+1}^{*d,l}\left(A_{\theta^{*}}(X^{*d,l-1}(t_{i}),t_{i})-A_{\theta^{*}}(X^{*l-1}(t_{i}),t_{i})\right)A_{\theta^{*}}(X^{*l-1}(t_{i}),t_{i})\\
		 +  \Delta R_{i+1}^{*d,l}\left(R_{i+1}^{*d,l}-R^{*l}(t_{i+1})\right)A_{\theta^{*}}(X^{*l-1}(t_{i}),t_{i})\\

		 +  \triangle\left(R^{*l}(t_{i+1})BU^{-1}B^{T}\left(R^{*l}(t_{i+1})-R_{\theta,i+1}^{d,l}\right)+\left(R^{*l}(t_{i+1})-R_{\theta,i+1}^{d,l}\right)BU^{-1}B^{T}R_{\theta,i+1}^{d,l}\right)\\
		 +  O_{n}(\triangle^{2})
		\end{array}
		\]
		\[
		\begin{array}{l}
		h_{i}^{*d,l}-h^{*l}(t_{i})  \\
=  h_{i+1}^{*d,l}-h^{*l}(t_{i+1})+\Delta A_{\theta^{*}}(X^{*d,l-1}(t_{i}),t_{i})^{T}\left(h_{i+1}^{*d,l}-h^{*l}(t_{i+1})\right)\\
		 +  \Delta\left(A_{\theta^{*}}(X^{*d,l-1}(t_{i}),t_{i})-A_{\theta^{*}}(X^{*l-1}(t_{i}),t_{i})\right)^{T}h^{*l}(t_{i+1})\\
		 +  \triangle R^{*,l}(t)BU^{-1}B^{T}\left(h^{*l}(t_{i+1})-h_{i+1}^{*d,l}\right)+\triangle\left(R^{*l}(t)-R_{i+1}^{*d,l}\right)BU^{-1}B^{T}h_{i+1}^{*d,l}+O_{n}(\Delta^{2})
		\end{array}
		\]
		\[
		\begin{array}{l}
		X^{*d,l}(t_{i+1})-X^{*l}(t_{i+1}) \\
 = X^{*d,l}(t_{i})-X^{*l}(t_{i})+\Delta A_{\theta^{*}}(X^{*d,l-1}(t_{i}),t_{i})\left(X^{*d,l}(t_{i})-X^{*l}(t_{i})\right)\\
		 +  \Delta\left(A_{\theta^{*}}(X^{*d,l-1}(t_{i})t_{i})-A_{\theta^{*}}(X^{*l-1}(t_{i}),t_{i})\right)X^{*l}(t_{i})\\
		 +  \triangle BU^{-1}B^{T}\left(R_{i+1}^{*d,l}\left(X^{*d,l}(t_{i})-X^{*l}(t_{i})\right)+\left(R_{i+1}^{*d,l}-R^{*l}(t_{i+1})\right)X^{*l}(t_{i})\right)\\
		 +  \triangle BU^{-1}B^{T}\left(h_{i+1}^{*d,l}-h^{*l}(t_{i+1})\right)+O_{n}(\triangle^{2}).
		\end{array}
		\]
	}Now, let us prove by induction that $\left\Vert R_{i}^{*d,l}-R^{*l}(t_{i})\right\Vert _{2}=O_{n}(\Delta)$,
	$\left\Vert h_{i}^{*d,l}-h^{*l}(t_{i})\right\Vert _{2}=O_{n}(\Delta)$
	and $\left\Vert X^{*d,0}(t_{i})-X^{*,0}(t_{i})\right\Vert _{2}=O_{n}(\Delta)$.
	For initialisation, let us consider the case $l=1$, we have $\left\Vert X^{*d,l-1}(t_{i})-X^{*l-1}(t_{i})\right\Vert _{2}=\left\Vert x_{0}^{*}-x_{0}^{*}\right\Vert =0$
	for all $i\in\left\llbracket 0,\,n\right\rrbracket $. From this,
	equations ruling $R_{i}^{*d,1}-R^{*1}(t_{i})$ , $h_{i}^{*d,1}-h^{*1}(t_{i})$
	{\small{}and $X^{*d,1}(t_{i+1})-X^{*1}(t_{i+1})$ become: 
		\[
		\begin{array}{l}
		R_{i}^{*d,1}-R^{*1}(t_{i}) \\
 =  R_{i+1}^{*d,1}-R^{*1}(t_{i+1})+\Delta A_{\theta^{*}}(x_{0}^{*},t_{i})^{T}\left(R_{i+1}^{*d,1}-R^{*1}(t_{i+1})\right)\\
		 +  \Delta\left(R_{i+1}^{*d,1}-R^{*1}(t_{i+1})\right)A_{\theta^{*}}(x_{0}^{*},t_{i})+\triangle R^{*1}(t_{i+1})BU^{-1}B^{T}\left(R^{*1}(t_{i+1})-R_{\theta,i+1}^{d,1}\right)\\
		 +  \triangle\left(R^{*1}(t_{i+1})-R_{\theta,i+1}^{d,1}\right)BU^{-1}B^{T}R_{\theta,i+1}^{d,1}+O_{n}(\triangle^{2})
		\end{array}
		\]
		\[
		\begin{array}{l}
		h_{i}^{*d,1}-h^{*1}(t_{i})  \\
=  h_{i+1}^{*d,1}-h^{*1}(t_{i+1})+\Delta A_{\theta^{*}}(x_{0}^{*},t_{i})^{T}\left(h_{i+1}^{*d,1}-h^{*1}(t_{i+1})\right)\\
		 +  \triangle R^{*1}(t)BU^{-1}B^{T}\left(h^{*1}(t_{i+1})-h_{i+1}^{*d,1}\right)+\triangle\left(R^{*1}(t)-R_{i+1}^{*d,1}\right)BU^{-1}B^{T}h_{i+1}^{*d,1}+O_{n}(\Delta^{2})
		\end{array}
		\]
		\[
		\begin{array}{l}
		X^{*d,1}(t_{i+1})-X^{*1}(t_{i+1})  \\
=  X^{*d,1}(t_{i})-X^{*1}(t_{i})+\Delta A_{\theta^{*}}(x_{0}^{*},t_{i})\left(X^{*d,1}(t_{i})-X^{*1}(t_{i})\right)\\
		 +  \triangle BU^{-1}B^{T}\left(\left(R_{i+1}^{*d,1}-R^{*1}(t_{i+1})\right)X^{*1}(t_{i})+h_{i+1}^{*d,1}-h^{*1}(t_{i+1})\right)+O_{n}(\triangle^{2}).
		\end{array}
		\]
	}By using triangular inequality, we derive:{\small{}
		\[
		\begin{array}{lll}
		\left\Vert R_{i}^{*d,1}-R^{*1}(t_{i})\right\Vert _{2} & \leq & \left(1+O_{n}(\Delta)\right)\left\Vert R_{i+1}^{*d,1}-R^{*1}(t_{i+1})\right\Vert _{2}+O_{n}(\triangle^{2})\\
		\left\Vert h_{i}^{*d,1}-h^{*1}(t_{i})\right\Vert _{2} & \leq & \left(1+O_{n}(\Delta)\right)\left\Vert h_{i+1}^{*d,1}-h^{*1}(t_{i+1})\right\Vert _{2}+O_{n}(\Delta)\left\Vert R_{i+1}^{*d,1}-R^{*1}(t_{i+1})\right\Vert _{2}+O_{n}(\triangle^{2})\\
		\left\Vert X^{*d,1}(t_{i+1})-X^{*1}(t_{i+1})\right\Vert _{2} & \leq & \left(1+O_{n}(\Delta)\right)\left\Vert X^{*d,1}(t_{i})-X^{*1}(t_{i})\right\Vert _{2}\\
		& + &O_{n}(\Delta)\left\Vert R_{i+1}^{*d,1}-R^{*1}(t_{i+1})\right\Vert _{2}+ O_{n}(\Delta)\left\Vert h_{i+1}^{*d,1}-h^{*1}(t_{i+1})\right\Vert _{2}
		\end{array}
		\]
	}and since $\left\Vert R_{0}^{*d,1}-R^{*1}(0)\right\Vert _{2}=\left\Vert h_{0}^{*d,1}-h^{*1}(0)\right\Vert _{2}=0$,
	we can use discrete Gronwall lemma \ref{lem:disc_Gronwall_lemma}
	to derive $\left\Vert R_{i}^{*d,1}-R^{*1}(t_{i})\right\Vert _{2}=O_{n}(\Delta)$,
	$\left\Vert h_{i}^{*d,1}-h^{*1}(t_{i})\right\Vert _{2}=O_{n}(\Delta)$
	, $\left\Vert X^{*d,1}(t_{i})-X^{*1}(t_{i})\right\Vert _{2}=O_{n}(\Delta).$
	
	Now let us assume the property holds up to $l-1$, triangular inequality
	gives in the general case:{\small{}
		\[
		\begin{array}{l}
		\left\Vert R_{i}^{*d,l}-R^{*l}(t_{i})\right\Vert _{2} \\
 \leq  \left(1+2\Delta\overline{A}+\triangle\left(\left\Vert R_{i+1}^{*l}\right\Vert _{2}+\left\Vert R_{i+1}^{*d,l}\right\Vert _{2}\right)\left\Vert BU^{-1}B^{T}\right\Vert _{2}\right)\left\Vert R_{i+1}^{*d,l}-R^{*l}(t_{i+1})\right\Vert _{2}\\
		+  \Delta\left(\left\Vert R_{i+1}^{*d,l}\right\Vert _{2}+\left\Vert R^{*l}(t_{i+1})\right\Vert _{2}\right)\left\Vert A_{\theta^{*}}(X^{*d,l-1}(t_{i}),t_{i})-A_{\theta^{*}}(X^{*l-1}(t_{i}),t_{i})\right\Vert _{2}+O_{n}(\triangle^{2}).
		\end{array}
		\]
	}By using the induction hypothesis and $(x,t)\longmapsto A_{\theta^{*}}(x,t)$
	continuity, we have $\left\Vert A_{\theta^{*}}(X^{*d,l-1}(t_{i}),t_{i})-A_{\theta^{*}}(X^{*l-1}(t_{i}),t_{i})\right\Vert _{2}=O_{n}(\Delta)$
	and we can use again lemma \ref{lem:disc_Gronwall_lemma} to prove
	$\left\Vert R_{i}^{*d,l}-R^{*l}(t_{i})\right\Vert _{2}=O_{n}(\Delta)$,
	from this we derive: {\small{}
		\[
		\begin{array}{lll}
		\left\Vert h_{i}^{*d,l}-h^{*l}(t_{i})\right\Vert _{2} & \leq & \left(1+\Delta\overline{A}+\triangle\left\Vert R_{i+1}^{*l}\right\Vert _{2}\left\Vert BU^{-1}B^{T}\right\Vert _{2}\right)\left\Vert h^{*l}(t_{i+1})-h_{i+1}^{*d,l})\right\Vert _{2}\\
		& + & \Delta\left\Vert BU^{-1}B^{T}\right\Vert _{2}\left\Vert h_{i+1}^{*d,l}\right\Vert _{2}\left\Vert R^{*l}(t)-R_{i+1}^{*d,l}\right\Vert _{2}\\
		& + & \Delta\left\Vert h^{*l}(t_{i+1})\right\Vert _{2}\left\Vert A_{\theta^{*}}(X^{*d,l-1}(t_{i}),t_{i})-A_{\theta^{*}}(X^{*l-1}(t_{i}),t_{i})\right\Vert _{2}+O_{n}(\triangle^{2})
		\end{array}
		\]
	}which leads to $\left\Vert h_{i}^{*d,l}-h^{*l}(t_{i})\right\Vert _{2}=O_{n}(\Delta)$.
	As in the initialisation phase, we easily derive $\left\Vert X^{*d,l}(t_{i})-X^{*l}(t_{i})\right\Vert _{2}=O_{n}(\Delta)$
	and we can conclude the induction proof. Thanks to condition 7, we
	can derive the sensitivity equations of $R^{*d,l}-R^{*l}$, $h^{*d,l}-h^{*l}$
	and $X^{*d,l}-X^{*l-1}$:
	
	{\small{}
		\[
		\begin{array}{l}
		 \frac{\partial}{\partial\upsilon}\left(R_{i}^{*d,l}-R^{*l}(t_{i})\right)=\frac{\partial}{\partial\upsilon}\left(R_{i+1}^{*d,l}-R^{*l}(t_{i+1})\right)\\
		+\Delta\left(\frac{\partial A_{\theta^{*}}}{\partial v}(X^{*d,l-1}(t_{i}),t_{i})^{T}+\frac{\partial A_{\theta^{*}}}{\partial x}(X^{*d,l-1}(t_{i}),t_{i})\frac{\partial X^{*d,l-1}(t_{i})}{\partial\upsilon}\right)\left(R_{i+1}^{*d,l}-R^{*l}(t_{i+1})\right)\\
		+  \Delta A_{\theta^{*}}(X^{*d,l-1}(t_{i}),t_{i})\frac{\partial}{\partial\upsilon}\left(R_{i+1}^{*d,l}-R^{*l}(t_{i+1})\right)\\
		+  \Delta\left(\frac{\partial A_{\theta^{*}}}{\partial\upsilon}(X^{*d,l-1}(t_{i}),t_{i})-\frac{\partial A_{\theta^{*}}}{\partial\upsilon}(X^{*l-1}(t_{i}),t_{i})\right)^{T}R^{*l}(t_{i+1})\\
		+  \left(\frac{\partial A_{\theta^{*}}}{\partial x}(X^{*d,l-1}(t_{i}),t_{i})-\frac{\partial A_{\theta^{*}}}{\partial x}(X^{*l-1}(t_{i}),t_{i})\right)\left(\frac{\partial X^{*d,l-1}(t_{i})}{\partial\upsilon}-\frac{\partial X^{*l-1}(t_{i})}{\partial\upsilon}\right)^{T}R^{*l}(t_{i+1})\\
		+  \Delta\left(A_{\theta^{*}}(X^{*d,l-1}(t_{i}),t_{i})-A_{\theta^{*}}(X^{*l-1}(t_{i}),t_{i})\right)^{T}\frac{\partial R^{*l}(t_{i+1})}{\partial\upsilon}\\
		+  \Delta\frac{\partial R_{i+1}^{*d,l}}{\partial\upsilon}\left(A_{\theta^{*}}(X^{*d,l-1}(t_{i}),t_{i})-A_{\theta^{*}}(X^{*l-1}(t_{i}),t_{i})\right)\\
+\Delta R_{i+1}^{*d,l}\left(\frac{\partial A_{\theta^{*}}}{\partial\upsilon}(X^{*d,l-1}(t_{i}),t_{i})-\frac{\partial A_{\theta^{*}}}{\partial\upsilon}(X^{*l-1}(t_{i}),t_{i})\right)\\
		+  \Delta R_{i+1}^{*d,l}\left(\frac{\partial A_{\theta^{*}}}{\partial x}(X^{*d,l-1}(t_{i}),t_{i})-\frac{\partial A_{\theta^{*}}}{\partial x}(X^{*l-1}(t_{i}),t_{i})\right)\left(\frac{\partial X^{*d,l-1}(t_{i})}{\partial\upsilon}-\frac{\partial X^{*l-1}(t_{i})}{\partial\upsilon}\right)\\
		+  \Delta\frac{\partial}{\partial\upsilon}\left(R_{i+1}^{*d,l}-R^{*l}(t_{i+1})\right)A_{\theta^{*}}(X^{*l-1}(t_{i}),t_{i})+\Delta\left(R_{i+1}^{*d,l}-R^{*l}(t_{i+1})\right)\frac{\partial}{\partial\upsilon}A_{\theta^{*}}(X^{*l-1}(t_{i}),t_{i})\\
		+  \triangle\frac{\partial R^{*l}}{\partial\upsilon}(t_{i+1})BU^{-1}B^{T}\left(R^{*l}(t_{i+1})-R_{\theta,i+1}^{d,l}\right)+\triangle R^{*l}(t_{i+1})BU^{-1}B^{T}\frac{\partial}{\partial\upsilon}\left(R_{i+1}^{*d,l}-R^{*l}(t_{i+1})\right)\\
		+  \triangle\frac{\partial}{\partial\upsilon}\left(R_{i+1}^{*d,l}-R^{*l}(t_{i+1})\right)BU^{-1}B^{T}R_{\theta,i+1}^{d,l}+\triangle\left(R^{*l}(t_{i+1})-R_{\theta,i+1}^{d,l}\right)BU^{-1}B^{T}\frac{\partial R_{\theta,i+1}^{d,l}}{\partial\upsilon}\\
+O_{n}(\triangle^{2}).
		\end{array}
		\]
		\[
		\begin{array}{l}
		 \frac{\partial}{\partial\upsilon}\left(h_{i}^{*d,l}-h^{*l}(t_{i})\right)=  \frac{\partial}{\partial\upsilon}\left(h_{i+1}^{*d,l}-h^{*l}(t_{i+1})\right)\\
		+\Delta\frac{\partial A_{\theta^{*}}}{\partial\upsilon}(X^{*d,l-1}(t_{i}),t_{i})^{T}\left(h_{i+1}^{*d,l}-h^{*l}(t_{i+1})\right)\\
		+  \Delta\frac{\partial A_{\theta^{*}}}{\partial x}(X^{*d,l-1}(t_{i}),t_{i})\frac{\partial X^{*d,l-1}(t_{i})}{\partial\upsilon}\left(h_{i+1}^{*d,l}-h^{*l}(t_{i+1})\right)\\
+\Delta A_{\theta^{*}}(X^{*d,l-1}(t_{i}),t_{i})^{T}\frac{\partial}{\partial\upsilon}\left(h_{i+1}^{*d,l}-h^{*l}(t_{i+1})\right)\\
		+  \Delta\left(A_{\theta^{*}}(X^{*d,l-1}(t_{i}),t_{i})-A_{\theta^{*}}(X^{*l-1}(t_{i}),t_{i})\right)^{T}\frac{\partial h^{*l}(t_{i+1})}{\partial\upsilon}\\
		+  \Delta\left(\frac{\partial A_{\theta^{*}}}{\partial\upsilon}(X^{*d,l-1}(t_{i}),t_{i})-\frac{\partial A_{\theta^{*}}}{\partial\upsilon}(X^{*l-1}(t_{i}),t_{i})\right)^{T}h^{*l}(t_{i+1})\\
		+  \Delta\left(\frac{\partial A_{\theta^{*}}}{\partial x}(X^{*d,l-1}(t_{i}),t_{i})-\frac{\partial A_{\theta^{*}}}{\partial x}(X^{*l-1}(t_{i}),t_{i})\right)\left(\frac{\partial X^{*d,l-1}t_{i})}{\partial\upsilon}-\frac{\partial X^{*l-1}(t_{i})}{\partial\upsilon}\right)^{T}h^{*l}(t_{i+1})\\
		+  \triangle\frac{\partial R^{*l}(t)}{\partial\upsilon}BU^{-1}B^{T}\left(h^{*l}(t_{i+1})-h_{i+1}^{*d,l}\right)+\triangle R^{*l}(t)BU^{-1}B^{T}\frac{\partial}{\partial\upsilon}\left(h_{i+1}^{*d,l}-h^{*l}(t_{i+1})\right)\\
		+  \triangle\left(\frac{\partial}{\partial\upsilon}\left(R_{i+1}^{*d,l}-R^{*l}(t_{i+1})\right)BU^{-1}B^{T}h_{i+1}^{*d,l}+\left(R_{i+1}^{*d,l}-R^{*l}(t_{i+1})\right)BU^{-1}B^{T}\frac{\partial h_{i+1}^{*d,l}}{\partial\upsilon}\right)\frac{\partial h_{i+1}^{*d,l}}{\partial\upsilon}\\
+O_{n}(\Delta^{2}).
		\end{array}
		\]
		\[
		\begin{array}{l}
		 \frac{\partial}{\partial\upsilon}\left(X^{*d,l}(t_{i+1})-X^{*l}(t_{i+1})=  \frac{\partial}{\partial\upsilon}\left(X^{*d,l}(t_{i})-X^{*l}(t_{i})\right)\right)\\
		+\Delta\left(\frac{\partial A_{\theta^{*}}}{\partial\upsilon}(X^{*d,l-1}(t_{i}),t_{i})+\frac{\partial A_{\theta^{*}}}{\partial x}(X^{*d,l-1}(t_{i}),t_{i})\frac{\partial X^{*d,l-1}(t_{i})}{\partial\upsilon}\right)\left(X^{*d,l}(t_{i})-X^{*l}(t_{i})\right)\\
		+  \Delta A_{\theta^{*}}(X^{*d,l-1}(t_{i}),t_{i})\frac{\partial}{\partial\upsilon}\left(X^{*d,l-1}(t_{i})-X^{*l-1}(t_{i})\right)\\
+\Delta\left(\frac{\partial A_{\theta^{*}}}{\partial\upsilon}(X^{*d,l-1}(t_{i}),t_{i})-\frac{\partial A_{\theta^{*}}}{\partial\upsilon}(X^{*l-1}(t_{i}),t_{i})\right)X^{*l}(t_{i})\\
		+  \Delta\frac{\partial}{\partial x}\left(A_{\theta^{*}}(X^{*d,l-1}(t_{i}),t_{i})-A_{\theta^{*}}(X^{*l-1}(t_{i}),t_{i})\right)\frac{\partial}{\partial\upsilon}\left(X^{*d,l-1}(t_{i})-X^{*l-1}(t_{i})\right)X^{*l}(t_{i})\\
		+  \Delta\left(A_{\theta^{*}}(X^{*d,l-1}(t_{i}),t_{i})-A_{\theta^{*}}(X^{*l-1}(t_{i}),t_{i})\right)\frac{\partial X^{*l}(t_{i})}{\partial\upsilon}\\
		+  \triangle BU^{-1}B^{T}\left(\frac{\partial R_{i+1}^{*d,l}}{\partial\upsilon}\left(X^{*d,l}(t_{i})-X^{*l}(t_{i})\right)+R_{i+1}^{*d,l}\frac{\partial}{\partial\upsilon}\left(X^{*d,l}(t_{i})-X^{*l}(t_{i})\right)\right)\\
		+  \triangle BU^{-1}B^{T}\left(\frac{\partial}{\partial\upsilon}\left(R_{i+1}^{*d,l}-R^{*l}(t_{i+1})\right)X^{*l}(t_{i})+\left(R_{i+1}^{*d,l}-R^{*l}(t_{i+1})\right)\frac{\partial X^{*l}(t_{i})}{\partial\upsilon}\right)\\
		+  \triangle BU^{-1}B^{T}\frac{\partial}{\partial\theta}\left(h_{i+1}^{*d,l}-h^{*l}(t_{i+1})\right)+O_{n}(\triangle^{2}).
		\end{array}
		\]
	}Application of triangular inequality gives us: {\small{}
		\[
		\begin{array}{l}
		\left\Vert \frac{\partial}{\partial\upsilon}\left(R_{i}^{*d,l}-R^{*l}(t_{i})\right)\right\Vert _{2}  \\ 
\leq  \left(1+O_{n}(\Delta)\right)\left\Vert \frac{\partial}{\partial\upsilon}\left(R_{i+1}^{*d,l}-R^{*l}(t_{i+1})\right)\right\Vert _{2}\\
		 +  O_{n}(\Delta)\left\Vert \frac{\partial A_{\theta^{*}}}{\partial\upsilon}(X^{*d,l-1}(t_{i}),t_{i})-\frac{\partial A_{\theta^{*}}}{\partial\upsilon}(X^{*l-1}(t_{i}),t_{i})\right\Vert _{2}\\
		 +  O_{n}(\Delta)\left\Vert \frac{\partial A_{\theta^{*}}}{\partial x}(X^{*d,l-1}(t_{i}),t_{i})-\frac{\partial A_{\theta^{*}}}{\partial x}(X^{*l-1}(t_{i}),t_{i})\right\Vert _{2}\left\Vert \frac{\partial X^{*d,l-1}(t_{i})}{\partial\upsilon}-\frac{\partial X^{*l-1}(t_{i})}{\partial\upsilon}\right\Vert _{2}\\
		 +  O_{n}(\Delta)\left\Vert A_{\theta^{*}}(X^{*d,l-1}(t_{i}),t_{i})-A_{\theta^{*}}(X^{*l-1}(t_{i}),t_{i})\right\Vert _{2}+O_{n}(\Delta^{2})
		\end{array}
		\]
		\[
		\begin{array}{l}
		\left\Vert \frac{\partial}{\partial\upsilon}\left(h_{i}^{*d,l}-h^{*l}(t_{i})\right)\right\Vert _{2} \\
 \leq  \left(1+O_{n}(\Delta)\right)\left\Vert \frac{\partial}{\partial\upsilon}\left(h_{i+1}^{*d,l}-h^{*l}(t_{i+1})\right)\right\Vert _{2}+O_{n}(\Delta)\left\Vert h_{i+1}^{*d,l}-h^{*l}(t_{i+1})\right\Vert _{2}\\
		 +  O_{n}(\Delta)\left\Vert \frac{\partial A_{\theta^{*}}}{\partial x}(X^{*d,l-1}(t_{i}),t_{i})-\frac{\partial A_{\theta^{*}}}{\partial x}(X^{*l-1}(t_{i}),t_{i})\right\Vert _{2}\left\Vert \frac{\partial X^{*d,l-1}(t_{i})}{\partial\upsilon}-\frac{\partial X^{*l-1}(t_{i})}{\partial\upsilon}\right\Vert \\
		 +  O_{n}(\Delta)\left\Vert \frac{\partial A_{\theta^{*}}}{\partial\upsilon}(X^{*d,l-1}(t_{i}),t_{i})-\frac{\partial A_{\theta^{*}}}{\partial\upsilon}(X^{*l-1}(t_{i}),t_{i})\right\Vert _{2}+O_{n}(\Delta^{2}).
		\end{array}
		\]
		\[
		\begin{array}{l}
		 \left\Vert \frac{\partial}{\partial\upsilon}\left(X^{*d,l}(t_{i+1})-X^{*l}(t_{i+1})\right)\right\Vert _{2}\\
		\leq  \left(1+O_{n}(\Delta)\right)\left\Vert \frac{\partial}{\partial\upsilon}\left(X^{*d,l}(t_{i})-X^{*l}(t_{i})\right)\right\Vert _{2}+O_{n}(\Delta)\left\Vert X^{*d,l}(t_{i})-X^{*l}(t_{i})\right\Vert _{2}\\
		+  O_{n}(\Delta)\left\Vert \frac{\partial A_{\theta^{*}}}{\partial\upsilon}(X^{*d,l-1}(t_{i}),t_{i})-\frac{\partial A_{\theta^{*}}}{\partial\upsilon}(X^{*l-1}(t_{i}),t_{i})\right\Vert _{2}\\
		+  O_{n}(\Delta)\left\Vert \frac{\partial A_{\theta^{*}}}{\partial x}(X^{*d,l-1}(t_{i}),t_{i})-\frac{\partial A_{\theta^{*}}}{\partial x}(X^{*l-1}(t_{i}),t_{i})\right\Vert _{2}\left\Vert \frac{\partial X^{*d,l-1}(t_{i})}{\partial\upsilon}-\frac{\partial X^{*l-1}(t_{i})}{\partial\upsilon}\right\Vert \\
		+  O_{n}(\Delta)\left\Vert A_{\theta^{*}}(X^{*d,l-1}(t_{i}),st_{i})-A_{\theta^{*}}(X^{*l-1}(t_{i}),t_{i})\right\Vert _{2}\\
		+  O_{n}(\Delta)\left\Vert \frac{\partial R_{i+1}^{*d,l}}{\partial\upsilon}\left(X^{*d,l}(t_{i})-X^{*l}(t_{i})\right)\right\Vert _{2}+O_{n}(\Delta)\left\Vert R_{i+1}^{*d,l}-R^{*l}(t_{i+1})\right\Vert _{2}\\
		+  O_{n}(\Delta)\left\Vert \frac{\partial}{\partial\upsilon}\left(h_{i+1}^{*d,l}-h^{*l}(t_{i+1})\right)\right\Vert _{2}+O_{n}(\triangle^{2}).
		\end{array}
		\]
	}Again from these inequalities, we can prove by induction $\left\Vert \frac{\partial}{\partial\upsilon}\left(R_{i}^{*d,l}-R^{*l}(t_{i})\right)\right\Vert _{2}=O_{n}(\Delta)$,
	$\left\Vert \frac{\partial}{\partial\upsilon}\left(h_{i}^{*d,l}-h^{*l}(t_{i})\right)\right\Vert _{2}=O_{n}(\Delta)$
	and $\left\Vert \frac{\partial}{\partial\upsilon}\left(X^{*d,l}(t_{i+1})-X^{*l}(t_{i+1})\right)\right\Vert _{2}=O_{n}(\Delta)$.
	Since 
\[
\begin{array}{l}
\nabla_{\theta}S^{l}(\upsilon^{*})=\left(x_{0}^{*}\right)^{T}\frac{\partial R^{*l}(0)}{\partial\theta}x_{0}^{*}+2\left(x_{0}^{*}\right)^{T}\frac{\partial h^{*l}(0)}{\partial\theta}-2\int_{0}^{T}h^{*l}(t)^{T}BU^{-1}B^{T}\frac{\partial h^{*l}(t)}{\partial\theta}dt
\end{array}
\]
	and
	\[
	\begin{array}{l}
	\nabla_{\theta}S_{n}^{l}(Y^{d*};\upsilon^{*})=\left(x_{0}^{*}\right)^{T}\frac{\partial R_{0}^{*d,l}}{\partial\theta}x_{0}^{*}+2\left(x_{0}^{*}\right)^{T}\frac{\partial h_{0}^{*d,l}}{\partial\theta}\\
	-2\triangle\sum_{i=0}^{n-1}\left(h_{i+1}^{*d,l}\right)^{T}BG(R_{i+1}^{*d,l})B^{T}\frac{\partial h_{i+1}^{*d,l}}{\partial\theta}-\triangle\sum_{i=0}^{n-1}\left(h_{i+1}^{*d,l}\right)^{T}B\frac{\partial G}{\partial\theta}(R_{i+1}^{*d,l})B^{T}h_{i+1}^{*d,l}\\
	=\left(x_{0}^{*}\right)^{T}\frac{\partial R_{0}^{*,d,l}}{\partial\theta}x_{0}^{*}+2\left(x_{0}^{*}\right)^{T}\frac{\partial h_{0}^{*,d,l}}{\partial\theta}-2\triangle\sum_{i=0}^{n-1}\left(h_{i+1}^{*d,l}\right)^{T}BU^{-1}B^{T}\frac{\partial h_{i+1}^{*d,l}}{\partial\theta}+O_{n}(\triangle)
	\end{array}
	\]
	we derive:
	\[
	\begin{array}{l}
	\nabla_{\theta}S_{n}^{l}(Y^{d*};\upsilon^{*})-\nabla_{\theta}S^{l}(\upsilon^{*})\\
	=\left(x_{0}^{*}\right)^{T}\left(\frac{\partial R_{0}^{*d,l}}{\partial\theta}-\frac{\partial R^{*l}(0)}{\partial\theta}\right)x_{0}^{*}+2\left(x_{0}^{*}\right)^{T}\left(\frac{\partial h_{0}^{*d,l}}{\partial\theta}-\frac{\partial h^{*l}(0)}{\partial\theta}\right)\\
	-2\sum_{i=0}^{n-1}\left(\int_{t_{i}}^{t_{i+1}}h^{*l}(t)^{T}BU^{-1}B^{T}\frac{\partial h^{*l}(t)}{\partial\theta}dt-\triangle\left(h_{i+1}^{*d,l}\right)^{T}BU^{-1}B^{T}\frac{\partial h_{i+1}^{*d,l}}{\partial\theta}\right)+O_{n}(\triangle)
	\end{array}
	\]
	and because we can approximate uniformly the terms in the last sum
	by: 
	\[\small{}
	\begin{array}{l}
	\int_{t_{i}}^{t_{i+1}}h^{*l}(t)^{T}BU^{-1}B^{T}\frac{\partial h^{*l}(t)}{\partial\theta}dt-\triangle\left(h_{i+1}^{*d,l}\right)^{T}BU^{-1}B^{T}\frac{\partial h_{i+1}^{*d,l}}{\partial\theta}\\
	=\triangle h^{*l}(t_{i+1})^{T}BU^{-1}B^{T}\frac{\partial h^{*l}(t_{i+1})}{\partial\theta}+O_{n}(\triangle^{2})-\triangle\left(h_{i+1}^{*d,l}\right)^{T}BU^{-1}B^{T}\frac{\partial h_{i+1}^{*d,l}}{\partial\theta}\\
	=\triangle\left(h^{*l}(t_{i+1})^{T}BU^{-1}B^{T}\frac{\partial h^{*l}(t_{i+1})}{\partial\theta}-\left(h_{i+1}^{*d,l}\right)^{T}BU^{-1}B^{T}\frac{\partial h_{i+1}^{*d,l}}{\partial\theta}\right)+O_{n}(\triangle^{2})\\
	=\triangle\left((h_{i+1}^{*d,l}+O_{n}(\triangle))^{T}BU^{-1}B^{T}(\frac{\partial h_{i+1}^{*,d,l}}{\partial\theta}+O_{n}(\triangle))-\left(h_{i+1}^{*d,l}\right)^{T}BU^{-1}B^{T}\frac{\partial h_{i+1}^{*d,l}}{\partial\theta}\right)+O_{n}(\triangle^{2})\\
	=O_{n}(\triangle^{2})
	\end{array}
	\]
	we conclude $\nabla_{\theta}S_{n}^{l}(Y^{d*};\upsilon^{*})-\nabla_{\theta}S^{l}(\upsilon^{*})=O_{n}(\triangle)$.
	Regarding $\nabla_{x_{0}}S^{l}(\upsilon^{*})$ and $\nabla_{x_{0}}S_{n}^{l}(Y^{d*};\upsilon^{*})$,
	we have: 
	\[
	\begin{array}{l}
	\nabla_{x_{0}}S^{l}(\upsilon^{*})\\
=2R^{*l}(0)x_{0}^{*}+\left(x_{0}^{*}\right)^{T}\frac{\partial R^{*l}(0)}{\partial x_{0}}x_{0}^{*}+2h^{*l}(0)+2\left(x_{0}^{*}\right)^{T}\frac{\partial h^{*,}(0)}{\partial x_{0}}\\
-2\int_{0}^{T}h^{*l}(t)^{T}BU^{-1}B^{T}\frac{\partial h^{*l}(t)}{\partial x_{0}}dt
\end{array}
	\]
	\[
	\begin{array}{l}
	\nabla_{x_{0}}S_{n}^{l}(Y^{d*};\upsilon^{*}) \\
 =  2R^{*d,l}(0)x_{0}^{*}+\left(x_{0}^{*}\right)^{T}\frac{\partial R^{*d,l}(0)}{\partial x_{0}}x_{0}^{*}+2h^{*d,l}(0)+2\left(x_{0}^{*}\right)^{T}\frac{\partial h^{*d,l}(0)}{\partial x_{0}}\\
	 -  2\triangle\sum_{i=0}^{n-1}\left(h_{i+1}^{*d,l}\right)^{T}BG(R_{i+1}^{*d,l})B^{T}\frac{\partial h_{i+1}^{*d,l}}{\partial x_{0}}-\triangle\sum_{i=0}^{n-1}\left(h_{i+1}^{*d,l}\right)^{T}B\frac{\partial G}{\partial x_{0}}(R_{i+1}^{*d,l})B^{T}h_{i+1}^{*d,l}\\
	 =  2R^{*d,l}(0)x_{0}^{*}+\left(x_{0}^{*}\right)^{T}\frac{\partial R^{*d,l}(0)}{\partial x_{0}}x_{0}^{*}+2h^{*d,l}(0)+2\left(x_{0}^{*}\right)^{T}\frac{\partial h^{*d,l}(0)}{\partial x_{0}}\\
	 -  2\triangle\sum_{i=0}^{n-1}\left(h_{i+1}^{*d,l}\right)^{T}BG(R_{i+1}^{*d,l})B^{T}\frac{\partial h_{i+1}^{*d,l}}{\partial x_{0}}+O_{n}(\triangle)
	\end{array}
	\]
	so their difference is given by: 
	\[
	\begin{array}{l}
	\nabla_{x_{0}}S_{n}^{l}(Y^{d*};\upsilon^{*})-\nabla_{x_{0}}S^{l}(\upsilon^{*})\\
	=2\left(R^{*d,l}(0)-R^{*l}(0)\right)x_{0}^{*}+2\left(x_{0}^{*}\right)^{T}\left(\frac{\partial R^{*d,l}(0)}{\partial x_{0}}-\frac{\partial R^{*l}(0)}{\partial x_{0}}\right)x_{0}^{*}+2\left(x_{0}^{*}\right)^{T}\left(\frac{\partial h^{*d,l}(0)}{\partial x_{0}}-\frac{\partial h^{*l}(0)}{\partial x_{0}}\right)\\
	-2\sum_{i=0}^{n-1}\left(\int_{t_{i}}^{t_{i+1}}h^{*l}(t)^{T}BU^{-1}B^{T}\frac{\partial h^{*l}(t)}{\partial x_{0}}dt-\triangle\left(h_{i+1}^{*d,l}\right)^{T}BU^{-1}B^{T}\frac{\partial h_{i+1}^{*d,l}}{\partial\partial x_{0}}\right)+O_{n}(\triangle)
	\end{array}
	\]
	and we derive from this $\nabla_{x_{0}}S_{n}^{l}(Y^{d*};\upsilon^{*})-\nabla_{x_{0}}S^{l}(\upsilon^{*})=O_{n}(\triangle)$,
	hence the conclusion for $\nabla_{\upsilon}S_{n}^{l}(Y^{d*};\upsilon^{*})-\nabla_{\upsilon}S^{l}(\upsilon^{*})$. 
\end{proof}
\begin{lem}
	\label{lem:Continuous_S_Gradient_Hessian_convergence}Under conditions
	C1 to C7, we have $\frac{\partial R^{*l}}{\partial\upsilon}=\frac{\partial R^{*\infty}}{\partial\upsilon}+o_{l}(1)$,
	$\frac{\partial h^{*l}}{\partial\upsilon}=\frac{\partial h^{*\infty}}{\partial\upsilon}+o_{l}(1)$,
	$\frac{\partial^{2}R^{*l}}{\partial^{2}\upsilon}=\frac{\partial^{2}R^{*\infty}}{\partial^{2}\upsilon}+o_{l}(1),\,\frac{\partial^{2}h^{*l}}{\partial^{2}\upsilon}=\frac{\partial^{2}h^{*\infty}}{\partial^{2}\upsilon}+o_{l}(1)$
	and
	\[
	\left\{ \begin{array}{l}
	\nabla_{\upsilon}S^{l}(\upsilon^{*})=\nabla_{\upsilon}S^{\infty}(\upsilon^{*})+o_{l}(1)\\
	\frac{\partial^{2}S^{l}(\upsilon^{*})}{\partial^{2}\upsilon}=\frac{\partial^{2}S^{\infty}(\upsilon^{*})}{\partial^{2}\upsilon}+o_{l}(1).
	\end{array}\right.
	\]
\end{lem}
\begin{proof}
	As in lemma \ref{lem:Riccati_optimal_traj_uniform_convergence}, we
	derive the differences $X^{*l}-X^{*\infty}$, $R^{*l}-R^{*\infty}$
	and $h^{*l}-h^{*\infty}$ are ruled by the equations:{\small{}
		\[
		\begin{array}{lll}
		\frac{d}{dt}\left(X^{*l}(t)-X^{*\infty}(t)\right) & = & A_{\theta^{*}}(X^{*l-1}(t),t)X^{*l}(t)+BU^{T}B^{T}(R^{*l}(t)X^{*l}(t)+h^{*l}(t))\\
		& - & A_{\theta^{*}}(X^{*\infty}(t),t)X^{*\infty}(t)-BU^{T}B^{T}(R^{*\infty}(t)X^{*\infty}(t)+h^{*\infty}(t))\\
		& = & \left(A_{\theta^{*}}(X^{*l-1}(t),t)+BU^{T}B^{T}R^{*l}(t)\right)\left(X^{*l}(t)-X^{*\infty}(t)\right)\\
		& + & \left(A_{\theta^{*}}(X^{*l-1}(t),t)-A_{\theta^{*}}(X^{*\infty}(t),t)\right)X^{*\infty}(t)\\
		& + & BU^{T}B^{T}\left(h^{*l}(t)-h^{*\infty}(t)+\left(R^{*l}(t)-R^{*\infty}(t)\right)X^{*\infty}(t)\right)\\
		\frac{d}{dt}\left(R^{*l}(t)-R^{*\infty}(t)\right) & = & \left(A_{\theta^{*}}(X^{*\infty}(t),t)-R^{*l}(t)BU^{-1}B^{T}\right)^{T}\left(R^{*l}(t)-R^{*\infty}(t)\right)\\
		& + & R^{*\infty}(t)\left(A_{\theta^{*}}(X^{*\infty}(t),t)-A_{\theta^{*}}(X^{*l-1}(t),t)\right)\\
		& + & \left(A_{\theta^{*}}(X^{*\infty}(t),t)-A_{\theta^{*}}(X^{*l-1}(t),t)\right)^{T}R^{*l}(t)\\
		& + & \left(R^{*l}(t)-R^{*\infty}(t)\right)\left(BU^{-1}B^{T}R^{*\infty}(t)-A_{\theta^{*}}(X^{*l-1}(t),t)\right)\\
		\frac{d}{dt}\left(h^{*l}(t)-h^{*\infty}(t)\right) & = & \left(A_{\theta^{*}}(X^{*\infty}(t),t)^{T}-R^{*l}(t)BU^{-1}B^{T}\right)\left(h^{*\infty}(t)-h^{*l}(t)\right)\\
		& + & \left(A_{\theta^{*}}(X^{*\infty}(t),t)-A_{\theta^{*}}(X^{*l-1}(t),t)\right)^{T}h^{*l}(t)\\
		& + & \left(R^{*l}(t)-R^{*\infty}(t)\right)BU^{-1}B^{T}h^{*\infty}(t).
		\end{array}
		\]
	}Differentiate the last two equations gives us: {\small{}
		\[
		\begin{array}{l}
		\frac{d}{dt}\frac{\partial}{\partial\upsilon}\left(R^{*l}(t)-R^{*\infty}(t)\right) \\
=  \left(\frac{\partial}{\partial\upsilon}\left(A_{\theta^{*}}(X^{*\infty}(t),t)-R^{*l}(t)BU^{-1}B^{T}\right)\right)^{T}\left(R^{*l}(t)-R^{*\infty}(t)\right)\\
		+  \left(A_{\theta^{*}}(X^{*\infty}(t),t)-R^{*l}(t)BU^{-1}B^{T}\right)^{T}\frac{\partial}{\partial\upsilon}\left(R^{*l}(t)-R^{*\infty}(t)\right)\\
		 +  \frac{\partial R^{*\infty}(t)}{\partial\upsilon}\left(A_{\theta^{*}}(X^{*\infty}(t),t)-A_{\theta^{*}}(X^{*l-1}(t),t)\right)\\
		 +  R^{*\infty}(t)\left(\frac{\partial A_{\theta^{*}}}{\partial\upsilon}(X^{*\infty}(t),t)-\frac{\partial A_{\theta^{*}}}{\partial\upsilon}(X^{*l-1}(t),t)\right)\\
		 +  R^{*\infty}(t)\left(\frac{\partial A_{\theta^{*}}}{\partial x}(X^{*\infty}(t),t)-\frac{\partial A_{\theta^{*}}}{\partial x}(X^{*l-1}(t),t)\right)\frac{\partial}{\partial\upsilon}\left(X^{*\infty}(t)-X^{*l-1}(t)\right)\\
		 +  \left(\frac{\partial A_{\theta^{*}}}{\partial\upsilon}(X^{*\infty}(t),t)-\frac{\partial A_{\theta^{*}}}{\partial\upsilon}(X^{*l-1}(t),t))\right)^{T}R^{*l}(t)\\
		 +  \frac{\partial}{\partial\upsilon}\left(X^{*\infty}(t)-X^{*l-1}(t)\right)^{T}\left(\frac{\partial A_{\theta^{*}}}{\partial x}(X^{*\infty}(t),t)-\frac{\partial A_{\theta^{*}}}{\partial x}(X^{*l-1}(t),t))\right)^{T}R^{*l}(t)\\
		 +  \left(\frac{\partial A_{\theta^{*}}}{\partial\upsilon}(X^{*\infty}(t),t)-\frac{\partial A_{\theta^{*}}}{\partial\upsilon}(X^{*l-1}(t),t))\right)^{T}\frac{\partial R^{*l}(t)}{\partial\upsilon}\\
		 +  \frac{\partial}{\partial\upsilon}\left(R^{*l}(t)-R^{*\infty}(t)\right)\left(BU^{-1}B^{T}R^{*\infty}(t)-A_{\theta^{*}}(X^{*l-1}(t),t)\right)\\
		 +  \left(R^{*l}(t)-R^{*\infty}(t)\right)\frac{\partial}{\partial\upsilon}\left(BU^{-1}B^{T}R^{*\infty}(t)-A_{\theta^{*}}(X^{*l-1}(t),t)\right)\\
\end{array}
		\]
\[
		\begin{array}{l}
		\frac{d}{dt}\frac{\partial}{\partial\upsilon}\left(h^{*l}(t)-h^{*\infty}(t)\right) \\
 = \frac{\partial}{\partial\upsilon}\left(A_{\theta^{*}}(X^{*\infty}(t),t)^{T}-R^{*l}(t)BU^{-1}B^{T}\right)\left(h^{*\infty}(t)-h^{*l}(t)\right)\\
		 +  \left(A_{\theta^{*}}(X^{*\infty}(t),t)^{T}-R^{*l}(t)BU^{-1}B^{T}\right)\frac{\partial}{\partial\theta}\left(h^{*\infty}(t)-h^{*l}(t)\right)\\
		+  \left(\frac{\partial A_{\theta^{*}}}{\partial\upsilon}(X^{*\infty}(t),t)-\frac{\partial A_{\theta^{*}}}{\partial\upsilon}(X^{*l-1}(t),t))\right)^{T}h^{*l}(t)\\
		 +  \frac{\partial}{\partial\upsilon}\left(X^{*\infty}(t)-X^{*l-1}(t)\right)^{T}\left(\frac{\partial A_{\theta^{*}}}{\partial x}(X^{*\infty}(t),t)-\frac{\partial A_{\theta^{*}}}{\partial x}(X^{*l-1}(t),t)\right)^{T}h^{*l}(t)\\
		 +  \left(A_{\theta^{*}}(X^{*\infty}(t),t)-A_{\theta^{*}}(X^{*l-1}(t),t)\right)^{T}\frac{\partial h^{*l}(t)}{\partial\upsilon}\\
		 +  \frac{\partial}{\partial\upsilon}\left(R^{*l}(t)-R^{*\infty}(t)\right)BU^{-1}B^{T}h^{*\infty}(t)\\
		+  \left(R^{*l}(t)-R^{*\infty}(t)\right)BU^{-1}B^{T}\frac{\partial}{\partial\upsilon}h^{*\infty}(t).
		\end{array}
		\]
	}Taking the norm and using triangular inequality gives us: {\small{}
		\[
		\begin{array}{l}
		 \frac{d}{dt}\left\Vert \frac{\partial}{\partial\upsilon}\left(R^{*l}(t)-R^{*\infty}(t)\right)\right\Vert _{2}\\
		\leq  \left\Vert \frac{\partial}{\partial\upsilon}\left(A_{\theta^{*}}(X^{*\infty}(t),t)-R^{*l}(t)BU^{-1}B^{T}\right)\right\Vert _{2}\left\Vert R^{*l}(t)-R^{*\infty}(t)\right\Vert _{2}\\
		+  \left\Vert \frac{\partial}{\partial\upsilon}\left(BU^{-1}B^{T}R^{*\infty}(t)-A_{\theta^{*}}(X^{*l-1}(t),t)\right)\right\Vert _{2}\left\Vert R^{*l}(t)-R^{*\infty}(t)\right\Vert _{2}\\
		+  \left\Vert A_{\theta^{*}}(X^{*\infty}(t),t)-R^{*l}(t)BU^{-1}B^{T}\right\Vert _{2}\left\Vert \frac{\partial}{\partial\upsilon}\left(R^{*l}(t)-R^{*\infty}(t)\right)\right\Vert _{2}\\

		+  \left\Vert BU^{-1}B^{T}R^{*\infty}(t)-A_{\theta^{*}}(X^{*l-1}(t),t)\right\Vert _{2}\left\Vert \frac{\partial}{\partial\upsilon}\left(R^{*l}(t)-R^{*\infty}(t)\right)\right\Vert _{2}\\

		+  \left\Vert \frac{\partial R^{*\infty}(t)}{\partial\upsilon}\right\Vert _{2}\left\Vert A_{\theta^{*}}(X^{*\infty}(t),t)-A_{\theta^{*}}(X^{*l-1}(t),t)\right\Vert _{2}\\

		+  \left\Vert R^{*\infty}(t)\right\Vert _{2}\left(1+\left\Vert \frac{\partial}{\partial\upsilon}\left(X^{*\infty}(t)-X^{*l-1}(t)\right)\right\Vert _{2}\right)\left\Vert \frac{\partial A_{\theta^{*}}}{\partial\upsilon}(X^{*\infty}(t),t)-\frac{\partial A_{\theta^{*}}}{\partial\upsilon}(X^{*l-1}(t),t)\right\Vert _{2}\\

		+ \left\Vert R^{*l}(t)\right\Vert _{2}\left(1+\left\Vert \frac{\partial}{\partial\upsilon}\left(X^{*\infty}(t)-X^{*l-1}(t)\right)\right\Vert _{2}\right)\left\Vert \frac{\partial A_{\theta^{*}}}{\partial\upsilon}(X^{*\infty}(t),t)-\frac{\partial A_{\theta^{*}}}{\partial\upsilon}(X^{*l-1}(t),t)\right\Vert _{2}\\

		+ \left\Vert \frac{\partial}{\partial\upsilon}\left(X^{*\infty}(t)-X^{*l-1}(t)\right)\right\Vert \left\Vert R^{*l}(t)\right\Vert _{2}\left\Vert \frac{\partial A_{\theta^{*}}}{\partial x}(X^{*\infty}(t),t)-\frac{\partial A_{\theta^{*}}}{\partial x}(X^{*l-1}(t),t)\right\Vert _{2}\\

+  \left\Vert \frac{\partial R^{*\infty}(t)}{\partial\upsilon}\right\Vert _{2}\left\Vert \frac{\partial A_{\theta^{*}}}{\partial x}(X^{*\infty}(t),t)-\frac{\partial A_{\theta^{*}}}{\partial x}(X^{*l-1}(t),t)\right\Vert _{2}

		\end{array}
		\]
		\[
		\begin{array}{l}
		 \frac{d}{dt}\left\Vert \frac{\partial}{\partial\upsilon}\left(h^{*l}(t)-h^{*\infty}(t)\right)\right\Vert _{2}\\
		\leq  \left\Vert \frac{\partial}{\partial\upsilon}\left(A_{\theta^{*}}(X^{*\infty}(t),t)^{T}-R^{*l}(t)BU^{-1}B^{T}\right)\right\Vert _{2}\left\Vert h^{*\infty}(t)-h^{*l}(t)\right\Vert _{2}\\
		+  \left\Vert A_{\theta^{*}}(X^{*\infty}(t),t)^{T}-R^{*l}(t)BU^{-1}B^{T}\right\Vert _{2}\left\Vert \frac{\partial}{\partial\upsilon}\left(h^{*\infty}(t)-h^{*l}(t)\right)\right\Vert _{2}\\
		+  \left\Vert h^{*l}(t)\right\Vert _{2}\left\Vert \frac{\partial A_{\theta^{*}}}{\partial\upsilon}(X^{*\infty}(t),t)-\frac{\partial A_{\theta^{*}}}{\partial\upsilon}(X^{*l-1}(t),t))\right\Vert _{2}\\
+\left\Vert \frac{\partial h^{*l}(t)}{\partial\upsilon}\right\Vert \left\Vert A_{\theta^{*}}(X^{*\infty}(t),t)-A_{\theta^{*}}(X^{*l-1}(t),t)\right\Vert _{2}\\
		+  \left\Vert \frac{\partial}{\partial\upsilon}\left(X^{*\infty}(t)-X^{*l-1}(t)\right)\right\Vert _{2}\left\Vert h^{*l}(t)\right\Vert _{2}\left\Vert \frac{\partial A_{\theta^{*}}}{\partial x}(X^{*\infty}(t),t)-\frac{\partial A_{\theta^{*}}}{\partial x}(X^{*l-1}(t),t)\right\Vert _{2}\\
		+  \left\Vert BU^{-1}B^{T}h^{*\infty}(t)\right\Vert _{2}\left\Vert \frac{\partial}{\partial\upsilon}\left(R^{*l}(t)-R^{*\infty}(t)\right)\right\Vert _{2}+\left\Vert BU^{-1}B^{T}\frac{\partial}{\partial\upsilon}h^{*\infty}(t)\right\Vert _{2}\left\Vert R^{*l}(t)-R^{*\infty}(t)\right\Vert _{2}.
		\end{array}
		\]
	}By using lemma \ref{lem:Riccati_optimal_traj_uniform_convergence}
	and proposition \ref{prop:E_bounded_probability}, we can simplify
	these inequalities:{\small{}
		\[
		\begin{array}{lll}
		\frac{d}{dt}\left\Vert \frac{\partial}{\partial\upsilon}\left(R^{*l}(t)-R^{*\infty}(t)\right)\right\Vert _{2} & \leq & o_{l}(1)+O_{l}(1)\left\Vert \frac{\partial}{\partial\upsilon}\left(R^{*l}(t)-R^{*\infty}(t)\right)\right\Vert _{2}\\
		\frac{d}{dt}\left\Vert \frac{\partial}{\partial\upsilon}\left(h^{*l}(t)-h^{*\infty}(t)\right)\right\Vert _{2} & \leq & o_{l}(1)+O_{l}(1)\left\Vert \frac{\partial}{\partial\upsilon}\left(h^{*\infty}(t)-h^{*l}(t)\right)\right\Vert _{2}\\
&+&O_{l}(1)\left\Vert \frac{\partial}{\partial\upsilon}\left(R^{*l}(t)-R^{*\infty}(t)\right)\right\Vert _{2}.
		\end{array}
		\]
	}Since $\frac{\partial}{\partial\upsilon}R^{*l}(0)=\frac{\partial}{\partial\upsilon}R^{*\infty}(0)=0$
	and $\frac{\partial}{\partial\upsilon}h^{*l}(0)=\frac{\partial}{\partial\upsilon}h^{*\infty}(0)=0$,
	the continuous version of the Gronwall lemma successively gives us
	$\left\Vert \frac{\partial}{\partial\theta}\left(R^{*l}-R^{*\infty}\right)\right\Vert _{L^{2}}=o_{l}(1)$
	then $\left\Vert \frac{\partial}{\partial\theta}\left(h^{*l}-h^{*\infty}\right)\right\Vert _{L^{2}}=o_{l}(1)$.
	In theorem \ref{thm:Consistency_theorem}, we already derived the
	expression of $S^{\infty}(\upsilon)-S^{l}(\upsilon)$, from which
	we obtain 
	\[
	\begin{array}{l}
	\nabla_{\theta}S^{\infty}(\upsilon^{*})-\nabla_{\theta}S^{l}(\upsilon^{*})=\left(x_{0}^{*}\right)^{T}\frac{\partial}{\partial\theta}\left(R^{*\infty}(0)-R^{*l}(0)\right)x_{0}^{*}+2\left(x_{0}^{*}\right)^{T}\frac{\partial}{\partial\theta}\left(h^{*\infty}(0)-h^{*l}(0)\right)\\
	+\int_{0}^{T}\left(\frac{\partial h^{*l}}{\partial\theta}(t)+\frac{\partial h^{*\infty}}{\partial\theta}(t)\right)BU^{-1}B^{T}\left(h^{*l}(t)-h^{*\infty}(t)\right)dt\\
	+\int_{0}^{T}\left(h^{*l}(t)+h^{*\infty}(t)\right)BU^{-1}B^{T}\left(\frac{\partial h^{*l}}{\partial\theta}(t)-\frac{\partial h^{*\infty}}{\partial\theta}(t)\right)dt.
	\end{array}
	\]
	{\small{}and 
		\[
		\begin{array}{l}
		\nabla_{x_{0}}S^{\infty}(\upsilon^{*})-\nabla_{x_{0}}S^{l}(\upsilon^{*})=2\left(R^{*\infty}(0)-R^{*l}(0)\right)x_{0}^{*}+\left(x_{0}^{*}\right)^{T}\frac{\partial}{\partial x_{0}}\left(R^{*\infty}(0)-R^{*l}(0)\right)x_{0}^{*}\\
		+2\left(h^{*\infty}(0)-h^{*l}(0)\right)+2\left(x_{0}^{*}\right)^{T}\frac{\partial}{\partial x_{0}}\left(h^{*\infty}(0)-h^{*l}(0)\right)\\
		+\int_{0}^{T}\left(\frac{\partial h^{*l}}{\partial x_{0}}(t)+\frac{\partial h^{*\infty}}{\partial x_{0}}(t)\right)BU^{-1}B^{T}\left(h^{*l}(t)-h^{*\infty}(t)\right)dt\\
		+\int_{0}^{T}\left(h^{*l}(t)+h^{*\infty}(t)\right)BU^{-1}B^{T}\left(\frac{\partial h^{*l}}{\partial x_{0}}(t)-\frac{\partial h^{*\infty}}{\partial x_{0}}(t)\right)dt
		\end{array}
		\]
	}and since we know $\frac{\partial R^{*l}}{\partial\upsilon}=\frac{\partial R^{*\infty}}{\partial\upsilon}+o_{l}(1)$,
	$\frac{\partial h^{*l}}{\partial\upsilon}=\frac{\partial h^{*\infty}}{\partial\upsilon}+o_{l}(1)$
	, we have $\nabla_{\upsilon}S^{\infty}(\upsilon^{*})-\nabla_{\upsilon}S^{l}(\upsilon^{*})=o_{l}(1)$.
	As sensitivity equations, the ODEs ruling the functions $\frac{\partial}{\partial\upsilon}\left(R^{*l}(t)-R^{*\infty}(t)\right)$,
	$\frac{\partial}{\partial\upsilon}\left(h^{*l}(t)-h^{*\infty}(t)\right)$are
	linear. By using C6 and classic existence and regularity results for
	linear ODEs, we know these functions are differentiables w.r.t to
	$\upsilon$ and $\frac{\partial^{2}}{\partial^{2}\upsilon}\left(R^{*l}(t)-R^{*\infty}(t)\right)$,
	$\frac{\partial^{2}}{\partial^{2}\upsilon}\left(h^{*l}(t)-h^{*\infty}(t)\right)$
	are defined on $\left[0,\,T\right]$. From the previous derived expressions,
	it is straightforward to see $\frac{\partial R^{*l}}{\partial\upsilon},\,\frac{\partial h^{*l}}{\partial\upsilon}$
	(resp. $\frac{\partial R^{*\infty}}{\partial\upsilon},\,\frac{\partial h^{*\infty}}{\partial\upsilon}$)
	are ruled by ODEs of the form $\dot{V^{l}}=F^{l}(t,\upsilon^{*})V^{l}+G^{l}(t,\upsilon^{*})$
	(resp. $\dot{V^{\infty}}=F^{\infty}(t,\upsilon^{*})V^{\infty}+G^{\infty}(t,\upsilon^{*})$)
	with $F^{l}$ , $F^{\infty}$, $G^{l}$, $G^{\infty}$ continuous
	w.r.t $t$ and $\left\Vert V^{l}-V^{\infty}\right\Vert _{L^{2}}=o_{l}(1)$,
	$\left\Vert F^{l}(t,\upsilon^{*})-F^{\infty}(t,\upsilon^{*})\right\Vert _{L^{2}}=o_{l}(1)$,
	$\left\Vert G^{l}(t,\upsilon^{*})-G^{\infty}(t,\upsilon^{*})\right\Vert _{L^{2}}=o_{l}(1)$,
	$\left\Vert \frac{\partial F^{l}}{\partial\upsilon}(t,\upsilon^{*})-\frac{\partial F^{\infty}}{\partial\upsilon}(t,\upsilon^{*})\right\Vert _{2}=o_{l}(1)$,
	$\left\Vert \frac{\partial G^{l}}{\partial\upsilon}(t,\upsilon^{*})-\frac{\partial G^{\infty}}{\partial\upsilon}(t,\upsilon^{*})\right\Vert _{2}=o_{l}(1)$.
	Here $V$ arbitrarily stands for $\frac{\partial R^{*}}{\partial\upsilon},\,\frac{\partial h^{*}}{\partial\upsilon}$.
	By differentiation, we obtain: {\small{}
		\[
		\begin{array}{l}
		\frac{d}{dt}\frac{\partial}{\partial\upsilon}\left(V^{l}(t)-V^{\infty}(t)\right) \\
 = \frac{\partial F^{l}}{\partial\upsilon}(t,\upsilon^{*})V^{l}(t)+F^{l}(t,\upsilon^{*})\frac{\partial V^{l}(t)}{\partial\upsilon}+\frac{\partial G^{l}}{\partial\upsilon}(t,\upsilon^{*})\\
		 -  \frac{\partial F^{\infty}}{\partial\upsilon}(t,\upsilon^{*})V^{\infty}(t)-F^{\infty}(t,\upsilon^{*})\frac{\partial V^{\infty}(t)}{\partial\upsilon}-\frac{\partial G^{\infty}}{\partial\upsilon}(t,\upsilon^{*})\\
		 =  \left(\frac{\partial F^{l}}{\partial\upsilon}(t,\upsilon^{*})-\frac{\partial F^{\infty}}{\partial\upsilon}(t,\upsilon^{*})\right)V^{l}(t)+\frac{\partial F^{\infty}}{\partial\upsilon}(t,\upsilon^{*})(V^{l}(t)-V^{\infty}(t))\\
		 +  \left(F^{l}(t,\upsilon^{*})-F^{\infty}(t,\upsilon^{*})\right)\frac{\partial V^{l}(t)}{\partial\upsilon}+F^{\infty}(t,\theta^{*})\left(\frac{\partial V^{l}(t)}{\partial\upsilon}-\frac{\partial V^{\infty}(t)}{\partial\upsilon}\right)\\
		+  \frac{\partial G^{l}}{\partial\upsilon}(t,\upsilon^{*})-\frac{\partial G^{\infty}}{\partial\upsilon}(t,\upsilon^{*}).
		\end{array}
		\]
	}By taking the norm and by using triangular inequality, we obtain:
	{\small{}
		\[
		\begin{array}{l}
		\frac{d}{dt}\frac{\partial}{\partial\upsilon}\left\Vert V^{l}(t)-V^{\infty}(t)\right\Vert _{2}  \\ 
\leq \left\Vert \frac{\partial F^{l}}{\partial\upsilon}(t,\upsilon^{*})-\frac{\partial F^{\infty}}{\partial\upsilon}(t,\upsilon^{*})\right\Vert _{2}\left\Vert V^{l}(t)\right\Vert _{2}+\left\Vert \frac{\partial F^{\infty}}{\partial\upsilon}(t,\upsilon^{*})\right\Vert _{2}\left\Vert V^{l}(t)-V^{\infty}(t)\right\Vert _{2}\\
		+  \left\Vert F^{l}(t,\upsilon^{*})-F^{\infty}(t,\upsilon^{*})\right\Vert _{2}\left\Vert \frac{\partial V^{l}}{\partial\upsilon}(t)\right\Vert _{2}+\left\Vert F^{\infty}(t,\upsilon^{*})\right\Vert _{2}\left\Vert \frac{\partial V^{l}}{\partial\upsilon}(t)-\frac{\partial V^{\infty}}{\partial\upsilon}(t)\right\Vert _{2}\\
		 +  \left\Vert \frac{\partial G^{l}}{\partial\upsilon}(t,\upsilon^{*})-\frac{\partial G^{\infty}}{\partial\theta}(t,\upsilon^{*})\right\Vert _{2}\\
		 \leq  o_{l}(1)+O_{l}(1)\left\Vert \frac{\partial V^{l}}{\partial\upsilon}(t)-\frac{\partial V^{\infty}}{\partial\upsilon}(t)\right\Vert _{2}
		\end{array}
		\]
	}and since $\frac{\partial V^{l}(0)}{\partial\upsilon}=\frac{\partial V^{\infty}(0)}{\partial\upsilon}=0$,
	we can conclude by using the continuous Gronwall lemma that $\frac{\partial^{2}R^{*l}}{\partial^{2}\upsilon}=\frac{\partial^{2}R^{*\infty}}{\partial^{2}\upsilon}+o_{l}(1),\,\frac{\partial^{2}h^{*l}}{\partial^{2}\upsilon}=\frac{\partial^{2}h^{*\infty}}{\partial^{2}\upsilon}+o_{l}(1)$.
	By differentiating $\nabla_{\theta}S^{\infty}(\upsilon^{*})-\nabla_{\theta}S^{l}(Y;\upsilon^{*})$
	with respect to $\theta$, we obtain:{\small{} }
	\[
	\begin{array}{l}
	\frac{\partial^{2}S^{\infty}(\upsilon^{*})}{\partial^{2}\theta}-\frac{\partial^{2}S^{l}(\upsilon^{*})}{\partial^{2}\theta} \\
=\left(x_{0}^{*}\right)^{T}\frac{\partial^{2}}{\partial^{2}\theta}\left(R^{*\infty}(0)-R^{*l}(0)\right)x_{0}^{*}+2\left(x_{0}^{*}\right)^{T}\frac{\partial^{2}}{\partial^{2}\theta}\left(h^{*\infty}(0)-h^{*l}(0)\right)\\
	+\int_{0}^{T}\left(\frac{\partial^{2}h^{*l}}{\partial^{2}\theta}(t)+\frac{\partial^{2}h^{*\infty}}{\partial^{2}\theta}(t)\right)BU^{-1}B^{T}\left(h^{*l}(t)-h^{*\infty}(t)\right)dt\\
	+2\int_{0}^{T}\left(\frac{\partial h^{*l}}{\partial\theta}(t)+\frac{\partial h^{*\infty}}{\partial\theta}(t)\right)BU^{-1}B^{T}\left(\frac{\partial h^{*l}}{\partial\theta}(t)-\frac{\partial h^{*\infty}}{\partial\theta}(t)\right)dt\\
	+\int_{0}^{T}\left(h^{*l}(t)+h^{*\infty}(t)\right)BU^{-1}B^{T}\left(\frac{\partial^{2}h^{*l}}{\partial^{2}\theta}(t)-\frac{\partial^{2}h^{*\infty}}{\partial^{2}\theta}(t)\right)dt.
	\end{array}
	\]
	and by differentiating $\nabla_{x_{0}}S^{\infty}(\upsilon^{*})-\nabla_{x_{0}}S^{l}(\upsilon^{*})$
	with respect to $\theta$, we obtain:{\small{} 
		\[
		\begin{array}{l}
		\frac{\partial^{2}S^{\infty}(\upsilon^{*})}{\partial x_{0}\partial\theta}-\frac{\partial^{2}S^{l}(\upsilon^{*})}{\partial x_{0}\partial\theta}=2\frac{\partial}{\partial\theta}\left(R^{*\infty}(0)-R^{*l}(0)\right)x_{0}^{*}+\left(x_{0}^{*}\right)^{T}\frac{\partial^{2}}{\partial x_{0}\partial\theta}\left(R^{*\infty}(0)-R^{*l}(0)\right)x_{0}^{*}\\
		+2\frac{\partial}{\partial\theta}\left(h^{*\infty}(0)-h^{*l}(0)\right)+2\left(x_{0}^{*}\right)^{T}\frac{\partial^{2}}{\partial x_{0}\partial\theta}\left(h^{*\infty}(0)-h^{*l}(0)\right)\\
		+\int_{0}^{T}\left(\frac{\partial^{2}h^{*l}}{\partial x_{0}\partial\theta}(t)+\frac{\partial^{2}h^{*\infty}}{\partial x_{0}\partial\theta}(t)\right)BU^{-1}B^{T}\left(h^{*l}(t)-h^{*\infty}(t)\right)dt\\
		+2\int_{0}^{T}\left(\frac{\partial h^{*l}}{\partial x_{0}}(t)+\frac{\partial h^{*\infty}}{\partial x_{0}}(t)\right)BU^{-1}B^{T}\left(\frac{\partial h^{*l}}{\partial\theta}(t)-\frac{\partial h^{*\infty}}{\partial\theta}(t)\right)dt\\
		+\int_{0}^{T}\left(h^{*l}(t)+h^{*\infty}(t)\right)BU^{-1}B^{T}\left(\frac{\partial^{2}h^{*l}}{\partial x_{0}\partial\theta}(t)-\frac{\partial^{2}h^{*\infty}}{\partial x_{0}\partial\theta}(t\right)dt
		\end{array}
		\]
		and:
		\[
		\begin{array}{l}
		\frac{\partial^{2}S^{\infty}(\upsilon^{*})}{\partial^{2}x_{0}}-\frac{\partial^{2}S^{l}(\upsilon^{*})}{\partial^{2}x_{0}}=2\left(\frac{\partial}{\partial x_{0}}\left(R^{*\infty}(0)-R^{*l}(0)\right)x_{0}^{*}+\left(R^{*\infty}(0)-R^{*l}(0)\right)\right)\\
		+2\frac{\partial}{\partial x_{0}}\left(R^{*\infty}(0)-R^{*l}(0)\right)x_{0}^{*}+\left(x_{0}^{*}\right)^{T}\frac{\partial^{2}}{\partial^{2}x_{0}}\left(R^{*\infty}(0)-R^{*l}(0)\right)x_{0}^{*}\\
		+2\frac{\partial}{\partial x_{0}}\left(h^{*\infty}(0)-h^{*l}(0)\right)+2\left(\frac{\partial}{\partial x_{0}}\left(h^{*\infty}(0)-h^{*l}(0)\right)+\left(x_{0}^{*}\right)^{T}\frac{\partial^{2}}{\partial^{2}x_{0}}\left(h^{*\infty}(0)-h^{*l}(0)\right)\right)\\
		+\int_{0}^{T}\left(\frac{\partial^{2}h^{*l}}{\partial^{2}x_{0}}(t)+\frac{\partial^{2}h^{*\infty}}{\partial^{2}x_{0}}(t)\right)BU^{-1}B^{T}\left(h^{*l}(t)-h^{*\infty}(t)\right)dt\\
		+2\int_{0}^{T}\left(\frac{\partial h^{*l}}{\partial x_{0}}(t)+\frac{\partial h^{*\infty}}{\partial x_{0}}(t)\right)BU^{-1}B^{T}\left(\frac{\partial h^{*l}}{\partial x_{0}}(t)-\frac{\partial h^{*\infty}}{\partial x_{0}}(t)\right)dt\\
		+\int_{0}^{T}\left(h^{*l}(t)+h^{*\infty}(t)\right)BU^{-1}B^{T}\left(\frac{\partial^{2}h^{*l}}{\partial^{2}x_{0}}(t)-\frac{\partial^{2}h^{*\infty}}{\partial^{2}x_{0}}(t)\right)dt.
		\end{array}
		\]
	}from this we can conclude that $\frac{\partial^{2}S^{l}(\upsilon^{*})}{\partial^{2}\upsilon}=\frac{\partial^{2}S^{\infty}(\upsilon^{*})}{\partial^{2}\upsilon}+o_{l}(1)$.
\end{proof}

\section*{Acknowledgements}
Quentin Clairon was supported by EPSRC award EP/M015637/1.

 \bibliographystyle{rss}
\bibliography{biblio_OCA_iterative_approach}